\documentclass[12pt, a4paper]{article}

\usepackage[utf8]{inputenc}
\usepackage[main=english, danish]{babel}

\usepackage{titlesec}
\titleformat{\section}
{\centering\normalfont\LARGE\bfseries}{\thesection}{1em}{}

\usepackage[a4paper, hmargin={55pt, 55pt}, vmargin={2.8cm, 2.8cm} ]{geometry}
\usepackage{eso-pic}
\usepackage{times}
\usepackage{ifthen}
\usepackage{graphicx}
\usepackage{float}

\usepackage{amsmath} 
\usepackage{amssymb} 
\usepackage{amsthm} 
\usepackage{relsize,exscale} 
\usepackage{amsfonts}
\usepackage{mathrsfs}
\usepackage{dsfont}
\usepackage{enumerate}
\usepackage{interval}
\usepackage[shortlabels]{enumitem}
\usepackage{esint}
\usepackage{graphics}
\usepackage{caption}
\usepackage{subcaption}
\usepackage{mathtools}
\usepackage{multirow}
\usepackage{comment}
\usepackage{hhline}
\usepackage{array}

\DeclareFontFamily{U}{mathx}{\hyphenchar\font45}
\DeclareFontShape{U}{mathx}{m}{n}{
      <5> <6> <7> <8> <9> <10>
      <10.95> <12> <14.4> <17.28> <20.74> <24.88>
      mathx10
      }{}
\DeclareSymbolFont{mathx}{U}{mathx}{m}{n}
\DeclareFontSubstitution{U}{mathx}{m}{n}
\DeclareMathAccent{\widecheck}{0}{mathx}{"71}
\DeclareMathAccent{\wideparen}{0}{mathx}{"75}

\usepackage{hyperref}
\hypersetup{
    colorlinks = true,
    citecolor = blue,
    linkcolor = blue,
    urlcolor = blue
}

\usepackage[font=small,labelfont=bf]{caption}

\newtheoremstyle{droit}
{}
{}
{\normalfont}
{}
{\boldface}
{}
{\newline}
{}
\newtheoremstyle{italique}
{}
{}
{\itshape}
{}
{\boldface}
{}
{\newline}
{}

\newtheorem*{thm*}{Theorem}
\newtheorem{thm}{Theorem}[section]
\newtheorem*{prop*}{Proposition}
\newtheorem{prop}[thm]{Proposition}
\newtheorem*{lem*}{Lemma}
\newtheorem{lem}[thm]{Lemma}
\newtheorem*{cor*}{Corollary}
\newtheorem{cor}[thm]{Corollary}

\theoremstyle{definition}

\newtheorem{rema}[thm]{Remark}
\newtheorem{defi}[thm]{Definition}
\newtheorem*{defi*}{Definition}

\newtheorem{claim}{Claim}[thm]

{\begin{enumerate}[\textup{(}i\textup{)}] }
{\end{enumerate}}

\newcommand{\R}{\mathbb{R}}
\newcommand{\N}{\mathbb{N}}
\newcommand{\Z}{\mathbb{Z}}
\newcommand{\C}{\mathbb{C}}

\newcommand{\ts}[1]{\binom{#1}{ 2 }}

\newcommand{\compl}{$\mathlarger{\mathlarger{\mathlarger{\bigoplus}}}$}
\newcommand{\noi}{\noindent}

\newcommand{\pws}[1]{\mathcal{P}\left( #1 \right)}
\newcommand{\dom}{\mathcal{D}}
\newcommand{\paths}{\mathcal{P}}
\newcommand{\face}[1]{\triangle\left(#1\right)}
\newcommand{\cface}[1]{\bigtriangledown\left(#1\right)}
\newcommand{\kface}[2]{\triangle_{#2}\left(#1\right)}
\newcommand{\kcface}[2]{\bigtriangledown_{#2}\left(#1\right)}
\newcommand{\rk}{\operatorname{rk}}
\newcommand{\Rk}{\operatorname{Rk}}
\newcommand{\rkb}{\overline{\rk}}

\newcommand{\Ne}{\mathcal{N}}
\newcommand{\E}{\mathcal{E}}
\newcommand{\V}{\mathcal{V}}
\newcommand{\Ce}{\mathcal{C}}
\newcommand{\dual}[1]{\overline{#1}}
\newcommand{\bual}[1]{\widetilde{#1}}
\newcommand{\bdual}[2]{\bual{#1}^{#2}}
\newcommand{\cdual}[2]{\overline{#1}^{#2}}

\newcommand{\all}{\dual{\emptyset}}
\newcommand{\sphi}{\rho}
\newcommand{\cphi}{\pi}
\newcommand{\red}{\mathrel{\raisebox{1pt}{$\prec$}}}
\newcommand{\der}{\mathrel{\raisebox{1pt}{$\succ$}}}
\newcommand{\col}{\vdash}
\newcommand{\loc}{\dashv}
\newcommand{\cpa}{\leftthreetimes}
\newcommand{\apc}{\rightthreetimes}
\newcommand{\refl}{\mathbin{\top}} 
\newcommand{\ccup}[1]{\stackrel{#1}{\cup}}

\newcommand{\bdiv}[1]{\widecheck{#1}}
\newcommand{\pathto}{\rightsquigarrow}

\newcommand{\kz}{K^{[0]}}
\newcommand{\ko}{K^{[1]}}
\newcommand{\kt}{K^{[2]}}
\newcommand{\kg}{K^{(1)}}
\newcommand{\dg}[1]{\mathcal{G}_{\dual{#1}}}

\newcommand{\ks}{K^{[3]}}
\newcommand{\kc}{K^{(2)}}

\newcommand{\kr}{K^{[r]}}
\newcommand{\kR}{K^{[R]}}
\newcommand{\kb}{\dual{K}}
\newcommand{\kzb}{\dual{K}^{[0]}}
\newcommand{\kob}{\dual{K}^{[1]}}

\newcommand{\ksb}{\dual{K}^{[3]}}

\newcommand{\vb}{\dual{v}}

\newcommand{\xb}{\dual{x}}

\newcommand{\kbb}{\dual{\dual{K}}}

\newcommand{\ce}[1]{\dual{#1}}

\newcommand{\st}{\operatorname{St}}
\newcommand{\lk}{\operatorname{Lk}}
\newcommand{\comp}[2]{{#1}^{#2}}
\newcommand{\conc}{\ast}
\newcommand{\ksec}[2]{\interval{#1}{#2}}
\newcommand{\ori}[1]{#1_{\rightarrow}}
\newcommand{\nsc}{\mathsf{C}}
\newcommand{\cob}{\mathsf{Cob}}
\newcommand{\mlc}{\mathsf{B}}

\newcommand{\Hom}{\mathsf{Hom}}

\newcommand{\Cob}{\mathscr{C}}
\newcommand{\sts}{\mathcal{S}}
\newcommand{\tim}{\mathtt{T}}
\newcommand{\parity}{\mathtt{P}}
\newcommand{\charge}{\mathtt{C}}
\newcommand{\ists}{\sts^{in}}
\newcommand{\osts}{\sts^{out}}
\newcommand{\rev}[1]{{#1}^{-1}}
\newcommand{\seqs}{\mathcal{Q}}
\newcommand{\Dom}{\mathfrak{D}}
\newcommand{\bra}[1]{ \langle #1 | }
\newcommand{\ket}[1]{ | #1 \rangle  }
\newcommand{\braket}[2]{\langle #1 | #2 \rangle}
\newcommand{\ktbr}[2]{\ket{#1}\bra{#2}}

\newcommand{\ins}{\leq_{\text{in}}}
\newcommand{\outs}{\leq_{\text{out}}}
\newcommand{\homc}{\Hom_{\Cob_d}}

\newcommand{\prt}{\mathcal{Z}}

\newcommand{\F}{\mathcal{F}}
\newcommand{\pr}{\mathbf{P}}

\newcommand{\G}{\mathcal{G}}
\newcommand{\M}{\mathcal{M}}

\newcommand{\A}{\mathcal{A}}
\newcommand{\B}{\mathcal{B}}

\newcommand{\id}{\operatorname{id}}

\newcommand{\im}{\mathfrak{Im}}
\newcommand{\abs }[1]{ \left| #1 \right| }
\newcommand{\dans}{\longrightarrow}

\def\ASSIGNMENT{Use \texttt{$\backslash$assignment$\lbrace \ldots \rbrace$}}
\def\AUTHOR{Use \texttt{$\backslash$author$\lbrace \ldots \rbrace$}}
\def\TITLE{Use \texttt{$\backslash$title$\lbrace \ldots \rbrace$}}
\def\SUBTITLE{Use \texttt{$\backslash$subtitle$\lbrace \ldots \rbrace$}}
\def\ADVISOR{Use \texttt{$\backslash$advisor$\lbrace \ldots \rbrace$}}
\def\DATE{Use \texttt{$\backslash$date$\lbrace \ldots \rbrace$}}
\def\FRONTPAGEIMAGE{...}

\renewcommand{\author}[1]{\def\AUTHOR{#1}}
\renewcommand{\title}[1]{\def\TITLE{#1}}
\renewcommand{\date}[1]{\def\DATE{#1}}

\newcommand{\assignment}[1]{\def\ASSIGNMENT{#1}}
\newcommand{\subtitle}[1]{\def\SUBTITLE{#1}}
\newcommand{\advisor}[1]{\def\ADVISOR{#1}}
\newcommand{\frontpageimage}[1]{\def\FRONTPAGEIMAGE{#1}}

\def\KUbold{\fontfamily{ptm}\fontseries{bx}\selectfont}
  \def\KUsemibold{\fontfamily{ptm}\fontseries{sb}\selectfont}

  \def\maketitle{
    \thispagestyle{empty}
    \AddToShipoutPicture*{\put(0, -35){\includegraphics*{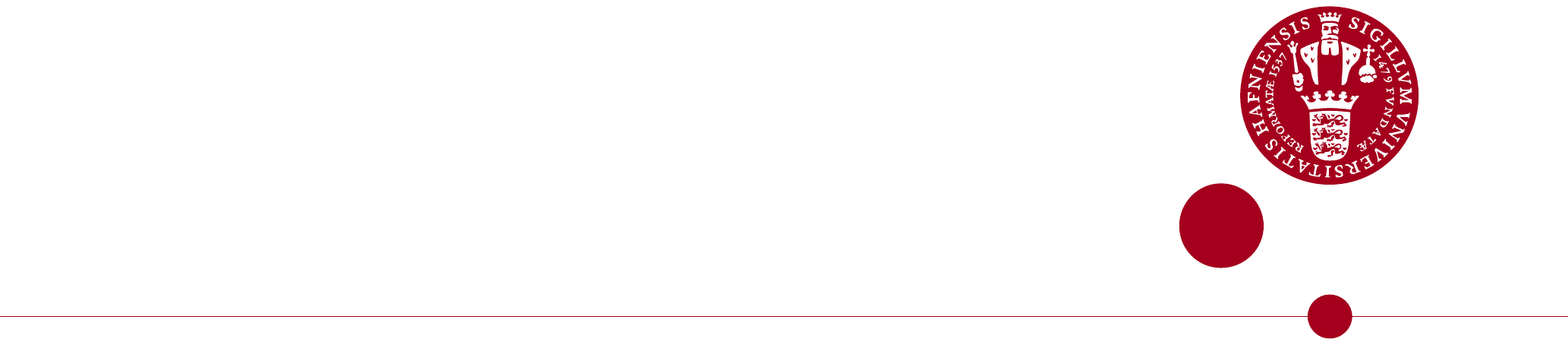}}}
    \AddToShipoutPicture*{\put(5, -10){\includegraphics*{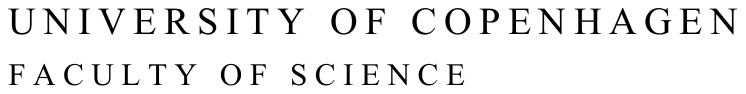}}}

    \ifthenelse{\equal{\FRONTPAGEIMAGE}{...}}{}{
      \AddToShipoutPicture*{
        \AtPageUpperLeft{\raisebox{-185mm}{\hspace{55pt}\includegraphics*[width=\textwidth, height=100mm, keepaspectratio, bb=0 0 1440 810]{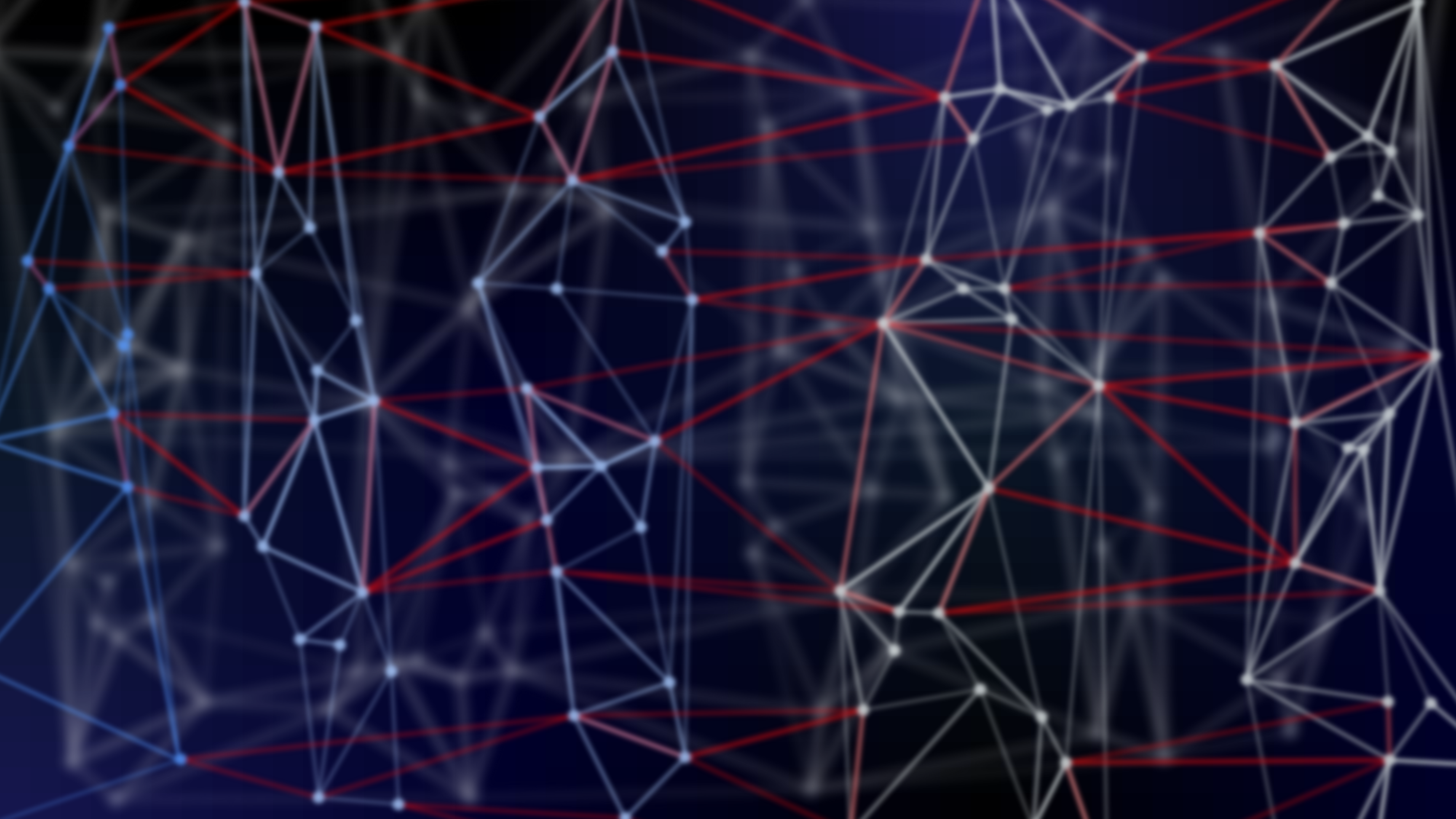}}}
      }
    }

    \AddToShipoutPicture*{\put(55, 270){\fontsize{20 pt}{22 pt} \KUbold \ASSIGNMENT}}
    \AddToShipoutPicture*{\put(55, 242){\fontsize{14 pt}{16 pt} \KUsemibold \AUTHOR}}

    \AddToShipoutPicture*{\put(55, 188){\fontsize{22 pt}{24 pt} \KUsemibold \TITLE}}
    \AddToShipoutPicture*{\put(55, 160){\fontsize{14 pt}{16 pt} \KUsemibold \SUBTITLE}}

    \AddToShipoutPicture*{\put(55, 85){\fontsize{11 pt}{12 pt} \KUsemibold \ADVISOR}}
    \AddToShipoutPicture*{\put(55, 57){\fontsize{11 pt}{12 pt} \KUsemibold \DATE}}

    \phantom{...}
    \newpage
    \noindent
  }

\assignment{PhD thesis}
\author{Maxime Savoy}

\title{Combinatorial Cell Complexes}
\subtitle{Duality, reconstruction and causal cobordisms} 
\date{Submitted on February 28, 2021}
\advisor{Advisor: Bergfinnur Durhuus}
\frontpageimage{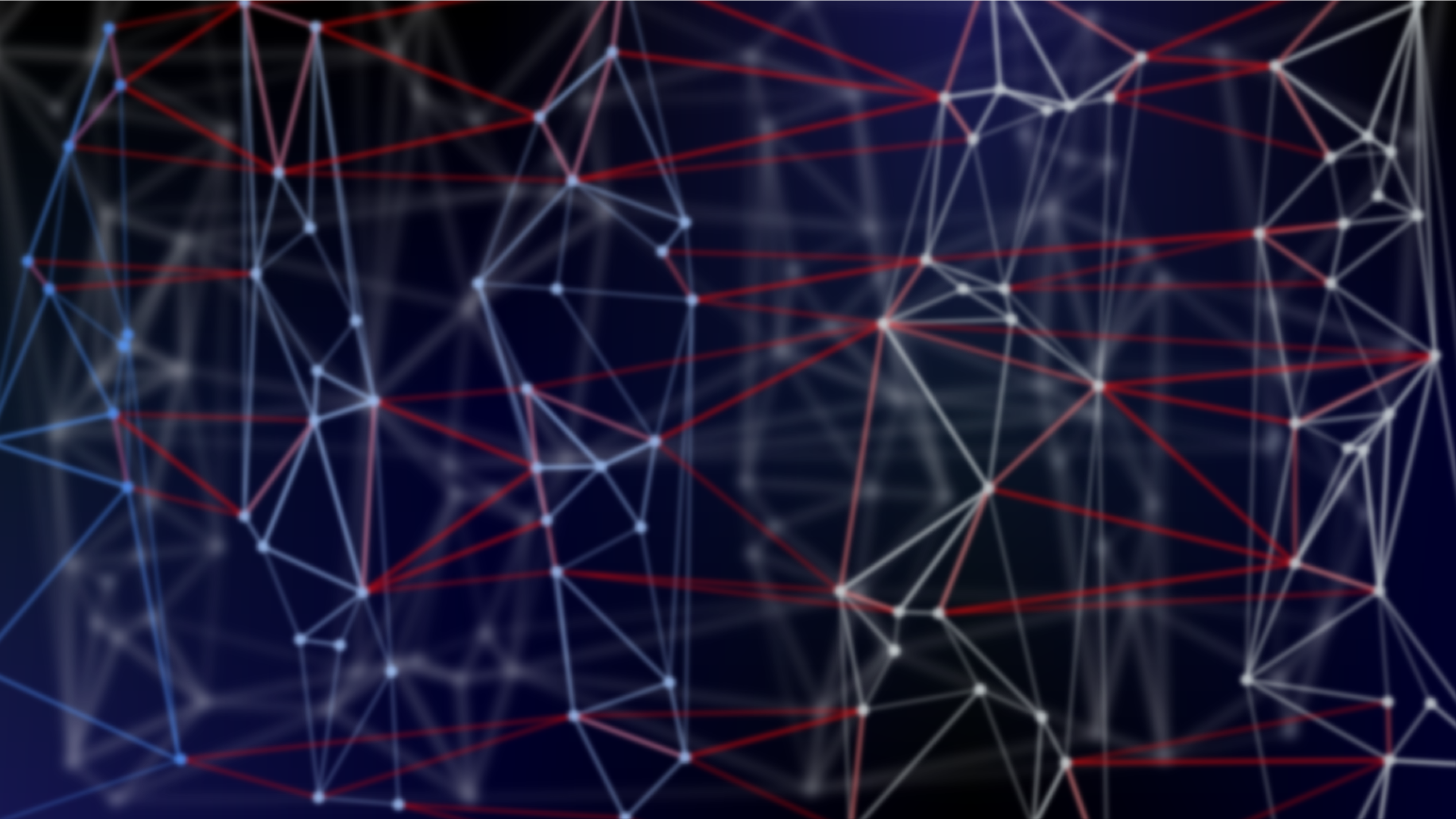}

\begin{document}

\maketitle

\thispagestyle{empty}
\vspace{2cm}

{\renewcommand{\arraystretch}{2}
\begin{tabular}{ll}
Assessment committee : & \textbf{John W. Barrett}, University of Nottingham\\
 & \textbf{Thordur Jonsson}, University of Island \\
 & \textbf{Nathalie Wahl} (Chair), University of Copenhagen
\end{tabular}}

\vspace{13cm}

\begin{tabular}{ll}

\copyright ~ by the author: 

&

Maxime Savoy \\

& Department of Mathematical Sciences, \\

& University of Copenhaguen,\\

& Universitetsparken 5,\\

& 2100 København Ø, Denmark.\\

& \texttt{maximegsavoy@gmail.com}\\[10pt]

ISBN: & 978-87-7125-041-1

\end{tabular}

\vspace{2cm}

This thesis has been submitted on February 28 2021 to the PhD School of 
the Faculty of Science, University of Copenhagen.
The present document is an updated version containing some corrections inserted before uploading it on arxiv.

\newpage

~
\vspace{8cm}
~

\thispagestyle{empty}
\begin{flushright}
\textit{À mon grand-père Fred}
\end{flushright}

\vfill

\newpage
~
\thispagestyle{empty}

\newpage

\begin{abstract}
This thesis proposes a framework based on a notion of combinatorial cell complex (cc) whose cells are defined simply as finite sets of vertices. The cells of a cc are subject to four axioms involving a rank function that assigns a rank (or a dimension) to each cell. Our framework focuses on classes of cc admitting an inclusion-reversing duality map. We introduce a combinatorial notion of cobordism that allows us to single out a category whose morphisms are cobordisms having a causal structure. Our aim is to offer an approach to look for a combinatorial notion of quantum field theory having a built-in duality operation acting on the underlying space and not relying on any manifold structure.

The introduction includes links with certain fields in Theoretical and Mathematical Physics related to Quantum Gravity and motivating our framework. We start by introducing cc and the duality map on a class of cc with empty boundary called closed cc. We then focus on the problem of reconstructing a certain class of cc from their cells of rank lower than or equal to 2. Such cc are in particular duals to simplicial complexes with no boundary and their reconstruction is realized using a discrete notion of connection. Our next main result establishes a way to extend the duality map we defined on closed cc to a class of cc with boundary. An important by-product of the study of this extended duality map is the combinatorial notion of cobordism used in this work. We also introduce a general notion of subdivision of a cc via a map called reduction, as well as the dual notion of reduction called collapse. These two types of map allow to characterize the structure of certain cc called slices, using sequences of maps called slice sequences. Slices are the basic building blocs of our definition of causal cobordisms and the dual of a slice sequence defines the composition of cobordisms, providing us with a category whose morphisms are causal cobordisms.
\end{abstract}

\vfill

\renewcommand{\abstractname}{Danish Resumé}

\begin{abstract}
I denne afhandling fremsættes en begrebsramme baseret på en general definition af det, vi kalder et kombinatorisk celle-kompleks (forkortet til cc), der består af endelige delmængder af en mængde af vertexer, som opfylder fire aksiomer, hvori indgår en funktion, der knytter en rang (eller en dimension) til hver celle. Et væsentligt fokuspunkt for diskussionen drejer sig om klasser af kombinatoriske celle-komplekser, der tillader en inklusions-omvendende dualitetsafbildning. Vi definerer kombinatoriske kobordismer, der giver anledning til en naturlig kategori, hvis morfier er kobordismer med en kausal struktur. Målet er at beskrive en ramme for en kombinatorisk version af kvantefeltteori, der har en indbygget dualitetsoperation, der virker på et underliggende rum uden a være afhængig af en traditionel mangfoldighedsstruktur. 

I introduktionen gives en beskrivelse af, hvordan forskellige problemstillinger  i  matematisk fysik generelt, og mere specielt vedrørende formuleringer af en kvanteteori for gravitation, tjener som motivation for de indførte begreber og resultater.  Derefter fokuseres i første omgang  på spørgsmålet om rekonstruktion af et kombinatorisk celle-kompleks ud fra dets celler af dimension 2 og lavere. Sådanne celle-komplekser er i særdeleshed duale til simpliciale komplekser uden rand, og rekonstruktionen gennemføres ved brug af en kombinatorisk udgave af en  kovariant afledet. Et vigtigt biprodukt af det efterfølgende studium af dualitet for kombinatoriske celle-komplekser med rand er den her indførte kombinatoriske definition af kobordismer, hvortil dualitetsafbildningen kan udvides. Videre defineres underinddelinger af kombinatoriske celle-komplekser via en såkaldt reduktionsafbildning, hvis duale modstykke i form af en kollapsafbildning også indføres. Disse to slags afbildninger tillader en karakterisering af en klasse af kombinatoriske celle-komplekser kaldet skiver ved hjælp af en følge af såkaldte skive-afbildninger. Disse skiver udgør grundelementerne i definitionen af kausale kombinatoriske kobordismer, og den duale til en følge af skive-afbildninger definerer sammensætning af kobordismer, og giver således anledning til en kategori, hvis morfier er kausale kobordismer. 
\end{abstract}

\vfill

\newpage

\vspace{3cm}


\vspace{2cm}

\renewcommand{\abstractname}{Acknowledgements}

\begin{abstract}
I would first like to thank my supervisor Bergfinnur Durhuus for letting me explore this uncharted avenue despite that for a long time I could only provide very vague indications that it could lead to this project. Thank you also for your patience, even when I needed many attempts to truly understand some of your valuable explanations, for the time you invested during our very regular meetings and for all the inputs you gave me along the way that taught me a lot about how to be a better mathematician.

It has been a great privilege to be a PhD student at the Department of Mathematical Sciences of the University of Copenhagen and to be part of the the QMATH Centre, funded by VILLUM FONDEN (Grant no. 10059). I admire the remarkable work of the heads of the QMATH center, Jan Philip Solovej, Matthias Christandl and Bergfinnur Durhuus and its center administrator Suzanne Anderson. They created a vibrant environment gathering distinguished expertise from diverse fields that I found very enriching.  I am grateful to all the people I shared office, lunch, coffee breaks, jumble games or jumps in the canal with for the good company and to all the participants of the weekly running for our completion of uncountable laps around Faelledparken under various weather conditions. I greatly benefited from discussions I had with members of the department and I am especially grateful for the supportive interactions I had with Laurent and Linard over coffee breaks. I would also like to thank Morten Risager, Henrik Laurberg Pederson and Thomas Vils Pederson for their time and helpful discussions, as well as Albert Werner and Fabien Pazuki for pleasant collaborations in teaching activities. I would like to thank Nina Weisse, Mette Fulling and all the members of the staff I encountered during my time at KU for always being helpful.

A big THANKS to Laurent Bétermin, Linard Hoessly and Davide Parise for the many valuable comments on this monograph.

Merci Clément Hongler pour m'avoir introduit au modèle d'Ising, pour m'avoir inspiré à explorer d'avantage la physique mathématique et pour votre confiance en mes capacités qui m'a conduite où je suis et m'a aussi aidé à mener à bien ce projet.

I am also grateful for all the great people I met in Copenhagen who helped me go along this journey. Patrick and Esteban, you were and will certainly remain very inspiring climbing mates and it has also been great sharing about our experiences as PhD students or discussing about the strange difficulties of some quadratic equations in finance or the intricacies of behavioural economics. Considering the amount of time we were all spending at home, I was also lucky to have Ken and Roì as my flatmates. We made a great home office team and you also helped me get through harder times. There are many people that I shared my climbing sessions with that I am unfairly leaving out as they also contributed to all the great times I had in and around Copenhagen and helped me keep my spirits up.
I also owe many of my older friends for their continuing encouragements and this unveils how much this list is doomed to remain incomplete.

Last but not least, my family and close ones have always been extraordinary in supporting me in my studies and this thesis is also a product of all their encouragements. Un gigantesque merci à mes parents pour votre soutient qui m'a toujours énormément aidé et qui m'a été très précieux durant ces trois dernières années.

\vspace{2cm}
I would like to thank John Barrett and Thordur Jonsson for their willingness to participate as members of the assessment committee and Nathalie Wahl for her time as a chair of the committee.
\end{abstract}

\newpage

\vfill

\tableofcontents

\vfill

\newpage
~
\thispagestyle{empty}

\newpage
\pagenumbering{arabic}
\setcounter{page}{1}

\vspace{3cm}

\section*{\centering \huge{Introduction}}

\addcontentsline{toc}{section}{Introduction}

\vspace{2cm}

How to best describe the geometry of space-time? The quest for an accurate model of our physical reality and its geometry has lead to considerable developments of our mathematical apparatus. These developments occurred at the interface between Mathematics and Physics, in the minds of scientists driven by ideas rooted on both sides. An emblematic figure among these scientists is the German mathematician Bernhard Riemann. While developing the new mathematical notion of manifold (in the text of his 1854 "Habilitation" lecture \cite{rie16}), famously generalizing the foundation of Euclidean geometry, Riemann was also guided by the intention of using this notion to describe our spacial reality. His construction in fact starts by directly making a clear distinction between the discrete and the continuous case. The main focus of his work was on the continuous case, leading to the usual modern definition of Riemannian manifold: a metric space in which each point is "homogeneous" in the sense that it has an open neighbourhood homeomorphic to the open Euclidean ball. Although the continuous case best captures the metric structure of space at macroscopic scales, Riemann also remarkably foresaw how some conceptual considerations would be needed when investigating which of these two hypothesis would still hold at microscopical scales. The notion of continuous manifold encompasses the notion of space-time introduced in 1908 by Minkowski, which was used by Einstein and Grossmann in 1913 when laying down the foundations of General Relativity. When physicists were later able to gather an overwhelming body of observational evidence supporting Einstein's revolutionary theory, it cemented these new ideas as the correct understanding of the physics of space-time at observable scales, and, in the same way, cemented Riemannian geometry as the most accurate way to formalize this theory. Riemannian and semi-Riemannian manifolds are now the default mathematical objects used to model the geometry of space-time and were instrumental in the development of the theoretical foundations of the two main pillars of theoretical physics, General Relativity and the Standard Model of particle physics.

Despite the secondary role played by Riemann's notion of "discrete manifoldness", discrete geometries have been extensively used within theoretical physics, in particular as a technical tool that helps circumvent certain difficulties arising in the analysis of models defined on continuous manifold. A common such difficulty being the divergences encountered when studying quantum fields at arbitrarily low scales, in which case using a discretization can be a way to regularize the theory, by defining a so-called ultraviolet (UV) cut-off. Triangulations are usually the simplest way to discretize a manifold, defining a simplicial complex with some homogeneity property. Square, cubic or hypercubic honeycomb lattices are also often chosen in contexts where one regularizes fields defined in flat space, for example in statistical mechanics or lattice gauge theories.

The search for a unified theory or a theory of quantum gravity has been the source of many attempts to use discretized spaces to reconcile the principles of General Relativity and Quantum mechanics. One main issue that usually arises in this case is that regularizing General Relativity by means of a discretization of the underlying manifold violates the symmetries of the theory: diffeomorphism invariance (or invariance under changes of coordinates) and Lorentz invariance. A UV cut-off indeed generally implies the existence of a type of minimal length that is inconsistent with coordinate or Lorentz transformations over the manifold. 
But how would we even define the notions of Lorentz transformation and change of coordinates if we were to consider field theories on discrete spaces defined with no reference to any continuous manifold or metric? And would the previous inconsistencies also arise?


The main interest of this work is to study the structure of discrete spaces with the properties of
\begin{itemize}[label= $\ast$, topsep=2pt, parsep=2pt, itemsep=1pt, leftmargin=1.5cm]
\item having a given dimension,
\item being defined independently of any prior continuous manifold or metric structure and
\item admitting a duality similar to Poincaré duality.
\end{itemize}
Our main goal is to understand how to define a category using a combinatorial notion of "cobordism" based on these spaces. Such framework is also used for the formulation of Topological Quantum Field Theory described in Section \ref{sectqft} of this introduction. 

Our starting point is a notion of "combinatorial cell complexes", defined essentially in the same way as Basak in \cite{bas10} using this terminology. Since we remain for the most part in a purely combinatorial context, we often choose to drop the term "combinatorial". To define a cell complex $K,$ we start with a finite set, whose elements we call \textit{vertices}. A \textit{cell complex} $K$ is then a collection of subsets of vertices called \textit{cells} together with a rank function $\rk : K \dans \N$ (where $\N = \{0,1,2, \dots \}$) assigning a \textit{rank} (or dimension) to each cell such that each vertex has rank $0$ and the following four axioms are satisfied:
\begin{enumerate}[label=\textbf{\roman*}), topsep=2pt, parsep=2pt, itemsep=1pt, leftmargin=1.5cm, rightmargin = 2cm]
\item the rank of a cell is strictly smaller than the rank of any cell strictly containing it,
\item the intersection of two cells is a cell or is empty,
\item if a cell $x$ of rank $r$ is strictly contained in a cell $y$ then there is a cell of rank $r+1$ strictly containing $x$ and contained in $y,$
\item if a cell $x$ of rank $r$ is strictly contained in a cell $y$ of rank $r+2$ then there are exactly two cells of rank $r+1$ containing $x$ and contained in $y.$
\end{enumerate}
   
One way to interpret this definition is to think of cell complexes as a generalization or an abstraction of certain discretizations of continuous manifolds. To compare these notions one uses the \textit{face lattice} of a discretization defined as the partially ordered set (poset) whose elements are the sets of vertices contained in each cells of the discretization. The face lattice of a discretization of a manifold is a cell complex with rank function defined by the dimension of the cells if its cells satisfies the second axiom and if no two of its cells share the same set of vertices. This corresponds to a particular type of combinatorial cell complexes satisfying an additional homogeneity property. The fact that a cell of a cell complex is characterized by its set of vertices is nevertheless important for our framework: a cell does not contain more information than the vertices it contains. This framework then offers a common ground on which to treat both spaces that are constructed from a continuous manifold, as is the case of most discrete geometries used in mathematical physics, and spaces for which one does not a priori associate an underlying continuous manifold and which are therefore not necessarily homogeneous.

From the above definition of cell complexes, and closely following the beginning of \cite{bas10}, one can readily give the definitions of purity, boundary, non-singularity and closeness for cell complexes (introduced in more detail in the beginning of Chapter \ref{secgennot} of the present work). A cell complex such that every maximal cell has the same rank is called \textit{pure} and for such cell complexes the \textit{boundary} is defined by: a cell is in the boundary if and only if it is contained in a sub-maximal cell $y$ such that $y$ is contained in only one maximal cell. \textit{Non-singular} cell complexes are pure cell complexes such that every \textit{edge}, i.e. cell of rank 1, contains exactly two vertices and such that every sub-maximal cell is contained in no more than 2 maximal cells. A non-singular cell complex is called \textit{closed} if its boundary is empty. \textit{Simplicial complexes} are then defined as cell complexes satisfying the additional property that every subset of a cell is a cell and can be thought of as "abstract" triangulations.

The duality on cell complexes we will focus on is a bijection on the set of cell complexes that reverses the inclusion relation between cells. For example, we discuss in Section \ref{ssecfullnsimplecc} the case of cell complexes that are dual to simplicial complexes, called \textit{simple cell complexes} (a terminology also used in the context of polytopes). The main feature of simple cell complexes can be understood as maximizing the adjacency relations among the cells containing any given cell. This property is dual to the characteristic property of simplicial complexes stating that any non-empty subset of the set of vertices of a simplex is a sub-simplex and this can be interpreted as maximizing the number of sub-cells of any given cell. 

Similar dualities also arise in models involving fields taking values on certain cells of a discrete space, where it extends to a duality on fields sometimes interpreted as a Fourier-like transformation. For example, we describe the Kramers-Wannier (KW) duality of the Ising model in Section \ref{secising}, which arguably constitutes the simplest non-trivial example of such a duality in dimension two, where the discrete space is a two-dimensional square lattice and the fields assign +1 or -1 to each vertex of the lattice. In Section \ref{sectqft} we also mention another occurrence of a type of Fourier transformation between so-called "state sum models". If two physical models are related by a duality interpreted as a Fourier transform, the primal and dual models are in principle considered as two equivalent mathematical descriptions of the same physical data. Applying this principle in a strict sense, one should obtain that the double dual of a theory is equivalent to the original theory and such an involution property is considered an important desirable feature for the duality map on cell complexes of our framework. The examples mentioned above generally focus on how the duality acts on fields and the underlying discrete space on which the duality also acts are usually defined as a discretization of a manifold. Despite these facts, their definition, based on assignments of data to the cells of the underlying discrete space, does not crucially rely on the assumption that this discrete space is a discretization of a manifold. Our notion of cell complex can then be thought of as some abstract generalization of the lattices or discrete spaces used as underlying spaces for these field theories that breaks free of any dependence on an underlying manifold. The framework we present in this thesis however only focuses on how to define a duality map on cell complexes of arbitrary finite dimension. The reason why we mention these examples of Fourier-like dualities is that we think they constitute a good basis to look for a notion of fields defined on cell complexes. We explain more concretely why we believe so in Section \ref{secising}. These examples therefore encapsulate well some of the motivations of this work and understanding how to include them in the formalism introduced here constitutes an interesting direction for further research.

It is relatively straightforward to define the duality map for the case of closed combinatorial cell complexes and verify that it defines an involution, which is part of the content of Chapter \ref{secgennot}. It is also easier to restrict the discussion to closed cell complexes while addressing the following uniqueness question raised for example by the formalism used in spin foam models: what are sufficient assumptions on a simple cell complex to allow the reconstruction of the complete cell complex using only the cells of dimensions lower than or equal to 2? This question is the main focus of Chapter \ref{secreconst}. Treating the duality of cell complexes with boundary and satisfying some locality assumptions requires a more general treatment and is the object of Chapter \ref{secdualcob}. What eventually came about as the most natural way to extend the duality to this case was to define it on some class of relative cell complexes, i.e.  couples $(K,J)$ where $K$ belongs to a certain class of non-singular cell complexes and $J$ is a subset of the boundary components of $K,$ with some additional conditions ensuring that the dual of such an element remains in the same class of relative cell complexes. It is then possible to interpret this duality map as acting on some discrete cobordisms interpolating between boundary components, a picture analogous to Atiyah's formulation of Topological Quantum Field Theory (TQFT) in \cite{ati88} that we describe briefly in Section \ref{sectqft}. Motivated by this observation, the remaining part of this thesis investigates how to make this analogy more precise by singling out an interesting category defined using a specific type of boundary objects and cobordisms as morphisms, that in particular have a causal structure. Defining how to compose these causal cobordisms requires a machinery that we develop in the last Chapter \ref{seccausalstruct}, refining the notion of duality on cobordisms introduced in the previous chapter. The structure of causal cobordisms turns out to be a generalization of the causal triangulations introduced in dimension 2 by Ambjørn and Loll in \cite{al98}. Durhuus and Jonsson introduced a notion of midsection in \cite{dj15} in order to establish that three-dimensional causal triangulations are universally exponentially bounded by their number of tetrahedrons. Coincidently, a generalized notion of midsection also plays an important role in the last two chapters. We briefly present in Section \ref{seccdt} some of the ideas underlying the statistical field theory approach to Quantum Geometry and the study of causal dynamical triangulations.
 
The motivation behind the choices of these particular discrete spaces and duality as a starting point is therefore a mixture of multiple observations and intuitions related to the different areas of physics and mathematical physics mentioned in this introduction. Some of these connections have been discovered along the way and there are certainly more to be found. As a consequence of our lack of expertise in many of these areas, we did not yet find truly satisfactory explanations or heuristic arguments to convey our intuition, other than the connections we sketch here and the example of Kramers-Wannier duality we describe in Section \ref{secising}. However such intuition is not required to understand the content presented in this thesis, which might also lead to different interpretations.

The core of this thesis starts with the definition of cell complexes from which the rest has been derived, driven by questions arising from the study of the duality and, unless specified otherwise, independently from known works in theoretical physics and discrete geometry. As a consequence, we introduce multiple definitions and concepts along the way, some are related to well-known notions and other have been chosen in order to best fit our needs. We aspired to use a terminology that is coherent with what we could find in the literature, mostly in the area of discrete geometry, but by no means we can ensure that this has been done thoroughly. In Section \ref{secdiscgeom}, we provide a sample of existing definitions and results in discrete geometry that are related to this work. The thesis is then entirely self-contained and does in particular not require to understand the physical ideas and models we mention in this introduction.  We therefore invite the reader to skip ahead to Section \ref{secresults} where we present the structure of the thesis if the content presented before is not relevant to her or his reading.



Each of the next five sections include a brief presentation of certain key concepts from topics in Mathematics and Physics that influenced our approach. We also mention more specific works we encountered in our research that have some relation to what we develop, as an inspiration or as a potential source of applications. We start with a section on the Ising model and the KW duality, which we take as an opportunity to expose a first example of cell complex in order to prepare the reader for the more general discussion we embark on after this introduction.

\subsection{The Ising Model and Kramers-Wannier (KW) duality} \label{secising}

The Ising model, introduced by the physicist Wilhelm Lenz in 1920 and solved in dimension one in 1924 by his student Ernst Ising as his thesis work, is one of the most studied models of statistical mechanics. Thanks to its relative simplicity, it is often the most tractable statistical mechanical model one can consider that exhibits interesting properties. An important such property is the so-called order-disorder phase transition that the Ising model undergoes when defined on the square or hypercubic lattices of dimension two or above. For the sake of brevity, we look at how the Ising model is defined on a simple portion of the square lattice $\Z^2$ seen as a subset of $\R^2.$ We then explain how the Kramers-Wannier (KW) duality is defined in this case in order to illustrate more concretely what the duality on cell complexes we study in this work is about and explain how in this case it extends to a duality on a field theory. One can find a mathematical treatment of the KW duality for example in Section 1.5 of \cite{pal07}.

A configuration of the 2D Ising model is given by associating a "spin value" $+1$ or $-1$ to each \textit{vertex} in $\V,$ the set of vertices of the square lattice $\Z^2$ included in a domain $\Omega \subset \R^2,$ as illustrated in Figure \ref{figising}, where $\Omega$ is simply a square of side length $5.$ We define the set of \text{edges} $\E$ associated to $\Omega$ to consist of all edges of $\Z^2,$ defined as sets containing two vertices, having both vertices in $\Omega.$ This defines a graph $\G_\Omega = (\V, \E)$ which represents a portion of the square lattice, to which we can associate a set of \text{faces} $\mathcal{F}$ containing each square of side length 1 having vertices in $\V.$ Here again, we simply consider an element in $\mathcal{F}$ as a set containing four vertices. Note that the term "face" has different meanings in other parts of this thesis. We define the portion of the square lattice $K = K(\Omega)$ to be the poset defined by
$$K := \kz \cup \ko \cup \kt, \quad \text{ where } \quad (\kz , \ko , \kt ) = (\V , \E ,\mathcal{F}).$$
We associate to $K$ the \textit{rank function} $\rk = \rk_K: K \dans \N $ defined by $\rk(x) = r$ whenever $x \in \kr,$ in which case $x$ is called an \textit{$r$-cell}. The \textit{rank of} $K,$ denoted $\Rk(K),$ is defined as the maximal rank among its cells. Defined as such, $(K, \rk)$ constitutes an example of a non-singular cell complex of rank 2. The \textit{boundary} $\partial K$ of $K$ is then made of all the edges in $\ko$ contained in only one face in $\kt$ and all the vertices contained in those edges.

\begin{figure}[!h]
\centering
\includegraphics[scale=0.48]{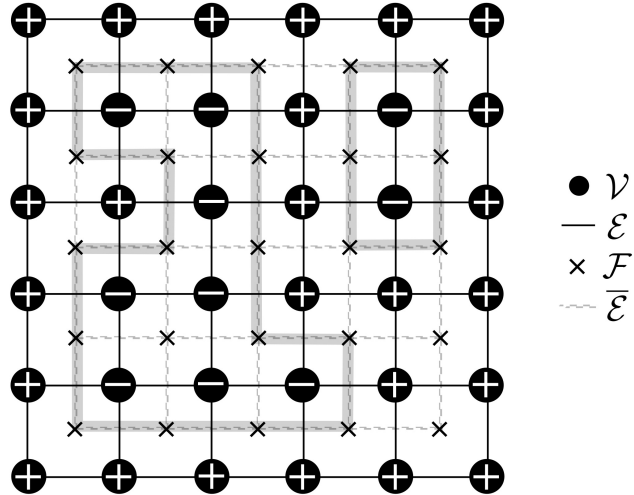}
\caption{\label{figising} In this picture, we represent the portion of lattice $K = (\V,\E, \F)$ on which the Ising model is defined by assigning a sign to each vertex in the vertex set $\V.$ The domain $\Omega$ is delimited by the outermost edges in $\E.$ For the purpose of this example, we define the dual lattice using "dual-vertices" corresponding to the faces $\mathcal{F}$, each "dual-edge" $\dual{e} \in \dual{\E}$ corresponds to two faces sharing an edge $e \in \E$ and has one "dual-face" for each vertex in $\V$ not on the boundary.}
\end{figure}

One can therefore consider an \textit{Ising configuration} to be a map $\sigma : \V \dans \{+1, -1\},$ i.e. a field taking values in $\Z/2\Z.$ We associate to each configuration $\sigma$ an "energy function" 
$$H(\sigma) = - \sum_{e =\{v,w\} \in \E} \sigma(v) \sigma(w) = N_\sigma - M_\sigma,$$
where $N_\sigma$ is the number of neighbouring spins in $\sigma$ with opposite values and $M_\sigma$ the number of neighbouring spins with identical values. The \textit{Ising model with parameter $\beta$} (and vanishing external magnetic field), or $\beta$-Ising for short, is defined by assigning a probability measure to the set of configurations, called the Boltzmann measure, defined by
$$ \pr[\sigma] = \frac{1}{\prt} \exp( - \beta H(\sigma)),$$
where $\beta = \frac{1}{T}$ is called the \textit{inverse temperature} (assuming that the Boltzmann constant is $1$ for simplicity). The \textit{partition function $\prt$} of the $\beta$-Ising is therefore given by
$$ \prt = \prt_\beta := \sum_{\sigma \in \{\pm 1\}^\V} \exp(- \beta H(\sigma)).$$
One can then assign different types of boundary conditions on the spins. For our purposes we will only use two types of boundary conditions: the \textit{$+$ boundary conditions,} fixing a value $+1$ to all vertices on the boundary, and the \textit{free boundary conditions,} corresponding to the case introduced above with no constraint on the spins on the boundary. We denote by $\prt^+$ the partition function of an Ising model with $+$ boundary conditions and by $\{\pm 1\}^{\V,+}$ the corresponding set of spin configurations.

Rudolf Peierls showed in 1936 (\cite{pei36}) that the Ising model with $+$ boundary conditions is disordered at high temperature, whereas at low temperature it exhibits ferromagnetic order. The duality due to Hendrik Kramers and Gregory Wannier (\cite{kw41}) described below allows to derive  the \textit{critical temperature} $T_c = \frac{2}{\ln(\sqrt{2})+1}$ of the model: the value of $T$ for which the model transits between the order and disorder phases (assuming such a value is unique). Lars Onsager found an exact solution of the model in dimension 2 in 1944 (\cite{ons44}) using transfer matrices.

In order to describe the KW duality, we will first define the dual of the cell complex $K.$ Generalizing this step to more general discrete spaces is a way to interpret the main focus of this thesis. We start by associating a \textit{dual cell $\dual{x}$} to each cell of $x \in K \setminus \partial K$ defined by $$\dual{x} := \{ z \in \kt ~|~ x \subset z\}.$$ 
The \textit{dual of $K$} is then defined as the set 
$$\dual{K \setminus \partial K} := \{ \dual{x} ~|~ x \in K \setminus \partial K\},$$
to which we also associate a rank function defined by $$  \rkb(\dual{x}) = \rk_{\dual{K\setminus \partial K}}(\dual{x}) := 2 -\rk(x).$$

As a result, the cells of $\dual{K \setminus \partial K},$ also illustrated in Figure \ref{figising}, can be described as follows. Each vertex $\dual{f}$ of $\dual{K \setminus \partial K}$ is of the form $\dual{f} = \{f\}$ where $f \in \mathcal{F}.$ An edge $\dual{e}$ of $\dual{K \setminus \partial K}$ is a set containing the two faces $f, f' \in \mathcal{F}$ such that $e = f \cap f'.$  And each face of $\dual{K \setminus \partial K}$ is of the form $\dual{v} = \{f_1, f_2, f_3, f_4\},$ where $v$ is a vertex not on the boundary $\partial K$ and $v \in f_i$ for all $1 \leq i \leq 4,$ i.e. the four faces of $K$ containing $v.$ We use the notation $\dual{\E}$ to denote the set of edges in $\dual{K \setminus \partial K}.$


The idea of the KW duality is that one can express - up to some constant factor - the partition function of the $\beta$-Ising with $+$ boundary conditions defined on $K$ as the partition function of a $\dual{\beta}$-Ising with free boundary conditions defined on $\dual{K \setminus \partial K}.$ The value of $\dual{\beta}= \dual{\beta}(\beta)$ is given by the formula 
\begin{equation}\label{equKW}
\sinh(2 \beta) \sinh(2 \dual{\beta}) = 1 \quad \Leftrightarrow \quad \tanh(\dual{\beta}) = \exp(-2 \beta). \tag{KW}
\end{equation}
The critical temperature $T_c$ is then obtained by choosing the value of $\beta$ satisfying $\dual{\beta}(\beta) = \beta.$ At critical temperature, one obtains the "self-duality" of the model in the infinite volume limit (i.e. when $\Omega$ tends to the entire plane).

More precisely, the KW duality is obtained by expanding the partition function of each model in two different ways, using the so-called low-temperature expansion for $\prt_\beta^+$ and the high-temperature expansion for $\prt_{\dual{\beta}}.$ For the latter, we will use the following important observation: a spin configuration of the $\beta$-Ising with $+$ boundary conditions uniquely defines a subset $\varepsilon \subset \dual{\E}$ called an \textit{even edge configuration of dual edges} satisfying that each vertex of $\dual{K \setminus \partial K}$ is contained in an even number of edges in $\varepsilon.$ This bijection between $\{\pm 1 \}^{\V,+}$ and the set $\mathcal{C}(\dual{\E})$ of even edge configurations of dual edges is defined by associating to a spin configuration $\sigma$ the set $\varepsilon$ of edges delimiting the sign clusters in $\sigma$ (as illustrated in Figure \ref{figising}). Conversely, we can interpret each element $\varepsilon \in \mathcal{C}(\dual{\E})$ as the sign clusters of a spin configuration in $\{\pm 1\}^{\V,+}$ with outermost spins being assigned the value $+1.$

Let us denote by $\abs{\E}$ the number of edges in $K.$ The definitions of $N_\sigma$ and $M_\sigma$ given above yield  $\abs{\E} = N_\sigma + M_\sigma$ for all $\sigma \in \{\pm 1\}^\V.$
We can then on the one hand express $\prt_\beta^+$ as follows to obtain the \textit{low-temperature expansion}:
\begin{align*}
\prt_\beta^+ &= \sum_{\sigma \in \{\pm 1\}^{\V,+}} \exp( - \beta (N_\sigma - M_\sigma)) \\
&= \exp(2 \beta \abs{\E}) \sum_{\sigma \in \{\pm 1\}^{\V,+}} \exp( -2 \beta N_\sigma)  \propto \sum_{\sigma \in \{\pm 1\}^{\V,+}} \exp( -2 \beta N_\sigma).
\end{align*}
Where $\propto$ stands for "equal up to a constant factor". On the other hand, we get the \textit{high-temperature expansion} for $\prt_{\dual{\beta}}$ using the following computation:
\begin{align*}
\prt_{\dual{\beta}} &= \sum_{\tau \in \{\pm 1\}^{\mathcal{F}}} \exp\left( \dual{\beta} \sum_{\{f,f'\} \in \dual{\E}} \tau(f) \tau(f')\right)\\
&= \sum_{\tau \in \{\pm 1\}^{\mathcal{F}}}  \prod_{\{f,f'\} \in \dual{\E}} \left( \cosh(\dual{\beta}) + \sinh( \tau(f) \tau(f') \dual{\beta})\right)\\
&= \cosh(\dual{\beta})^{\abs{\dual{\E}}} \sum_{\tau \in \{\pm 1\}^{\mathcal{F}}}  \prod_{\{f,f'\} \in \dual{\E}} \left( 1 + \tau(f) \tau(f') \tanh(\dual{\beta})\right) \\
&= \cosh(\dual{\beta})^{\abs{\dual{\E}}} \sum_{\tau \in \{\pm 1\}^{\mathcal{F}}} \sum_{\varepsilon \subset \dual{\E}} \left( \prod_{\{f,f'\} \in \varepsilon}  \tau(f) \tau(f') \tanh(\dual{\beta}) \right)\\
&= 2^{\abs{\V}}\cosh(\dual{\beta})^{\abs{\dual{\E}}} \sum_{\varepsilon \in \mathcal{C}(\dual{\E})}  \tanh(\dual{\beta})^{\abs{\varepsilon}} \propto  \sum_{\varepsilon \in \mathcal{C}(\dual{\E})}  \tanh(\dual{\beta})^{\abs{\varepsilon}}
\end{align*}
where, in the last equality, we exchanged the order of summation and used that, after summing over $\tau \in \{\pm 1\}^{\mathcal{F}},$  the product over edges in $\varepsilon$ is non-zero if and only if $\varepsilon \in \mathcal{C}(\dual{\E}).$  Using the correspondence described above and the equation (\ref{equKW}) we indeed obtain
$$ \sum_{\varepsilon \in \mathcal{C}(\dual{\E})}  \tanh(\dual{\beta})^{\abs{\varepsilon}} = \sum_{\sigma \in \{\pm 1\}^{\V,+}}  \tanh(\beta)^{N_\sigma} = \sum_{\sigma \in \{\pm 1\}^{\V,+}}  \exp( -2 \beta N_\sigma).$$

It is possible to make an analogy between the KW duality and the Fourier transform. Such an analogy appears in a recent work of Freed and Teleman \cite{ft19}, where they also make a connection between the KW duality and the electromagnetic duality of finite gauge theories in dimension 3, an indication that KW duality is a manifestation of a more general duality.

The correspondence between even configurations of dual edges $\varepsilon \in \mathcal{C}(\dual{\E})$ and a spin configurations $\sigma \in \{\pm 1\}^{\V,+}$ implicitly uses that the domain $\Omega$ is simply connected and that $\R^2$ is an orientable manifold. In Section \ref{ssectionsimplyconccc}, we introduce a notion of simply connected cell complex and Basak's work \cite{bas10} is based on a notion of orientability for cell complexes. It would therefore be interesting to understand whether these are sufficient assumptions in order to adapt the steps in this derivation to a certain type of cell complexes from our framework. Doing so would also require to introduce a notion of fields on a cell complex and we believe that this would involve a generalization of the notion of connexion we introduce in Chapter \ref{secreconst}.

One can see for example in \cite{weg14} that the KW duality holds on more general lattices although it does only directly lead to a self-duality argument allowing the computation of the critical temperature as in the case of the square lattice. Nevertheless, the triangular lattice for example also allows to derive the critical temperature using a so-called decimation process, an operation seemingly related to the notions of reductions and collapses we introduce in Chapter \ref{seccausalstruct}.

Fixing a simply connected bounded planar domain $\Omega \subset \R^2,$ one can study the properties of the Ising model on increasingly finer discretizations of $\Omega,$ using for example $\delta \Z^2$ for some $\delta >0$ small instead of $\Z^2.$ The scaling limit of the model is then defined as the limit of models defined on discretization of $\Omega$ using a lattice spacing $\delta$ tending to $0.$ The scaling limit of the critical Ising model, with $\beta = \beta_c := \frac{1}{T_c},$ has been extensively studied. Many of its "observables", defined as the expectation value of certain functions of the spin variables, have been shown to admit a well-defined scaling limit. The limiting value of these observables has the remarkable property of being invariant under conformal transformations. 
Rigorous proofs of some of these results have only been given recently, with for example the work of Smirnov \cite{smi06} on fermionic observables and later works for other fields of the model like the spin correlation functions (\cite{chi15} 
\cite{hkv17}
) and for the interfaces of the sign clusters (
\cite{bh19}). 
Setting the parameter $\beta$ at its critical value therefore allows to derive a continuous field theory having a large class of symmetries. Moreover, one observes the phenomenon of universality: different discretizations of a given domain lead to the same continuum limit (although the specific value of the critical parameter depends on the choice of discretization).
Systematizing these ideas is related to the next topic of this introduction: statistical field theories.



\subsection{Statistical field theory and Causal Dynamical Triangulations (CDT)} \label{seccdt}

The statistical theory of fields uses the framework of phase transitions to establish connections between the properties of field theories describing some physical system at microscopic scales, often defined as lattice models, and the properties of continuous field theories describing the system at macroscopic scales. The most common lattice models studied in this context are originally based on physical models from statistical mechanics. Among them, the Ising model, related to the phenomenon of ferromagnetism, has played a prominent role in the development of the techniques used in statistical field theory.

An important such technique that first emerged as a regularization procedure in Quantum Field Theory (QFT) is the notion of renormalization. In statistical mechanics, this technique consists in defining the notion of "coarse graining" of a lattice model to obtain a rescaled model defined on a larger lattice. Applying this process repeatedly leads to a "renormalization flow" within a certain space of lattice models described by a given set of parameters, a picture elaborated in large part by the theoretical physicist Kenneth Wilson. This flow is defined by differential equations describing how these parameters change when the model is rescaled. Using these equations, one can also try to run the flow backwards, which amounts to look for models that describe the theory at smaller scales (or lower temperature).
By running a renormalization flow one can obtain fixed points corresponding to parameters of a system invariant under renormalization. These can be of three kinds: an infrared fixed point for high scales, an ultraviolet fixed point for low scales or a critical fixed point corresponding to a phase transition. 
These fixed points are universal among a certain class of models, meaning that all models in this class flow towards the same fixed point. The idea of looking for a non-trivial UV fixed point for a theory describing gravity originates in the work of Steven Weinberg and is at the root of an approach to Quantum Gravity called "asymptotic safety" (see e.g. \cite{eic19} for a review on this topic). 
Our framework provides the definitions of the notions of reductions (defining subdivisions of cell complexes) and collapses (duals of reductions in a certain sense), given in the first section of Chapter \ref{seccausalstruct}. Although the relation with fields is yet to be determined, these constitute good candidates as formalizations of the concepts of coarse graining and refinement in the context of cell complexes.

The study of phase transitions in statistical field theory provided a new understanding of renormalization techniques in QFT and, reciprocally, the notion of Feynman diagrams and path integral from QFT  influenced the interpretation and study of statistical models as quantum systems. 
To illustrate this, let us consider a discretization of a metric space defining a cell complex $K$ with two boundary components $J$ and $L,$ obtained for example from a lattice embedded in a Riemannian manifold. QFT motivates us to look at a lattice model on $K$ as a functional integral in "imaginary time" (i.e. using a Wick rotation $t \mapsto i \tau$) over the space $\Phi(K)$ of field configurations taking values on certain classes of cells . Formally, one could write the propagator of such a theory as
$$ G( \phi_J; \phi_L) = \int_{\substack{\phi \in \Phi(K), \\ ~\phi|_J = \phi_J, ~ \phi|_L = \phi_L}} \mathscr{D} \phi~ \exp( - S[\phi] ),$$ 
where the integration is used symbolically and would usually reduce to a sum. For example, in the case of the Ising model, a field configuration would correspond to a spin configuration $\phi = \sigma \in \{\pm 1\}^{\kz},$ the action would be $S[\sigma] = \beta H(\sigma)$ and the integral $\int \mathscr{D} \phi$ would simply be a sum over $\sigma \in \{\pm 1\}^{\kz}.$ One could also imagine that $K$ is defined independently of any background manifold and that a field $\phi \in \Phi(K)$ induces a metric on $K.$
With the intention of quantizing gravity, this raises the following questions: can one extend the previous integral to a functional integral over some set of discrete geometries with field configurations thereon? More explicitly, such a propagator would have the form
\begin{equation} \label{equCDT}
G( J, \phi_J; L, \phi_L) = \int_{K \in \mathsf{C}(J,L)} \int_{\substack{\phi \in \Phi(K), \\ ~\phi|_J = \phi_J, ~ \phi|_L = \phi_L}}  \mathscr{D}  K ~ \mathscr{D} \phi~ \exp( - S[\phi, K] ), \tag{FI}
\end{equation}
where $\mathsf{C}(J,L)$ is a set of cell complexes with boundary components $J$ and $L.$ Moreover, would such a model also undergo a phase transition of the kind needed to define a continuum limit?

This leads us to the topic of Dynamical Triangulations (DT) and the statistical field theory approach of quantum geometry as presented in \cite{adj97}, which in particular answers positively the latter questions in dimensions 1 and 2. This approach investigates the critical behaviour of models that implement the integration over a space of triangulations made out of flat simplices of fixed edge length $a >0,$ defining a metric space. Such a model would correspond to setting $\mathsf{C}(J,L)$ in (\ref{equCDT}) to be some set of simplicial complexes with boundary $J$ and $L$ and defining $\Phi(K)$ as the set only containing one field associating the value $a$ to each edge of $K.$ Such an integration can be seen as a way to regularize the integral over the space of metrics on a Riemannian manifold. The weight assigned to a triangulation is typically defined using the so-called discrete Regge action, the Euclidean discrete equivalent to the Einstein-Hilbert action obtained in the influential paper from Tullio Regge \cite{reg61}.
One attractive idea of this approach is that the existence of a continuum limit would provide a measure on the space of continuous metrics on some manifold and the universality of this limit would imply the independence of the choice of regularization. This program leads to many analytic results in dimension 2, and one technique used to derive exact solutions of this case is to obtain triangulations of surfaces as Feynman expansion of matrix models (see for example \cite{fgzj95} for a review on this topic). 

First introduced in dimension 2 in \cite{al98}, the model of Causal Dynamical Triangulations (CDT) in dimension $d$ restricts the functional integral (\ref{equCDT}) to triangulations having the topology of a cylinder $S^{(d-1)} \times \left[ 0,1 \right]$ decomposed into "causal slices" having the same topology. A causal slice has the property that all its vertices are on the boundary and none of its simplices of maximal dimension has all its vertices in the same boundary component. A precise definition of causal slice can be found in \cite{dj15} and physical descriptions of CDT can be found in the reviews \cite{agjl12}, \cite{lol20}.  
In the last chapter of this thesis we introduce a generalization of the notion of causal slices, simply called a "slice". We show in Section \ref{ssecslice} that such a slice is characterized by some sequence of the reduction and collapse maps mentioned above. We also introduce a generalization of the notion of causal triangulation that we call causal cobordism. Our reason for choosing the name "cobordism" comes from a similar notion appearing in the formulation of TQFT, presented in the next section. 


\subsection{Topological Quantum Field Theory (TQFT) and state sum models} \label{sectqft}

Topological quantum field theories have the particular property that the observables associated to the fields of the theory do not depend upon any metric structure on their underlying space. As presented in \cite{bbrt91}, the study of TQFT was initiated on the one hand by mathematical works related to the classification of manifolds using topological invariants, notably the work of Albert Schwarz on the Ray-Singer torsion and, on the other hand, by the search for super-symmetric extensions of the Standard Model by Edward Witten. The first instances of TQFT where therefore relatively complicated and quite diverse in nature. This motivated the introduction of a comparatively simple axiomatic definition of TQFT given by Atiyah in \cite{ati88}. These axioms have been inspired by a similar setup given for CFT by Graeme Segal and are designed to allow a general mathematical treatment of these theories. We refer the reader to \cite{bar95} for a more complete discussion of TQFT in the context of Quantum Gravity. 

One can conveniently formulate the definition of a TQFT in the language of category theory, using the notion of symmetric monoidal category. Other more refined notions of categories have also been developed to encompass the different examples of state sum models we mention later on. We chose to simply use the formulation used by Atiyah here in order to not be carried away introducing categorical notions we will not need for the purpose of this work. The discussion of the duality on the type of combinatorial cobordisms presented in the core of this thesis nonetheless calls for an appropriate formulation in categorical terms and this could be a topic for further investigations.  In what follows we will only refer to the standard notions of category and functor between categories, which we assume familiar to the reader.

Let $\Lambda$ be $\R$ or $\C$ and let $d$ be an integer larger or equal to $1.$ A TQFT in dimension $d$ over $\Lambda$ is defined by associating a finite dimensional vector space $Z(\Sigma)$ over $\Lambda$ to each oriented closed smooth $(d-1)$-dimensional manifold $\Sigma,$ setting $Z(\emptyset) = \Lambda$ and an element $Z(M) \in Z(\partial M)$ to each oriented smooth $d$-dimensional manifold $M$ in such a way that the three following axioms are satisfied:
\begin{enumerate}[label=\textbf{(\arabic*}), topsep=2pt, parsep=2pt, itemsep=1pt, leftmargin=1.5cm, rightmargin = 2cm]
\item \label{axiom1tqft} each orientation preserving diffeomorphism $f: \Sigma \dans \Sigma'$ induces an isomorphism $Z(f): Z(\Sigma) \dans Z(\Sigma')$ which satisfies the composition rule 
$$Z(g \circ f) = Z(g) \circ Z(f)$$
for any other orientation preserving diffeomorphism $g: \Sigma' \dans \Sigma'';$
\item $Z(\Sigma^*) = Z(\Sigma)^*$ where $\Sigma^*$ is $\Sigma$ with opposite orientation and $Z(\Sigma)^*$ is the vector space dual to $Z(\Sigma);$
\item \label{axiom3tqft}for every disjoint oriented closed smooth $(d-1)$-dimensional manifolds $\Sigma_1, \Sigma_2,$
$$Z(\Sigma_1 \sqcup \Sigma_2) = Z(\Sigma_1) \otimes Z(\Sigma_2).$$
\end{enumerate}

Theses axioms in particular implies that for each oriented smooth $d$-dimensional manifold $M$ such that $\partial M = \Sigma_1 \sqcup \Sigma_0^*$ (where $\Sigma_0,\Sigma_1$ can be empty), we have
$$ Z(M) \in Z(\Sigma_0)^* \otimes Z(\Sigma_1) = \Hom(Z(\Sigma_0),Z(\Sigma_1)).$$ 
We can therefore view $M$ as a cobordism from $\Sigma_0$ to $\Sigma_1$ which induces a linear transformation $Z(M) : Z(\Sigma_0) \dans Z(\Sigma_1).$ We give another example of cobordism in Figure \ref{figtqft} and one can find simple examples of TQFT in \cite{cr18}. Note that in particular if $M$ has empty boundary then $$Z(M) \in Z(\emptyset) = \Lambda,$$ i.e. $Z(M)$ corresponds to a topological invariant of the manifold $M.$

\begin{figure}[!h]
\centering
\includegraphics[scale=0.35]{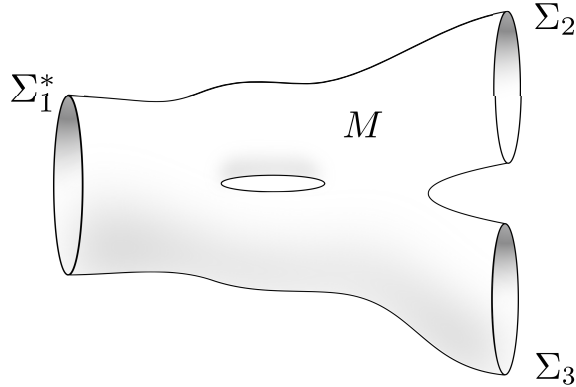}
\caption{\label{figtqft} A cobordism $M$ of a two-dimensional TQFT with ingoing component $\Sigma_1^*$ and outgoing components $\Sigma_2$ and $\Sigma_3.$ The axioms of TQFT imply that $Z(M)$ induces a linear transformation $Z(\Sigma_1) \dans Z(\Sigma_2) \otimes Z(\Sigma_3).$}
\end{figure}

In particular, a TQFT provides a functor $Z$ from a geometrical category, whose objects are oriented closed smooth $(d-1)$-dimensional manifolds and morphisms are $d$-dimensional cobordisms, to the algebraic category of vector spaces. Such a functor well captures the fundamental physical concepts of quantum states and processes: a state of a quantum system corresponds to a background space in which one incorporates physical data and the morphism associated to a cobordism acts on a state as a linear transformation on these data. Axiom \ref{axiom3tqft} implies the usual rule from quantum mechanics, stating that the state space associated to a quantum system composed of two isolated sub-systems corresponds to the tensor product of the state spaces associated to the sub-systems. In this sense, TQFT provides a way to consider oriented smooth manifolds as "natural carriers" of quantum fields and proposing a similar definition using purely combinatorial objects instead of manifolds is a central motivation of the last two chapters of this thesis. 


In practice, many TQFT are defined using discretizations of manifolds. 
Loosely speaking, starting from a functional integral as in (\ref{equCDT}), the Dynamical Triangulations approach fixes a field and sums over a certain set of discrete spaces. In the case of a discretized TQFT, one would fix a discretization of a manifold while summing over a certain set of field configurations with the requirement that taking any other triangulation of the same manifold does not affect the result of the integral.
More generally one can define the notion of state sum model, as given in \cite{bar01}. Such a model is defined on a triangulation of a manifold of given topology by assigning a set of states to each simplex with compatibility conditions whenever simplices overlap, where each set of states can be seen as the possible field configurations on a given simplex, usually defined using group theoretic data. These states allow to associate a complex number called weights to each simplex and the partition function defining the model is computed by summing the product of all weights over all possible state configurations, with no a priori assumption on the invariance of the resulting value under changes of triangulation. 

Numerous state sum models have been defined in relation to Quantum Gravity. One of the first examples is the Ponzano-Regge model (\cite{pr69}, \cite{bn09}), which is defined by assigning irreducible representation of the group $SU(2)$ to the edges of a 3-dimensional triangulation and with weights defined using intertwiners (i.e. equivariant linear maps between representations) called Racah's coefficients or 6-j symbols. The partition function of the latter model diverges in many cases and a way to obtain a regularization is to use representations of the quantum group $SU_q(2),$ where $q$ is a root of unity, instead of $SU(2),$ and this corresponds to the Tuarev-Viro (TV) model (\cite{tv92}). This last model in fact defines an example of a TQFT: the partition function of the TV model associated to a triangulation of a manifold produces an invariant of the manifold. An analogue of the TV model in four dimensions is the Crane-Yetter (CY) model (\cite{cy93}, \cite{cky97}). It is also an example of TQFT and it provides invariants of 4-manifolds. In the original version of the CY model, these invariants are actually functions of the Euler characteristic and signature of the manifold, but it has recently been generalized to give a larger class of invariants (\cite{bb18}).
State sum models defining a quantized discrete version of General Relativity using representations of a quantum deformation of the rotation group $SO(4)$ include the Barrett-Crane (BC) model (\cite{bc98}) for the case of Euclidean signature and its Lorientzian counterpart (\cite{bc00}). These models were called relativistic spin networks, as they can be seen as 4-dimensional versions of the spin-networks originally introduced by Roger Penrose in \cite{pen71}. The BC model is in particular not topological, as it corresponds to a (deformed) discretized version of a continuous field theory called $BF$-theory, which is topological, but with an additional local constraint sometimes referred to as the "simplicity constraint". This constitutes an attractive aspect of the BC model since there is a relation between $BF$-theory with the simplicity constraint and gravity described e.g. in \cite{fs12}.

It is interesting to note the many occurrences, in the analysis of state sum models, of a type of complex dual to the underlying triangulation as well as dualities seen as generalized Fourier transformations between the observables of certain models. This suggests that a number of the models mentioned above would be good candidates to consider as particular examples of the field theory one could obtain using the discrete notion of cobordism introduced in this work. Among the occurrences of Fourier-type dualities in state sum models, we can mention the relation between the TV or CY models and the relativistic spin networks described in \cite{bar03} and \cite{bfmg07}.

We can also point out the recent construction of a "cellular TQFT" in \cite{cmr20}, corresponding to a discretized $BF$-theory defined on a discrete notion of cobordism. These discrete spaces and their duals are constructed using well-known notions from piecewise linear topology we mention in Section \ref{secdiscgeom}. Even though the fields are therefore defined on cobordisms constructed from a manifold, it would also be interesting to understand if their approach can be adapted to constitute an example of fields defined on the notion of cobordism defined here.


The spin foam models we discuss in the next section and developed in the context of LQG constitute another type of state sum models closely related to the BC models and defined on a dual structure.

\subsection{Loop Quantum Gravity (LQG) and spin foam models} \label{seclqg}

The Loop Quantum Gravity approach to Quantum Gravity is originally based on the Hamiltonian formulation of General Relativity given by Abhay Ashtekar, using a connection and its conjugate momentum as fundamental variables. As presented in \cite{arr15}, this so-called canonical quantization approach dates back to the work of physicists like Dirac, Bergmann, Arnowitt, Desner and Misner and sets General Relativity into the framework of gauge theory. During the 1990s many physicists followed the lead of Lee Smolin, Ted Jacobson and Carlo Rovelli and went on to develop the first formulation of LQG from Ashtekar's variables, using loops - or holonomies - to represent the type of solutions obtained in this approach. Their formulation included the idea that loops are a way to implement a "loop transform" acting on the connection variables like a Fourier transform \cite{ai92}. The theory then evolved in two directions: the first direction stayed closer to the algebraic treatment of the original canonical formulation and the second opted for a "covariant approach" based on the path integral formulation related to the functional integration we described above. Although the results on reconstruction presented in Section \ref{ssecreconstthm} of this thesis suggest the interpretation that a certain class of simple cell complexes are characterized by the loops defined by their set of 2-cells, we believe that the spin foam models at the basis of the second "covariant" approach of LQG might be more directly related to our work. Features of spin foam models, as structures based on the dual of a triangulation, where in fact the first inspiration of our choice to include a Poincaré-like duality at the basis of our framework, as an attempt to bridge ideas of CDT and LQG.

As mentioned previously, spin foams are a type of state sum models and they are also designed to implement a functional integral over 4-dimensional metrics interpolating between 3-dimensional boundary states called spin networks. Their definition is based on a 2-dimensional discrete space called "2-complex" and assigns "colourings" to each cell corresponding to an irreducible representation of some rotation group for each 2-dimensional face as well as an intertwiner for each vertex and each edge. Such 2-complexes are similar to but more general than the combinatorial cell complexes introduced here, as for example the underlying graph can have edges containing only one vertex or two vertices can be contained in more than one edge. The 2-complex of a spin foam is often taken to be the 2-skeleton (i.e. the cells of rank lower or equal to 2) of the dual of a triangulation of a 4-dimensional manifold. In this case spin foams are essentially defined as a non-deformed BC model, with more general choices of intertwiners.

There exist many variants of spin foam models corresponding to different choices of intertwiners as well as different ways to enforce the simplicity constraint mentioned for the case of the BC model (see e.g. \cite{per13} for a review on the topic).
A spin foam model allows to compute "amplitudes" - an analogue of the propagator (\ref{equCDT}) - assigning a number to boundary spin network states with fixed colouring that are seen as quantum states of 3-dimensional geometry. Similarly to the partition functions of state sum models, such an amplitude is typically defined by integrating over the colourings on a 2-complex interpolating between the boundary states and "tracing over" the representations assigned to each face of the 2-complex by contracting indices according to the pairings induced by the intertwiners. 

Spin foam models in four dimensions are generally "non-topological", implying that the related amplitudes in principle depend on the choice of a triangulation and this constitutes a major difficulty in making sense of these models as a non-perturbative theory of Quantum Gravity. Looking for a so-called "background independent" formulation, some spin foam models are obtained in \cite{dfkr00} as a Feynman expansion of models called "Group Field Theories". Such a theory can be interpreted as a formal generalization of the matrix model techniques used to derive analytic results for 2 dimensional Dynamical Triangulations. At the end of Chapter \ref{secreconst},  we briefly explain how the results on reconstruction of cell complexes could be interesting if seen as diagrams in such Feynman expansions.

The term spin foam is due to John Baez, who introduced these objects using the language of category theory in \cite{bae98}. His definition includes a notion of dual spin foam, but the duality acts only on the colourings of the model and does not affect the structure of the underlying 2-complex. Also, the cells of such a spin foam are defined using the notion of piecewise linear cell complex given in \cite{rs72} and the composition of spin foams is formulated using some embedding in Euclidean space. In this sense, our formulation is "even more combinatorial" as the composition of cobordisms defined in Section \ref{sseccompcob} does not involve any embedding and is based on our definition of cells merely as sets of vertices.

In \cite{rs12} and \cite{rov11}, the idea of using a setup similar to Atiyah's axioms of TQFT for a specific spin foam model is proposed and motivates well the aim of this work. However, the boundary of a spin foam is defined in this context as the graph composed of edges contained in only one 2-cell (or face) and therefore differs fundamentally from the definition of the boundary of cell complexes used in our framework. In particular it does not allow for boundary edges to be included in two faces of the 2-complex, a feature we think might represent a too restrictive and somewhat arbitrary constraint on the discrete geometry in dimensions higher than two. Moreover, the composition of spin foams defined in these references implies that the graph corresponding to the intermediate boundary object is removed in the process, a feature that seems to represent a drawback in order to obtain a formulation using categories.

A notion of spin foam models with causal structure was discussed by Gupta in \cite{gup00}, which is presented as a generalization of the concept of causal evolution for spin networks given by Markopoulou and Smolin in \cite{ms97} (with an emphasis on the duality in \cite{mark97}). These models also motivate well the notion of causal cobordism introduced in Section \ref{sseccatcausalcob}, which can be seen as discrete spaces on which one could in principle define a notion of causal spin foams having a more general structure.



In \cite{ds14}, Dittrich and Steinhaus discuss the notions of coarse graining and refinement of spin foams and establish a relation with their time evolution. These notions are used to argue for the existence of a continuum limit of certain spin foam models. These very notions happen to be the central ingredients of the main result of Chapter \ref{seccausalstruct} which characterizes the structure of causal cobordisms using some sequences of reductions and collapses.

As already mentioned, spin foams are usually based on a 2-complex constructed as the dual of a triangulation and corresponding to (the 2-skeleton of) the notion of simple cell complexes used here. The geometry associated to a spin foam is obtained by attributing volumes and areas to the simplices of the triangulations using the data from the colourings. These constructions are generalized in  \cite{bds11} and \cite{bby18} using the notion of polyhedra to define more general types of cells and this brings us to the last topic of this introduction.

\subsection{Combinatorial cell complexes in the context of discrete geometry} \label{secdiscgeom}

Making connections with existing works in discrete geometry or other areas of mathematics has not been our main concern, nevertheless we encountered a number of notions that are related to the framework we develop here and influenced our choice of terminology. As a first general remark, we point out that the notion of "face", particularly in the context of polytopes, has a more general meaning than the definition of face we use in our work, in which it designates only the sub-cells of rank one lower than a given cell (as opposed to any sub-cell).


As mentioned before, the starting point for our choice of a suitable discrete space to study a Poincaré-like duality is the definition of combinatorial cell complexes given by Basak in \cite{bas10}. This definition indeed has the advantage of not invoking any embedding in some manifold and this reflects well an idea underlying our work and also rooted in the physics of LQG: a vertex of a cell complex is often not best interpreted as a point in a manifold. In physical contexts related to space-time, it can be more illuminating to see a vertex as representing a region of space or "quanta of space" of a given dimension, having in particular a non-zero volume. One reason for adopting this interpretation in our framework comes from the principle that both the primal and dual theories should in principle give an equivalent description of the physics it models. As a consequence, a vertex would then be "coupled" to a maximal cell of the cell complex, representing a space-time region having a quantized volume. It would only be after defining some continuum limit, where the volume of each maximal cell tends to zero, that these two entities become equivalent and can be seen as points in a continuous space which, in "good cases", would have the structure of a manifold.
Basak's combinatorial cell complexes are studied from the point of view of finite topological spaces and we refer the reader to the introduction of his paper for more background and references on this topic.

A different notion of combinatorial cell complex is given by Aschbacher in \cite{asc96}. This definition seems to offer a more general treatment which appears in some respects too complicated for the purposes of our work.

Our framework being purely combinatorial and not based on the assumption of homogeneity, we did not knowingly make use of techniques or results from other discrete frameworks designed to tackle the problem of quantum space-time and based on the notion of manifold. We can for example mention the framework of "Quantum triangulations" given in \cite{cm17}.

Cell complexes are common to the field of piecewise linear topology, formulated using subspaces of the Euclidean space. In \cite{rs72} for example, cells are defined as subsets of $\R^n$ and in the definition of "ball complex" (on p.27), two cells can intersect on more than one cell. On this same page, and also in more details in \cite{bry02}, the notion of dual complex is defined. As explained in Section \ref{ssecbdiv}, this definition of dual differs from ours, although when studied as topological spaces embedded in $\R^n,$ the two definitions are essentially equivalent.


A definition of complexes that had an impact on our work is the notion of polyhedral complexes. Polytopes are one of the oldest topics  mathematics (see e.g. \cite{gru03} or \cite{zi95} for classic references on this topic) and although their definition is usually equivalent to a convex hull of a finite set of points in $\R^n,$ the combinatorial theory of polytopes represents an important part of the related research. So much so that there exists a theory of "abstract polytopes" based on a purely combinatorial definition using a poset with a rank function (see \cite{ms02} finor a reference on this topic). Abstract polyhedral complexes share features with the notion of combinatorial cell complexes and we discuss their relation in Section \ref{secpolytope}. In this section, we also present a definition of shellable cell complexes, which is directly borrowed from the theory of polytopes and essentially means that such cell complexes can be built by successively adding maximal cells while preserving the topology of a Euclidean ball, except possibly for the last step where one can obtain a sphere-like object. This notion was introduced by the mathematician Ludwig Schläfli to compute the Euler characteristic of polytopes and it is also used here as a sufficient condition in order to derive a version of the Euler-Poincaré formula for cell complexes.




The problem of polytope reconstruction, understood as determining the face lattice of a polytope from its $k$-skeleton (the faces of dimension $k$ and lower), gave rise to a number of results and we refer the reader to the review \cite{bay18} from Bayer for details on this topic. Notably, a well-known result (\cite{bm87}, \cite{kal88}) states that the face lattice of a simple polytope of arbitrary dimension is determined by its graph (i.e. its 1-skeleton). The proof relies on shellability of polytopes and so does another result using the 2-skeleton that applies to "almost simple" polytopes presented in the review from Bayer. The result presented here concerns the 2-skeleton of simple cell complexes of arbitrary dimension with no assumption on shellability and no equivalent result related to complexes of polytopes is known to the author. Our proof uses a notion of connection, which defines a mapping between the edges incident on neighbouring vertices. We recently found out about the notion of "$k$-frames" and "$k$-systems" defined in \cite{jkk02} that seem closely related. The theme of reconstruction is also present in the context of piecewise linear manifolds, where "n-graphs", are used to encode the face lattice of a triangulation (see e.g. \cite{lm85}). This notion was motivated by the Poincaré Conjecture and used for example in the work of Vince (\cite{vin85}). This last work inspired us to define the notion of 2-shellability for cell complexes and prove at the end of Section \ref{secpolytope} that, in the case of simple cell complexes, this notion is equivalent to shellability of the dual triangulation.

Diagrams of spaces, described for example in \cite{zz93}, allow to encode the construction of a topological space from basic building blocs using a functor. This can be considered as a way to dualize a given topological space, but in this case the dual has then a rather different structure than the original space. A central feature of our framework is that the dual of a cell complex is also a cell complex and our focus is on classes of cell complexes that are invariant under duality. This naturally incited us not to use the assumption of homogeneity in our framework, as in general this property is not preserved under taking the dual, as exposed in the last example of Section \ref{ssecexamdcob}. Diagrams of spaces are related to the notion of model category and it would be interesting to understand if this notion can be applied to the category we obtain at the end of this work in Section \ref{sseccatcausalcob}. 

Finally, we introduce in Section \ref{ssecredncol} a novel formulation of the notion of subdivision. A subdivision of a cell complex $K$ is defined as the domain of type of map called reduction with image $K.$ The duality map defined on closed cell complexes leads to a dual notion of a reduction map that we call collapse. The definitions of reductions and collapses are given using five conditions that are formulated as combinatorial relations on the cells. A subdivision in our context is more general than the usual examples of subdivisions given for example in \cite{bry02} as in particular it can introduce changes in the topology of the underlying cell complexes.

\subsection{Structure and main results of the thesis} \label{secresults}

The first chapter is dedicated to the introduction of general notions from our framework which are used throughout the thesis. We embark on the study of the duality operation by proving that it defines an involution on closed cell complexes (cc). This result exposes a structure of proof used repeatedly to introduce more complicated instances of cell complexes constructed from the cells of a given cc: we start with a lemma establishing an equivalence between the inclusion relations of cells in the newly defined cc and the inclusion relations of certain cells from the cc given for the construction and the main proof consists in checking the axioms for the new cc using the lemma. In the later part of Chapter \ref{secgennot}, we discuss the notion of barycentric subdivision and adapt a result given by Bayer (\cite{bay88}) to show that a cc is determined by its barycentric subdivision up to duality. We close the chapter by briefly introducing the topology and geometric realization of a cc, although these do not play a central role in what is developed in the next chapters.

The central task of Chapter \ref{secreconst} is to establish under which conditions one can reconstruct a simple complex from its 2-skeleton, a question raised for example by the 2-complexes used in spin foam models. Our result, given by the Corollary \ref{correconst1} of Section \ref{ssecreconstthm}, involves assumptions called evenness, fullness, monodromy-freedom and simple-connectedness.  All these notions are therefore introduced in the preceding sections of this second chapter with the exception of Section \ref{secpolytope}, where we make connection with the notions of polytopes and shellability.

Our main result of Chapter \ref{secdualcob} in particular generalizes the notion of dual to a class of cell complexes with boundary. In order to preserve a framework where the duality acts as a map on a set of discrete spaces, we are lead to introduce the notion of relative cc and to interpret the duality as acting on a discrete notion of cobordisms. We start by introducing midsections, before giving the definitions of specific types of relative cc used to define cobordisms and formulate the main results in Section \ref{ssectiondcob}. The last two sections of this chapter are devoted to the discussion of certain examples of cobordisms and their duals using illustrations and give the proof of the main theorem on the duality of cobordisms.

The last chapter of this thesis starts by introducing the notions of reduction and collapse. An important result about these maps is that they induce partial orders on the set of pure cell complexes. We close this first section by discussing the notions of relative reduction and relative collapse designed for the study of cobordisms.
Section \ref{ssectrans} introduces transitions, a poset constructed from the cells in a boundary component of a cobordism. This poset is used as a way to formulate a notion of uniformity for cobordisms. Uniformity constitutes an important assumption underlying the definition of the category of causal cobordism given at the end of this chapter.
The notions of compatibility and orthogonality between two poset homomorphisms are introduced in Section \ref{sseccompnorthhom} as a condition on the case where these poset homomorphisms have a common domain and codomain, respectively.
In Section \ref{ssecslice}, we introduce slice sequences, composed of reductions and collapses using the compatibility and orthogonality conditions. The core of this section is devoted to a particular type of cobordism, called slice, and to its characterization using slice sequences. This constitutes the main result of this fourth chapter that we refer to as the "Correspondence Theorem". 
We then introduce connecting sequences in Section \ref{sseccompcob}, the dual notion of slice sequences, on which the definition of the composition of cobordisms is based, formulated as a union of cell complexes. 
We close this chapter by introducing the category of causal cobordisms in Section \ref{sseccatcausalcob} constructed from connecting sequences and slice sequences. This category exhibits two other dualities arising from the causal structure analogous to a type of "time reversal" and "parity transformation". With such an interpretation, the duality map we studied in this work can be compared to a "charge conjugation".

\newpage
\vspace{3cm}

\section{General definitions and terminology}\label{secgennot}

\vspace{2cm}

In this first chapter, we lay down the basic definitions used in this work and present a number of simple facts related to this framework. We begin by giving some notations and a list of symbols appearing in this work and then turn to the introduction of cell complexes in Section \ref{ssecccc}. The Section \ref{ssecgraphncc} focuses on the terminology related to graphs and paths that are used particularly in Chapter \ref{secreconst}. In Section \ref{secdualncc}, we define the dual of closed cell complexes, show that the duality map is an involution on closed cell complexes and introduce the notions of strongly connected and non-pinching cell complexes. In Section \ref{ssecbdiv}, we introduce a combinatorial notion of barycentric subdivision and adapt a result from M. Bayer. This result states that the barycentric subdivision characterises the cell complex it subdivides up to duality, which has some implications for the discussion of the results on reconstruction of Chapter \ref{secreconst}. We use the barycentric subdivision in the last Section \ref{sectopoconsid} where we briefly describe how to associate a topology to a cell complex, although topological considerations do not play a central role here.

\subsection{Notation}

In these notes $\N = \{0,1,2, \dots\}$ are the natural numbers and $\N^* = \{1,2,\dots \}.$ We usually use the letters $i,j,k,n,r$ to denote indices in $\N.$ We use $:=$ to denote an equality where the left hand side is defined by the right hand side and we use italic words whenever we \textit{define a notion for the first time.}

Let $(S, \leq)$ be a \textit{partially ordered set} or \textit{poset}, which means that $\leq$ is a reflexive antisymmetric transitive binary relation on the set $S.$ For two elements $x,y \in S,$ we say that $x$ is \textit{below} $y$ or $y$ is \textit{above} $x,$ written $x<y$ if $x \leq y$ and $x \neq y.$ An element in $S$ is \textit{maximal} if it is not below any other element and \textit{minimal} if it is not above any other element. A \textit{totally ordered subset} of a poset $S$ is a sub-poset $T$ of $S$ such that for all $x,y \in T,$ either $x \leq y$ or $y \leq x.$ In this case we say that $\leq$ is a total or linear ordering of $T.$ We denote by $\{x_0 < \dots < x_k \}$ the finite totally ordered set with elements $x_i$ such that $x_i < x_j$ if and only if $i<j.$ If $T \subset S,$ is a subset of $S,$ an \textit{upper bound of $T$} is an element $a \in S$ such that $b \leq a$ for all $b \in T$ and a \textit{lower bound of $T$} is an element $a \in S$ such that $a \leq b$ for all $b \in T.$ If there is a unique upper bound of $T$ which is below all upper bounds of $T$ then this element is called the \textit{least upper bound of $T$} and is denoted by $\bigvee T$. Similarly if there is a unique lower bound of $T$ which is above all lower bounds of $T$ then this element is called the \textit{greatest lower bound of $T$} and is denoted by $\bigwedge T.$ We use the notation $\abs{S}$ for the number of elements in a set $S.$  The set of subsets of $S$ is denoted by $\pws{S}$ and the set of (unordered) 
2-element subsets in $S$ by $\ts{S} \subset \pws{S}.$ \textbf{The disjoint union of the sets $S$ and $S'$ is denoted by $S \sqcup S'$ and is used purely to emphasize that $S$ and $S'$ are disjoint.} The set $S \setminus S'$ denotes the difference between $S$ and $S',$ i.e. the elements in $S$ that are not in $S'.$
A \textit{map} $f: A \dans B$ is a function with domain equal to $A$ and the composition of two maps $f:A \dans B,$ $g: B \dans C$ is denoted by $g \circ f.$ We denote by $f^{-1}$ the inverse image under $f$ seen as a function from $\pws{B}$ to $\pws{A}.$ If $f^{-1}(\{x\}) = \{y\}$ we will often simply write $f^{-1}(x) = y.$

\newpage

\paragraph{List of symbols with [the number or location of their definition]}

\begin{itemize}[label= $\cdot$,topsep=2pt, parsep=2pt, itemsep=1pt]
\item $\face{x},$ $\cface{x}:$ set of faces and cofaces of a cell $x$ [ \ref{defccc}]
\item $\kr, ~~ K^{(r)}:$ set of $r$-cells and $r$-skeleton of $K$ [ Section \ref{ssecccc}]
\item $J \leq K:$ sub-cell complex [ Section \ref{ssecccc}]
\item $K \cap A:$ restriction of $K$ to $A \subset \kz$ [ Section \ref{ssecccc}]
\item $K \cong J:$ isomorphic cc [ Section \ref{ssecccc}]
\item $\V,$ $\E, \Ne_v:$ set of resp. vertices, edges and neighbours of a vertex $v$ in a graph $\G$ [ Section \ref{ssecgraphncc}]
\item $p: v \pathto w:$ path from $v$ to $w$ in a graph [ Section \ref{ssecgraphncc}]
\item $l_C(v):$ loop associated with a connected component of a 2-cell [\ref{defiloops2c}]
\item $[v_1 \dots v_k]:$ representation of a connected component of a 2-cell [\ref{lemrepr2cell}]
\item $\partial K:$ boundary of a non-singular cc $K$ [ Section \ref{ssecccc}]
\item $\dual{A}:$ dual of the set $A$ [\ref{defdualset}]
\item $\dg{K}:$ dual graph [\ref{defdualgraph}]
\item $\bdiv{K}:$ barycentric subdivision [\ref{defbdiv}]
\item $A(x),$ $B(x):$ set of cells respectively above and below $x$ [ Section \ref{ssecbdiv}]
\item $\E_v^x, \E_v^\A:$ set of edges incident to $v$ contained in one or several cells [ intro Chap. \ref{secreconst}]
\item $\ksec{x}{y}:$ section in a c.c between the cells $x$ and $y$ [ Section \ref{secpolytope}]
\item $\nabla:$ connexion [\ref{nabla}], $\nabla_p$ defined on a path [\ref{definablapaths}]
\item $\comp{p}{C}:$ complementary path in a connected component of a 2-cell [ Section \ref{ssectionsimplyconccc}]
\item $m_C^p,$ $m_{e,v} \in \M:$ moves of paths and set of moves [\ref{defmovespath}]
\item $\A_B$ ($\A_B^C$): collar of a set of cells $\A$ into another set $B$ (restricted to the cells in $C$) [\ref{defcollar}]
\item $\nsc$ ($\nsc^R$): set of non-singular non-pinching cc (of rank $R$) [ Section \ref{susecmidsection}]
\item $\mlc$ ($\mlc^R$): set of closed non-pinching cc (of rank $R$) [ Section \ref{susecmidsection}]
\item $(K - J) \in \cob:$ set of cobordisms[ Section \ref{ssectiondcob}]
\item $\bual{A}:$ generalized dual set [\ref{defbdualset}]  
\item $J \red K:$ reduction [\ref{defreduction}]
\item $J \col K:$ collapse [\ref{defcollapse}]
\item $K_\sphi, K^\sphi:$ for $K \geq J \red_\sphi L$ or $J \red_\sphi L \leq K$ [\ref{defikphi}]
\item $J(K):$ for $J \leq K,$ transition [\ref{defreducedbdry}]
\item $K_J(A), K_J[A]:$ set of cells in $K_J$ containing $A \subset J^{[0]}$ or intersecting $J^{[0]}$ on $A$ [ Section \ref{ssectrans}]
\item $\sphi_K^J, \cphi_K^J:$ reduction and collapse associated to a transition [\ref{proptransit}]
\item $j \cpa l:$ compatible cc-homomorphisms [\ref{defcoorthom}]
\item $j\refl l,$ $j \perp l:$ reflective and orthogonal cc-homomorphisms [\ref{defrefnorth}]
\item $(\phi_J \sqcup \phi_L)(M):$ augmented poset associated to $\phi_J$ and $\phi_L$ [\ref{lemsqcup}]
\item $\nsc_s:$ set of slices in $\nsc$ [above \ref{defsliceseq}]
\item $ J \red J' \loc M \col L' \der L:$ slice sequence [\ref{defsliceseq}]
\item $ M \col J \der I \red L \loc M:$ connecting sequence [\ref{defconnseq}]
\item $K \ccup{\sphi_J ~ \sphi_L} K':$ union of cc in $\nsc$ [\ref{defunioncc}]
\item $M \col \der J:$ semi-sequence [Section \ref{sseccatcausalcob}]
\end{itemize}

\subsection{A combinatorial definition of cell complexes (cc)} \label{ssecccc}
 
In this section, we start by introducing the definition of combinatorial cell complexes. Throughout this work, \textbf{we will drop the term "combinatorial"} and simply call such objects cell complexes, abbreviated cc, as we focus on a purely combinatorial treatment. 
In \cite{bas10} Basak gives a definition of combinatorial cell complexes using a set of four axioms on a poset with a rank function. We formulate an equivalent definition that makes it clear from the beginning that every cell is defined as a subset of a set of vertices and that the intersection of two cells is a cell or is empty. We give simple illustrations related to the axioms in Figure \ref{axiomtwopic}. The reader is invited to look at the illustrations given later in this chapter or in the Figures \ref{figdualbitetrahedra}, \ref{figdualcylinder} and \ref{figtorus} of Section \ref{ssecexamdcob} to find more elaborated examples of cell complexes. 


\begin{defi}[Cell complex (cc)] \label{defccc}
Let $S$ be a finite set and let $K \subset \pws{S} \setminus \{\emptyset\} $ together with a \textit{rank function} \mbox{$\rk = \rk_K: K \rightarrow \N$} such that $S = \rk^{-1}(\{0\}).$ We define $ \kr := \rk^{-1}(\{r\})$ for $r \in \N.$ We call an element of $K$  a \textit{cell} and an element of $\kr$ an \textit{r-cell.} The set of \textit{faces} of $x \in K$ is defined by
$$ \face{x} := \{ y \in K ~|~ y \subset x, ~\rk(y) = \rk(x) -1\}$$
and the set of \textit{co-faces} of $x$ is defined by 
$$ \cface{x} :=\{ y \in K ~|~ x \subset y, ~ \rk(y) = \rk(x) +1\}.$$
The data $(K, \rk)$  defines a \textit{cell complex,} or cc for short, if it satisfies the following four \textit{axioms}:
\begin{enumerate}[label=\textbf{\roman*}),topsep=1pt, parsep=2pt, itemsep=1pt]
\item \label{cccrank}$\rk$ is strictly compatible with the partial order $\subset:$ for all $x,y \in K,$ if $x \subsetneq y$ then $\rk(x) < \rk(y),$
\item \label{cccinter} $K \cup \{ \emptyset \}$ is closed under intersection: for all $x,y \in K,$ $x \cap y \in K \cup \{\emptyset\}.$
\item \label{cccenough} there is no gap in the value of the rank: for all $x,y \in K,$ if $x \subsetneq y$ then $x$ is a face of a cell contained in $y,$
\item \label{cccdiamond} $K$ satisfies the \textit{diamond property}: for all $x,y \in K,$ if $x \subset y$ and $\rk(x) = \rk(y) - 2$ then there are exactly two cells $z_1,z_2$ such that
$$ \cface{x} \cap \face{y} =\{z_1,z_2\}.$$ 
\end{enumerate}

\begin{figure}[!h]
\centering
\includegraphics[scale=0.4]{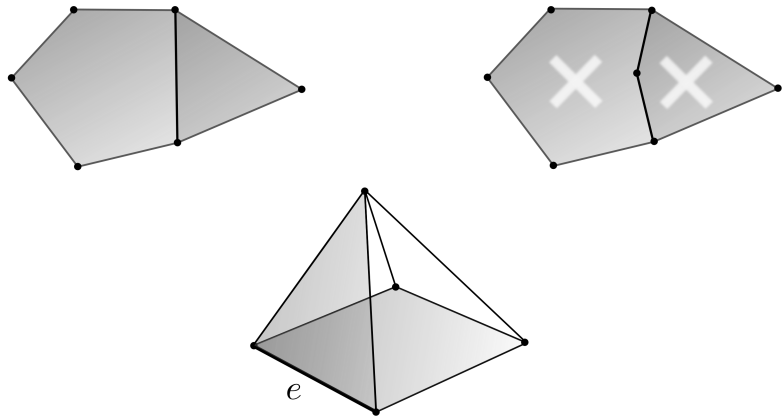}
\caption{\label{axiomtwopic} The top-left picture gives an example of a 2-cc. The top-right pictures is a non-example of a 2-cc as Axiom \ref{cccrank} and \ref{cccinter} imply for example that the intersection of two 2-cells is a cell of rank 1 or lower and hence it cannot contain two (or more) edges. The lower picture consists in one 3-cell defining an example of 3-cc. The diamond property \ref{cccdiamond} implies that the edge $e$ is included in exactly two 2-cells, depicted in gray.}
\end{figure}

\end{defi}

Let $K$ be a cc and let $R \in \N$ be the minimal integer such that $K = \bigcup_{r=0}^{R}\kr.$ $R$ is called the \textit{rank} or \textit{dimension} of $K,$ noted $\Rk(K),$ and we call $K$ a $R$-cc. For $0 \leq k \leq R,$ we define the \textit{$k$-skeleton} of $K$ to be 
$$K^{(k)} := \bigcup_{r=0}^{k} \kr,$$
which is clearly a $k$-cc. Two cells $x$ and $y$ are called \textit{incident} if either $x \subset y$ or $y \subset x.$ We call $0$-cells \textit{vertices} and $1$-cells \textit{edges}. We will sometimes call a 2-cell a \textit{triangle} if it contains 3 vertices or a \textit{square} if it contains 4 vertices. Similary, it will be convenient to use other common names such as a \textit{tetrahedron} for a 3-cell containing four vertices when describing illustrations.

Let $A \subset \kz$ then the poset $$K \cap A := \{x \in K ~|~ x \subset A \},$$ called the \textit{restriction} of $K$ to $A,$ is a cc. For example, every cell $x$ in $K$ defines an underlying cc, namely $K \cap x.$

A cc $K$ is said to be  \textit{pure} or \textit{equidimensional} if every maximal cell in the poset $(K, \subset )$ has rank $R$ and in this case we define the \textit{boundary} of $K$, noted $\partial K$ to be the cells in $K^{(R-1)}$ contained in a $(R-1)$-dimensional cell, also called a \textit{sub-maximal cell}, having only one maximal cell above itself. In other words $$\partial K := \bigcup_{\substack{y \in K^{[R-1]} \\ \abs{\cface{y}} = 1}} K \cap y.$$

The only axiom from Definition \ref{defccc} which is not self-explanatory is the diamond property \ref{cccdiamond} (named after the shape $\diamond$ formed by the face lattice obtained from the four cells involved in the statement). As we will see later in Lemma \ref{rembdrycc}, this axiom can be understood as a condition that in particular implies $\partial(\partial K) = \emptyset.$
Similarly to the terminology used for polytopes, we will sometimes call \textit{facets} the maximal faces of a cc. We say that a cc $J$ is a \textit{sub-cc} of $K$ and write $J \leq K$ if $J \subset K$ and $\rk_J = \rk_K|_J.$

A map $\phi : J \dans K$ is a \textit{cc-homomorphism} if it is a poset homomorphism, i.e. $\phi(x) \subset \phi(y)$ if $x \subset y,$ such that for all $y \in \phi(J),$ there exists $x \in J^{[\rk_K(y)]}$ such that $\phi(x) = y.$ A \textit{cc-isomorphism} $\phi$ is a bijective cc-homomorphism such that $\phi^{-1}$ is a cc-homomorphism. Note that the inverse of bijective poset isomorphism is not always a poset homomorphism, a counter-example is given in Figure \ref{figexccchom}. Two cc $J$ and $K$ are \textit{isomorphic}, written $J \cong K,$ if there exists a cc-isomorphism between them. A cc-isomorphism $\phi : J \dans J$ is called an \textit{cc-automorphism} of $J.$

\begin{figure}[!h]
\centering
\includegraphics[scale=0.47]{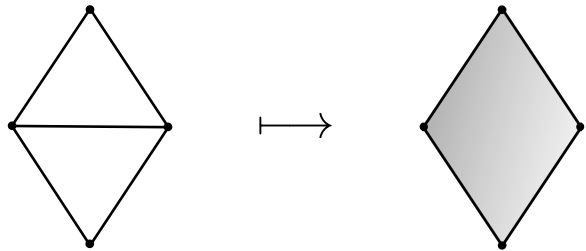}
\caption{\label{figexccchom} Example of a bijective poset homomorphism from a cc of rank 1 to a cc of rank 2 such that the inverse is not a poset homomorphism. The outer edges from the 1-cc are "mapped to themselves" and the middle edge is mapped to the square.}
\end{figure}

\begin{lem}\label{lemisom}
Let $K,J$ be cc and $\phi: J \dans K$ be a poset homomorphism. The following three statements are equivalent.
\begin{enumerate} [label=\arabic*),topsep=2pt, parsep=2pt, itemsep=1pt]
\item \label{lemiso1}$\phi$ is a cc-isomorphism.
\item \label{lemiso2}$\phi$ is bijective and $\phi(x) \in K^{[\rk_J(x)]}$ for all $x \in J.$
\item \label{lemiso3}$\phi$ is bijective and $\phi^{-1}$ is a poset homomorphism.
\end{enumerate}
\begin{proof}
\ref{lemiso1} implies \ref{lemiso2} since we have $(\phi^{-1})^{-1}(x) = \{ \phi(x) \}$ and $\phi^{-1}$ is a cc-homomorphism, therefore it implies that $\rk_K(\phi(x)) = \rk_J(x).$ The reverse implication is clear.

\ref{lemiso3} implies \ref{lemiso2} by the following argument. If $x_r \in J$ is a $r$-cell then by Axioms \ref{cccrank} and \ref{cccenough} for $J,$ there exists $x_0, x_1 \dots, x_{r-1}$ in $J$ such that $\rk_J(x_i) = i$ and $$x_0 \subsetneq x_1 \subsetneq \dots \subsetneq x_{r-1} \subsetneq x_r.$$
By Axiom \ref{cccrank} for $K$ and since $\phi$ is a bijective poset homomorphism, the previous inclusions imply
$$ \phi(x_0) \subsetneq \phi(x_1) \subsetneq \dots \subsetneq \phi(x_{r-1}) \subsetneq \phi(x_r)$$
and therefore $\rk_K(\phi(x_i)) \geq i + 1$ for all $0 \leq i \leq r.$ 

We have that $\rk_K(\phi(x_0)) = 0$ by the following argument. If $\rk_K(\phi(x_0)) > 0$ then by Axiom \ref{cccenough} for $K,$ there exists $v \in \kz$ such that $v \subsetneq \phi(x_0).$ Since $\phi^{-1}$ is a poset homomorphism, the latter implies that $\phi^{-1}(v) \subset \{x_0\},$ hence $\phi(x_0) = v$ and this contradicts $\rk_K(\phi(x_0)) > 0.$ 
Moreover, if we assume that there exists $1 \leq i \leq r$ such that $\rk_K(\phi(x_i)) > i + 1$ then, by Axiom \ref{cccenough} for $K,$ there exists $y \in K$ such that $\phi(x_{i-1}) \subsetneq y \subsetneq \phi(x_i).$ Since $\phi^{-1}$ is a bijective poset homomorphism, this implies $x_i \subsetneq \phi^{-1}(y) \subsetneq x_{i+1},$ a contradiction with the Axiom \ref{cccrank} of $J.$
Therefore $\rk_K(\phi(x_i)) + 1 = \rk_K(\phi(x_{i+1}))$ for all $1 \leq i \leq r$ and $\rk_K(\phi(x_r)) = r$ and this proves the point \ref{lemiso2} as we assumed that $x_r \in J^{[r]}.$

\ref{lemiso1} clearly implies \ref{lemiso3} and this concludes the proof.
\end{proof}
\end{lem}

The following remark is an observation from \cite{bas10} which we will refer to later in Chapter \ref{secdualcob}, where we will in effect find a way around this fact.

\begin{rema}[Remark 3.1 in \cite{bas10}]\label{remcobdry}
An edge of a cc $K$ has more than one vertex as the definition of cc asserts that all singletons are vertices (i.e. have rank $0$), so the \textit{co-boundary} of $K,$ defined as the set of cells containing a 1-cell made of only one vertex, is empty.
\end{rema}

The next remark is an observation that allows to better compare our definition of cc with the definition in \cite{bas10} and which will be used later in Chapter \ref{seccausalstruct} when discussing the notion of collapse.

\begin{rema}\label{remdefccc}
The definition of cc in \cite{bas10} demands the additional condition that if $\rk(x) \geq 1$ then $ x = \bigvee \face{x},$ where the least upper bound $\bigvee$ is taken in the poset $(K, \subset).$ This condition is always satisfied in our context since cells are defined as subsets of a vertex set. Indeed if $y$ is an upper bound of $\face{x}$ then in particular for every vertex $v \in x$ there is a face  $x' \in \face{x}$ such that $v \in x' \subset y$ and therefore $x \subset y.$ So $x$ always fulfils the criteria to be the least upper bound of $\face{x}.$ As a consequence, a cell of non-zero rank is always equal to the union of its faces. In summary, if $\rk(x) \geq 1$ then
\begin{equation}\label{relcface}
x = \bigvee \face{x} = \bigcup_{y \in \face{x}} y.
\end{equation} 
\end{rema}

The definition of cc allows edges to have more than two vertices. Adding the requirement that edges have exactly two vertices constitutes the notions of graph-based cc. Adding what will later be understood as the dual of the latter requirement gives the notion of non-singular cc introduced next. 

\begin{defi}[Graph-based, non-singular and closed cc]
A cc $K$ of rank $R$ is said to be \textit{graph-based} if every edge of $K$ has (no more than) two vertices. $K$ is \textit{non-singular} if it is graph-based, pure and \textit{non-branching,} i.e. every $(R-1)$-cell is a face of at-most two $R$-cells. 

$K$ is said to be \textit{closed} if it is non-singular and has empty boundary, i.e. every $(R-1)$-cell is contained in exactly two $R$-cells.
\end{defi}

For example, as a consequence of the diamond property \ref{cccdiamond}, if $K$ is a non-singular cc then for all $x \in K$ the cc $\partial(K \cap x)$ is closed. In \cite{bas10}, what we call a closed cc is called a \textit{manifold-like} cc, but we prefer the former term as it better encapsulates the idea of having an empty boundary and is not related to the notion of manifold. If $K$ contains a sub-maximal cell $y$ such that $\abs{\cface{y}} > 2$ then we say that $K$ is \textit{branching} and that $y$ is a \textit{($K$-)branching}.

The following technical lemma is used repeatedly in this work and, in a sense that will become clear in the next section, shows when the property dual to (\ref{relcface}) is verified.

\begin{lem}\label{lemcface}
Let $x$ be a cell of a cc $K$ such that $\abs{\cface{x}} \geq 2$ then $$x = \bigwedge \cface{x} = \bigcap_{y \in \cface{x}} y.$$
\begin{proof}
Since $\abs{\cface{x}} \geq 2$ Axiom \ref{cccinter} implies $x' :=  \bigcap_{y \in \cface{x}} y \in K.$ If $y'$ is a lower bound of $\cface{x}$ then $y' \subset y, ~\forall y \in \cface{x}$, hence $y' \subset x'.$ The latter implies $$x' = \bigcap_{y \in \cface{x}} y = \bigwedge \cface{x}.$$ We also have that $x \subset x'$ since $x$ is a lower bound of $\cface{x}.$ If by contradiction $x' \neq x$ then by \ref{cccrank} $\rk(\bigwedge \cface{x}) > \rk(x),$ hence $\rk(\bigwedge \cface{x}) = \rk(x) + 1.$ But if $y \in \cface{x}$ then $ \bigwedge \cface{x} \subset y$ and $\rk(y) = \rk( \bigwedge \cface{x})$ so that $y = \bigwedge \cface{x}$ by \ref{cccrank}, which implies $\abs{\cface{x}} = 1$ and contradicts $\abs{\cface{x}} \geq 2.$
\end{proof}
\end{lem}

Introducing a greatest lower bound $\bigwedge$  raises the question of whether it exists at all. This next remark therefore helps us clarify this point in Chapter \ref{seccausalstruct}, when discussing the notions of reduction and collapse.

\begin{rema} \label{remcface}
A consequence of the proof of Lemma \ref{lemcface} is that $\bigwedge \cface{x}$ exists for all cell $x$ in a cc $K$ and 
\begin{align*}
\bigwedge \cface{x} = \begin{cases} x, &\text{if } \abs{\cface{x}} \geq 2,\\
y \in \cface{x}, &\text{if } \cface{x} = \{y\},\\
\emptyset &\text{if } x \text{ is maximal.}
\end{cases}
\end{align*}
\end{rema}

\subsection{Simplicial complexes, graphs, paths and local cc} \label{ssecgraphncc}

In this section, we first introduce simplicial complexes and graphs. We then introduce a number of definitions and some terminology related to paths that in particular allows us to define the notion of connected, cell-connected and local cc. We  end by a lemma we use to introduce the notion of representation of a cell. These definitions are used especially in the next Chapter \ref{secreconst}, whose main goal is to prove our results on reconstruction.

The definition of simplicial complex presented now, sometimes called "abstract simplicial complex" in other contexts, constitutes a particular case of cc and encompasses the notion of graph introduced afterwards.

\begin{defi}[Simplicial complex]
Let $S$ be a finite set. A collection $K \subset \pws{S}$ is called a \textit{simplicial complex} if it satisfies the two following properties:
\begin{enumerate}[label=\textbf{\alph*}),topsep=2pt, parsep=2pt, itemsep=1pt]
\item If $ x \subset y \in K$ then $x \in K,$ \label{simpcplxa}
\item every one element subset (vertex) of $S$ is in $K.$ \label{simpcplxb}
\end{enumerate}
An element $x \in K$ is called a \textit{simplex} and we define the rank function of $K$ by
$$ \rk_K(x) = \abs{x} - 1.$$
\end{defi}

Conditions \ref{simpcplxa} and \ref{simpcplxb} characterizing simpicial complexes clearly imply the four axioms of cell complexes, hence our choice of definition: it is sufficient to check these two conditions to establish whether a given collection of subsets of a set is a simplicial complex. We could have also defined a simplicial complex to be a cc $K$ with rank function $\rk_K$ satisfying condition \ref{simpcplxa}. In this case, as a consequence of Axiom \ref{cccenough}, the rank function satisfies $\rk_K(x) = \abs{x} - 1.$

The notion of graph used in our context can naturally be defined as follows.

\begin{defi}[Graph]
A graph $\G$ is a simplicial complex of rank 1. We denote by $\V$ the set of vertices of a graph and $\E$ its set of edges.
\end{defi}

In particular a graph $\G$ is a 1-cc where we use the notation $\V = \G^{[0]}$ and $\E = \G^{[1]}$ and there cannot be multiple edges between two vertices nor self-edge (linking a vertex with itself) in a graph. As expected, if a cc $K$ is graph-based then $\kg$ is a graph and we will sometimes use the notation $\E$ for $\ko,$ the set of edges of $K,$ when there is no ambiguity on the cc $K.$
        
Let $\G$ be a graph. For $v \in \V,$ we define $$\Ne_v := \{ w \in \V ~|~ \{v,w\} \in \E\}$$ the set of \textit{neighbours} of $v$ in $\G$ and $$\E_v := \{ e \in \E ~|~ v \in e \}$$ the set of \textit{edges of $\G$ incident to $v.$} We define the \textit{degree} of a vertex $v$ to be $$ \deg(v):= \abs{\E_v} = \abs{\Ne_v}.$$ More generally, if $A \subset \V,$ we will sometimes use the notation $$\E_A :=\{ e \in \E ~|~ \abs{A \cap e} = 1\}.$$

We will often use paths in Chapter \ref{secreconst}, where for example in Section \ref{ssectionsimplyconccc} we define a notion of homotopy for paths in cell complexes of rank higher or equal to 2. For this, we chose to use the following definitions related to paths in a graph $\G$ and, by extension, in any graph-based cc. 

Let $k \in \N$ and let $v_0,v_k \in \V.$ We say that $p \in (\E)^k$ is a \textit{path in $\G$ from $v_0$ to $v_k$} of \textit{length} $ \abs{p} := k$ if $p = (e_1, \dots, e_k)$ and there exists $v_1, \dots v_{k-1}$ such that $e_i = \{v_{i-1}, v_i\}$ for all $1 \leq i \leq k.$ In this case we sometimes write $p : v_0 \pathto v_k.$ We conventionally denote by $\emptyset_v$ the \textit{empty path} of length 0 from $v$ to $v.$ We say that an edge $e$ is contained in $p$ and write $e \in p$ if $e = e_i$ for possibly multiple $1 \leq i \leq k.$ We define the vertices of the path $p$  to be $p^{[0]} = \{v \in \V ~|~ \exists e \in p, ~ v \in e\}$ and its edge to be $p^{[1]} = \{e \in \E ~|~ e \in p\}.$ The \textit{inverse} of $p$ is defined by $p^{-1} =( e_k, e_{k-1}, \dots, e_1),$ which is a path from $v_k$ to $v_0.$ A path $p'$ is called a \textit{sub-path} of $p,$ noted $p' \leq p,$ if it is a path such that $(p')^{[1]} \subset p^{[1]}.$ A path from $v$ to $v$ is called a \textit{cycle based at $v$}. We say that a path of length $k$ is \textit{simple} if $\abs{p^{[0]}} = \abs{p} + 1,$ which implies that the path $p = (e_1, \dots, e_k)$ "does not pass through each $v \in p^{[0]}$ more than once": 
$$\abs{\{1 \leq i \leq k ~|~ v \in e_i \}} \leq 2 \quad \forall v \in p^{[0]},$$
and that $p$ "never uses an edge more than once": 
$$\abs{\{ 1 \leq i \leq k ~|~ e = e_i \}} \leq 1 \quad \forall e \in p^{[1]}.$$
A simple cycle is called a \textit{loop.}

If $p_1$ is a path from $v_0$ to $v_1$ and $p_2$ is a path from $v_1$ to $v_2$ then we denote by $p_1 \conc p_2$ the \textit{concatenation} of the paths $p_1$ and $p_2$, which is the path from $v_0$ to $v_2$ composed of the edges of $p_1$ followed by the edges of $p_2.$ The concatenation of paths is thereby associative and has \textit{the empty paths $\emptyset_v$ based at $v \in \V$} as neutral elements, i.e. if $p$ is a path from $v$ to $v'$ then $p \conc \emptyset_{v'} = \emptyset_v \conc p = p.$ We set $\emptyset_v^{-1} = \emptyset_v$ by convention.

A graph is \textit{connected} if any two of its vertices can be linked by a path. For our purposes, we choose to define connectedness for cc in the following way, implying in particular that a connected cc is graph-based.

\begin{defi}[Connected and cell-connected cc] \label{defconnccc}
A cc $K$ is \textit{connected} if it is graph-based and $\kg$ is connected. $K$ is said to be \textit{cell-connected} if it is graph-based and $K \cap x$ is connected for all $x \in K.$
\end{defi}

The following notion of locality for cc is often a case of focus in this work, hence we also chose an abbreviation for local cc. 

\begin{defi}[Local cell complex (lcc)]
A \textit{local cell complex,} abbreviated lcc, of dimension $R$ or \textit{$R$-lcc} is a connected cell-connected $R$-cc 
\end{defi}

It will be convenient to have the following notion of representation of a connected component of a 2-cell in a graph-based cc, in particular when using what we will call the loop associated to a connected component of a 2-cell defined afterwards. Not assuming locality in this lemma allows us to formulate some results of Chapter \ref{secreconst} in more generality, however thinking of $C$ in the next lemma simply as a $2$-cell will be sufficient for most of this work.

\begin{lem}[Representation of a 2-cell] \label{lemrepr2cell}
Let $K$ be a graph-based $R$-cc with $R \geq 2.$ If $C$ is a connected component of 2-cell of $K$ then there exits $v_1,\dots, v_k \in \kz$ such that $C= \{v_1, \dots, v_k \}$ and such that $\{v_i, v_{i+1}\},$  $1 \leq i \leq k-1$ and $\{v_k,v_1\}$ are the only edges of $\kg \cap C.$ In this case we write $C = [v_1\dots v_k]$ and call $[v_1 \dots v_k]$ a representation of $C.$
\begin{proof}
It is sufficient to prove the claim when $C \in \kt$ is connected. The diamond property \ref{cccdiamond} implies that if $C \in \kt$ and $v \in C$ then 
\begin{equation}\label{equlemrepr2cell}
\abs{\Ne_v \cap C} = 2.
\end{equation} 
Let $v_2 \in C$ and let $v_1, v_3 \in \Ne_{v_2} \cap C.$ If $v_1 \in \Ne_{v_3} \cap C$ then $C = \{v_1,v_2,v_3\}$ and 
$$\{\{v_1,v_2\},\{v_2,v_3\}, \{v_1,v_3\}\} \subset \ko,$$ is the set of all the edges in $\kg \cap C.$ If $v_1 \notin \Ne_{v_3} \cap C$ then, since $\abs{\Ne_{v_3} \cap C} = 2,$ there exists $v_4 \in \kz \setminus \{v_1,v_2,v_3\}$ such that $\Ne_{v_3} \cap C = \{v_2,v_4\}.$ By repeating this argument, if at step $j$ we have $v_i \notin \Ne_{v_{j-1}} \cap C$ for all $1 \leq i \leq j-3,$ we can add $v_j \in \kz \setminus \{v_1, \dots, v_{j-1}\}$ to the set $\{v_1, \dots, v_{j-1}\} \subset C$ such that $\Ne_{v_{j-1}} \cap C = \{v_{j-2}, v_j\}.$ Since $\kz$ is finite, this process must terminate after a finite number $k$ of steps such that for all $4 \leq j \leq k$ we have $\{v_{j-1},v_{j}\} \in \ko$ by step $j$ and $\{v_k, v_l\} \in \ko$ by the step $k,$ for some $1 \leq l \leq k-2.$ If $1 < l \leq k-2$ then $\{v_{l-1}, v_{l+1}, v_k\} \subset \Ne_{v_l} \cap C,$ which is a contradiction with (\ref{equlemrepr2cell}), therefore $l = 1.$ The same contradiction occurs if there exists an edge in $\kg \cap \{v_1, \dots, v_k\}$ different than $\{v_j,v_{j+1}\}$ for $1 \leq j \leq k-1$ or $\{v_k, v_1\}.$ Moreover, we have $C = \{v_1, \dots, v_k\}$ since if $v \in C \setminus \{v_1, \dots, v_k\}$ then by the same argument $\{v,v_j\} \notin \ko$ for all $1\leq j \leq k$ which is in contradiction with the fact that $\kg \cap C$ connected.
\end{proof}
\end{lem}

\begin{defi}[Loops associated with a connected component of a 2-cell] \label{defiloops2c}
Let $K$ be a graph-based $R$-cell complex with $R \geq 2.$ A representation $[v_1\dots v_k]$ of a connected component $C$ of a 2-cell induces the definition of $2k$ loops $l_C(v_j), l_C^{-1}(v_j)$ $1 \leq j \leq k$ of length $\abs{C}$ based at $v_j$ and defined by $$l_C(v_j) := (\{v_j, v_{j+1}\}, \dots, \{v_{j-1}, v_j\}),~~l_C^{-1}(v_j) := (\{v_j, v_{j-1}\}, \dots, \{v_{j+1}, v_j\}).$$
\end{defi}

\subsection{Duality map for closed cc} \label{secdualncc}

In this section we first introduce the notion of duality map at the center of the focus in this work. We first define it on closed cc and then generalize on what can be considered as more general cc, called cobordisms, in Chapter \ref{secdualcob}.
We start by defining the general notion of dual set defined for pure cc.

\begin{defi}[Dual set] \label{defdualset}
Let $K$ be a pure $R$-cc and let $A \subset \kz$ be non-empty. We define the \textit{dual set} of $A$ to be $$\ce{A} := \{ z \in K^{[R]} ~|~ A \subset z \}.$$
If we need to specify the cc relative to which one is taking the dual, we use the notation $\cdual{A}{K} = \ce{A}.$
\end{defi}

Using the notion of dual set, we can then simply associate a poset to any pure cc in the following way. 

\begin{defi}[Dual of a cc] \label{defdualccc}
Let $K$ be a pure $R$-cc We define the poset $\left( \kb, \subset \right)$ by $$\kb := \{ \ce{x} ~|~ x \in K \}.$$
\end{defi}

We illustrate examples of dual cells in cc of rank 3 in Figure \ref{figexdual}.
Before proving that $\dual{K}$ is a cell complex for $K$ closed, we  start by proving the following lemma introducing the duality map. Similar lemmas showing an equivalence between inclusion relation of certain cells will repeatedly be a first step when proving that a newly introduced poset satisfies the axioms of a cell complex.

\begin{figure}[!h]
\centering
\includegraphics[scale=0.43]{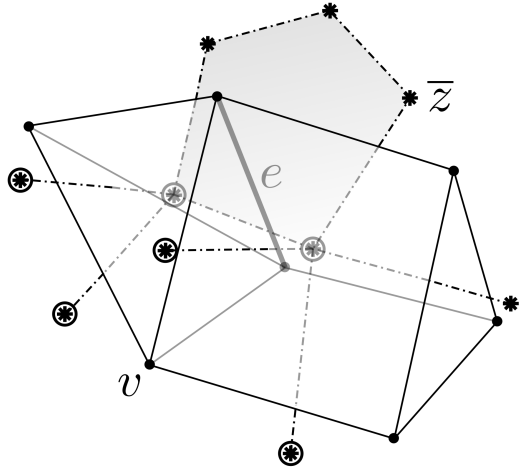}
\caption{\label{figexdual} This figure illustrates a portion of a closed 3-cc $K$ made of two 3-cells. One of these 3-cell is a tetrahedra $z_1$ and the second is a 3-cell $z_2$ on 6 vertices sharing a triangle with the tetrahedra. The dual of the 3-cells in $K$ are represented using star-shaped vertices. We indicated one of these vertices $\dual{z},$ representing one of the 3-cells around $z_1$ and $z_2.$ Two of these vertices are linked by an edge whenever the corresponding 3-cells share a face. The picture includes the dual of the edge $e$ (thicker line) forming a 2-cell on five vertices represented with a grey shade. We also use circles to indicate the vertices included in $\dual{v},$ the 3-cell dual to the vertex $v \in K.$ }
\end{figure}

\begin{lem}\label{claim2propdual}
Let $K$ be a closed cc and $x,y \in K.$ Then we have the following equivalence
$$x \subsetneq y \quad \text{ if and only if } \quad \ce{y} \subsetneq \ce{x},$$
and the \textit{duality map} $x \mapsto \ce{x}$ is a bijection from $K$ to $\kb.$
\begin{proof}
We first prove the following statement: if $x,y \in K$ are such that there exists $v \in y \setminus x$ then $\ce{x}$ contains a maximal cell which does not contain $v,$ and thus $\ce{x} \not\subset \ce{y}.$ 

This can be shown by induction on $n := R - \rk(x).$ For $n = 0,$ $x$ has rank $R$ and is then a maximal cell in $\ce{x}$ which does not contain $v$ by assumption. Assume $n \geq 1,$ therefore $\rk(x) < R$ and by Lemma  \ref{lemcface}, we have $ x = \bigwedge \cface{x}.$ This implies that there exists $x' \in \cface{x}$ such that $v \notin x',$ since otherwise $v \in x'$ for all $x' \in \cface{x}$ and in particular $v \in \bigwedge \cface{x} = x,$ which is a contradiction. By induction $x'$ is contained in a $R$-cell which does not contain $v,$ and so does $x,$ and this shows the statement.

As a consequence, if there exists $v \in y \setminus x$ then $\ce{x}$ contains a maximal cell $z$ such that $v \notin z.$ Together with the obvious implication $x \subset y ~\Rightarrow~ \ce{y} \subset \ce{x}$ this shows that $x \subsetneq y ~\Rightarrow~ \ce{y} \subsetneq \ce{x}.$ Conversely if $\ce{x},\ce{y} \in \kb$ are such that $\ce{y} \subsetneq \ce{x}$ then $x \setminus y = \emptyset$ by the first statement, and since $x \neq y,$ this implies $x \subsetneq y.$

Similarly if $\ce{x} = \ce{y}$ then $x \setminus y = y \setminus x = \emptyset,$ so that $x = y.$ The duality map $x \mapsto \ce{x}$ is then injective, hence bijective. 
\end{proof}
\end{lem}

We could rephrase the latter lemma by saying that the duality map is an anti-isomorphism of posets.
The following proposition shows that the above definitions define a duality map on closed cell complexes. The absence of boundary in closed cc allows us to prove this result in an otherwise rather broad generality, as for example no assumption related to the notion of locality has to be included. This allows us to introduce several notions of Chapter \ref{secreconst} also without assuming locality and also to prove certain results for example only using the assumption that the cc is graph-based. Nevertheless, the main results Chapter \ref{secreconst} related to reconstruction of cc are proven under the assumption of locality. 
Focusing on the local case also prevents to run into too many complications in Chapters \ref{secdualcob} and \ref{seccausalstruct} even though some results from this chapters could probably be formulated in more general terms.

\begin{prop} \label{propdual}
If $K$ is a closed $R$-cc then so is $\kb$ with rank function $$\rk_{\kb}(\xb) = \rkb(\xb) := R - \rk_K(x)$$ and we have $$\kbb \cong K.$$
\begin{proof}
The rank function $\rkb$ is well defined as a consequence of the bijectivity of the duality map showed in Lemma \ref{claim2propdual}. The latter lemma together with condition \ref{cccrank} for $K$ also imply that for $x,y \in K,$ if $\ce{y} \subsetneq \ce{x}$ then $ x \subsetneq y,$ hence $\rk(x) < \rk(y)$ and therefore $\rkb(\ce{y}) < \rkb(\ce{x}).$ This shows that $\rkb$ is strictly compatible with the inclusion in $\kb.$

Axiom \ref{cccinter} is satisfied for $\kb$ by the following argument. Suppose that there is no cell of $K$ containing $x \cup y$ then in particular $\dual{x \cup y} = \dual{x} \cap \dual{y} = \emptyset.$ If not then let $w$ be the cell containing $x \cup y$ with minimal rank. The cell $w$ is well defined, since if $w'$ is a cell of $K$ such that $x \cup y \subset w'$ and $\rk(w') = \rk(w)$ then $x \cup y \subset w \cap w' \in K.$ If $\rk(w \cap w') < \rk(w),$ it contradicts the minimality of $\rk(w),$ so $\rk(w \cap w') = \rk(w)$ and this implies $w \cap w' = w = w'$ by \ref{cccrank}. We therefore obtain $\dual{x} \cap \dual{y} = \dual{w} \in \kb,$ since the definition of $w$  implies the following equivalences: a maximal cell $z$ contains $w$ if and only if $z$ contains $x \cup y,$ if and only if $z \in \dual{x} \cap \dual{y}.$

In order to prove Axiom \ref{cccenough} for $\kb,$ we pick $\ce{x}, \ce{y} \in \kb$ such that $\ce{y} \subsetneq \ce{x}.$ Lemma \ref{claim2propdual} implies $x \subsetneq y$ and as a direct consequence of Axiom \ref{cccenough} for $K$ there exists a face $y'$ of $y$ containing $x.$ By Lemma \ref{claim2propdual}, we have the following identity.
\begin{align*}
 \{ \dual{z} ~|~ z \in K, ~z \in \cface{y'} \} &= \{ \ce{z} ~|~ z \in K,~ y' \subset z,~ \rk(z) = \rk(y') +1 \}\\
&= \{ \ce{z} \in \kb ~|~ \ce{z} \subset \ce{y'}, ~ \rkb(\ce{z}) = \rkb(\ce{y'}) -1 \} = \face{\ce{y'}}.
\end{align*}
Therefore we have that $\ce{y}$ is a face of $\ce{y'} \subset \ce{x}.$

Similarly, the following argument shows that the diamond property \ref{cccdiamond} for $\kb$ is equivalent to the diamond property for $K.$ Take $\ce{y} \subsetneq \ce{x}$ such that $\rkb(\ce{x}) = \rkb(\ce{y}) + 2$ then we have $x \subsetneq y$ and $\rk(x) = \rk(y) -2$ and this implies that $$\cface{x} \cap \face{y} =\{z_1,z_2\}$$ for some $z_1,z_2 \in K^{[\rk(x) + 1]} = \kb^{[\rkb(\ce{y}) + 1]}.$ This directly implies that
$$  \face{\ce{x}} \cap \cface{\ce{y}} = \{\dual{ z } ~|~ z \in \cface{x}\} \cap \{ \dual{z} ~|~ z \in \face{y} \} = \{\ce{z_1}, \ce{z_2}\}.$$

Therefore $\kb$ is a cc, it is also pure and satisfies $\Rk(\kb) = \Rk(K)$ by the following reasons: the maximal cells of $\kb$ of rank $R$ are the images of vertices of $K$ under the duality map and since every cell of $K$ contains at least one vertex, we have that every cell of $\kb$ is contained in $\vb$ for some $v \in K.$ Since $K$ is graph-based, Lemma \ref{claim2propdual} implies that every sub-maximal cell in $\kb$ is contained in exactly two maximal cells, hence $\kb$ is non-singular and has empty boundary.
Lemma \ref{claim2propdual} also implies that the map $x \mapsto \dual{\dual{x}}$ is a poset isomorphism from $K$ to $\kbb,$ which defines a cc-isomorphism by Lemma \ref{lemisom} since $\rk_{\kbb}(\ce{\ce{x}}) = \rk_K(x).$ 
\end{proof}
\end{prop}

By Lemma \ref{claim2propdual}, we have in particular that if $v,w \in \kz$ are distinct vertices then $\dual{v} \not \subset \dual{w}.$ Therefore if $\dual{v} \subset \dual{w}$ then $v = w.$ This implies that $$\dual{\dual{v}} = \{ \dual{w} \in \kR ~|~ \dual{v} \subset \dual{w} \} = \{ \dual{v}\}.$$ The vertices of $\kbb$ are then also sets containing one element.

The following notion of dual graph is convenient to formulate the notion of pinch and strongly connected cell complex that are introduced afterwards. It is also used later in the proof of Theorem \ref{thmdualcob}.

\begin{defi}[Dual graph] \label{defdualgraph}
Let $K$ be a non-singular cc. We defined $$ \dg{K}:= ( \kR, \{ \dual{y} ~|~ y \in (K \setminus \partial K)^{[R-1]}\} )$$ to be the \textit{dual graph} of $K.$ 
\end{defi}

We point out that the non-singularity of $K$ and the diamond property \ref{cccdiamond} guarantee that $\dg{K}$ and $\dg{\partial K}$ are graphs. The notion of strong connectedness defined next is common to different treatments of simplicial complexes and is often used to defined a \textit{pseudo-manifold} as a strongly connected closed simplicial complex.

\begin{defi}[Strongly connected cc]
Let $K$ be a non-singular cc. $K$ is said to be \textit{strongly-connected} if $\dg{K}$ is connected.  
\end{defi}

The following notions of pinch and non-pinching cc will be important when defining our notion of cobordisms introduced in Chapter \ref{secdualcob}.

\begin{defi}[Non-pinching cc]
A \textit{pinch}, or \textit{$K$-pinch,} is a cell $x \in K$ such that $\dg{K} \cap \dual{x}$ is \textit{not} connected. $K$ is said to be \textit{non-pinching} if it contains no $K$-pinch and no $\partial K$-pinch.
\end{defi}

Figure \ref{pinchpic} illustrates a simple example of pinch. As explained in the caption, this illustration can also help us understand the reason for adding in our definition of a non-pinching cc $K$ the condition that $K$ also do not contain any $\partial K$-pinch.

\begin{figure}[H]
\centering
\includegraphics[scale=0.55]{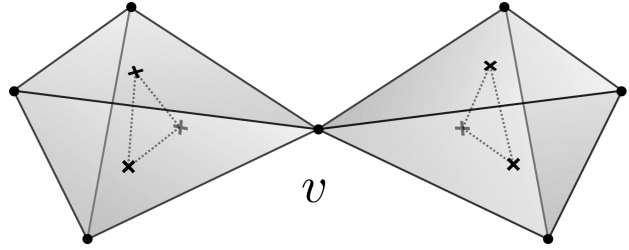}
\caption{\label{pinchpic} There are three different ways to look at this illustration. First, one can see it as a portion of a closed 2-cell complex $K$ and the central vertex $v$ gives an example of a $K$-pinch. The graph $\dg{K} \cap \dual{v}$ would then be constituted of two disconnected cycles, each containing three vertices and each such vertex corresponding to a 2-cell (triangle) of $K$ containing $v.$ Second, one could look at this illustration as a portion of a non-singular 3-cell complex $L,$ in which $v$ is included in two tetrahedra considered as two different 3-cells of $L.$ In this case $v$ is both a $L$-pinch and a $\partial L$-pinch. The third way to look at this illustration is to consider that the two tetrahedra of $L$ containing $v$ form a single cell of a non-singular 3-cell complex $L'.$ In this case $v$ is a $\partial L'$-pinch, but not a $L'$-pinch.}
\end{figure}

This next lemma and the remark that follows will be useful in Chapter \ref{secdualcob}, when considering the boundary of cobordisms. 

\begin{lem} \label{rembdrycc}
If $K$ is a non-singular cc such that there is no $K$-pinch in $\partial K$ then $\partial K$ is closed.
\begin{proof}
Let $R = \Rk(K).$ It is sufficient to show that if $\partial K$ contains no $K$-pinch then 
$$\abs{\cface{x} \cap \partial K} = 2, \quad \forall x \in (\partial K)^{[R-2]}.$$ 
We start by showing that $\abs{ \cface{x} \cap \partial K} \geq 2,$ which in particular implies that $\partial ( \partial K) = \emptyset.$ Since $x \in \partial K$ there is $y_0 \in (\partial K)^{[R-1]}$ and $z_0 \in \kR$ such that $x \subset y_0 \subset z_0$ and by the diamond property there exists $y_1 \in K^{[R-1]}$ such that $ \face{z_0} \cap \cface{x} = \{y_0, y_1\}.$ If $y_1 \not\in \partial K$ then there exist $z_1 \in \kR, ~y_2 \in K^{[R-1]}$ such that $\cface{y_1} =\{z_0, z_1\}$ and $\face{z_1} \cap \cface{x} =\{y_1,y_2\}.$ By repeating this argument we can construct a sequence of elements in $K^{[R-1]},$ distinct by non-singularity of $K,$ until at some step $k$ we get $y_k \in \partial K$ and this implies $\{y_0, y_k\} \subset \cface{x} \cap \partial K.$ If by contradiction we assume that there is more that two elements in $\cface{x} \cap \partial K$ then the previous construction would show that it contains at least four elements and that the corresponding two sequences of elements in $K^{[R]}$ sharing an element of $K^{[R-1]}$ constitute two disconnected components of $\dg{K} \cap \dual{x}.$ This would imply that $x \in \partial K$ is a $K$-pinch, a contradiction with our hypothesis.
\end{proof}
\end{lem}

\begin{rema} \label{remassumptlocal}
Proposition \ref{propdual} allows us to make the following observations.
A closed cc $K$ is strongly-connected  and non-pinching if $\kb$ is connected and cell-connected. 
Hence a closed $R$-cc $K$ is local if and only if $\kb$ is a closed strongly-connected non-pinching $R$-cc. We can also note that if a non-singular cc $K$ is connected and non-pinching then it is strongly connected.
Therefore, in order to define a class of cc invariant under the duality map, we need the assumptions of locality and non-pinchingness. This will justify our choices of definitions in Chapter \ref{secdualcob}.
\end{rema}

\subsection{Barycentric subdivision and how it determines the underlying cc} \label{ssecbdiv}

In this section, we first introduce a combinatorial definition of barycentric subdivision that one can associate to any poset and explain how it is sometimes used to defined a different notion of dual cell than the one used in our framework. In Chapter \ref{seccausalstruct} we will introduce a generalization of the notion of barycentric subdivision that we will simply call subdivision. We then show Theorem \ref{thmbayer}, stating the barycentric subdivision characterizes the cc from which it is constructed under relatively general conditions.

We use the following abstract analogue of the first derived barycentric subdivision of a poset, also called order complex. 

\begin{defi}[Barycentric subdivision] \label{defbdiv}
Let $(K, \subset)$ be a poset. The \textit{(first derived) barycentric subdivision} or \textit{order complex} of $K$ is the poset noted $\bdiv{K}$ of totally ordered subsets of $K.$ 
\end{defi}

Figure \ref{dualcellpic} gives examples of cells in the barycentric subdivision of a 2-cc. In our context, the term barycentric subdivision is therefore interpreted in a broad sense, but for the case of a simplicial complex obtained from a triangulation of a manifold, this definition corresponds to the face lattice of the barycentric subdivision of the triangulation as defined for example in \cite{bry02}.
The barycentric subdivision of a cc is a simplicial complex since condition \ref{simpcplxa}  follows from the observation that a subset of a totally ordered set is a totally ordered set and condition \ref{simpcplxb} is a direct consequence of the definition of the barycentric subdivision.

The barycentric subdivision provides an interesting way to relate a closed cc with its dual, as shown in the next proposition.

\begin{prop}\label{propsubdivdual}
The barycentric subdivision of a closed cc is isomorphic to the barycentric division of its dual. In other words, $\bdiv{K} \cong \bdiv{\kb}.$
\begin{proof}
This follows from the observation that,  by Lemma \ref{claim2propdual}, the bijection $\{x\} \mapsto \{\ce{x}\}$ from $\bdiv{K}^{[0]}$ to $\bdiv{\kb}^{[0]}$ induced by the duality map defines an isomorphism of cc. This bijection indeed maps totally ordered subsets of $K$ to totally ordered subsets of $\kb$ as  
\begin{equation} \label{mappropdual}
\{x_1 \subsetneq \dots \subsetneq x_k\} \mapsto \{\ce{x_k} \subsetneq \dots \subsetneq \ce{x_1}\}.
\end{equation}
\end{proof}
\end{prop}

The barycentric subdivision is sometimes used to define the notion of dual cell, as for example in \cite{bry02}. This definition can also be defined in a purely combinatorial manner, but it corresponds to different cells than the cells obtained by the duality map leading to the involution property of Proposition \ref{propdual}. We explain here how to precisely compare these two definitions.

Let us extend a notion we use for a cc $K,$ to arbitrary subsets $\A \subset K$ by setting $$\A^{[r]} :=  \A \cap \kr.$$ 
We can also introduce the sets of elements respectively \textit{above} and \textit{below} an element $w$ in a poset $(L, \subset):$
$$ A_L(w) := \{ y \in L ~|~ w \subset y \}, \quad B_L(w) := \{ y \in L ~|~ y \subset w \}. $$

Let $K$ be a closed cc. We will more specifically consider the sets $A_L(w),B_L(w)$ in the cases where $L= \bdiv{K}$ and $w = \{x\} \in \bdiv{K}^{[0]}$ and the case $L = \bdiv{\kb}$ and $w = \{\dual{x}\} \in \bdiv{\kb}^{[0]}$ respectively. These sets can be more explicitly expressed as   
$$A_{\bdiv{K}}(\{x\}) = \{ \{y_1 \subsetneq \dots \subsetneq y_k\} \in \bdiv{K} ~|~ x \subset y_1 \}, \quad B_{\bdiv{\kb}}(\{\ce{x}\}) = \{ \{\ce{y_1} \subsetneq \dots \subsetneq \ce{y_k} \} \in \bdiv{\kb} ~|~ \ce{x} \supset \ce{y_k}\}.$$ 
Then $A_{\bdiv{K}}(\{x\})$ and $ B_{\bdiv{\kb}}(\{\ce{x}\})$ are isomorphic as posets using the map (\ref{mappropdual}), so that we have 
$$ A_{\bdiv{K}}(\{x\}) \cong \cap B_{\bdiv{\kb}}(\{\ce{x}\}) = \bdiv{\kb \cap \ce{x}}.$$

For the case of a simplicial complex $K,$ the map $K \dans \pws{\bdiv{K}}$ defined by
$$ x \longmapsto \dual{x}^B := \bdiv{\kb \cap \ce{x}} \cong \{\{y_1 \subsetneq \dots \subsetneq y_k \} \in \bdiv{K} ~|~ x \subset y_1 \}$$
is equivalent to the definition of dual cell used in \cite{bry02} (which explains our choice for the letter $B$ in the notation $\dual{x}^B$). The dual cell $\dual{x}^B$ is therefore seen as a sub-simplicial complex of $\bdiv{K}.$ 

\begin{figure}[!h]
\centering
\includegraphics[scale=0.46]{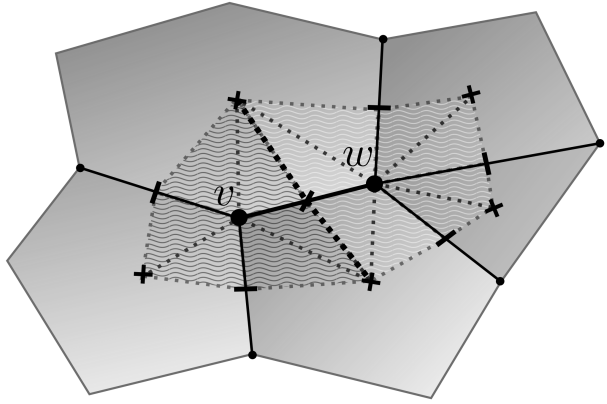}
\caption{\label{dualcellpic}This illustration represents certain cells of the barycentric subdivision $\bdiv{K}$ of a 2-cc $K$ and gives three examples of cells in $\kb^B.$ Being a simplicial complex, each 2-cell of the barycentric subdivision is a triangle and the vertices of the barycentric subdivision are represented using three different symbols depending of the rank of the cells it represent: a vertex of $\bdiv{K}$ of the form $\{C\}$ where $C \in \kt$ is represented by a cross, a vertex of the form $\{e'\}$ where $e' \in \ko$ is represented by a dash and a vertex of the form $\{v\}$ simply corresponds to a vertex of $K.$ There are two 2-cells of $\kb^B,$ the first $\dual{v}^B$ is filled the black wavy pattern, the other $\dual{w}^B$ is filled with the white wavy pattern. We also have the 1-cell $\dual{e}^B,$ the dual of $e = \{v,w\},$ represented by the two thick dashed lines. As a comparison, the dual cell $\dual{v}$ used in our framework simply corresponds to the triangle whose vertices are the three 2-cells containing $v.$}
\end{figure}

We can then define the poset $\kb^B := \{ \dual{x}^B ~|~ x \in K \}$ with rank function 
$$\rk_{\kb^B}(\dual{x}^B) := \Rk(\bdiv{\kb \cap \ce{x}}) = \rkb(\ce{x}).$$ 
It is possible to show that the map $ \phi(\dual{x}^B) := \ce{x}$ defines an isomorphism of posets from $\kb^B$ to $\kb$ such that $\rkb(\phi(\dual{x}^B)) = \rk_{\kb^B}(\dual{x}^B).$ Nevertheless $\kb^B$ is not a cc since the cells in $\kb^B$ are not subsets of $(\kb^B)^{[0]}$ as these cells also contain elements in $\bdiv{K}^{[0]} \setminus (\kb^B)^{[0]}.$

The remainder of this section aims to present an adaptation of the Theorem 3 given by Bayer in \cite{bay88}, stating that the barycentric subdivision of a polyhedral complex determines the underlying complex up to duality. The notion of polyhedral complex from \cite{bay88} is identical to the notion of complex of polytopes defined for example in \cite{zi95}, which will be discussed in more detail in Section \ref{secpolytope}. We sometimes designate such convex polytopes in $\R^n$ as "embedded polytopes". Bayer's proof can nevertheless be easily adapted to our purely combinatorial framework. Theorem \ref{thmbayer} stated below can then be seen as a relatively straightforward generalization of Bayer's result to local cell complexes. The main point from Bayer's proof that needs an additional justification is an observation about the existence of a type of ordering of the faces of an embedded polytope. This observation is also valid for the sub-cells of a cell in a cell complex and this constitutes our Lemma \ref{claimthmbayer}. Note that by "face of a polytope", we mean any sub-polytope and not necessarily of dimension one lower than the original polytope (as is defined for the cells of a cc).

We start by summarizing the notions needed to state the proof of Theorem \ref{thmbayer}, closely following the section 3 of \cite{bay88}. First, we define what is an orientation in a graph-based cc. For a set of edges $\E \subset \ts{\V}$ we use the notation $\ori{\E}$ to denote the set of ordered two-sets of $\E.$ An element $e = (x,y) \in \ori{\E}$ is called an \textit{directed edge}, $e_s := x$ is called its \textit{source} and $e_t := y$ its \textit{target}.

\begin{defi}[Orientation of a cc] \label{deforient}
An \textit{orientation} of a graph-based cc $K$ is a map $g: \ko \mapsto \ori{\ko}$ such that $g(e)^{[0]} = e.$

We say that an orientation $g$ is \textit{transitive} if, whenever 
$$g(\{x,y\}) = (x,y), \quad g(\{y,z\}) = (y,z)$$ 
for $x,y,z \in \kz$ such that $\{x,y\},\{y,z\} \in \ko$ then $\{x,z\} \in \ko$ and $$g(\{x,z\}) = (x,z).$$
\end{defi}

In particular an \textit{oriented graph} $(\G, g)$ is a graph $\G = (\V, \E),$ for example the 1-skeleton of a graph-based cc, with an orientation $g: \E \dans \ori{\E}.$ An isomorphism of oriented graphs is a cc-isomorphism $\phi: \G \dans \G'$ between two oriented graphs $(\G,g)$ and $(\G',g')$ such that $g'$ is the image of $g$ through $\phi$ in the sense that $$g'( \phi(e) ) = (\phi(g(e)_s), \phi(g(e)_t)).$$

We now turn to the notion of uniquely partially orderable cc, a key notion in the proof of the Theorem \ref{thmbayer}, in which we essentially prove that the barycentric subdivision of a local cell complex is uniquely partially orderable. Let $g$ be a transitive orientation of a graph-based cc $K$ and let $(x,y), (w,z) \in \ori{\ko}.$ We say that $(x,y)$ \textit{directly forces} $(w,z)$ if either $x = w$ and $\{y,z\} \notin \ko$ or $y = z$ and $\{x,y\} \notin \ko.$ More generally, we say that $(x,y)$ \textit{forces} $(y,z)$ if there is a finite sequence of directed edges starting from $(x,y)$ and ending at $(y,z)$ such that each directed edge in the sequence directly forces the next one. This relation is an equivalence relation, which allows to decompose $\ori{\ko}$ into equivalence classes called \textit{implication classes.} $K$ is called \textit{uniquely partially orderable} if and only if $K$ admits exactly two transitive orientations, each being the reverse of the other, or equivalently, if and only if there exists only two possible implication classes in $\ori{\ko}.$ For example, a transitive orientation on $\bdiv{K}$ is simply given by $g_K : \bdiv{K}^{[1]} \dans \ori{\bdiv{K^{[1]}}},$ the orientation induced by the inclusion relation in $K,$ i.e. if $x,y \in K$ satisfy $x \subsetneq y,$ or equivalently $\{x \subsetneq y\} \in (\bdiv{K})^{[1]},$ then $g_K( \{x \subsetneq y \} ) = (x,y).$ (Recall that our notation $\{x \subsetneq y\}$ is a shorthand for an - unordered - set $\{x , y\}$ whose elements are totally ordered by the inclusion relation in $K,$ so that we also specify this total order by writing $\{ x \subsetneq y\}$). A second transitive orientation of $\bdiv{K}$ is the \textit{reverse of $g_K,$} denoted $g^{-1}_K,$ and defined by $g^{-1}_K(\{x \subsetneq y\}) := (y,x).$ In case where $K$ is a closed cc then Proposition \ref{propsubdivdual} implies $$\bdiv{K} \cong \bdiv{\dual{K}}.$$ Using the later cc-isomorphism, we can identify $\{x \subsetneq y \} \in \bdiv{K}^{[1]}$ with $\{\dual{y} \subsetneq \dual{x} \} \in \bdiv{\dual{K}}^{[1]}$ and using this identification we have that $g^{-1}$ corresponds to the transitive orientation $g_{\dual{K}}$ induced by the inclusion relation in $\dual{K}.$

We prove the following lemma, which, in the proof from \cite{bay88}, corresponds to an observation about the faces of convex polyhedra, where by faces we mean any sub-polyhedra. We will see that it also holds in the case of cc of the form $K \cap z,$ for $z \in K^{[R]},$ as a consequence of Axioms \ref{cccdiamond} and \ref{cccinter}.

\begin{lem} \label{claimthmbayer}
Let $K$ be a $R$-lcc and let $z \in \kR.$ For any $0 \leq r \leq R-2$ and any $r$-cell $y \in K \cap z,$  there is a $(r+1)$-cell $w \in K \cap z$ such that $y \not\subset w.$
\begin{proof}
It is sufficient to prove the claim for the case $r= R-2.$ The case $r = 0,$ i.e. $R = 2$ and $y$ is a vertex of a 2-cell $z,$ is direct since by Lemma \ref{lemrepr2cell} $z$ is a cycle of at least 3 edges. Suppose that $r > 0$ and take $x \in \face{y},$ i.e. such that $\rk(x) = r - 1.$ By the diamond property \ref{cccdiamond} there exist $\{w_1,w_2\} \subset K^{[r+1]}$ and $y' \in \kr$ such that $$\cface{y} \cap \face{z} =\{w_1,w_2\}, \quad \cface{x} \cap \face{w_1} = \{y,y'\}.$$ Moreover there exists $w \in K^{[r+1]}$ such that $$\cface{y'} \cap \face{z} =\{w_2,w\}$$ and $w \in K^{[r+1]}\setminus \{w_1,w_2\}$ since otherwise $\{y,y'\} \subset w_1 \cap w_2$ which contradicts \ref{cccinter}. Also $y \not\subset w$ since $\cface{y} \cap \face{z} =\{w_1,w_2\}$ and this concludes the proof.
\end{proof}
\end{lem}

We also add the following lemma, which is considered as an obvious statement in \cite{bay88} for the case of face lattices of polytopes.

\begin{lem} \label{lembdivoriented1skel}
A cell complex $K$ is determined up to cc-isomorphism by the oriented 1-skeleton of its barycentric subdivision $(\bdiv{K}^{(1)},g_K).$ 
\begin{proof}
Let $\{x\}$ be a vertex of $\bdiv{K}.$ Lets us prove this result by first indicating how to associated to $\{x\}$ a set of vertices of $\bdiv{K}$ of the form $\{v\}$ where $v \in \kz$ corresponding to the vertices of $K$ contained in $x.$ The elements of the form $\{v\} \in \bdiv{K}^{[0]}$ such that $v \in x$ are the starting points of all directed paths (i.e. paths using directed edges which target is equal to the source of the next directed edges or the endpoint of the directed path) in $(\bdiv{K}^{(1)},g_K)$ having $\{x\}$ as an endpoint and which length is maximal with this property. The rank of $\{x\} \in \bdiv{K}^{[0]}$ is the length of any of such directed path. This is well defined by Axiom \ref{cccenough}.
\end{proof}
\end{lem}

\noi We now have the required notions and preliminary results to state and prove the following theorem. This result seems important to us as a general fact about lcc and will be interesting for us in relation to the reconstruction results of Chapter \ref{secreconst}, as explained in Remark \ref{remareconst}. It also lead us to introduce a notion of orientation in a cc, but this notion as defined here will not be used again in this work. The work of Basak \cite{bas10} introduces another notion of orientation which is central to his treatment of the Poincaré duality for his definition of combinatorial cell complexes.

\begin{thm} \label{thmbayer}
Let $K$ and $J$ be two lcc having at least two maximal cells and such that $\bdiv{K} \cong \bdiv{J}.$ Then we have the following possibilities:
\begin{itemize}
\item if $\partial K$ is non-empty then $\partial J$ is also non-empty and $K \cong J;$
\item if $K$ and $J$ are closed then either $K \cong J$ or $\kb \cong J.$
\end{itemize}
  
\begin{proof}
Let $(\bdiv{K}^{(1)}, g_K)$ and $(\bdiv{J}^{(1)},g_J)$ be the oriented 1-skeletons of the barycentric subdivision of $K$ and $J,$ where both $g$ and $h$ are the orientation induced by the inclusion relation in $K$ and $J$ respectively. Our assumption implies in particular that $\bdiv{K}^{(1)}$ and $\bdiv{J}^{(1)}$ are cc-isomorphic, we therefore consider these two graphs as identical for simplicity. The strategy of the proof is to show that $\bdiv{K}$ is uniquely partially orderable, which would directly imply that $g_J = g_K$ or $g_J = g_K^{-1}$ as defined above.

In case $K$ and $J$ are closed, the two possibilities $g_J = g_K$ and $g_J = g_K^{-1}$ lead consistent cases since $g_K^{-1} = g_{\dual{K}}$ and Lemma \ref{lembdivoriented1skel} lead $K \cong J$ or $\dual{K} \cong J.$

If $\partial K$ is non-empty then $g^{-1}_K$ cannot be a transitive orientation induced by the inclusion relation of a local cell complex. To justify this statement, let us assume by contradiction that $$g_K^{-1} = g_L : \bdiv{L}^{[1]} \dans \ori{\bdiv{L}}^{[1]}$$ is a transitive orientation induced by the inclusion relation of a lcc $L.$ Consider a maximal cell $y \in \partial K,$ which is therefore included in a unique maximal cell $z$ of $K.$ Then, as shown in Lemma \ref{lembdivoriented1skel}, $\{z\}$ corresponds to a vertex of $L$ and $\{y\}$ corresponds to an edge of $L$ and this is a contradiction with the fact that every edge of $L$ has two vertices. As a consequence, it only leaves $g_K = g_J$ as a possibility and Lemma \ref{lembdivoriented1skel} implies $K \cong J.$  

We therefore turn to the proof that $\bdiv{K}$ is uniquely partially orderable by induction on the rank of $K.$

For $\Rk(K) = 1$ we have that $K$ is a connected graph and so is $\bdiv{K}.$ Every edge $e = \{v,w\}$ of $K$ is subdivided into two "half-edges" $\{\{v\} , \{e\}\}$ and $\{ \{w\}, \{e\}\}$ in $\bdiv{K}^{[1]}.$ Moreover, any two edges $\varepsilon, \varepsilon'$ of $\bdiv{K}$ sharing a vertex in $\bdiv{K}^{[0]}$ correspond either to a subdivided edge of $K$ or to two "half-edges of $K$" i.e. $ \varepsilon =\{\{v\}, \{e\}\},$ $\varepsilon' = \{\{v\}, \{e'\}\}$ where $e,e' \in \ko$ and $ v = e \cap e' \in \kz.$ In both cases there is no edge between the other ends of $\varepsilon$ and $\varepsilon'$ and this means that an orientation of $\varepsilon$ directly forces an orientation of $\varepsilon'.$ The connectedness of $\bdiv{K}$ then implies that an orientation of any edge in $\bdiv{K}$ force an orientation of any other edge in $\bdiv{K},$ which shows that $K$ is uniquely partially orderable.

Let $K$ be a $R$-lcc with $\abs{K^{[R]}} \geq 2$ and assume that $\bdiv{(K')}^{(1)}$ is uniquely partially orderable for any lcc $K'$ of rank smaller or equal to $R-1$ with more than one maximal cell. Let $K' := K^{(R-1)}$ and $G' := (\bdiv{K'})^{(1)}.$ Then $K'$ is a $(R-1)$-lcc with more than one maximal cell and $G'$ is the sub-graph of $G$ restricted to the vertices with corresponding cell in $K$ having a rank smaller or equal to $R-1.$ By induction $G'$ is uniquely partially orderable and we show that an orientation of $G'$ forces an orientation of $G.$

Let $z \in K^{[R]}$ and let $\{z\}$ be the corresponding vertex in $G.$ Since there are at least two maximal cells in $K,$ there is one vertex $v$ of $z$ contained in an edge $e \in \ko$such that $e \not \subset z.$ Therefore we obtain $$\{\{v\},\{z\}\} \in G^{[1]}, \quad \{\{v\}, \{e\}\} \in (G')^{[1]} \text{ and } \{\{z\},\{e\}\} \notin G^{[1]}$$ and this implies that $(\{v\},\{e\})$ directly forces $(\{v\},\{z\}).$ We conclude the proof by showing that $(\{v\},\{z\})$ directly forces all edges of $G$ containing $\{z\}.$

Indeed, all edges of $G$ containing $\{z\}$ are of the form $\{\{z\},\{w\}\}$ where $w$ is a sub-cell of $z.$ By Lemma \ref{claimthmbayer} we can find an ordering of the cells in $K \cap z,$ starting with any vertex in $z,$ so that no two consecutive sub-cells contain one another. This corresponds to an ordering of the vertices in $G$ neighbouring $\{z\},$ so that no two consecutive vertices are linked by an edge. We can suppose that this vertex ordering starts with $v$ and this means that $(\{v\},\{z\})$ forces each edge containing $\{z\}$ in G and this concludes the proof.
\end{proof}
\end{thm}

\subsection{Topological considerations} \label{sectopoconsid}

In this section, we mention two different ways to assign a topology to a cc: the first one is from \cite{bas10} and the second, using a notion of a geometrical realization, is from \cite{wac07}. It seems clear that these two topologies are equivalent, but we have not proven this here. We then give the definitions of star, link and homogeneity in a cc and define the notions of cellular and simplicial manifold. All homeomorphisms here are piecewise-linear homeomorphisms (as defined e.g. in \cite{bry02}). 

In \cite{bas10} (setion 2.4) the topology of a cc $K$ is defined by setting all closed subsets $C$ of $K$ to be the subsets of $K$ such that if $x \in K,~ y \in C, ~x \subset y$ then $x \in C.$ If we define the closure of a subset $S \subset K$ to be the set of cells of $K$ that are contained in some cell in $S$ then the closed subsets of $K$ are given by the closures of subsets of $K.$ As a consequence, $K$ is not Hausdorff if it has rank one or above, since each set of the form $\{x\}$ where $x \in K$ has a rank higher or equal to 1 is a non-closed "point" in this topology. For example if  $x = \{v,w\} \in \ko,$ this topology implies that $\{\{v,w\}\}$ is an "open edge" whereas the closure of $\{x\}$ corresponds to a "closed edge" $\{\{v,w\},\{v\},\{w\}\}$ and this is coherent with the following definitions obtained using a geometrical realization.

Another way to define a topology on a cc $K,$ described for example in \cite{wac07}, is to consider $K$ simply as a poset to which one associates the topology obtained from the geometrical realization of its barycentric subdivision $\bdiv{K}.$ The geometrical realization of a simplicial complex is defined in the following way. A $d$-dimensional geometric simplex $\sigma$ in $\R^n$ is defined as the convex hull of $d+1$ points in $\R^n,$ each such point being a vertex of $\sigma.$ The set of sub-simplicies of $\sigma,$ $B(\sigma),$ contains all the convex hulls of subsets of vertices of $\sigma.$ A geometric simplicial complex $\Lambda$ in $\R^n$ is defined as a non-empty collection of geometric simplicies satisfying that if $\sigma \in \Lambda$ and $\tau \in B(\sigma)$ then $\tau \in \Lambda$ and if $\tau, \sigma \in K$ then $\tau \cap \sigma \in B(\sigma) \cap B(\sigma)$ (the roles of $K$ and $\Lambda$ are exchanged compared to the notations used in \cite{wac07}). From a geometric simplicial complex $\Lambda,$ one can define a (abstract) simplicial complex $K(\Lambda)$ corresponding to the poset of non-empty faces of $\Lambda$ and say that $\Lambda$ realizes $K.$ It is an exercise in \cite{rs72} (2.27(1)) to show that every simplicial complex can be realized as a geometric simplicial complex and such a realization is unique up to homeomorphism. We can therefore define the topology of $K$ to be the topology of $\Lambda$ as a subspace of $\R^n,$ with usual topology. As a direct consequence of Proposition \ref{propsubdivdual}, $K$ and $\kb$ are homeomorphic as topological spaces.


The notions of star and link are usually defined in the context of simplicial complexes. These definitions are easily adapted to the case of cc in the following way.

\begin{defi}[Star and link of a cell] \label{defstarlink}
The \textit{star} of a cell $x$ of a cell complex $K$ is defined by 
$$ \st_K(x) = \{  y \in K ~|~ \exists w \in K \text{ such that } x \vee y = w \}.$$
The \textit{link} of $x$ is defined as
$$ \lk_K(x) = \{ y \in \st_K(x) ~|~  x \cap y = \emptyset \}.$$
\end{defi}

The previous definitions of star and link are then equivalent to the definitions given for example in \cite{bry02} for the case of simplicial complex, since if $\sigma$ and $\eta$ are sub-simplicies of a simplex in a simplicial complex $K$ then $\sigma \vee \eta = \sigma \cup \eta \in K.$

Homogeneity in a cc is then defined using the star of each vertex and the topology associated to a cc introduced before.

\begin{defi}[Homogeneous vertex in a cc]
A vertex $v$ in a $R$-cc is \textit{homogeneous} if any geometrical realization of the cc $K \cap \st_K(v)$ is homeomorphic to the closed $R$-dimensional Euclidean ball.
If every vertex in $K$ is homogeneous, we say that $K$ is \textit{homogeneous}.
\end{defi}

An example of non-homogeneous vertex is given in Figure \ref{figtorus} of Chapter \ref{secdualcob}. It is then natural to introduce the following definitions of cellular manifold and simplicial manifold.

\begin{defi}[Cellular and simplicial manifold]
A \textit{cellular $R$-manifold} is a homogeneous $R$-cc. A \textit{simplicial (cellular) manifold} is a cellular manifold such that the underlying cc is simplicial.
\end{defi}

A \textit{discretization} of a piecewise-linear manifold $M$ is a cellular manifold $K$ which geometrical realization is homeomorphic to $M$ and a \textit{triangulation} of $M$ is a simplicial discretization of $M.$

\newpage
\vspace{3cm}

\section{Structures in a simple cell complex and reconstruction} \label{secreconst}

\vspace{2cm}

The aim of this chapter is to prove Theorem \ref{thminducedcc} as well as Corollaries \ref{correconst1} and \ref{correconst2} of Section \ref{ssecreconstthm}. These results can be interpreted as a method for reconstructing a local even simply connected simple cell complex from its 2-skeleton.  The core of this chapter is thus devoted to the introduction of the notions of fullness, evenness, monodromy-freedom and simple connectedness and prove some related basic facts, with the exception of Section \ref{secpolytope} where we make links with some well-known notions of discrete geometry. A key technical result that allows to prove the results on reconstruction is the "Extension Lemma" \ref{lemextension}. This lemma will be the first result involving all notions required for the main results of Section \ref{ssecreconstthm} as well as two newly introduced notions of this chapter defined in the context of a cc of rank 2 or higher: a connection and covariant edge fields.
We start with a short section introducing these two new notions, followed by Section \ref{ssecfullnsimplecc} where we introduce the notions of full and simple cell complexes that characterize the duals of simplicial cell complexes. In Section \ref{secpolytope}, we make a digression from the reconstruction problem about an abstract definition of polytopes and shellability, using the notion of simplicity. We first show the Euler-Poincaré formula for shellable local cell complexes. We then introduce a notion of 2-shellability, only involving the cells of the 2-skeleton, and show in Theorem \ref{thm2shell} that the 2-shellability of a simple cell complex is equivalent to the shellability of its dual. We then come back to the discussion of reconstruction-related notions in Section \ref{ssectionsimplyconccc}, a mostly technical section whose principal goal is to define simply connected cc. Simple connectedness in a cc is defined using a combinatorial notion of homotopy that naturally arise from the structure of 2-cc. In Section \ref{ssecmonod}, we introduce the concept of monodromy-free and even cc. We then have all the ingredients needed to give a relatively quick proof of the Extension Lemma at the beginning of Section \ref{ssecreconstthm}, which allows us to prove the main results on reconstruction in the remainder of this last section.

In this chapter 2-cc play an important role. We often use the letter $C$ to denote a $2$-cell of a cc and we also recall that if $K$ is a pure 2-cc then the dual set of $A \subset \kz$ corresponds to
$$ \ce{A} := \{ C \in \kt ~|~ A \subset C \}.$$
In what follows, we will choose to state certain results using a 2-cc $K$ in order to formulate the proof in terms of dual sets as above. In these cases $K$ often play the role of the 2-skeleton of a cc of higher rank for which the corresponding result is also true.

It will also be convenient to use the following notation to denote the set of edges included in a cell $x \in K$ and containing a vertex $v \in x,$
$$ \E_v^x := \{ e \in \E_v ~|~ e \subset x \}.$$
More generally, if $\A \subset K,$ we define
$$\E_v^\A := \bigcup_{x \in \A} \E_v^x.$$

As an example used in Section \ref{ssecindedgemap}, if $v,w$ are vertices of a pure graph-based $2$-cc $K$ such that $\{v,w\} \in \ko,$ we consider the following subset of $\E_v:$
$$\E_v^{\dual{w}} := \{ e \in \E_v ~|~ \exists C \in \dual{w} \text{ s.t. } e \subset C \} = \{ e \in \E_v ~|~ \dual{\{v,w\}} \cap \dual{e} \neq \emptyset \}.$$
In Section \ref{ssecfullnsimplecc} we introduce the notion of full cc and this notion in particular implies that $\E_v^{\dual{w}} = \E_v$ for all $w \in \Ne_v.$ Since these edge sets will for the most part be used in the case of full cc, the set $\E_v^{\dual{w}}$ is almost exclusively used to introduce a general definition of connection in the next section.

\subsection{Connection and covariant edge fields} \label{ssecindedgemap}

In this first section, we define the notion of a connection as a mapping between edges of neighbouring vertices canonically associated to a pure graph-based cc of rank 2 or higher. This connection is an important tool for this chapter, as well as the notion of covariant edge fields introduced next which is the main object of the Extension Lemma \ref{lemextension}.

The next Lemma is used to introduce the \textit{connection} $\nabla$ for the case of a pure graph-based cc $K$ of rank 2. The connection associated to a pure graph-based cc $L$ of rank higher or equal to 2 is simply defined as the connection associated to its 2-skeleton $K = L^{(2)}.$

\begin{lem} \label{nabla}
Let $K$ be a pure graph-based 2-cc and $\{v,w\} \in \ko.$  There is a bijection 
\begin{align*}
\nabla_w^v : \E_v^{\dual{w}} &\longrightarrow \E_w^{\dual{v}} \\
e_v &\longmapsto  \nabla_w^v e_v =:e_w,
\end{align*}
as illustrated in Figure \ref{figlemnabla}, such that
\begin{enumerate}[label=\textbf{\arabic*}),topsep=2pt, parsep=2pt, itemsep=1pt]
\item $\nabla_w^v \{v,w\} = \{v,w\},$ \label{nabla0}
\item  If $C \in \ce{\{v,w\}}$ then  $e_v \subset C$ if and only if $e_w \subset C,$ \label{nabla1}
\item $(\nabla_w^v)^{-1} = \nabla_v^w.$ \label{nabla2}
\end{enumerate}

\begin{figure}[!h]
\centering
\includegraphics[scale=0.35]{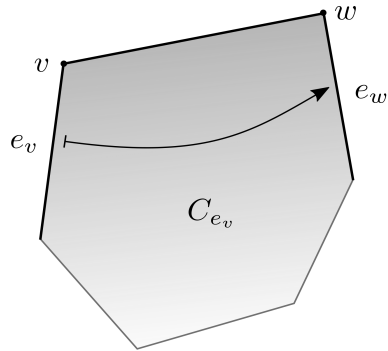}
\caption{\label{figlemnabla}This figure illustrates how the connection $\nabla_w^v$ maps an edge $e_v \in \E_v \setminus \{v,w\}$ to an edge $e_w \in \E_w \setminus \{v,w\}$ using a $2$-cell $C_{e_v}$ defined as the unique $2$-cell in $\dual{v}$ containing $e_v$ and $\{v,w\}.$} 
\end{figure}

\begin{proof}
We set $\nabla_w^v \{v,w\} = \{v,w\}.$ Let $e_v \in \E_v^{\dual{w}} \setminus \{v,w\}$ and let $C_{e_v} \in \ce{e_v} \cap \ce{\{v,w\}}.$ Since $e_v \neq \{v,w\}$ we have $\abs{\ce{e_v} \cap \ce{\{v,w\}}} = 1$  as two 2-cells cannot intersect on more than one edge by Axiom \ref{cccinter}. $C_{e_v}$ is then the unique 2-cell containing $e_v$ and $\{v,w\}.$ In particular $C_{e_v} \in \ce{w},$ so there exists a unique edge $e_w \in \E_w \setminus \{v,w\}$ such that $\{\{v,w\}, e_w\} \subset C_{e_v},$ which implies $e_w \in \E_w^{\dual{v}},$ so that we can set $\nabla_w^v e_v = e_w.$  The uniqueness of $C_{e_v}$ and $e_w$ ensures \ref{nabla1} and $\nabla_v^w e_w = e_v.$
\end{proof}
\end{lem}

The connection is also a convenient tool in the formulation of Lemma \ref{lemfullk} and in Definition \ref{defmonodromyfreecc} of monodromy-free cc.  From this mapping, one is also naturally lead to consider the notion of covariant edge fields, introduced next.

\begin{defi}[Covariant Edge field]
Let $K$ be a cc a let $\dom \subset \kz$ be non-empty. An \textit{edge field on $\dom$} is a map $$\phi : \dom \dans \ko$$ such that $\phi(v) \in \E_v$ for all $v \in \dom.$ We will sometimes denote the \textit{domain of $\phi$} $\dom$ by $\dom_\phi.$

An edge field $\phi$ in a full graph-based cc $K$ is said to be \textit{covariant} if, whenever an edge $\{v,w\} \in \ko$ is included in $\dom_\phi$ then $\nabla_w^v \phi(v) = \phi(w).$
\end{defi}

Covariant edge fields, encoding the edge mapping induced by the connection, are a key tool in the proof of the Theorem \ref{thminducedcc} where they allow to construct "induced cells" by identifying edge fields that can be interpreted as "tangent" and "normal" to a given cell. In this sense the connection can in fact be interpreted as defining a notion of "parallel transport" in a cc.

\subsection{Full and simple cc} \label{ssecfullnsimplecc}

In this section we introduce the notions of regularity and edge-regularity in a cc and we explain how these notions are implied by the notion of fullness in a cc that is defined next. Fullness is the object of two lemmas and is used as a characteristic property in our definition of simple cc. We finally show that simple cc are precisely the duals of simplicial complexes.

The following notion of $n$-regularity is common to many contexts involving graphs.

\begin{defi}[$n$-regular cc]
A graph-based cc $K$ is said to be $n$-regular if $\abs{\E_v} = n$ for all $v \in \kz$ or equivalently if $\kg$ is a $n$-regular graph.
\end{defi}

When considering cc of rank 2, it is natural to also define the notion of edge-regularity as follows.

\begin{defi}[$n$-edge-regular 2-cc]
A 2-cc $K$ is \textit{$n$-edge-regular} if $\abs{\ce{e}} = n$ for all $e \in \ko.$ 
\end{defi}

As a consequence of the diamond property \ref{cccdiamond}, the notions of regularity and edge-regularity in a cc are related. The precise relation is given in the next remark mentioned in Lemma \ref{lemfullk}.

\begin{rema}\label{remdegv}
Let $n \geq 1$ and $ n-1 \geq m \geq 0.$ An $n$-regular 2-cc is $m$-edge regular if and only if $\abs{\ce{v}} = nm/2$ for all $v \in \kz,$ as showed by the following counting argument:
\begin{align*}
2 \abs{\ce{v}} &= \sum_{C \in \ce{v}} 2 = \sum_{C \in \ce{v}} \abs{ \E_v^C}\\
&= \abs{\{(e,C) \in \E_v \times \ce{v}~|~ e \subset C \}} \\
&= \sum_{e \in \E_v} \abs{\ce{e}} = m \abs{\E_v} = m \deg(v) = mn,
\end{align*}
where the second equality uses that $v$ as two neighbors in $C$ by the diamond property \ref{cccdiamond} and the penultimate equality uses that $K$ is $m$-edge-regular. In particular, in a 2-edge-regular 2-cc we have $\deg(v) = \abs{\ce{v}}.$
\end{rema}

Fullness in a cc, introduced next, is a key characteristic of cc dual to simplicial complexes and formalizes the idea that one cannot add more cell to the set of cells containing a given vertex while complying with Axiom \ref{cccinter}. This dualizes the property of (maximal) simplices in a simplicial complex for which one cannot add more faces, as their faces are already given by all subsets of their set of vertices. Since in what follows we will particularly need to use this assumption for the case of 2-cc, or 2-skeletons of higher dimensional cc, we chose to reserve the terminology "full" to denote the fullness property for the case of 2-cc and only specify a rank $r$ when using the more general notion of $r$-fullness.

\begin{defi}[Full cc]
Let $K$ be a cc of rank $R \geq 2$ and let $2 \leq r \leq R.$ We say that $K$ is \textit{$r$-full} if for all $v \in \kz,$ for all $2 \leq k \leq r$ and for all $S \subset \E_v$ such that $\abs{S}=k,$ there is a (unique, by Lemma \ref{lemaedgesetfullcc}) $r$-cell $x \in K$ such that $ \E_v^x = S.$ We simply say that $K$ is \textit{full} if $K$ is 2-full.
\end{defi}

In particular a $R$-full $R$-cc is pure. The following lemma constitutes a first useful observation one can make about full connected cc.

\begin{lem} \label{lemfullk}
A connected cc $K$ is full if and only if there exists $n\in \N^*$ such that $K$ is a $(n-1)$-edge-regular $n$-regular 2-cc.
\begin{proof}
First, as a consequence of Remark \ref{remdegv}, a $n$-regular 2-cc is $(n-1)$-edge regular if and only if $\abs{\ce{v}} = n(n-1)/2.$
Therefore if $K$ is full and $\{v,w\} \in \ko$ then $\nabla_w^v$ is a bijection from $\E_v$ to $\E_w.$  Also, since $\kg$ is connected, $\kg$ is $n$-regular where $n= \abs{\E_v}$  and 
$$\abs{\ce{v}} = \abs{\ts{\E_v}} = \ts{\abs{\E_v}} = \frac{n(n-1)}{2}.$$
Conversely, if $K$ is a $(n-1)$-edge-regular $n$-regular 2-cc and $v \in \kz$ then $$\frac{n(n-1)}{2} = \abs{\ce{v}} \leq \abs{\ts{\E_v}} = \frac{n(n-1)}{2},$$ so $\abs{\ce{v}} = \abs{\{\{e_1,e_2\} \subset \E_v\}}$ and $K$ is full.
\end{proof}
\end{lem}

An important consequence of $k$-fullness is the following useful result, which will be used several times in the remainder of this chapter.

\begin{lem} \label{lemaedgesetfullcc}
Let $K$ be an $R$-cc and suppose that $K$ is $k$-full for a given $k$ such that $2 \leq k \leq R.$ Then for all $x \in K$ such that $1 \leq \rk(x) \leq k$ and all $v \in x,$ we have 
$$ \abs{\E_v^x} = \rk(x)$$
and $x$ is the unique $\rk(x)$-cell $y$ such that $\E_v^y = \E_v^x.$
\begin{proof}
We proceed by induction on $1 \leq \rk(x) \leq k.$

The claim for $\rk(x) = 1$ is direct. For $x \in \kt$ and $v \in x,$ we have $\abs{ \E_v^x} = 2$ by the diamond property \ref{cccdiamond} and if $x,y \in \kt$ satisfy $\E_v^x = \E_v^y$ then $e \subset x \cap y$ for all $e \in \E_v^x.$ This implies $\rk(x \cap y) \geq 2$ and therefore $x \cap y = x = y.$

Suppose that the claim holds for all cells in $K^{(r)}$ for some $2 \leq r \leq k-1$ and let $x \in K^{[r+1]}$ be such that $v \in x.$ By the property \ref{cccenough}, there exists $x_s \in K^{[s]},$ $s = r-1, r$ such that $$v \in x_{r-1} \subset x_r \subset x.$$ By the diamond property there exists $x_r' \in \kr \setminus \{x_r\}$ such that $x_{r-1} \subset x_r' \subset x$ and by the induction hypothesis there is an edge $e_r \in \E_v^{x_r'} \setminus \E_v^{x_r} \subset \E_v^x.$ Since $K$ is $k$-full, there is $y \in K^{[r+1]}$ such that $\E_v^y = \E_v^{x_r} \cup \{e_r\}.$ Therefore we have that $ \abs{\E_v^{x \cap y}} = r + 1$ and this implies $ \rk(x \cap y) \geq r + 1$ by induction hypothesis. Hence $x \cap y = x = y$ and we have $ \abs{\E_v^x} = r+1 = \rk(x).$ Similarly, the uniqueness part is proven as follows. Let $y \in K^{[r+1]}$ such that $\E_v^y = \E_v^x.$ By the previous argument $\abs{\E_v^{x \cap y}} = \rk(x \cap y)$ and $\abs{\E_v^{x \cap y}} = \abs{\E_v^x} = r + 1$ and this implies that $x \cap y = x = y.$
\end{proof}
\end{lem}

The following then serves as our definition of a simple cc.

\begin{defi}[Simple cc]
A $R$-cc such that $R \geq 2$ is \textit{simple} if it is closed and $R$-full. A $R$-cc such that $R \leq 1$ is \textit{simple} if and only if it is closed.
\end{defi}

The next lemma provides an alternative definition of a simple cc as a graph-based $R$-full $(R+1)$-regular $R$-cc.

\begin{lem} \label{lemregfullcc}
A connected $R$-full $R$-cc is closed if and only if it is graph-based and $(R+1)$-regular.
\begin{proof}
Let $K$ be a connected $R$-full cell complex of rank $R.$ By Lemma \ref{lemfullk} $K$ is $n$-regular and Lemma \ref{lemaedgesetfullcc} implies that $n \geq R$ and that for all $1 \leq r < s \leq R$ the number of $s$-cells containing a given $r$-cell $x$ is given by 
$$\abs{\{S \subset \E_v \setminus \E_v^{x_r} ~|~ \abs{S} = s-r \}} = \binom{n-r}{s-r},$$
where $v \in x_r.$ For the case $s = R, ~r = R-1$ and for any $y \in K^{[R-1]},$ we obtain that
$$\abs{ \cface{y}} = \binom{n-r}{s-r} = \binom{n-R + 1}{1} = n - R + 1 = 2 \quad \text{if and only if} \quad n = R+1.$$
\end{proof}
\end{lem}

Using the previous lemma, we are then able to prove the main result of this section. 

\begin{prop}\label{propsimplecc}
A connected closed cc $K$  is simple if and only if $\kb$ is simplicial.
\begin{proof}
Suppose that $K$ is a connected simple $R$-cc. In order to show that $J = \kb$ is simplicial it is sufficient to prove that any subset of a maximal cell $z$ of $J$ is a cell of $J.$ Since $K$ is $(R+1)$-regular by Lemma \ref{lemregfullcc}, $\dual{z} \in \kz$ is contained in $(R+1)$ edges, which implies that $\abs{\face{z}} = R+1.$ 
By Lemma \ref{lemaedgesetfullcc} we have that
$$ \abs{z} = \abs{ \dual{\dual{z}}} = \abs{ \{y \in \kR ~|~ \dual{z} \in y\}} = \binom{R + 1}{R} = R+1.$$ 
As a consequence of Lemma \ref{lemaedgesetfullcc}, we also have that $\abs{x} = R$ for all $x \in \face{z},$ since each vertex $v \in x \subset z$ corresponds to a configuration $ \dual{z} \in \dual{x} \subset \dual{v} \in K^{[R]},$ which is bijectively mapped to a set $S \subset \E_{\ce{z}}$ such that $\abs{S} = \rk_K(\dual{v}) = R$ and $\ce{x} \in S$ and there are $\binom{R}{R-1} = R$ such sets. Therefore we have $$ \face{z} =\{ z \setminus \{v\} ~|~ v \in z \}.$$
Using that $J \cup \{\emptyset\}$ is closed under intersection, this implies that $J \cap z = \pws{z}$ and thus $J$ is simplicial.

Conversely, if $J = \kb$ is a connected closed simplicial $R$-cc and $v \in \kz,$ let $S \subset \E_v, ~ \abs{S} = r.$ If we set $z := \dual{v}$ we obtain the following inclusions
$$ \{ \dual{e} ~|~ e \in S \} := \{ y_1, \dots , y_r \} \subset \face{z} \subset J^{[R-1]}.$$
Since $J$ is simplicial, for all $1 \leq i \leq r$ we have $y_i = z \setminus \{x_i\}$ for some (distinct) $x_i \in z$ and therefore
$$ y_1 \cap \dots \cap y_r = z \setminus \{x_1, \dots, x_r\} \in J^{[R-r]},$$
which implies that $\dual{y_1 \cap \dots \cap y_r}$ is the $r$-cell of $K$ containing all edges in $S.$
\end{proof}
\end{prop}

\subsection{Digression on polytopes and shellability} \label{secpolytope}

Now that we introduced the notion of simple cell complex, we would like to discuss some links between the cell complexes considered in this thesis and some common notions of discrete geometry that could be relevant when stepping into physical considerations: polytopes and the related notion of shellability for cell complexes. Polytopes or polyhedral complexes appear as a natural choice to generalize simplices or simplicial complexes. For example, in \cite{bds11} three dimensional polyhedra are used as natural generalizations of a tetrahedron to define a model of three-dimensional quantum geometry related to the spin foam or Loop Quantum Gravity approach. The notion of shellability is an assumption one can make on some discretizations to ensure that it satisfies the Euler-Poincaré formula. In spirit, shellable objects are iteratively defined as being constructed using pieces of a given dimension $d$ glued together in a given order, satisfying that the next piece is intersecting the previous ones on a shellable object of dimension $d-1.$ 

Convex polytopes in $\R^n$ (or equivalently polyhedra), can be defined in many equivalent ways (see e.g. Theorem 2.15 \cite{zi95}). We sometimes designate such convex polytopes in $\R^n$ as \textit{embedded polytopes,} as opposed to the notion of abstract polytopes we introduce below. The notion of \textit{face of an embedded polytope} (defined on p.51 of \cite{zi95}) is more general than the notion of face used in this work. Polytopes can be arranged into a \textit{polyhedral complex}: a finite collection $\mathcal{C}$ of polyhedra in $\R^n$ containing every face of each of its polyhedron and such that the intersection of two polyhedra $P,Q \in \mathcal{C}$ is a face of both $P$ and $Q,$ with the convention that the empty polyhedron is a face of any polyhedron. 

Within their combinatorial theory, convex polytopes are often reduced to their face lattice, i.e. the poset of their faces ordered by inclusion. The boundary of a polytope forms a polyhedral complex and the face lattice of a polytope is characterized by the face lattice of its boundary. The face lattices obtained in this way constitute examples of cell complexes as defined here. The term "simple" is used in this context to denote polytopes such that (the face lattice of) their boundary is dual to a simplicial complex. Therefore our notion of simplicity for cell complexes is equivalent to the corresponding notion for polytopes when restricted to cc that can be realized as such.
Typical examples of cc that cannot be realized as a face lattice of a polyhedral complex are given using cc having cells with non-trivial topology. We will come back to this point in the discussion given after the Proposition \ref{propEPformula} on the Euler-Poincaré formula.

In \cite{ms02}, McMullen and Schulte introduced the notion of abstract polytopes, a notion that is designed to possess the main characteristics of the face lattice of polytopes, but without any a priori notion of embedding in Euclidean space nor convexity and thus share some properties with our definition of cell complexes. Their definition of polytopes is however more general than the definition we give next, as it does not require faces to be characterized by their vertices nor to be stable under intersection (including $\emptyset$ as a face). 

In order to introduce the definition of polytopes, we need the following notion of section.

\begin{defi}
A \textit{section} in a cc $K$ is a sub-poset of $K$ of the form 
$$ \ksec{x}{y} := \{ w \in K ~|~ x \subset w \subset y \},$$
where $x \in K \cup \{\emptyset\}$ and $y \in K$ such that $x \subset y.$ A section $\ksec{x}{y}$ is \textit{connected} if $\rk(y) - \rk(x) \leq 2,$ where we set $\rk(\emptyset) = -1,$ or, for all $x \subsetneq w,w' \subsetneq y,$ there exists a sequence $w_0 = w, w_1, \dots, w_k=w'$ such that $x \subsetneq w_i \subsetneq y$  and $w_i \subset w_{i+1}$ or $w_{i+1} \subset w_i$ for all $i = 0, \dots, k-1.$ 
\end{defi}

\begin{defi}[Polytope]\label{defpolytope}
An pure $R$-cc $K$ is a \textit{polytope of rank $R$} if $\abs{K^{[R]}} = 1$ and every section in $K$ is connected.
\end{defi}

For example, if $K$ is a polytope of rank $R$ then in particular $K$ has no $\partial K$-pinch $x$ such that $\rk(x) \geq R-2,$ as it would imply that the section $\ksec{x}{z}$ is not connected, where $z$ is the unique maximal cell of $K.$ In \cite{ms02}, the property that each section is connected is called strong-connectedness. We already use this terminology here to denote the characteristic property of a pseudo-manifold.

The following result gives sufficient conditions for the cells of a local cell complex to be seen as abstract polytopes. It is the only result in this work that do not hold for cell complexes of arbitrary rank.

\begin{prop} \label{proppolytope}
Let $K$ be a local simple cell complex of rank $R \leq 4.$ Then $K \cap x,$ and thus $\kb \cap \ce{x},$ are polytopes for all $x \in K.$
\begin{proof}
It is sufficient to prove the case $R = 4.$ Let $\ksec{a}{b}$ be a section in $K \cap x.$ If $\rk(b)-\rk(a) \leq 2$ then $\ksec{a}{b}$ is connected by definition and if $a = \emptyset, b \in K$ then $\ksec{a}{b}$ is connected since $K$ is cell-connected. 

In the case $a \in \kz, b \in \ks \cup K^{[4]},$ $\ksec{a}{b}$ is connected since $K$ is full and therefore for all $e_1,e_2 \in \E_v$ there is a 2-cell $C \in \kt$ containing $v$ as well as $e_1$ and $e_2.$ 

Finally, for the case $a \in \ko, b \in K^{[4]},$ the section $\ksec{a}{b}$ is connected if for all $C_1, C_2 \in \kt$ such that $a \subset C_i \subset b, i = 1,2,$ there is $B \in \ks$ such that $C_i \subset B \subset b,~ i=1,2.$ Finding such $B$ amounts to find an edge $e \in \kob$ such that 
$$\kzb \ni \dual{b} \in e \subset \dual{a} \in \ksb$$
and such that $e \in \dual{C_1} \cap \dual{C_2}.$ Such an edge exists as $\kb$ is a simplicial complex by Proposition \ref{propsimplecc} and therefore $\dual{a}$ is a tetrahedron and $\dual{C_i}$ is a triangle of $\dual{a}$ for $i=1,2.$
\end{proof}
\end{prop}

We now turn to the notion of shelling of a cc and the proof the Euler-Poincaré formula for shellable complexes in Proposition \ref{propEPformula}. Both the proof of the latter result and the following definition of shelling are closely inspired by the content of Sections 8.1 \& 8.2 of \cite{zi95}, containing also many examples of shellable polyhedral complexes that also constitute examples for our framework.

\begin{defi}[Shelling and shellable cell complex] \label{defshelling}
A shelling of a non-singular $R$-cell complex $K$ is defined in the following recursive manner.

\begin{enumerate}[topsep=2pt, parsep=2pt, itemsep=1pt]
\item For $R = 0,$ $K$ is a point and has a shelling which is the trivial linear order on the 1-element set $\kz = K.$
\item For $R \geq 1,$ a shelling of $K$ is a linear ordering of its facets $\kR = \{z_1, \dots, z_N\}$ such that $\partial(K \cap z_1)$ has a shelling and for all $1 \leq k \leq N,$ the cc
\begin{equation} \label{intersdefshelling}
K \cap \left( z_k \cap \left( \bigcup_{j=1}^{k-1} z_j \right) \right)
\end{equation}
is a non-singular sub-$(R-1)$-cell complex of $K$ which has a shelling that can be completed into a shelling of $\partial(K \cap z_k)$ and such that (\ref{intersdefshelling}) has non-empty boundary for $1<k<N$ and empty boundary for $k = N$ (i.e. (\ref{intersdefshelling}) is equal to $\partial(K \cap z_N)$) only if $k = N$ and $K$ is closed. 
\end{enumerate}
A cc is shellable if it is non-singular and has a shelling.

\end{defi}

We make the following remark about what are 1-dimensional shellings according to our definition, which will be used in the proof of the Euler-Poincaré formula \ref{propEPformula}.

\begin{rema} \label{remashell1cc}
The previous definition differs slightly from the one given in \cite[p.233]{zi95} as in our setup a shellable 1-cell complex is either a tree, if it has a non-empty boundary, or a cycle, if it is closed, as opposed to being any connected graph as allowed by Ziegler's definition.
This allows to dervie the Euler-Poincaré formula also for shellable cell complexes that are not the boundary complex of a polytope.
Also, we require the cc (\ref{intersdefshelling}) to have a non-empty boundary for all $1 < k < N,$ as it is not clear a priori that it is the case in general. This allows to prove that the reverse order of a shelling is also a shelling.
\end{rema}

The barycentric subdivision of a shellable cell complex is shellable, as shown in the next lemma. This is a relatively basic result that we will use in Corollary \ref{corECbdiv}. It is also for example shown in \cite{ab17} that the second derived subdivision of any convex polytope is shellable and that a triangulated manifold is a PL-sphere or PL-ball if and only if for some $m \geq 0$ its $m$-th derived subdivision is shellable. 
There exists however an example of a non-shellable triangulation of a PL-ball (\cite{lic91}) or even simple examples of convex polytopes such as a tetrahedra (\cite{ru58}) admitting non-shellable triangulations. 

\begin{lem} \label{lemshellbdiv}
The barycentric subdivision of a shellable cell complex is shellable.
\begin{proof}
Let $K$ be a shellable cell complex. We prove the statement by induction on the rank $R = \Rk(K).$ We first argue that it is sufficient to prove that the statement holds for each sub-cell complex $K \cap z$  where $z \in \kR.$ This follows from the observation that if we have a shelling order of $K \cap z$ for each maximal cell $z$ then we can obtain a shelling of $\bdiv{K}$ by combining the shellings of each sub-cell complex $K \cap z$ according to the order on $\kR$ defined by the shelling of $K.$ For this we use that $\bdiv{K}$ is a simplicial complex, therefore the cell complex defined at each step using (\ref{intersdefshelling}), with $\bdiv{K}$ instead of $K$, is simply the restriction of $\bdiv{K}$ to a $(R-1)$-simplex. By the previous observation, the proof for the cases $R = 0,1$ is clear. This same observation also implies the induction step relatively directly, by the following argument. By Definition \ref{defshelling} a maximal cell $z \in \kR$ satisfies that $\partial (K \cap z ) = K \cap \face{z}$ is shellable and by induction we also have that its barycentric subdivision $\bdiv{\partial (K \cap z)}$ is shellable. A shelling of the latter cell complex directly provides a shelling of $K \cap z,$ as the map between $\left( \bdiv{\partial(K \cap z)} \right)^{[R-1]}$ and  $\bdiv{(K \cap z)}^{[R]}$ defined by $$ \{x_0 \subsetneq \dots \subsetneq x_{R-1}\} \mapsto \{x_0 \subsetneq \dots \subsetneq x_{R-1} \subsetneq z \}$$ is clearly a poset isomorphism. Hence the image of a shelling of $\bdiv{\partial(K \cap z)}$ through this map is a shelling of $\bdiv{(K \cap z)}.$ Therefore we obtain that $\bdiv{K \cap z}$ is shellable for all $z \in \kR$ and we can conclude that $\bdiv{K}$ is also shellable.
\end{proof}
\end{lem}

For our purposes, the Euler characteristic is defined as follows.

\begin{defi}[Euler characteristic] \label{defeulerchar}
The \textit{Euler characteristic} of a cell complex $K$ of rank $R$ is defined by
$$\chi(K) = \sum_{r=0}^R (-1)^r \abs{\kr}.$$
\end{defi}

As a direct consequence of the latter definition, we can make the following observation, used later in Corollary \ref{corthm2shell}.

\begin{rema} \label{remaeulerchardual}
Since $\abs{\kb^{[R-r]}} = \abs{\kr}$ we have $\chi(K) = (-1)^R \chi(\kb),$ for $K$ closed.
\end{rema}

We are now able to state the Euler-Poincaré formula for cc, which proof is essentially identical to the proof of Corollary 8.17 of \cite{zi95}.

\begin{prop}[Euler-Poincaré formula]\label{propEPformula}
Let $K$ be a shellable $R$-cell complex. Then if $\partial K \neq \emptyset,$ we have
$$ \chi(K) = 1,$$
otherwise, i.e. if $K$ is closed, we have
$$ \chi(K) = 1 + (-1)^R.$$
\begin{proof}
We prove the claim by induction on $R \geq 0$ and $ N = \abs{\kR},$ using that for any sub-cell complex $J,J'$ of $K,$ $\chi$ satisfies 
\begin{equation}\label{chiadditive}
\chi(J) + \chi(J') = \chi(J \cup J') + \chi(J \cap J').
\end{equation}

The claim for $R \leq 1$ and all $N \geq 0$ is clear by Remark \ref{remashell1cc}. For $R \geq 2,$ let $\kR = \{z_1, \dots, z_N\}$ be a shelling of $K.$ 
For $1 \leq k \leq N,$ let $J_k := K \cap \left( \bigcup_{j=1}^k z_j \right).$ By (\ref{chiadditive}), we get
$$ \chi\left( J_k \right) = \chi( K \cap z_k ) + \chi\left( J_{k-1} \right) - \chi\left( (K \cap z_k) \cap J_{k-1} \right).$$
By the induction hypothesis, we have that 
$$  \chi\left( J_{k-1} \right) = 1,$$
as well as 
$$\chi( K \cap z_k) ) = \chi(\{z_k\}) + \chi(\partial(K \cap z_k)) = (-1)^R + 1 + (-1)^{R-1} = 1,$$
since $\partial(K \cap z_k)$ is a closed $(R-1)$-cell complex and the shellability of $K$ implies in particular that the boundary of each of its facet is shellable.

For $1\leq k < N,$ the shelling property and the induction hypothesis also implies
$$ \chi\left( (K \cap z_k) \cap J_{k-1} \right)  = 1.$$
For $k = N,$ we have the following two cases:
\begin{align*}
\chi\left( (K \cap z_N) \cap J_{N-1} \right) = \begin{cases}
\chi\left( \partial(K \cap z_N) \right) = 1 + (-1)^{R-1}, &\text{if } K \text{ is closed,}\\ 
1, &\text{ otherwise.} \end{cases}
\end{align*}
This concludes the proof.
\end{proof}
\end{prop}

We can for example consider the boundary of a $R$-simplex, which clearly defines a closed shellable simplicial complex $K$ of rank $R-1$ and Proposition \ref{propEPformula} therefore implies that $\chi(K) = 1 + (-1)^{(R-1)}.$ As a consequence, the Euler characteristic of a simplex (seen as a polytope) is $1.$

As noted in the next corollary, the previous results imply that the Euler characteristic is invariant under taking the barycentric subdivision, provided that the complex we started with is shellable.

\begin{cor} \label{corECbdiv}
If $K$ is a shellable cell complex then $\chi(\bdiv{K}) = \chi(K).$
\begin{proof}
This is a direct consequence of Lemma \ref{lemshellbdiv} and Proposition \ref{propEPformula}.
\end{proof}
\end{cor}

In our context, cc are not shellable in general and their Euler characteristic can therefore take various values. For example, a 3-cc $T$ could be defined  so that $\partial T$ is a discretization of a torus, seen as a surface of genus 1, and $\abs{T^{[3]}}=1,$ hence $\chi(T) = -1$ (such an cc is given in Figure \ref{figtorus} as an example of cobordism in Chapter \ref{secdualcob}). It is a standard result from Bruggesser and Mani (Theorem 8.11 in \cite{zi95}) that convex polytopes are shellable. Therefore $T$ also provides a case of cc that cannot be realized as the face lattice of  an embedded convex polytope, although it would satisfy the conditions of (abstract) polytope defined in Definition \ref{defpolytope} (the homogeneity of $\partial T$ in particular implies that $\partial T$ is local and non-pinching, which implies the connectedness of all sections in $T$). Similarly, a $R$-cc $K$ having a facet $z \in K^{[R]}$ such that $K \cap z = T$ cannot be realized as the face lattice of a polyhedral complex. In such an example, the $3$-cell $z$ would be a center of an non-homogeneous vertex in $\bdiv{K}$


In the remainder of this section we introduce a definition of 2-shellability that applies to cc of rank higher or equal to 2 and prove that the 2-shellability of a simple cc is equivalent to the shellability of its dual. 2-shellability is characterized only using the cells of the 2-skeleton and is inspired by the notion of "weak-shellability" given by Vince in \cite{vin85}.

\begin{defi}[2-shellable cell complex] \label{def2shellable}
A cell complex $K$ is \textit{2-shellable} if there is a linear ordering on $\kz := \{v_1, \dots, v_N\}$ such that for all $C \in \kt$ the graph $$ K_j \cap (C \setminus \{v_{j+1}\})$$ is connected, where $K_j := K \cap \{v_1, \dots, v_j\}$ is a cell complex. Such linear ordering is called a \textit{2-shelling} of $K.$
\end{defi}

\begin{rema} \label{rem2shellable}
We notice that Definition \ref{def2shellable} implies that the 2-skeleton of a 2-shellable cell complex can be constructed by successively adding each 2-cell in such a way that at each step of the construction, the new 2-cell intersects the previous ones on a simple path. Indeed, using the notation of Definition \ref{def2shellable}, a 2-shelling satisfies the following property. At step $j+1,$ if $C \in \kt$ contains $v_{j+1}$ and $C \cap \{v_1, \dots, v_j\} \neq \emptyset$ then there is no gap between $v_{j+1}$ and the previous vertices in $C,$ in the sense that there exists $1 \leq k \leq j$ such that $ C \supset \{v_{j+1},v_k\} \in \ko.$ The reason for this is that if there were a gap and if we take the first $l > j+1$ such that $$v_l \in C \setminus \{v_1, \dots , v_{j+1}\}$$ we have that $K_l \cap (C \setminus \{v_l\})$ is disconnected, contradicting that $\{v_1, \dots, v_N\}$ is a 2-shelling. We can therefore suppose, without loss of generality, that $l = j+2,$ i.e. the shelling order continues by adding all the vertices in $C$ first, before completing another 2-cell.
\end{rema}

\noi This following technical result will be used in the proof of the next Theorem \ref{thm2shell}.

\begin{lem} \label{claim3prop2shel}
Let $K$ be a simple cc such that $\dual{K}$ is shellable. Let $\{\sigma_1, \dots, \sigma_N\}$ be a shelling of $\dual{K}$ and define $v_i := \dual{\sigma_i},$ for $i = 1, \dots N.$ We also set $K_j := K \cap \{v_1, \dots, v_j\}$ and define the following simplicial complex:
$$\kb_j := \kb \cap \left( \bigcup_{i=1}^j \sigma_i \right).$$
Then we have that $K_j \cap \dual{\omega}$ is connected for all $\omega \in \kb_j$ and all $1 \leq j \leq N.$
\begin{proof}
Suppose by contradiction that there exists $1 < j \leq N$ such that $K_j \cap \dual{\omega}$ is not connected for some $\omega \in \kb_j,$ which implies $\rk(\dual{\omega}) \geq 2.$ Since $\dual{\omega} \in K = K_N$ is connected, we can take $j$ maximal with the property that $K_j \cap \dual{\omega}$ is not connected, so that $K_{j+1} \cap \dual{\omega}$ is connected. Hence there are at least two vertices $v_k,v_l$ in $\dual{\omega} \cap \Ne(v_{j+1})$ that belong to two different connected components of $K_j \cap \dual{\omega}.$ Since $K$ is full, there is a 2-cell $C \in \kt$ containing the edges $\{v_{j+1},v_k\}$ and $\{v_{j+1},v_l\}$ and thus $K_j \cap (C \setminus \{v_{j+1}\})$ is not connected, which contradicts that $\{v_1, \dots, v_N\}$ is a 2-shelling of $K.$ 
\end{proof}
\end{lem}

The following theorem is analogous to Theorem 7 given in \cite{vin85} stated in terms of $n$-graphs. Our proof is novel as the definition of $n$-graph allows to use some equivalence relations between edges that do not exist in general for simple cell complexes.  This result given in terms of simple cc is nevertheless not a generalization of Vince's result since $n$-graphs are more general than our notion of graphs, as in particular they allow multiple edges to contain the same two vertices.

\begin{thm} \label{thm2shell}
A simple cell complex $K$ of rank $R \geq 2$ is 2-shellable if and only if $\kb$ is shellable.
\begin{proof}
Let $K$ be a simple cell complex such that $\kb$ is shellable with shelling $\{\sigma_1, \dots, \sigma_N\}$ and define
$v_i,$ $K_j$ and $\kb_j$ as in Lemma \ref{claim3prop2shel}.

We show that $\{v_1, \dots , v_N \}$ is a 2-shelling of $K.$
For this, let us assume by contradiction that there is $C \in \kt$ such that $ K_j \cap (C \setminus \{v_{j+1}\})$ is not connected. We consider $j$ to be maximal with this property for $C$ fixed ($j$ is well defined since $K \cap C$ is connected). The maximality of $j$ implies that $K_{j+1} \cap C$ is connected, therefore $v_{j+1}$ has two neighbours in $K_j \cap C,$ linked to $v_{j+1}$ by the edges $e,e'.$ Since $\sigma_{j+1}$ is a simplex, the two faces $\dual{e}, \dual{e'}$ of $\sigma_{j+1}$ share the $(R-2)$-cell $\dual{C}.$ But $K_{j+1} \cap C \neq C$ by definition of $j,$ therefore $\dual{C} \notin \kb_{j+1},$ contradicting the intersecting property \ref{cccinter} of the $(R-1)$-cell complex
$$ \kb_{j+1} \cap \left( \sigma_{j+1} \cap \left( \bigcup_{i=1}^{j} \sigma_i \right) \right).$$

Conversely, assume that $\{v_1, \dots, v_N\}$ is a 2-shelling of $K.$ We show that $\{\sigma_1, \dots, \sigma_N\}$ is a shelling of $\kb$ where $\sigma_i = \dual{v_i}.$

Since any subset of the faces $\face{\sigma}$ of an $R$-simplex $\sigma \in \kb$ is a $(R-1)$-sub-cell complex that can be completed into a shelling of $\partial(\kb \cap \sigma),$ it is sufficient to show that for all $1 < j \leq N,$ if $\omega := \sigma_k \cap \sigma_j \neq \emptyset$ for some $1 \leq k < j$ then there exists  $1 \leq l < j$ such that $$w \subset \sigma_l \cap \sigma_j \in \face{\sigma_j}.$$
The latter implies that $\kb_{j-1} \cap \sigma_j$ is indeed a $(R-1)$-cc that has a shelling which can be completed in to a shelling of $\partial(\kb \cap \sigma_J).$
By Lemma \ref{claim3prop2shel}, such $\omega$ satisfies that $K_j \cap \dual{\omega}$ is connected, therefore there is $1 \leq k < j$ such that $\sigma_k \in \dual{\omega}$ and $\{ \sigma_k, \sigma_j\} \in K_j^{[1]},$ i.e. $\rk_{\kb}(\sigma_k \cap \sigma_j) = R-1,$ and this completes the proof.
\end{proof}
\end{thm}

\begin{cor} \label{corthm2shell}
If $K$ is a local simple 2-shellable $R$-cell complex then $K$ satisfies Euler-Poincaré formula $$ \chi(K) = 1 + (-1)^R.$$
\begin{proof}
By Proposition \ref{propsimplecc} and Theorem \ref{thm2shell}, $\kb$ is a shellable simplicial complex, which satisfies Euler-Poincaré formula by Proposition \ref{propEPformula}. As noticed in Remark \ref{remaeulerchardual} and since $K$ is in particular closed, we have 
$$\chi(K) = (-1)^R \chi(\kb) = (-1)^R + 1.$$
\end{proof}
\end{cor}

\subsection{Simply connected cc} \label{ssectionsimplyconccc}

In this section, we introduce a combinatorial notion of homotopy in cell complexes of rank higher or equal to 2 defined using $2$-cells and edges. Although the idea is simply to allow paths to be "deformed through" $2$-cells, the exact formulation of these notion is somewhat technical. For example, in order to define such deformation on every paths, we have to deal with the case were a path goes several times through a given edge. This can create ambiguities on which part of the path is to be deformed through a given $2$-cell. Two avoid this, we added a condition in the definition of such "path move" so that it acts on the part of the path which first goes through the $2$-cell that the move is based on. Similar questions arise when defining the move of paths through an edge $e,$ which acts non-trivially on paths that contains a sub-path of the form $(e,e)$ (using the notations from Section \ref{ssecgraphncc}), i.e. when a path goes through the edge $e$ in both directions in a row.

Let $K$ be a graph-based 2-cc and let $v,w \in \kz.$ We denote by $\paths$ the set of paths in $\kg$ and $\paths(v,w)$ the set of paths  from $v$ to $w$ in $\kg.$ For a 2-cell $C \in \ce{v} \cap \ce{w},$ a simple path $p \in \paths(v,w)$  in $\kg \cap C$ and $C_p$ the connected component of $C$ containing $p,$ we define the \textit{complementary path $\comp{p}{C}$ of $p$ in $C$} to be the sub-path of $l_{C_p}(v)$ from $v$ to $w$ such that $$l_C(v) = p \conc (\comp{p}{C})^{-1},$$
as in Figure \ref{complementarypathpic}, where $l_{C_p}(v)$ was introduced in Definition \ref{defiloops2c}.

\begin{figure}[!h]
\centering
\includegraphics[scale=0.38]{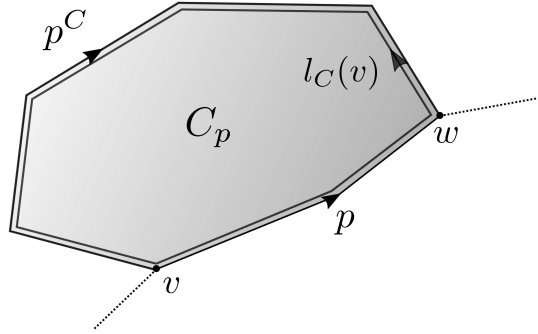}
\caption{\label{complementarypathpic}On this illustration we have a simple path $p$ of length 2 from the vertex $v$ to the vertex $w$ contained in a 2-cell $C.$ The path $p$ can be seen as as sub-path of a longer path going through other edges before $v$ and after $w$ (including the edges in $C$). Since we do not a priori assume cells to be connected, a $2$-cell can have more than one connected component, $C_p$ is therefore the connected component of $C$ containing $p.$ By Lemma \ref{lemrepr2cell} we have that the edges in $C_p$ form a loop $l_C(v)$ that we can define as based at the vertex $v$ and such that $p \leq l_C(v).$ The complementary path $p^C$ of $p$ through $C_p$ is then defined as the sub-path of $l_C(v)$ going from $v$ to $w$ and disjoint from $p.$}
\end{figure}

\begin{defi}[Moves of a path] \label{defmovespath}
Let $C \in \kt$ and let $p$ be a simple path in $\kg \cap C.$ The \textit{2-cell move of $p$ through $C$} is a function $m_C^p : \paths \dans \paths$ defined by
\begin{align*}
m_C^p(q) = \begin{cases}
q_1 \conc \comp{p}{C} \conc q_2, &\text{if } q = q_1 \conc p \conc q_2 \text{ with } q_1,q_2 \in \paths, p \not\leq q_1; \\
q &\text{otherwise.}
\end{cases}
\end{align*}
We can see that $m_C^p$ is a bijection with inverse $m_C^{\comp{p}{C}}.$

Let $e \in \ko$ and $v \in e.$ The \textit{edge-move through $e$ based at $v$} is a function $m_{e,v} : \paths \dans \paths$ defined by
\begin{align*}
m_{e,v}(q) = \begin{cases}
q_1 \conc q_2, &\text{if } q = q_1 \conc (e,e) \conc q_2 \text{ with } q_i \in \paths, e \cap q_i^{[0]} = {v}, i = 1,2,  (e,e) \not\leq q_1; \\
q &\text{otherwise.}
\end{cases}
\end{align*}
In this case, the inverse of $m_{e,v}(q)$ is the maps $m^{-1}_{e,v}$ which adds $e$ to a path $q$ if $e \cap q^{[0]} = \{v\}$ and leaves $q$ invariant otherwise.
We denote by $\M_1$ the set of edge-moves and their inverses and by $\M_2$ the set of 2-cell-moves and their inverses and define the \textit{set of moves} to be $\M = \M_1 \cup \M_2.$
Note that if $q$ is a path from $v$ to $w$ then so is $m(q)$ for all $m \in \M.$ Therefore moves of paths can be seen as functions on $\paths(v,w).$
\end{defi}

\begin{defi}[Homotopy, homotopy class of paths and contractible cycle] \label{defhomclasspath}
An \textit{homotopy} is a map $H : \paths(v,w) \dans \paths(v,w)$ given by a finite composition of moves in $\M.$ Two paths $p_1,p_2 \in \paths(v,w)$ are \textit{(homotopically) equivalent,} noted $p_1 \sim p_2,$ if there exists an homotopy $H$ such that $H(p_1) = p_2.$ The homotopy class of a path $p$ is denoted by $[p].$ A cycle $p$ based at $v$ is \textit{contractible} if $[p] = [\emptyset_v].$
\end{defi}

\begin{defi}[Simply connected cc]
A graph-based cc $K$ is \textit{simply connected} if every cycle in $K$ is contractible in $\kc.$
\end{defi}

The next remarks on simple connectedness will be used in the discussion at the end of this chapter.

\begin{rema} \label{remahomotopy}
It seems clear that the notion of homotopy used here is equivalent to the usual notion of homotopy of paths on a manifold, in the sense that if $K$ is a cell complex which geometrical realization is a discretization of a cellular manifold (implying that $K \cap x$ is geometrically realized as a ball for all $x \in K$) then $K^{(2)}$ is simply connected (in the sense of Definition \ref{defhomclasspath}) if and only if $M$ is simply connected in the usual topological sense. A sketch of an argument proving this can be formulated as follows, using that in this case $\dual{K}$ also admits a geometrical realization discretizing $M$ (this can be argued using the fact that $\bdiv{K} \cong \bdiv{\kb}$ by the Proposition \ref{propsubdivdual}). A path $p$ in $M$ can be projected onto $\kg$ by decomposing $p$ into sub-paths such that there is one sub-path for each crossing of $p$ from one maximal cell of $\dual{K}$ to another. One can then continuously deform each sub-path as to be included into the edge of (the geometric realization of) $K$ composed of the two corresponding maximal cells of $\dual{K}$ and the combination of these deformations defines a usual homotopoy of paths on $M.$ The resulting path can be interpreted as a path on the graph $\kg$ and an homotopy as defined in these notes also corresponds to a usual homotopy on $M.$
\end{rema}

\begin{rema} \label{remashellissimplyconn}
It is a well-known result (\cite{bj84}) that shellable pseudomanifolds are homeomorphic to a sphere (using the notion of geometrical realization mentioned in Section \ref{sectopoconsid}).
As a consequence of Theorem \ref{thm2shell},  a 2-shellable simple local cell complex $K$ satisfies that $\dual{K}$  is simply connected.
\end{rema}

\subsection{Monodromy-free and even cc} \label{ssecmonod}

We now turn to the notions of monodromy-freedom and evenness in a cc. The latter corresponds simply to the property of having an even number of vertices in each (connected component of each) $2$-cell. However we do not yet fully understand how evenness can be interpreted in the context of the reconstruction results of Section \ref{ssecreconstthm}, although we suspect that it has something to do with some notion of orientation. Monodromy-freedom  essentially means that the connection does not produce a monodromy around any $2$-cell and it is the notion we focus on first.

In the case of full graph-based cc, the notation for the connection introduced in Section \ref{ssecindedgemap} can be extended to paths in the following way.

\begin{defi}[$\nabla_p$] \label{definablapaths}
For a full graph-based 2-cc $K,$ we have that, for all $\{v,w\} \in \ko,$
$$\E_v^{\dual{w}} = \{ e \in \E_v ~|~ \ce{e} \cap \ce{\{v,w\}} \neq \emptyset \}= \E_v.$$ 
Therefore if $p = (e_1,\dots, e_k)$ is a path in $\kg,$ the map  $$\nabla_p := \nabla_{v_k}^{v_{k-1}} \dots \nabla_{v_1}^{v_0}$$ is a bijection from $\E_{v_0}$ to $\E_{v_k},$ where $e_i = \{v_{i-1},v_{i}\}$ for $1 \leq i \leq k.$ Also, $\nabla_p^{-1} = \nabla_{p^{-1}},$ $\nabla_{p \conc p'} = \nabla_{p'} \nabla_p$ and $\nabla_{\emptyset_v} = \id_{\E_v}.$
We say that $\nabla_p e \in \E_{v_k}$ is \textit{associated to} $e \in \E_{v_0}$ by $p.$

If $K$ is a full graph-based cc and $p$ is a path in $\kg$ then the connection $\nabla_p$ in $K$ is defined as the corresponding connection in $\kc.$
\end{defi}

Let $C$ be a connected component of a 2-cell of a full graph-based cc $K.$ Recall that $$\E_C := \{ e \in \ko ~|~ \abs{e \cap C} = 1 \}.$$ For $\{v,w\} \subset C,$  $\nabla_w^v$ maps $\E_v \cap \E_C$ bijectively onto $\E_w \cap \E_C,$ since if $e \not\subset C$ then $\nabla_w^ve \not\subset C$ by Lemma \ref{nabla} \ref{nabla1}. Therefore $\nabla_{l_C(v)}$ restricted to $\E_v \cap \E_C$ is a bijection from $\E_v \cap \E_C$ onto itself.

We can now introduce our definition of monodromy-free cell, illustrated in Figure \ref{monodromyfreecellpic} and define the notion of monodromy-free cc.

\begin{defi}[Monodromy-free cc] \label{defmonodromyfreecc}
A connected component $C$ of a 2-cell in a full graph-based cc $K$ is said to be \textit{monodromy-free} if $\nabla_{l_C(v)}$ restricted to $\E_v \cap \E_C$ is the identity map on $\E_v \cap \E_C$ for some $v \in C$ (and hence for all $v \in C$). We say that a full graph-based cc $K$ is \textit{monodromy-free} if every connected component of every 2-cell in $\kt$ is monodromy-free.
\end{defi}

\begin{figure}[H]
\centering
\includegraphics[scale=0.38]{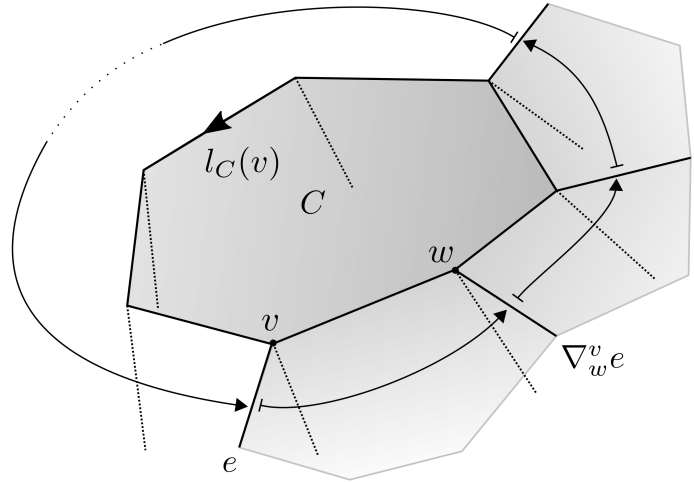}
\caption{\label{monodromyfreecellpic} In this figure we see an example of a monodromy free (connected component of a) 2-cell $C,$ seen as a cell in a 3 dimensional cc. We also pictured the loop $l_C(v)$ based at the vertex $v$ and how the connexion $\nabla_{l_C(v)}$ maps the edges of $\E_C$ around $C.$ In this case, since there are only two edges in $\E_{v'} \cap \E_C$ for all $v' \in C,$ it is sufficient to check that say $e$ as on the figure is mapped to itself by $\nabla_{l_C(v)}$ to deduce that $C$ is monodromy-free.}
\end{figure}

The following Lemma \ref{lemmonodrfree} shows that the monodromy-freedom condition naturally arises when considering $R$-full cell complexes for $R \geq 3.$ This fact motivates our definition of monodromy-free cell complexes, as it is a key assumption when considering the reconstruction of simple cell complexes from their 2-skeleton in Section \ref{ssecreconstthm}. We recently found that a similar notion was introduced in the last chapter of \cite{dpp00}, where several conditions are discussed in order to establish whether some space defined as "simplicial glueings" in dimensions 3 and 4 define a manifold. A simplicial glueing is simply a set of simplices of a given dimension where faces are identified according to some face-pairing. One of these conditions, called "$\mathsf{Cycl}$" is stated in dimension 4 and gives a condition on each triangle (i.e. 2-simplices) that is precisely the condition of monodromy-freedom as defined here. It is suggested that one can easily devise an extension to all dimensions and we believe that the definition of monodromy-free cell complex indeed corresponds to the extensions of the condition $\mathsf{Cycl}$ to dimensions higher or equal to 3.

\begin{lem} \label{lemmonodrfree}
A 3-full graph-based cell complex is monodromy-free.
\begin{proof}
Let $[v_1 \dots v_N]$ be a representation of  the connected component of a 2-cell $C \in \kt$ and let $l: v_1 \pathto v_1$ be the loop associated with this representation based at $v_1$ and starting with the edge $\{v_1,v_2\}.$ By the diamond property \ref{cccdiamond}, if $e \in \E_{v_1} \cap \E_C,$ there are 2 edges $e_0,e_1 \in \E_v$ included in $C$ and since $K$ is 3-full, there is a 3-cell $B$ containing $e_0,e_1$ and $e.$ Since $K$ is full, there is a unique 2-cell $C_e$ containing $e$ and $e_1$ and 
by Axiom \ref{cccinter} since $e,e_1 \subset B \cap C_e$
we have $C_e \subset B.$ Therefore $\nabla_{v_2}^{v_1} e \subset B.$ By the same reasoning, there is a unique edge $e_2 \in \E_{v_2} \setminus \{e_1\}$ included in $C$ and there is a unique 2-cell $C_e'$ containing $\nabla_{v_2}^{v_1}e$ and $e_2$ and contained in $B,$ which implies $\nabla_{v_3}^{v_2}e \in B.$ By repeating this argument $N$ times we get $\nabla_{l} e \in B$ and this implies that $\nabla_{l} e = e,$ by Lemma \ref{lemaedgesetfullcc}, since $K$ is 3-full.
\end{proof}
\end{lem}

The notion of evenness of a $2$-cell and the corresponding notion for cc, introduced in the following definition, will also be crucial for the reconstruction in Section \ref{ssecreconstthm}. The idea is that the assumption of evenness leads to the invariance of the connection $\nabla_p$ under homotopies, as shown in Lemma \ref{nablamove}. One can better picture the effect of even $2$-cell versus odd (non-even) $2$-cell on the connection using Figure \ref{edgemapping2cellpic}.

\begin{defi}[Even $2$-cell and even cc]
A connected component of a 2-cell is said to be \textit{even} if it contains an even number of vertices. A cc $K$ is said to be \textit{even} if every connected component of every 2-cell in $\kt$ is even.
\end{defi}

\begin{rema}\label{remaeven2cell}
If $C$ is even and monodromy-free then for all $v \in C$ we have $$\nabla_{l_C(v)} = \id_{\E_v},$$ since $\nabla_{l_C(v)}$ induces a pairing of the edges contained in $C.$
\end{rema}

\begin{figure}[!h]
\centering
\includegraphics[scale=0.43]{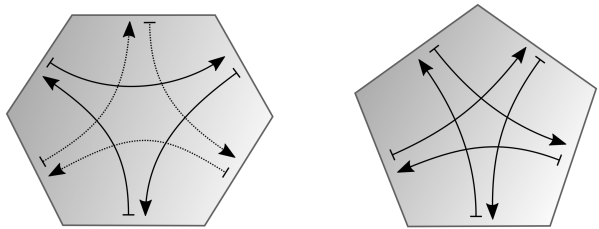}
\caption{\label{edgemapping2cellpic}In these pictures, we see  how the connection acts on the edges within an even 2-cell (on the left) versus within an odd 2-cell (on the right). The even case then leads to a pairing of the edges induced by the connection.}
\end{figure}

An important example of even cc is given in the following lemma.

\begin{lem} \label{lemdualbdiviseven}
The dual of the barycentric subdivision of a closed cc is an even cc.
\begin{proof}
Let $J$ be a closed $R$-cc and let $K = \dual{\bdiv{J}}.$ A 2-cell $C \in \kt$ is then of the form $C = \dual{y}$ where 
$$y = \{x_1 \subsetneq \dots \subsetneq x_{R - 1}\} \in \bdiv{J}^{[R-2]}.$$ 
We can then see that elements in $\dual{y}$ come in pairs by the following observation. Since $\bdiv{J}$ is a simplicial complex, every $z \in \dual{y} \subset \bdiv{J}^{[R]}$ satisfies $\abs{z \setminus y} = 2.$ Since $z$ defines a totally ordered subset of $J,$ the minimum rank of an element in $z \setminus y =: \{x, x'\},$ say $\rk(x),$ is strictly less than $R.$ This means that $z = y \cup \{x \subsetneq x'\}$ and by the diamond property \ref{cccdiamond} (and since $J$ is graph-based for the case $\rk(x) = 0$), there exists a unique $x''$ such that $\{ y \cup \{x \subsetneq x''\}, z\} \subset \dual{y}.$
\end{proof}
\end{lem}

The following lemma states under what assumptions the connection $\nabla_p$ is invariant under homotopies, which is a key property in order to prove the Extension Lemma \ref{lemextension} in the next section. 

\begin{lem}\label{nablamove}
Let $K$ be a full, graph-based, monodromy-free, even cc and let $q \in \paths$ and $m\in \M.$ Then we have $$\nabla_q = \nabla_{m(q)}.$$ In particular if $p \sim q$ then $\nabla_p = \nabla_q.$
\begin{proof}
If $m(q) = q$ it is trivial. If not then either $m = m_C^p$ for some 2-cell $C$ and a simple path $p$ in $\kg \cap C$ and $q = q_1 \conc p \conc q_2,$ or $m = m_{e,v}$ (or $m = m_{e,v}^{-1}$) for some edge $e$ and vertex $v$ such that $v \in e \cap q^{[0]}$ and $q = q_1 \conc (e,e) \conc q_2.$

In the first case $\nabla_{l_C(v)} = \id_{\E_v}$ for all $v \in C,$ by Remark \ref{remaeven2cell}, therefore we can choose a representation of $C$ and some $v \in C$ such that $l_C(v) = (\comp{p}{C})^{-1} \conc p.$ One then obtains the following relations
\begin{align*}
\nabla_{m(q)} &= \nabla_{q_1} \nabla_{\comp{p}{C}} \nabla_{q_2} \\
&= \nabla_{q_1} \nabla_{l_C(v)} \nabla_{\comp{p}{C}} \nabla_{q_2}\\
&= \nabla_{q_1} \nabla_p \nabla_{\comp{p}{C}}^{-1} \nabla_{\comp{p}{C}} \nabla_{q_2} = \nabla_q.
\end{align*}

In the second case, we have that $\nabla_{(e,e)} = \nabla_v^w \nabla_w^v = \id_{\E_v}$ assuming $e = \{v,w\}$ and $(e,e)$ in $q = q_1 \conc (e,e) \conc q_2$ is a path from $v$ to $v,$ so that we have
\begin{align*}
\nabla_{m(q)} &= \nabla_{q_1} \nabla_{q_2}\\
&= \nabla_{q_1}  \nabla_{(e,e)} \nabla_{q_2} = \nabla_q.
\end{align*}
The case $m = m_{e,v}^{-1}$ is similar to the second case.
\end{proof}
\end{lem}

\subsection{Extension Lemma  and  reconstruction} \label{ssecreconstthm}

In this section, we prove that local even full monodromy-free simply connected 2-cell complexes can be associated a corresponding "ambient" cell complex of rank $R \geq 3$ having the original 2-cell complex as its 2-skeleton. The main step for this construction, carried out in Theorem \ref{thminducedcc}, is to define a collection of "induced" sub-2-cell complexes  whose vertices then serve as higher dimensional cells of the ambient cell complex. In particular, Corollary \ref{correconst2} shows that the ambient cell complex associated to the 2-skeleton of a local even simple simply connected cell complex is equal to the original cell complex. This result is in line with a number of results about the reconstruction of the face lattice of (embedded) polytopes, see e.g. \cite{bay18} for a recent review on this topic. A well-known result from \cite{bm87} and \cite{kal88} states that the face lattice of a simple polytope, i.e. a polytope which dual is simplicial, is determined by its 1-skeleton. As we discussed in Section \ref{secpolytope}, here we are dealing with local cell complexes that are more general than the face lattice of complexes of polytopes but we have the additional assumptions of simple connectedness and evenness. We are not aware of an equivalent result in the literature for complex of polytopes.

The first key result of this section is the next "Extension Lemma" \ref{lemextension} which will be essential for the proof of Theorem \ref{thminducedcc}. The assumptions of fullness, evenness, monodromy-freedom and simple connectedness introduced in this chapter can therefore be considered as chosen so that this lemma holds. The proof of the Extension Lemma is efficiently forumulated using covariant edge fields. One reason for this is that when defined on an appropriate domain, covariant edge fields  satisfy the property that their value over their entire domain is determined by their value at one vertex. We formulate this fact more precisely in the following remark on which the proof of the Extension Lemma relies.

\begin{rema}\label{remedgefield}
If $\phi$ is a covariant edge field in a full graph-based cc $K$ and if $p$ is a path in $\kg \cap \dom_\phi,$ from $v$ to $w$ then $\nabla_p \phi(v) = \phi(w).$ In particular if $\kg \cap \dom_\phi$ is connected then $\phi$ is determined by its value at any given vertex in $\dom_\phi.$  
\end{rema}

We now formulate the Extension Lemma and its proof in terms of growing domains of covariant edge fields as it incorporates well the idea that will be used in our formulation of the proof of Theorem \ref{thminducedcc}. Considering that we assume the entire cc $K$ to be simply connected, one can also read the result simply using $\dom_\phi = \{v\}$ for any $v \in \kz$ and $D = \kz.$

\begin{lem}[Extension Lemma]\label{lemextension}
Let $K$ be a graph-based, full, even, monodromy-free, simply connected cc and let $\phi : \dom_\phi \dans \ko $ be an covariant edge field in $K$ such that $\kg \cap \dom_\phi$ is connected. If $D \subset \kz$ such that $\kg \cap (\dom_\phi \cup D)$ is connected then there exists a unique covariant edge field $\hat{\phi}: \dom_{\hat{\phi}} \dans \ko$ such that $\dom_{\hat{\phi}} = \dom_\phi \cup D$ and $\hat{\phi}|_{\dom_\phi} = \phi.$
\begin{proof}
Let $v \in D,$ $w \in \dom_\phi$ and let $p$ be a path from $w$ to $v$ in $\kg \cap (\dom_\phi \cup D),$ as in Figure \ref{extensionlemmapic}. The main idea of the proof is to define $\hat{\phi}(v) = \nabla_{p}\phi(w)$ and to show that it is well defined. This amounts to show that if $w' \in \dom_\phi$ and $p'$ is a path from $w'$ to $v$ then there is a path $q$ from $w$ to $w'$ in $\dom_\phi$ and we have
\begin{align*}
\nabla_{p} \phi(w) &= \nabla_{p^{-1} \conc q \conc p'} \nabla_{p} \phi(w) \\
&= \nabla_{p'} \nabla_q \phi(w) = \nabla_{p'} \phi(w').
\end{align*}
This is indeed the case by the following argument. Since $K$ is simply connected and even, Lemma \ref{nablamove} implies that every cycle $q'$ in $K$ satisfies $\nabla_{q'} = \nabla_{\emptyset_v} = \id_{\E_v}$ for any $v \in (q')^{[0]}.$  The first equality then follows from the fact that $p^{-1} \conc p \conc p'$ is a cycle. For the second equality, we recall that $\nabla_{p_1 \conc p_2} = \nabla_{p_2} \nabla_{p_1}$ and the last equality is a consequence of $\nabla_q \phi(w) = \phi(w')$ as noticed in Remark \ref{remedgefield}. The latter remark also gives an argument for the uniqueness of $\hat{\phi}.$ Finally $\hat{\phi}$ defined thereby is a covariant edge field since if $\{v,v'\}$ is an edge of $\kg \cap D$ and $p$ is a path linking $v$ with $w \in \dom_\phi$ then $p' = p \conc (\{v,v'\})$ is a path linking $v'$ with $w$ therefore 
$$\hat{\phi}(v') = \nabla_{p'} \phi(w) = \nabla_{v'}^v \nabla_{p} \phi(w) = \nabla_{v'}^v \hat{\phi}(v).$$
\end{proof}
\end{lem}

\begin{figure}[!h]
\centering
\includegraphics[scale=0.43]{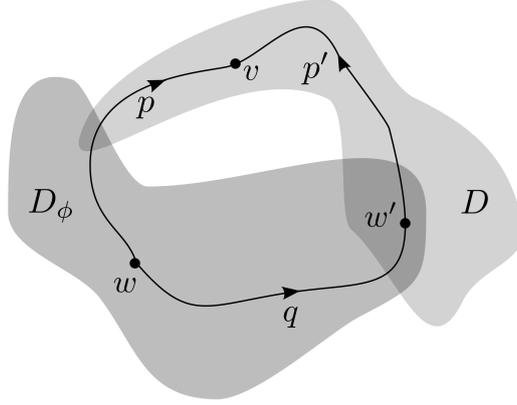}
\caption{\label{extensionlemmapic}Configuration of Lemma \ref{lemextension}.}
\end{figure}

\begin{rema}
As a consequence of Lemma \ref{lemextension}, a full, even, monodromy-free, simply connected and connected cc $K$ is \textit{properly $(R+1)$-edge-colorable,} meaning that one can color each edge of $K$ with $R+1$ colors so that no two edges sharing a vertex have the same color. Another way to say this is that $\kg$ is an $(R+1)$-graph, where the notion of $n$-graph is a tool used in works on so-called crystallizations of topological manifolds (see e.g. in \cite{lm85}, \cite{vin85} or \cite{aphos13}), introduced in the research related to minimal atlases of topological manifold initiated by Pezzana and al. (see  \cite{fgg86} for a review on this topic). However, not all $n$-regular graphs are $n$-graph, i.e. properly $n$-edge-colorable.
\end{rema}

Now comes the main result of this chapter, in which we define the notion of \textit{induced sub-cell complex.}

\begin{thm} \label{thminducedcc}
Let $K$ be a local even full monodromy-free simply connected 2-cell complex. Then for all $v \in \kz$ and all $S_v \subsetneq \E_v$ such that $r := \abs{S_v} \geq 2$ there is a unique sub-cell complex $J(v,S_v)$ of $K$ such that $J(v,S_v)$ is a $r$-regular local even full monodromy-free connected 2-cell complex and
\begin{equation} \label{conditionlemreconstr}
J(v,S_v)^{[1]} \cap \E_v = S_v.
\end{equation}
Moreover, $J(v,S_v)$ satisfies the following properties:
\begin{enumerate}[label=\textbf{\arabic*}),topsep=2pt, parsep=2pt, itemsep=1pt]
\item \label{thminducedcc1} $J(v,S_v) = J(v',S'_{v'})$ if and only if $v' \in J(v,S_v)^{[0]}$ and $S'_{v'} = J(v,S_v)^{[1]} \cap \E_{v'},$
\item \label{thminducedcc2} $S_v \subsetneq S'_v$ if and only if $J(v, S_v) \lneq J(v,S'_v).$
\item \label{thminducedcc3} If $\abs{S_v}=2$ then $S_v = \E_v^C$ for some $C \in \kt$ and $J(v,S_v) = K \cap C.$
\end{enumerate}
We call such $J(v,S_v)$ an induced sub-cell complex of $K.$
\begin{proof}
By Lemma \ref{lemfullk}, $K$ is $n$-regular and $(n-1)$-edge-regular for some $n \geq 2.$ 
Let $v \in \kz$ and $S_v \subset \E_v.$ By the Extension Lemma \ref{lemextension}, for any given $e_v \in S_v$ and for all $w \in \kz$ there is an edge $e_w$ in $\E_w$ associated to $e_v,$ meaning that the map $w \mapsto e_w$ is an induced edge field. Hence for all $w \in \kz$ we can define the set $S_w \subset \E_w$ of edges in $\E_w$ associated to the edges in $S_v.$ Moreover, for each $w \in \kz,$ we can define the following set of 2-cells
\begin{equation} \label{equdefCe}
\Ce_w := \{ C \in \ce{w} ~|~ \E_w^C \subset S_w\},
\end{equation}
where we used again the notation $\E_w^C := \{ e \in \E_w ~|~ e \subset C\}$ and $\dual{w}$ is the dual set of $w$ in the 2-cc $K.$

The strategy of the proof is to inductively define a finite sequence of connected domains $\dom_k \subset \kz, ~ 1 \leq k \leq N$ such that for all $1 \leq k \leq N$ we have
\begin{align}
\E^k_w &:= \E_w \cap \left( K \cap \dom_k \right)^{[1]}  \subset S_w, \label{lemedgewcond}\\
\Ce_w^k &:= \{ C \in \ce{w} ~|~ \E_w^C \subset \E_w^k\} \subset \Ce_w, \label{lemcellwcond}
\end{align}
for all $w \in \dom_k$ and such that $\dom_N = J(v,S_v)^{[0]}.$ Since the inclusion in (\ref{lemedgewcond}) implies the inclusion on line (\ref{lemcellwcond}), we only need to check that the first inclusion holds.

Define $\dom_1 = \{v\},$ let $k \geq 1$ and suppose that we defined the domains $\dom_j$ for $1 \leq j \leq k$ such that (\ref{lemedgewcond}) is satisfied for all $w \in \dom_j$ and for all $1 \leq j \leq k.$ Let us fix $w_0 \in \dom_k$ and a 2-cell
$$C_0 \in \Ce_w \setminus \Ce_w^k.$$ 
We then define $\dom_{k+1} := \dom_k \cup C_0,$ which implies that for all $w \in \dom_{k+1},$ $$\E_w^{k+1} = \E_w^k \cup \E_w^{C_0},$$
where $\E_w^{C_0} = \emptyset$ if $w \notin C_0.$ The definition of the sets $S_w, ~w \in \kz$ (relying on the induced edge field $w \dans \E_w$) and the fact that $C_0 \in \Ce_{w_0}$ imply that $C_0 \in \Ce_w$ for all $w \in C_0.$ We then have that $\E_w^{C_0} \subset S_w$ for all $w \in C_0$ and therefore (\ref{lemedgewcond}) is satisfied for all $w \in \dom_{k+1}.$

We can then define $N \in \N^*$ to be the first integer such that 
\begin{equation}\label{lemfinalcondcell}
\Ce_w^N = \Ce_w \quad \forall w \in \dom_N.
\end{equation}
An consquence of (\ref{lemfinalcondcell}) we also get 
\begin{equation} \label{lemcndtnedges}
\E_w^N = S_w \quad \forall w \in \dom_N,
\end{equation}
by the following argument. Since $\abs{S_w} \geq 2$ for all $w \in \kz,$ if $e \in S_w$ then there exists $e'\in S_w \setminus \{e\}.$ Since $K$ is full, we can also pick $C \in \ce{e} \cap \ce{e'}.$ Hence the definition of $\Ce_w$ implies that  $C \in \Ce_w$ and (\ref{lemfinalcondcell}) in turn implies that $C \in \Ce_w^N.$ The definition of $\Ce_w^k$ in (\ref{lemcellwcond}) for $k=N$ yields $\E_w^C \subset \E_w^N$ which in particular implies that $e \in \E_w^N$ since $\E_w^C = \{e,e'\}.$

We therefore have that $J = J(v, S_v) := (J^{[0]}, J^{[1]}, J^{[2]})$ defined by
$$J^{[0]} := \dom_N, \quad J^{[1]} := \bigcup_{w \in J^{[0]}} S_w, \quad J^{[2]} := \bigcup_{w \in J^{[0]}} \Ce_w,$$
is a sub-2-cell complex of $K.$ Axiom \ref{cccrank} is clearly satisfied for $J,$  as well as Axioms \ref{cccinter}, \ref{cccenough} and \ref{cccdiamond} since they are satisfied for $K$ and if $C \in J^{[2]}$ then the cc $K \cap C \subset J.$ The cc $J$ is clearly connected, cell-connected, even, full and monodromy-free. The property (\ref{lemcndtnedges}) then implies that 
\begin{equation} \label{equpropedgesetJ}
J^{[1]} \cap \E_w = S_w \quad \forall w \in J^{[0]}.
\end{equation}

Let's first note that point \ref{thminducedcc3} is clear as it corresponds to the case where $N=2$ in the above construction.

We can now argue that property \ref{thminducedcc1} is satisfied. Fix $(v,S_v)$ and $J= J(v,S_v)$ as defined above. Assume that we pick
$v' \in \kz$ and $S'_{v'} \subset \E_{v'}$ and let $J' = J(v',S'_{v'})$ be defined by the construction above, using $(v',S'_{v'})$ instead of $(v,S_v).$ It is clear that if $J = J'$ then $v' \in J^{[0]}$ and $S'_{v'} = J^{[1]} \cap \E_{v'}$ by property (\ref{equpropedgesetJ}) applied to $J' = J.$ Therefore it suffices to show to converse statement. For this we assume that $v' \in J^{[0]}$ and $S'_{v'} = J^{[1]} \cap \E_{v'}$ and we show that $J'= J(v', S'_{v'}) = J.$ First we notice that $S'_{v'} = S_{v'},$ the set of edges in $\E_{v'}$ associated to $S_v.$ Therefore, for all $w \in \kz,$ the set of edges in $\E_w$ associated to $S'_{v'}$ is identical to $S_w,$ the set of edges in $\E_w$ associated to $S_v.$ The corresponding sets of 2-cells also remain unchanged, i.e. 
$$\Ce'_w = \{ C \in \dual{w} ~|~ \E^C_w \subset S'_w \} = \Ce_w, \quad \forall w \in \kz.$$ 
Let's define an equivalence relation on $\kz$ from the set of edges in $\ko$ associated to a given set $S \subset \E_v$ as follows. Two vertices $v_0$ and $v_1$ are called "$S$-equivalent" if they can be linked by a sequence of 2-cells in the following manner: there are sequences $\{C_i\}_{i = 1}^M \subset K^{[2]}$ and $\{w_i\}_{i=0}^M \subset K^{[0]}$ such that $v_0 := w_0 \in C_1,$ $v_1 = w_M$ and 
$$C_i \in \Ce_{w_i} \cap \Ce_{w_{i-1}} \quad \forall i = 1, \dots, M,$$ where $\Ce_w$ for all $w \in \kz$ is defined by the formula (\ref{equdefCe}) choosing $S_w$ to be the set of edges in $\E_w$ associated to $S.$
We then have that $v_0$ and $v_1$ are vertices of $J$ if and only $v_0 \in J$ and $v_1$ is $S_v$-equivalent to $v_0.$
By our assumptions, $v$ and $v'$ are $S_v$-equivalent and, as we noticed before, $S'_{v'}$-equivalence classes are identical to $S_v$-equivalence classes and this indeed implies that $J' = J.$

Finally, the implication that $S_v \subsetneq S'_v$ if $J(v, S_v) \lneq J(v,S'_v)$ of point \ref{thminducedcc2} is clear and the reversed implication is a consequence of the following observation: the inclusion $S_v \subsetneq S'_v$  implies that for each vertex $w \in \kz$ the set $S_w$ of edges associated to $S_v$ is also strictly included in the set $S'_w$ of edges associated to $S'_v.$ Similarly, the corresponding sets of 2-cells satisfy $\Ce_w \subsetneq \Ce'_w$ and this implies that $J(v, S_v)^{[i]} \subsetneq J(v, S'_v)^{[i]}$ for $i = 0,1,2.$ 
\end{proof}
\end{thm}

Our formulation of the proof of Theorem \ref{thminducedcc} directly provides an algorithm to construct the induced sub-cell complexes associated to a local even full monodromy-free simply connected 2-cell complex. Using the same notations as in the proof, the algorithm simply consists in the following steps, starting with all couples of the form $(v,S_v)$ marked as "unselected" and running until all such couples have been selected:
\begin{itemize}[topsep=2pt, parsep=2pt, itemsep=1pt]
\item select one unselected couple $(v,S_v),$
\item construct $J(v,S_v)$ by iteratively adding 2-cells to a domain $\dom_k$ as in the proof,
\item at each step, mark each couple $(w,S_w)$ as "selected" where $w \in \dom_{k+1}$ and $S_w$ are the edge in $\E_w$ associated to $S_v.$
\end{itemize}

We introduce the following notion in order to simplify the formulation of the next proofs. The idea is simply to connect the induced cell complex defined in Theorem \ref{thminducedcc} to the Definition \ref{defccc} of a cell complex, where a cell is a subset of vertices.

\begin{defi}[Collection of induced $r$-cells]
Let $K$ be a local even full monodromy-free simply connected 2-cell complex. $K$ is $n$-regular for some $n \in \N^*$ by Lemma \ref{lemfullk}) and let us also fix an integer $2 \leq r \leq n -1.$ We denote by $\Ce^K_r$ the collection of vertex sets of all $r$-regular induced sub-cell complexes of $K$ given by Theorem \ref{thminducedcc}. $\Ce^K_r$  is called the \textit{set of induced $r$-cells of $K.$} By convention we set $\Ce_r^K = \emptyset$ for all $r > n.$
\end{defi}

Note that the point \ref{thminducedcc3} of Theorem \ref{thminducedcc} implies in particular that $\Ce_2^K = \kt.$ We also need the following condition on the induced cells of a local even full monodromy-free simply connected 2-cell complex in order to define an higher dimensional cell complex from these cells.

\begin{defi}[Non-singular collection of cells] \label{defnonsingcells}
Let $K$ be a local cell complex and let $\Ce$ be a collection of subsets of $\kz.$ $\Ce$ is \textit{non-singular} if  $K \cap (x \cap y)$ is connected or empty for all $x,y \in \Ce.$
\end{defi}

\noi We are now ready to formulate the following two corollaries of Theorem \ref{thminducedcc}. The Corollary \ref{correconst1} gives the sufficient conditions on a $2$-cc $K$ for the existence of an \textit{ambient cell complex $K_{amb}$} such that the $2$-skeleton of $K_{amb}$ is $K.$ Corollary \ref{correconst2} gives the uniqueness of $K_{amb}$ and can be interpreted as reconstructing a local even simply connected simple cc of rank higher or equal to 3 from the data of its 2-skeleton.

\begin{cor} \label{correconst1}
Let $K$ be a local even full monodromy-free simply connected 2-cell complex such that $\Ce^K_r$ is non-singular for all integer $2 \leq r \leq R,$ where $R$ is the largest integer $r$ such that $\Ce_r^K \neq \emptyset.$ Then there exists a local even simply connected simple $R$-cc $K_{amb}$ such that $$(K_{amb})^{(2)} = K.$$
\begin{proof}
By Lemma \ref{lemfullk}, $K$ is $n$-regular for some $n \geq 2.$ Define $(K_{amb},\rk_{K_{amb}})$ by $K_{amb}^{[r]} = \kr$ for $0 \leq r \leq 2$ and $K_{amb}^{[r]} = \Ce^K_r$ for $r \geq 3.$ By Theorem \ref{thminducedcc}, $\Rk(K_{amb}) = \abs{\E_v} - 1 =  n -1,$ where $v \in \kz.$
The  Axiom \ref{cccrank} is satisfied for $K_{amb}$ by point \ref{thminducedcc2} of Theorem \ref{thminducedcc}. In order to show Axiom \ref{cccinter}, consider $x, y \in K_{amb},$ where $x \supsetneq x \cap y \subsetneq y$ and, say, $\rk(x) \leq \rk(y).$ 

We first consider the case where  $x \cap y$ contains no 2-cell. Since $\Ce^K_r$ is non-singular for all $1 \leq r \leq n - 1$  then $x \cap y$ is either a point or contains a unique edge by the following arguments. If there are two or more edges in $x \cap y$ then $\rk(y) \geq 3$ (as otherwise it would contradict \ref{cccinter} for $K$) and the non-singularity of induced cells implies in particular that $x \cap y$ is connected, therefore there exist at least two edges $e$ and $e'$ in $x \cap y$ incident to the same vertex $v.$ Since $K$ is full and by definition of the induced cc $J(v, \{e,e'\}),$ the induced cell $y$ contains the 2-cell $C$ such that $\E_v^C = \{e,e'\}.$ If $\rk(x) \geq 3$ then we also have $C \subset x$ which contradicts our assumption and if $x \in \kt$ then there are two edges in $x \cap C$ and this contradicts Axiom \ref{cccinter} for $K.$ 

We then consider the case were there exists $C \in \left( K \cap (x \cap y) \right)^{[2]}$ and let $v \in C.$ 
By Theorem \ref{thminducedcc} there is an induced cell $w$ such that $\E^C_v \subset \E_v^w = \E_v^{x \cap y} = \E_v^x \cap \E_v^y.$ Suppose by contradiction that $w \subsetneq x \cap y.$ By the non-singularity of induced cells  $x \cap y$ is connected and this implies that there is a path $p:v \pathto v'$ in $ x \cap y$ to a vertex $v'$ such that $\E_{v'}^w \subsetneq \E_{v'}^{x \cap y}.$ This is a contradiction since $\abs{\E_v^{x'}} = \abs{\E_{v'}^{x'}}$ for all $x' \in K$ by the bijectivity of the connection $\nabla_p : \E_v \dans \E_{v'}.$  
Therefore $w = x \cap y$ and \ref{cccinter} is satisfied. Properties \ref{cccenough} and \ref{cccdiamond} are direct consequences of the point \ref{thminducedcc2} from Theorem \ref{thminducedcc}. $K_{amb}$ is local and full since so is $K.$
By Lemma \ref{lemregfullcc}, $K_{amb}$ is simple since it is $n$-regular with $n= \Rk(K_{amb}) + 1.$  $K_{amb}$ is clearly even, cell-connected and simply connected.
\end{proof}
\end{cor}

\begin{cor} \label{correconst2}
$K_{amb}$ given by Corollary \ref{correconst1} is unique, i.e. if $L$ is a local even simply connected simple $R$-cell complex such that $R \geq 3$ then $$(L^{(2)})_{amb} = L.$$
\begin{proof}
Let $K= L^{(2)}.$ By Lemma \ref{lemmonodrfree} $K$ is monodromy-free. By Lemma \ref{lemaedgesetfullcc} the induced cells of $K$ are exactly the cells of $L$ of rank greater or equal to $3.$
\end{proof}
\end{cor}

\begin{rema} \label{remareconst}
By Theorem \ref{thmbayer} on the determination of local cell complexes by their barycentric subdivision and Lemma \ref{lemdualbdiviseven}, a consequence of Corollary \ref{correconst1} is that local simply connected cell complexes are determined by the 2-skeleton of the dual of their barycentric subdivision.
\end{rema}

Physicists are sometimes looking at discrete spaces as "Feynman diagrams" of some generalized matrix models called group field theories (see e.g. \cite{fre05} for a review about this topic, \cite{dpp00} for a discussion in dimension 4 and \cite{dfkr00} with a definition of 2-cell complexes similar to ours). Here, we would like to point out that, as a consequence of the results of these first two chapters, the set of 2-shellable even local full 2-cell complexes appears to be an interesting class of cell complexes to consider in this context. By Remark \ref{rem2shellable} we can see this class of 2-cc as being constructed by successively adding 2-cells while respecting some natural connectedness assumptions that involves only local conditions. On the other hand Corollary \ref{correconst1} allows to uniquely associate a higher dimensional "ambient" cell structure to such 2-cell complexes, which also satisfy the axioms of a cc if they fulfil the condition of non-singularity from Definition \ref{defnonsingcells}. If this is the case, Theorem \ref{thm2shell} and Remark \ref{remashellissimplyconn} imply that the dual of the resulting ambient cell complex is a shellable simplicial complex which is then a triangulation of an Euclidean ball or sphere.

\newpage
\vspace{3cm}

\section{Duality map for cobordisms} \label{secdualcob}

\vspace{2cm}

The goal of this chapter if to introduce the definition of a duality map that applies to non-singular cell complexes with arbitrary boundary and which is equivalent to the Definition \ref{defdualccc} for closed cell complexes. Since we are mainly interested in the case of local cell complexes and, as noticed in Remark \ref{remassumptlocal}, the dual of a closed local cc (lcc) is local only if the original lcc is also non-pinching, we will restrict our treatment to the set, noted $\nsc,$ of non-singular non-pinching lcc. We also denote by $\nsc^R$ the elements of rank $R$ in $\nsc.$ Boundaries of elements in $\nsc$ will also play an important role in what follows and, by Lemma \ref{rembdrycc}, these consist of closed but not necessarily connected cell complexes. We therefore also introduce the notations $\mlc$ for the set of closed non-pinching cell-connected cc and $\mlc^R$ for the set of elements in $\mlc$ of rank $R.$ Since $\nsc$ also includes the case of closed cell complexes, we assume by convention that the \textit{empty cell complex $\{\emptyset\}$} is included in $\mlc$ and $\mlc^R$ for all $R \geq 0.$

As we mentioned in the introduction, our goal is to mimic the behaviour of a Fourier-like transformation in certain physical models that provides a duality between two equivalent formulations of the same physical data. In this case, one has in particular that the double dual of a model is equivalent to the original model, an involution property we obtained for the cases of closed cc that one would like to preserve when the duality map acts on cell complexes in $\nsc,$ or at least for a certain sub-class thereof. As we will explain here, this idea is what lead us to consider the concept of cobordism, a geometrical space seen as interpolating between its boundary components, originally defined in the context of differentiable manifolds by Thom in \cite{tho54} and defining a category used in the formulation of the axioms of TQFT by Atiyah in \cite{ati88}. The parallel with TQFT motivated us to look for a suitable category of cobordisms defined from our combinatorial notion of cell complexes. This is the aim of the next Chapter \ref{seccausalstruct}, where we define a composition operation for cobordisms leading us to the introduction of the category of causal cobordisms.

In order to obtain an involutive duality map on a certain sub-class of cc in $\nsc$, we have to focus on how this duality acts on the cells of the boundary. This property is indeed already satisfied for the other cells, as a consequence of Lemma \ref{claim2propdual}. For this purpose, we introduce a different type of dual set, noted $\bual{A}$ for $A \subset \kz.$ If $K \in \nsc^R$ and $y \in (\partial K)^{[R-1]} $ is a maximal cell in the boundary, it must be in particular that the dual of $y$ is an edge of the dual of $K,$ as it becomes therefore again a sub-maximal cell of the double dual of $K.$ Since $\bual{y}$ therefore has to contain at least two elements, it is natural to define
$$ \bual{y} = \{ y , z \},$$
where $z$ is the unique element of $\kR$ containing $y.$ In fact, the "generalized dual" of a cell $x \in \partial K$ will be defined as
$$ \bual{x} := \{ y \in (\partial K)^{[R-1]} ~|~ x \subset y\} \sqcup \{z \in \kR ~|~ x \subset z\} = \cdual{x}{\partial K} \sqcup \cdual{x}{K},$$
where this definition also applies to any set $A \subset \kz,$ with the convention that $\cdual{A}{\partial K} = \emptyset$ whenever $A \not \subset \partial K^{[0]}.$

The first issue we encounter then is that $\bual{K}$ is not a cell complex, as for example in the case of $y \in (\partial K)^{[R-1]},$ $\bual{K}$ does not contain the vertex $\{y\} = \cdual{y}{\partial K}.$ One can easily fix this by considering the dual of $K$ to be $\bual{K} \sqcup \dual{\partial K}$ but then this notion of dual does not satisfy any involution property as it "adds" more cells to the original cell complex: for example the dual of a graph made of one edge, therefore having two boundary components each made of one vertex, is a graph made of two edges forming a simple path of length 2. The double dual would then be a graph forming a simple path of length 3, etc. The radically opposite alternative would be to "ignore" the cells in $\partial K$ and define the dual of $K$ to be $\dual{K \setminus \partial K},$ with obvious drawbacks when it comes to the involution property, as this definition "removes" cells. It turns out that alternating between these two definitions is essentially the idea behind the construction of the duality of this chapter: the dual of a path of length 1 is a path of length 2 and the double dual is again a path of length 1.

But how can we signal when to use the definition that adds cells around the boundary versus the definition that removes cells? We chose to do this by defining our duality map on "relative cell complexes", i.e. couples of the type $(K,J)$ where $J \leq K,$ in the particular case where $K \in \nsc$ and $J \in B ^{[R-1]},~ J \leq \partial K.$ The dual of such a relative cell complex would use the definition that adds cells on the cells in $\partial K \setminus J$ and use the definition that removes cell on the cells in $J.$

This is where the picture with cobordisms arises: an element $(K,J)$ where $K \in \nsc^R$ and $J \in \mlc^{[R-1]}, ~ J \leq \partial K,$ can be interpreted as a cobordism from some "ingoing components" $J$ to some "outgoing components" $L := \partial K \setminus J.$ In this picture, the duality maps the ingoing components to the outgoing components of the dual. For example, the two-dimensional lattice we used to define the Ising model in Section \ref{secising} can be seen as a cobordism $(K,J)$ where $J = \partial K$ is the only boundary component of $K,$ for which we indeed defined the dual using the definition that "ignores" the cells in $\partial K.$

To make this idea precise, we will need the notion of midsection, defined in the next Section \ref{susecmidsection}, a crucial notion for this chapter and Chapter \ref{seccausalstruct}. In  Section \ref{ssecrelcc}, we introduce the notions of non-degenerate, local and exactly collared relative cell complexes. Non-degenerate and local relative cc are used in Section \ref{ssectiondcob} to introduce our definition of cobordism and the corresponding notion of duality. Exactly collared relative cc will single out a sub-class of cobordisms having additional regularity around the ingoing boundary component. We state our main result in Section \ref{ssectiondcob} implying in particular that the dual of a cobordism is also a cobordism. As a corollary of this result we show that the double dual of an exactly collared cobordism is isomorphic to the original cobordism. In Section \ref{ssecexamdcob}, we illustrate some examples of duals of cobordisms, before giving the proof of the main theorem in Section \ref{secthmdualcob}.

\subsection{Midsections} \label{susecmidsection}

We recall that the notation $\nsc$ is used for the set of non-singular non-pinching lcc and $\nsc^R$ for the set of elements in $\nsc$ of rank $R.$ $\mlc^{R}$ is the set of closed non-pinching cell-connected $R$-cc and by convention contains the empty cc $\{\emptyset\},$ which is a sub-cc of every cc. Similarly $\mlc$ denotes the union of all elements in $\mlc^R,$ for $R \geq 0.$ As a consequence, the boundary of an element in $\nsc^R$ is an element of $\mlc^{R-1}.$

Let us introduce a generalization of the notations previously used, named "collar". This tool is a central element of the definition of a midsection and will be used throughout the rest of this work.

\begin{defi}[Collar] \label{defcollar}
Let $A$ be a finite set and $\A$ be a collection of subsets of $A.$ If $B \subset A$ we define the \textit{collar of $B$ in $\A,$} by
$$ \A_B := \{ x \in \A ~|~ \emptyset \neq x \cap B \neq x \}.$$
Moreover, if $C \in \A,$ we use the following definition and notation for the collar of $B$ in $\A$ restricted to $C:$
$$ \A_B^C := \{ x \in \A_B ~|~ x \subset C \}.$$
By abuse of notation, in case $\B$ is also a collection of subsets of $A,$ the notation $\A_\B$ is used to denote $\A_{B},$ where $B := \cup_{x \in \B} x$ is the set of elements of $A$ contained in a set in $\B.$
\end{defi}

For example if $K \in \nsc$ and $J \leq K$ we have $$K_J =  \{ x \in K ~|~ \emptyset \neq x \cap J^{[0]} \neq x \}.$$ 
Since we use the notation $\E$ for the set of edges in $K$ when there is no ambiguity on what the cell complex $K$ is, we will often use the following notation for $x \in K:$ $$\E_J^x = \{ e \in \ko ~|~ \emptyset \neq e \cap J^{[0]} \neq e, \quad e \subset x \}.$$

The following lemma is an important first step in order to endow the notion of midsection with the structure of cc.

\begin{lem} \label{lemmidsec}
Let $K$ be a local cell complex and $J \leq K.$ Then the map $$x \longmapsto \E_J^x, \quad x \in K_J$$ is an injective poset homomorphism such that the inverse is a poset homomorphism.
\begin{proof}
We first notice that since $K$ is cell-connected we have that $\E^x_J \neq \emptyset$ for all $x \in K_J.$ 
Let $x,y \in K_J$ be such that $\E_J^x = \E_J^y = \E_J^{x \cap y},$ where the second equality is simply a consequence of the definition of a collar. By this equality, in order to show that $\E_J^x = \E_J^y$ implies $x = y,$ it is sufficient to consider the case where say $x \subset y.$ For this we proceed by induction on $r = \rk(x) \geq 1.$

If $\rk(x)=1$ then $x = \{v,w\} \in \E_J,$ where $v \in J^{[0]}, ~ w \in \kz \setminus J^{[0]}$ and if we suppose by contradiction that $y \neq x$ then by axiom \ref{cccenough} we can without loss of generality suppose that $\rk(y) = \rk(x) + 1 = 2,$ so that $y \in \kt_J.$ By the diamond property \ref{cccdiamond} there exists $x' \in \ko $ such that $\cface{v} \cap \face{y} =\{x,x'\}.$ The edge $x'$ is not in $\E_J$ since otherwise it would imply that $\{x,x'\} \subset \E_J^y$ which contradicts $\E_J^x = \E_J^y.$ By Lemma \ref{lemrepr2cell} the 2-cell $y$ admits a representation such that $y = [x,x',e_1, \dots, e_k]$ where $e_k \cap x = \{w\}$ thus in particular $e_k \notin J^{[1]}.$ Therefore there is $1 \leq i \leq k$ such that $e_i \in \E_J^y$ which is also a contradiction with $\E_J^x = \E_J^y$ and this shows that $x = y.$

Let's assume the claim true for all $x, y$ such that $x \subset y$ and $\rk(x) \leq r-1.$ Take $x \in K_J^{[r]}$ and $y \in K_J$ such that $\E_J^x = \E_J^y$ and $x \subsetneq y.$ Like before we can without loss of generality assume that $\rk(y) = r + 1.$ By induction we have that $\E_J^w \subsetneq \E_J^x$ for all $w \subsetneq x,$ and, using this inclusion with $w = x \cap x',$ this implies in particular that $\E_J^{x'} \not \subset \E_J^x$ for all $x' \in K_J \cap K_x$ (since $\E_J^{x'} \neq \emptyset$ by locality of $K$). Since $\E_J^x \neq \emptyset$ there exists $w \in \face{x} \cap K_J.$ By the diamond property \ref{cccdiamond}, there exists $x' \in K^{[r]}$ such that $ \cface{w} \cap \face{y} = \{x,x'\}.$ But $x' \in K_J$ since so does $w$ and $w \subset x',$ therefore $\emptyset \neq \E_J^{x'} \subset \E_J^y = \E_J^x.$ This is a contradiction with our induction hypothesis and therefore $x = y.$

The definition of $\E_J^x$ and the injectivity proven above imply that if $\E_J^x \subset \E_J^y$ then 
$$\E_J^{(x \cap y)} = \E_J^x \cap \E_J^y = \E_J^x$$
and therefore $x = x \cap y$ or equivalently $x \subset y$ and this concludes the proof.
\end{proof}
\end{lem}

The previous lemma ensures that the rank function of a midsection given in the following definition is well defined. This definition can be seen as a generalization of the notion of midsection used in \cite{dj15} to classify "causal slices" of causal triangulations and which was initially introduced in the form of a coloured graph dual to a causal slice in \cite{ajl01}.

\begin{defi}[Midsection]
Let $K$ be a local cell complex and $J \leq K.$ We define the \textit{midsection of $J$ in $K$} to be $(M_J^K, \rk_{M_J^K})$ where
$$M_J^K := \{\E_J^x  ~|~ x \in K_J \} \quad \text{and} \quad \rk_{M_J^K}(\E_J^x) = \rk_K(x) - 1.$$
\end{defi}

\begin{figure}[!h]
\centering
\includegraphics[scale=0.42]{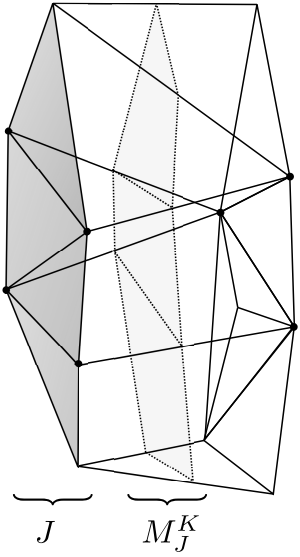}
\caption{\label{slicepic} Midsection $M_J^K$ associated to the boundary part $J$ of a 3-cc $K$ composed of three maximal cells (the center one having dots on its vertices).}
\end{figure}

We give an illustration of a midsection in Figure \ref{slicepic}. The next proposition shows that midsections are cc under rather general assumptions.

\begin{prop} \label{propmidsec}
Let $K$ be a local cell complex and $J \lneq K$ such that $x \not \subset J^{[0]}$ for all $x \in K \setminus J$ and such that $K_J \neq \emptyset.$ Then $(M_J^K, \rk_{M_J^K})$ is a cc and if $K$ is pure then so is $M_J^K$ and $\Rk(M_J^K) = \Rk(K)-1.$ Moreover if $J \cap x$ is connected for all $x \in \kt_J$ then $M_J^K$ is graph-based.
\begin{proof}
Since $K$ is cell-connected, $J \neq K,$ and by the condition $x \not \subset J^{[0]}$ for all $x \in K \setminus J,$ there is at least one edge of $K$ having one vertex in $J^{[0]}$ and one vertex in $\kz \setminus J^{[0]},$ therefore $K_J \neq \emptyset.$

By definition of $\E_J^x$ and $\rk_{M_J^K},$ the elements in $M_J^K$ of rank $0,$ i.e. vertices, are indeed 1-element sets and every set $\E_J^x$ is indeed a set of vertices of $M_J^K.$ Since for $x,y \in K,$ if $x \in K_J$ and $x \subset y$ then $y \in K_J,$ the axioms of the Definition \ref{defccc} are a direct consequence of Lemma \ref{lemmidsec} and the axioms for $K.$

Let $R = \Rk(K).$ By purity of $K,$ every cell in $x \in K_J$ is contained in a $R$-cell of $K,$ which must then belong to $K_J$ and constitute an $R-1$ cell in $M_J^K$ containing $\E_J^x,$ showing that $M_J^K$ is also pure.

By Lemma \ref{lemrepr2cell}, and since $K$ is local, $K \cap x$ is a cycle for every 2-cell $x$ of $K.$ Therefore, for every edge $\E^x_J$ of $M_J^K,$ i.e. for every $x \in \kt_J,$ if $J \cap x$ is connected then there are exactly two edges in $\ko_J$ included in $K \cap x,$ and therefore there are exactly two vertices of $M_J^K$ in $\E_J^x.$ 
\end{proof}
\end{prop}

\subsection{Non-degenerate, local and exactly collared relative cc} \label{ssecrelcc}

In this section we introduce what it means in our context to consider a "cell complex relatively to another cell complex containing it". We first explicitly define relative cell complexes, homomorphism between such objects as well as the notion of "relative homomorphism" between relative cell complexes. Relative homomorphism are used mostly in the next chapter, to define the notion of relative reduction and relative collapses at the end of Section \ref{ssecredncol}. 

\begin{defi}[Relative cc and homomorphism] \label{defirelccc}
A couple $(K,J)$ of cc is called a \textit{relative cc} if $J \leq K.$
We say that a map $\phi: (K, J) \dans (K', L)$ is a \textit{homomorphism of relative cc} if it is a cc-homomorphism $\phi : K \dans K'$ such that $\phi(J) = L.$ A \textit{relative cc-homomorphism} or a \textit{homomorphism of $J$ relative to $K$} is a homomorphism $\phi: (K,J) \dans (K',L)$ of relative cc such that $\phi|_{K \setminus J}: K \setminus J \dans K' \setminus L$ is a poset isomorphism. 
\end{defi}

Such "relative objects" are present in many different contexts, notably in topology with based space or algebraic topology with relative homology group. However our definition in effect reverses the focus that is usually put on such relative objects. While in topology one tends to "ignore" the structure contained in the subspace, here the cell complex $K$ in a relative cc $(K,J)$ can be interpreted merely as a "background" cell complex structure on which a relative homomorphism acts trivially. This point of view remains valid in the following definitions of the specific types of relative cell complexes we introduce next.

\begin{defi}[Non-degenerate relative cc]
We say that a relative cc $(K,J)$ is \textit{non-degenerate} if for all $x \in K \setminus J$ then $x \not \subset J^{[0]}.$ We say that $(K,J)$ is \textit{degenerate} if it is not non-degenerate. 
\end{defi}

Examples of degenerate relative cell complexes $(K,J)$ can be easily constructed when $K$ has relatively few maximal cells and $J =\partial K.$ For instance if $K$ has only one maximal cell then such cell has automatically all its vertices on the boundary.

The following definition extends the notion of locality to relative cell complexes. It is the first time where the notion of midsection enters crucially in our formalism. As will be discussed after this definition, it was not obvious to us that this was the right notion of locality in this context. It will however become apparent from the proof of Theorem \ref{thmdualcob} that this notion of locality for relative cc contains exactly the assumptions needed for the duality to be well-behaved around the ingoing boundary of a cobordism.

\begin{defi}[Local relative cc]
We say that a relative cc $(K,J)$ is \textit{local} if it is non-degenerate, $K$ is local, every connected component $J_0$ of $J$ satisfies that $J_0 \cap x$ is connected for all $x \in K_{J_0}$ and $M_{J_0}^K$ is cell-connected. In particular $(K,\emptyset)$ is local for any lcc $K.$
\end{defi} 

The case in which $K \in \nsc^R$ and $J \in \mlc^{R-1}$ is a sub-cc of $\partial K$ will be our main focus in what follows. In this case, the assumption in the definition of local relative cc that the midsection is cell-connected is in fact not redundant with the assumption that $K \in \nsc^R.$ In other words, knowing that $K$ is both non-pinching and cell-connected together with the assumption that $(K,J)$ is non-degenerate is not sufficient to prove that $M_J^K$ is cell-connected. This apparently remains the case even when adding an assumption on the connectedness of 2-cells: for example we could require that $ \left[ v, x \right]^{(2)} := \{ y \in K^{(2)} ~|~ v \in y \subset x \}$ is connected (as a section, c.f. definition in Section \ref{secpolytope}) for all $x \in J^{[0]}.$ Also, the connectedness of the cells in $M_J^K$ seems closely related to the connectedness of $ \partial (J \cap x)$ when this cc is well-defined, i.e. when $J \cap x$ is pure. The purity of $J \cap x$ is however not guaranteed in general and not necessary in order for the cell in $M_J^K$ to be connected (this purity assumption will actually come up later in Section \ref{ssectrans} as an assumption allowing to define a rank function on a poset called transition, but is not required for the results of this chapter).
These observations support the interpretation of the locality assumption on a relative cc $(K,J)$ (more specifically the connectedness of the midsection $M_J^K$) as a condition allowing to freely choose the endpoint of any path in $K$ having its last edge in $K_J.$ If $p$ is such a path, an edge in $M_{J_0}^K$ indeed corresponds to a 2-cell $C \in K_{J_0},$ which defines moves of the form $m_{p_0}^C$ where $p_0$ is a simple subpath of $C \cap p$ (cf. Definition \ref{defmovespath} of moves).

We now introduce the notion of exactly collared relative cc. This notion will be used later on in Corollary \ref{corbidualcob} in order to prove that the duality map defines an involution on the set of cobordism introduced in the next section.

\begin{defi}[Exactly collared relative cc]
Let $(K,J)$ be a non-degenerate cc. We say that $(K,J)$ is \textit{exactly collared} if the map from  $M_J^K$ to $J$ defined by $\E_J^x \mapsto J \cap x$ is a cc-isomorphism. In particular an exactly collared relative cc $(K,J)$ such that $K$ is local is a local relative cc and $(K,\emptyset)$ is also exactly collared for all lcc $K.$
\end{defi} 

We give an illustration of a simple example of an exactly collared cobordism in Figure \ref{figdualcylinder}.
In the next sections of this chapter, we will be particularly interested in the case of relative cc $(K,J)$ where $K \in \nsc$ and $J \leq K$ is a (possibly empty) union of boundary components of $\partial K.$ We therefore mention the following lemma, as it can shed light on certain definitions given thereafter.

\begin{lem} \label{lemconcomp}
Let $K \in \nsc^R.$ Then $J \leq \partial K$ and $J \in \mlc^{R-1}$ if and only if $J$ is a (possibly empty) union of connected components of $\partial K.$
\begin{proof}
Let $R = \Rk(K).$ A union of boundary components of $K$ is clearly an element in $\mlc^{R-1},$ by Lemma \ref{rembdrycc}. In order to prove the reverse implication it is sufficient to prove the case where $\partial K$ is connected and to show that $J \neq \emptyset$ implies $J = \partial K.$ If $J \in \mlc^{R-1}$ is non-empty then it is a closed $(R-1)$-cc. If by contradiction there exists a vertex in $\partial K \setminus J,$ the connectedness of $\partial K$ implies that this vertex is linked to a vertex in $J$ via a path in $(\partial K)^{(1)}.$ Hence there is an edge $e = \{v,w\}$ in $\partial K$ such that $v \in J$ and $w \in \partial K \setminus J.$ Now we have that $\dual{J} \leq \dual{\partial K}$ and $\cdual{e}{\partial K}$ (the dual of $e$ in $\partial K$) is a $(R-2)$-cell of $\dual{\partial K}$ contained in $\cdual{v}{\partial K},$ an $(R-1)$-cell of $\dual{J}.$ Therefore it implies that $\cdual{e}{\partial K}=\cdual{e}{J} \in J^{[R-2]}$ but since $\dual{J}$ is closed we have that $\dual{e}$ is contained in two maximal cells of $J.$ Since $\cdual{w}{\partial K}$ is a maximal cell of $\dual{\partial K}\setminus \dual{J}$ containing $\cdual{e}{J}$ this contradicts the non-singularity of $\dual{\partial K}.$
\end{proof}
\end{lem}

\subsection{Duality map for cobordisms and main results} \label{ssectiondcob}

The following set of relative cell complexes constitutes our notion of cobordism.

\begin{defi}[Cobordism] \label{deficob}
The set of \textit{combinatorial cobordisms}, or simply \textit{cobordisms}, is denoted and defined as 
$$ \cob := \{ (K-J) ~|~ K \in \nsc, ~ J \leq \partial K, ~ J \in \mlc^{\Rk(K)-1}, ~(K,J) \text{ non-degenerate and local}\}.$$
We also use the notation $\cob^R$ to denote the elements in $\cob$ such that $\Rk(K) = R,$ in which case $R$ is also called the \textit{rank} or \textit{dimension} of the cobordisms in $\cob^R.$ The boundary of a cobordism is simply defined as $\partial (K-J) := \partial K.$ The boundary components $J$ of a cobordism $(K-J)$ are called \textit{ingoing (boundary) components} or \textit{removed components} of the cobordism and the cells in $J$ are sometimes called the \textit{removed cells}. The boundary components in $\partial K \setminus J$ are called \textit{outgoing (boundary) components} of $(K -J).$ We say that $(K-J) \in \cob$ is \textit{exactly collared} if $(K,J)$ is exactly collared.
\end{defi}

We use the sign "$-$" in $(K-J)$ to distinguish such an element from a usual relative cc and we use the term "removed" for $J$ and its cells because the duality map we define later in Definition \ref{defdualcob} indeed acts on $(K-J)$ as if the cells in $J$ are removed from $K.$


There are numerous examples of cobordisms, as in particular every element $(K - J)$ where $K$ is a non-singular (i.e. non-branching) and non-pinching simplicial complex and $J$ is a union of connected components of $\partial K$ is an element of $\cob.$ In other words, the assumptions of locality and non-degeneracy are always satisfied for relative complexes $(K,J)$ where $K \in \nsc$ is simplicial.

Before introducing its generalization, we recall that in Definition \ref{defdualset} we defined the dual of a non-empty set $A \subset \kz$ for a cc $K$ by
$$ \cdual{A}{K} := \{ z \in K^{[\Rk(K)]} ~|~ A \subset z \},$$ also denoted $\dual{A}$ in case there is no ambiguity on the cc $K.$ The following definition constitutes a canonical way to extend the latter definition in order to associate a dual cell to the cells on the boundary of a non-singular cell complex.

\begin{defi}\label{defbdualset}
Let $K$ be a cc and let $A \subset \kz$ be non-empty. We define the \textit{$\sim$-dual of $A$} to be
$$\bdual{A}{K}:= \cdual{A}{K} \sqcup \cdual{A}{\partial K },$$
where by convention $\cdual{A}{\partial K} = \emptyset$ if $A \not\subset (\partial K)^{[0]}.$ We also write $\bual{x}$ for $\bdual{x}{K}$ when there is no ambiguity on $K.$
\end{defi}

As explained in the introductory paragraph of this chapter, this definition for example implies that the $\sim$-dual of a maximal cell of the boundary $y \in \partial K^{[R-1]}$ of a $R$-cc $K$ is given by
$$ \bual{y} = \{y, z\}$$
where $z \in \kR$ is the unique $R$-cell of $K$ containing $y.$ Other examples of $\sim$-dual are given in Figures \ref{figdualbitetrahedra} and \ref{figdualcylinder}.

For a collection of sets $\A \subset \pws{\kz},$ we use the notation $$ \bdual{\A}{K} := \{ \bdual{A}{K} ~|~ A \in  \A\}.$$
For example if $(K-J) \in \cob,$ we obtain $ \bdual{K \setminus J}{K} = \{ \bdual{x}{K} ~|~ x \in K \setminus J \}.$ We simply write $\bual{K \setminus J}$ instead of $\bdual{K \setminus J}{K}.$ Note that if $K$ is closed then $\bual{K} = \bual{K \setminus \emptyset} = \dual{K}.$

We would like to also define a rank function for elements in $\bual{K}$ when $K \in \nsc$ and this is done by mean of the following lemma, which generalizes Lemma \ref{claim2propdual}.

\begin{lem} \label{leminclbdual}
Let $K$ be a non-singular cc and let $x,y \in K.$ Then 
$$\bdual{x}{K} \subsetneq \bdual{y}{K} \quad \text{ if and only if } \quad y \subsetneq x.$$
\begin{proof}
First we note that Lemma \ref{lemcface} and Lemma \ref{claim2propdual} imply that if $x \in K \setminus \partial K$ then 
\begin{equation} \label{claim5lembdual}
\cdual{y}{K} \subsetneq \cdual{x}{K} \quad \Leftrightarrow \quad y \subsetneq x.
\end{equation}

We suppose first that $\bdual{x}{K} \subsetneq \bdual{y}{K},$  i.e. $\cdual{x}{K} \subset \cdual{y}{K}$ and $ \cdual{x}{ \partial K } \subset \cdual{y}{ \partial K}.$ If $x \in K \setminus \partial K$ then (\ref{claim5lembdual}) implies that $y \subset x.$ If $x \in \partial K$ then $\cdual{x}{\partial K} \subset \cdual{y}{\partial K}$ implies $y \subset x$ as a consequence of Lemma \ref{claim2propdual} for the case of the closed cc $\partial K.$ Since our assumption implies that $x \neq y$ we obtain $y \subsetneq x.$

Conversely, if $y \subsetneq x$ then the definition of dual set implies $\bdual{x}{K} \subset \bdual{y}{K}.$ If $x \in K \setminus \partial K$ then the equivalence (\ref{claim5lembdual}) implies that $\cdual{x}{K} \subsetneq \cdual{y}{K}$ and in particular $\bdual{x}{K} \neq \bdual{y}{K}.$ Similarly if $x \in \partial K$ then Lemma \ref{claim2propdual} implies that $\cdual{x}{\partial K} \subsetneq \cdual{y}{\partial K}$ and $\bdual{x}{K} \neq \bdual{y}{K}.$
\end{proof}
\end{lem}

As a consequence of Lemma \ref{leminclbdual}, the map $x \mapsto \bual{x}$ is a bijection from the non-singular cc $K$ to $\bual{K}$ and the rank function $$ \rk_{\bual{K}}( \bual{x} ) := \Rk(K) - \rk_K(x)$$  where $x \in \bual{K}$ is well defined for all $\bual{x} \in \bual{K}.$

Let $K, K'$ be disjoint posets with rank functions $\rk_K$ and $\rk_{K'}.$ We define the disjoint union $K \sqcup K'$ to be a poset with rank function
$$\rk_{K \sqcup K'} := \rk_K + \rk_{K'}.$$
For an element $(K-J) \in \cob,$ this definition allows us to have a rank function on the poset defined as the disjoint union
$$ \bual{K \setminus J} \sqcup \dual{\partial K \setminus J} = \bdual{K \setminus J}{K} \sqcup \cdual{\partial K \setminus J}{\partial K}$$
used below in our definition of dual cobordism.
The latter union is indeed  disjoint, since by non-degeneracy of $(K,J)$ every cell $x \in \partial K \setminus J$ satisfies $\cdual{x}{\partial K} \subsetneq \bdual{x}{K},$ which implies in particular that $\cdual{x}{\partial K} \neq \bdual{x}{K}$ for all such $x$ and therefore
$\bdual{K \setminus J}{K} \cap \cdual{\partial K \setminus J}{\partial K} = \emptyset.$

More explicitly, if $(K,J)$ is non-degenerate our notations imply the following expressions for the rank function on $\bual{K \setminus J} \sqcup \dual{\partial K \setminus J}:$
$$ \rk_{\bual{K \setminus J} \sqcup \dual{\partial K \setminus J}}(x') = \begin{cases}
\rk_{\kb}(x') = R - \rk_K(x) &\textit{if } x' = \cdual{x}{K} \textit{ for } x \in K \setminus \partial K,\\
\rk_{\bual{K}}(x') = R - \rk_K(x) &\textit{if } x' = \bdual{x}{K} \textit{ for } x \in \partial K,\\
\rk_{\dual{\partial K}}(x') = (R - 1) - \rk_{\partial K}(x) &\textit{if } x' = \cdual{x}{\partial K} \textit{ for } x \in \partial K.
\end{cases}$$
where $R = \Rk(K).$ The particular case $J = \emptyset$ yields 
$$ (\bual{K} \sqcup \dual{\partial K})^{[0]} = \cdual{\kR}{K} \sqcup \cdual{(\partial K)^{[R-1]}}{\partial K},$$
for the dual sets of rank $0.$ For dual sets of rank $1 \leq r \leq R-1$ we have
$$(\bual{K} \sqcup \dual{\partial K})^{[r]} = \cdual{(K \setminus \partial K)^{[R-r]}}{K}  \sqcup \cdual{\partial K^{[R-r-1]}}{\partial K} \sqcup \bdual{ (\partial K)^{[R-r]}}{K}.$$
For example, if $e' \in (\bual{K} \sqcup \dual{\partial K})^{[1]}$ then $e'$ is a dual set having one of the three following forms:
\begin{itemize}[topsep=2pt, parsep=2pt, itemsep=1pt]
\item the "usual" dual of a sub-maximal cell of $K$ i.e. $e' = \cdual{y}{K}$ for $y \in K^{[R-1]},$
\item the "usual" dual of a sub-maximal cell of $\partial K,$ i.e. $e' = \cdual{y}{\partial K}$ for $y \in (\partial K)^{[R-2]},$
\item the $\sim$-dual of a maximal cell in $\partial K$, i.e. $e' = \bdual{y}{K} = \{ y, z\}$ where $y \in (\partial K)^{[R-1]}$ and $z \in \kR$ is the unique maximal cell of $K$ containing $y.$
\end{itemize}
With this in mind, the next definition of duality map for cobordism is more transparent.

\begin{defi}[Dual cobordism] \label{defdualcob}
Let $(K-J) \in \cob,$ we define the \textit{dual cobordism} to be 
$$ \dual{(K-J)} := \left( ( \bual{K \setminus J} \sqcup \dual{\partial K \setminus J}  ) - \dual{\partial K \setminus J} \right).$$ 
\end{defi}

This definition, above all, raises the following question: is $\dual{(K - J)}$ also a cobordism? This question is answered positively by the main results of this chapter stated in the following theorem. 

\begin{thm} \label{thmdualcob}
If $(K- J) \in \cob^R$ then:
\begin{itemize}
\item $\dual{(K - J)} \in \cob^R,$ 
\item $\dual{(K - J)}$ is exactly collared,
\item $\dual{K_J} \in \mlc^{R-1}$ and  
$$\partial \dual{(K -J)} = \dual{\partial K \setminus J} \sqcup \dual{K_J}.$$
\end{itemize}
\begin{proof} \renewcommand{\qedsymbol}{}
See Section \ref{secthmdualcob}.
\end{proof}
\end{thm}

In the proof of Theorem \ref{thmdualcob} we will see that, as a consequence of Lemma \ref{lemmidsec}, we have $\dual{K_J} \cong \dual{M_J^K}.$ In other words, the dual outgoing component $\dual{K_J}$ is better understood as the dual of the midsection associated to the ingoing components $J.$ We point out that although $\dual{K_J}$ is a (closed) cell complex, it is not the case of $K_J$ and in particular $\dual{\dual{K_J}} \not\cong K_J.$

Definition \ref{defdualcob} leads to the interpretation that the action of the duality map on cobordisms in effect exchanges the role of the ingoing and outgoing boundary components. This can be seen particularly well in the case of the cylinder depicted in Figure \ref{figdualcylinder} of the next section.

As a first example, we can look at what Definition \ref{defdualcob} implies for the special cases where $J = \emptyset$ and $J = \partial K.$ In the first case we get:
$$ \dual{(K - \emptyset)} =  \left( ( \bual{K} \sqcup \dual{\partial K}  ) - \dual{\partial K} \right).$$
We can think of this case as duality acting on a non-singular cc $K,$ as no other data is specified. Such an example is depicted in Figure \ref{figdualbitetrahedra}. As a consequence of Theorem \ref{thmdualcob}, we have that $\bual{K} \sqcup \dual{\partial K} \in \nsc$ and 
$$ \partial ( \bual{K} \sqcup \dual{\partial K}  )  = \dual{\partial K}.$$
The dual of $(K - \emptyset)$ is therefore a cobordism with ingoing components equal to the entire boundary. For such cobordisms, corresponding to our second case $J = \partial K,$ the duality gives
$$  \dual{(K - \partial K)} =  \left(  \bual{K \setminus \partial K}  - \emptyset \right).$$
An example of such a duality was given when discussing the Ising model in Section \ref{secising}. Another consequence of Theorem \ref{thmdualcob} is therefore that $\bual{K \setminus \partial K} \in \nsc$ for all $(K - \partial K) \in \cob.$

Finally, the following corollary concerns the special case of exactly collared cobordisms, for which the involution property holds. To state this result, we defined the isomorphism $(K - J) \cong (K' - J')$ between cobordism as $K \cong K'$ with an isomorphism mapping $J$ to $J'.$ A consequence of this corollary is that $(K-J)$ is exactly collared if and only if there exists $(K'-J') \in \cob$ such that $\dual{(K'-J')} = (K - J),$ a property that motivated us to include the notion of "exactness" in the term used to denote such cobordisms.

\begin{cor} \label{corbidualcob}
If $(K-J) \in \cob$ is exactly collared then we have
$$\partial \dual{(K -J)} \cong \dual{J} \sqcup \dual{\partial K \setminus J} \quad \text{and} \quad \dual{\dual{(K-J)}} \cong (K - J).$$
\begin{proof}
By Theorem \ref{thmdualcob} and using the assumption that $(K-J)$ is exactly collared, i.e. that $ J \cong M_J^K,$ and the identity $ \dual{M_J^K} \cong \dual{K_J},$ we directly obtain that $\partial \dual{(K -J)} \cong \dual{J} \sqcup \dual{\partial K \setminus J}.$

The identity $\dual{\dual{(K-J)}} \cong (K - J)$ is a consequence of the following computation, using the previous point and the assumption that $(K-J)$ is exactly collared:
\begin{align*}
\dual{\dual{(K-J)}} &= \dual{ ( \bual{ K \setminus J} \sqcup \dual{ \partial K \setminus J} - \dual{ \partial K \setminus J} )}\\
&= \left( \bual{ \left( ( \bual{ K \setminus J } \sqcup \dual{ \partial K \setminus J} ) \setminus \dual{ \partial K \setminus J} \right)} \sqcup \dual{\dual{ K_J }}  - \dual{ \dual{ K_J}} \right)\\
&\cong ( \bual{ \bual{ K \setminus J}} \sqcup \dual{ \dual{ J}} - \dual{\dual{J}} ) \cong ( K - J ),
\end{align*}
where we also used Lemma \ref{leminclbdual} and Proposition \ref{propdual} for the last identity.
\end{proof}
\end{cor}

\newpage

\subsection{Basic examples in pictures} \label{ssecexamdcob}

In this section, we use illustrations to give simple examples of the duality map on cobordisms of rank 2 and 3 and we describe the example of a torus-like $3$-cell having a dual representing a typical example of an inhomogeneous cell complex. We chose examples in these dimensions to give some intuition of the generality under which the duality is defined, while keeping the dimension of the representation under 3. In order to make the following examples more transparent, we use different symbols for the vertices of the dual cobordism, as explained in Figure \ref{figdcellconv}.

\begin{figure}[!h]
\centering
\includegraphics[scale=0.75]{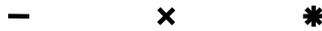}
\caption{\label{figdcellconv} These are the different illustrations chosen to represent the vertices in the dual of a cell complex, when such vertex is the dual of a 1-cell (stripe), 2-cell (cross) and 3-cell (star). }
\end{figure}

The first example is a $3$-dimensional cobordism with empty ingoing component made out of two tetrahedra sharing a triangle. We can think of this case simply as a non-singular cell complex. 
The dual of this cobordism is composed of five $3$-cells, one of which is illustrated below in the Figure \ref{figdualbitetrahedra}.

\begin{figure}[!h]
\centering
\includegraphics[scale=0.63]{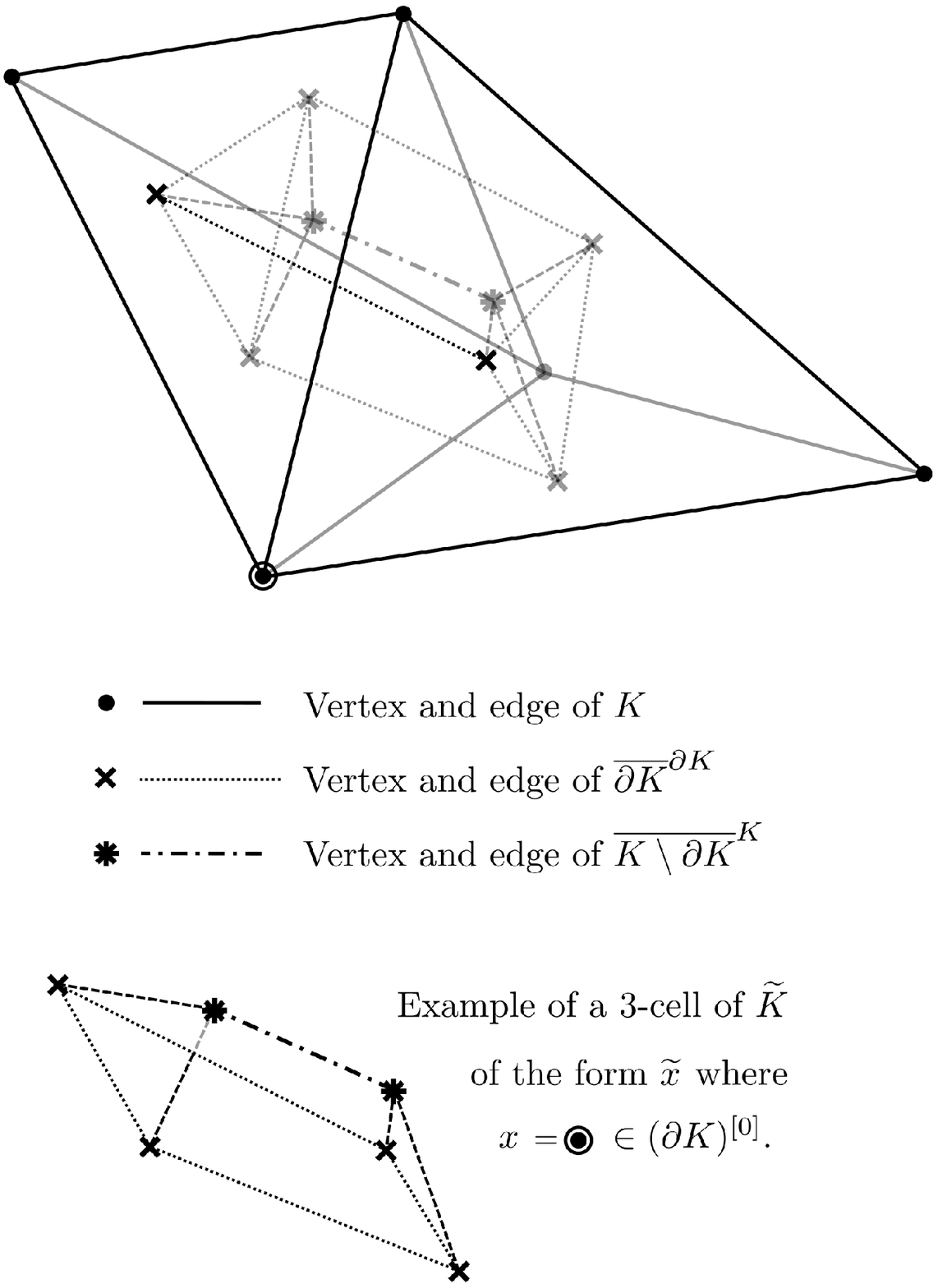}
\caption{\label{figdualbitetrahedra}A cc $K \in \nsc$ composed of two tetrahedra sharing a triangle and seen as a 3-dimensional cobordism $(K - \emptyset),$ i.e. with empty ingoing component. Its dual $\dual{(K - \emptyset)}$ is represented inside $K.$ In this example the entire boundary of the dual is made of removed cells (i.e. cells in the ingoing component). The cobordism $\dual{(K - \emptyset)}$ has only two vertices not lying on the boundary corresponding to the dual of the two tetrahedra in $K.$}
\end{figure}

\newpage

The second example is a 2-dimensional cobordism  $(K - J)$ where $K$ is a discretization of  a 2-dimensional cylinder where the ingoing boundary component $J$ is on the left of the illustration.
In this example, $\dual{(K - J)} \cong (K - J)$  and both $(K-J)$ and its dual have nine $2$-cells each containing four vertices. This is also an example of an exactly collared cobordism. In the lower part of Figure \ref{figdualcylinder} we illustrate a cell obtained as the dual of a vertex lying on the outgoing boundary component.

\begin{figure}[!h]
\centering
\includegraphics[scale=0.72]{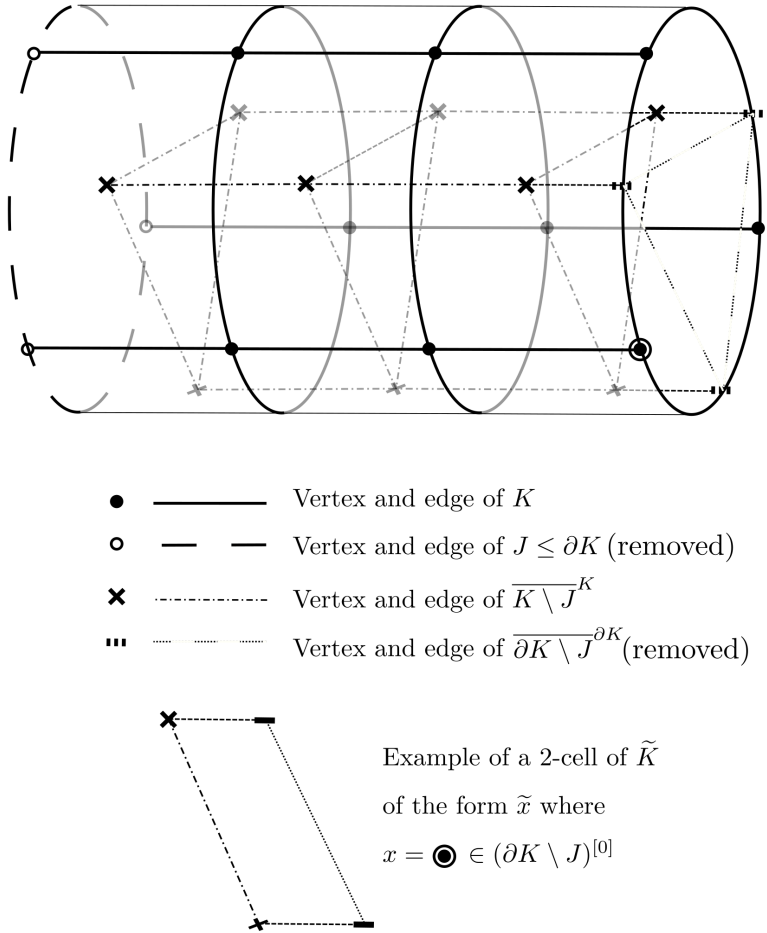}
\caption{\label{figdualcylinder}An example of a 2-dimensional cobordism that can be seen as a discretization of the boundary of a cylinder. In this case the primal and dual cobordisms are isomorphic, with the location of the ingoing and outgoing components exchanged, the ingoing component being represented using cells with white stripes. This is an example of an exactly collared cobordism.}
\end{figure}

\newpage

The last example is given by a simple discretization of a solid torus using a single $3$-cell, twelve 2-cells, twenty-four edges and twelve vertices, represented on the left of Figure \ref{figtorus}. The dual of this cell complex (seen as a cobordism of the form $(K- \emptyset)$), represented on the right of the figure, gives an example of a non-homogeneous cell complex which therefore does not admit a geometrical realization as a manifold. A similar construction is given in Figure 3 of \cite{dav95} and, if subdivided appropriately, constitutes an example of an inhomogeneous vertex in a closed cc.

\begin{figure}[!h]
\centering
\includegraphics[scale=0.72]{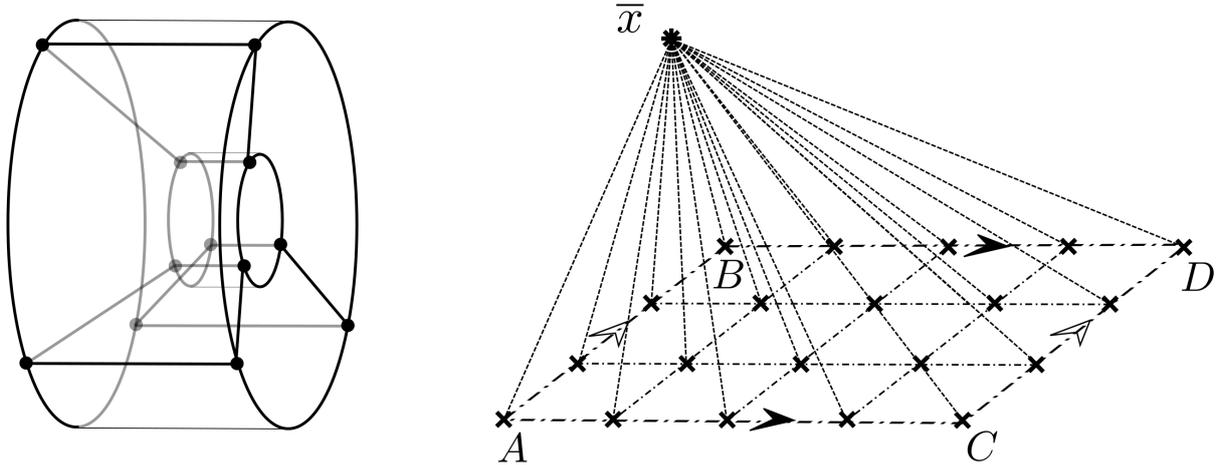}
\caption{\label{figtorus}The illustration on the left gives an example of a 3-cc $K$ composed of a single $3$-cell $x$ and having a boundary that defines a discretization of a 2-dimensional torus. On the right, we get a representation of $\dual{(K - \emptyset)} =  \left( ( \bual{K} \sqcup \dual{\partial K}  ) - \dual{\partial K} \right)$ by making an identification of the vertices on the boundary of the discretized square $ABDC$ following the usual construction of a 2-dimensional torus from a square hinted by the black and white arrows. That is to say: we identify the four vertices from the vertex $A$ to the vertex $B$ with the four vertices between $C$ and $D$ in this order and similarly for the vertices between $A$ and $C$ with the vertices between $B$ and $D.$ This identifications of vertices then automatically produce identifications of the edges having both vertices identified with two other vertices and similarly for the 2-cells (in this case triangles) having $\dual{x}$ as a vertex and an edge on the boundary of the square $ABDC.$ The top vertex $\dual{x}$ corresponds to a non-homogeneous point in a geometric realization of $\bual{K} \sqcup \dual{\partial K}.$ }
\end{figure}

\newpage

\subsection{Proof of the Theorem on duality} \label{secthmdualcob}

This section simply contains the proof of Theorem \ref{thmdualcob} stating that if $(K- J) \in \cob^R$ then:
\begin{itemize}
\item $\dual{(K - J)} \in \cob^R,$ 
\item $\dual{(K - J)}$ is exactly collared,
\item $\dual{K_J} \in \mlc^{R-1}$ and  $\partial \dual{(K -J)} = \dual{\partial K \setminus J} \sqcup \dual{K_J}.$
\end{itemize}

\begin{proof}[Proof of Theorem \ref{thmdualcob}]
Let $R = \Rk(K).$ We split the proof into three statements \ref{thmdualcobspta}, \ref{thmdualcobsptb} and \ref{thmdualcobpt2} as follows.

Suppose that $(K - J) \in \cob.$
\begin{enumerate}
\item \label{thmdualcobpt1}Then we have $K' := \bual{K \setminus J} \sqcup \dual{\partial K \setminus J} \in \nsc,$ 
\begin{enumerate}
\item \label{thmdualcobspta} for the special case $J = \emptyset,$
\item \label{thmdualcobsptb} and for $J$ arbitrary as a consequence of the following statement: if $K' \in \nsc$ and $J_0$ is a connected component of $\partial K \setminus J$ such that $(K, J \sqcup J_0)$ is local then
$$K'_0 := \bual{K \setminus (J \sqcup J_0)} \sqcup \dual{\partial K \setminus (J \sqcup J_0)} \in \nsc;$$
\end{enumerate}
\item \label{thmdualcobpt2} If we define $ J' := \dual{\partial K \setminus J},$ then $ J' \in \mlc^{R-1},$ $(K',J')$ is exactly collared and $$\partial K' = J' \sqcup \dual{K_J}.$$
\end{enumerate}
To prove the point \ref{thmdualcobspta}, we show that $K' := \bual{K} \sqcup \dual{\partial K}$ satisfies Axioms \ref{cccrank}, \ref{cccinter}, \ref{cccenough} and \ref{cccdiamond} and is a local non-singular non-pinching cell complex of rank $R.$

Axiom \ref{cccrank} is a consequence of Lemma \ref{leminclbdual} for the cells of $K'$ included in $\bual{K}$ and is a consequence of Lemma \ref{claim2propdual} for the cells included in $\dual{\partial K}.$ To prove \ref{cccinter}, we take $x', y' \in K'$ such that $x' \cap y' \neq \emptyset.$ If $x'= \bual{x}$ and $y' = \bual{y}$ are included in $\bual{K}$ then there exists a cell $w \supset x \cup y $ of minimal rank, and by Lemma \ref{leminclbdual} we get $\bual{w} = \bual{x} \cap \bual{y}$ (using the same arguments as in the proof of Proposition \ref{propdual}). If $x' = \dual{x}$ and $y' = \dual{y}$ are cells in $\dual{\partial K}$ then \ref{cccinter} is a consequence of Proposition \ref{propdual}. If say $x' = \bual{x} \in \bual{K}$ and $y' = \dual{y} \in \dual{\partial K}$ then
$$ \bual{x} \cap \dual{y} = (\cdual{x}{K} \cup \cdual{x}{ \partial K}) \cap \cdual{y}{ \partial K} = \cdual{x}{\partial K} \cap \cdual{y}{\partial K},$$
so $\cdual{x}{\partial K} \cap \cdual{y}{\partial K} \neq \emptyset.$ Therefore $x,y \in \partial K$ and
$$\cdual{x}{\partial K} \cap \cdual{y}{\partial K} = \cdual{w}{\partial K},$$ where $w = \dual{ \dual{x} \cap \dual{y}} \in \partial K.$ Axiom \ref{cccenough} is satisfied for elements in $\bual{K \setminus \partial K}$ or in $\dual{\partial K}$ as a consequence of Axiom \ref{cccenough} for $K$ and $\partial K$ and by Lemma \ref{leminclbdual} and Proposition \ref{propdual} respectively. The case $\dual{x} \subsetneq \bual{y}$ is shown by finding a cell $x' \in K'$ included in $\bual{y}$ and having $\dual{x}$ as a face by the following arguments. We have $$ \cdual{x}{\partial K} \subsetneq ( \cdual{y}{K} \cup \cdual{y}{\partial K})$$ which in particular implies $\cdual{x}{\partial K} \subset \cdual{y}{\partial K}.$ This in turn implies that $\cdual{y}{\partial K} \neq \emptyset$ and thus $y \in \partial K.$ If $\dual{x} \subsetneq \dual{y}$ then the existence of $x'$ is given by Axiom \ref{cccenough} for $\dual{\partial K}.$ If $\dual{x} = \dual{y}$ then $x = y$ and since $\rk_{K'}(\bual{y}) = \rk_{\dual{\partial K}}(\dual{y}) + 1$ then $\dual{x}$ is a face of $\bual{y}$ and $x'= \bual{y}$ satisfies the desired conditions. The diamond property \ref{cccdiamond} is proven in a similar way. The only non-direct case to check is when $\dual{x} \in \dual{\partial K},$ $\bual{y} \in \bual{K},$ $\dual{x} \subset \bual{y}$ and 
\begin{equation} \label{diamondrkprfpropbdual}
\rk_{K'}(\dual{x}) = \rk_{K'}(\bual{y}) - 2.
\end{equation}
By the same argument as before, $y \in \partial K$ and (\ref{diamondrkprfpropbdual}) implies that $y$ is a face of $x.$ We also have that $\bual{x}$ is a coface of $\dual{x}$ and a face of $\bual{y}$ by (\ref{diamondrkprfpropbdual}). Since $y \subset x$ this implies that 
\begin{equation} \label{diamondfacepropbdual}
\{\dual{y}, \bual{x}\} \subset \cface{\dual{x}} \cap \face{\bual{y}}.
\end{equation}
If $z \in \cface{\dual{x}} \cap \face{\bual{y}}$ then either $z \in \dual{\partial K}$ and $$\dual{x} \subsetneq z \subset \dual{y}$$ which implies $z= \dual{y}$ or $z \in \bual{K}$ and $$\bual{x} \subset z \subsetneq \bual{y}$$ which implies $z = \bual{x}.$ Therefore the inclusion in (\ref{diamondfacepropbdual}) is in fact an equality and this shows Axiom \ref{cccdiamond} as well as that $K'$ is a cc. 

We then prove that $K'$ is cell-connected.  $K'$ is graph-based since the non-singularity of $K$ implies that the edges of $K'$ are either composed of two maximal cells of $K$ or of one maximal cell $y$ of $\partial K$ and the unique maximal cell of $K$ containing $y.$ The cc $K'$ is cell-connected by the following argument. If $x' \in \bual{K}$ is a cell of the form $x' = \cdual{x}{K}$ where $x \in K \setminus \partial K$ then $x'$ is connected since $K$ is non-pinching. If $x' = \bual{x}$ where $x \in \partial K$ then $x'$ is connected if and only if there is and edge $e'$ of $K'$ with a vertex in $\cdual{x}{K}$ and in $\cdual{x}{\partial K},$ as both of these sets are connected since $K$ is non-pinching. Let $y$ be a maximal cell of $\partial K$ containing $x.$ The cell $y$ is then  included in a maximal cell $z$ of $K,$ we can then take $e' = \{y,z\}.$ Cells in $\dual{\partial K}$ are connected since $\partial K$ is non-pinching and this shows that $K'$ is cell-connected. 

The cc $K'$ is connected by the following argument. If $v_1',v_2'$ are vertices of $K'$ then by purity of $K$ these vertices are included in maximal cells $z_1'$ and $ z_2'$ of $\bual{K}.$ These maximal cells are duals to vertices $v_1,v_2$ of $K$ which are connected by a path in $\kg.$ This path provides a finite sequence of maximal cells of $\bual{K}$ starting with $z_1'$ and ending with $z_2'$ such that each element in the sequence shares a sub-maximal cell with the next element in the sequence and since the cells in $\bual{K}$ are connected this ensures that there is a path linking $v_1'$ to $v_2'$ in $\bual{K}.$ 

The following arguments show that $K'$ is a non-pinching non-singular cell complex of rank $R.$ In order to show that $K'$ is non-pinching, we first note that $\partial K' = \dual{\partial K}$ therefore there is no $\partial K'$-pinch in $K'$ since $\partial K$ is cell-connected. Let us assume by contradiction that there is a $K'$-pinch in $K',$ i.e. a cell $x' \in K'$ such that $\dual{K'}^{(1)} \cap \dual{x'}$ is not connected. If we take $x \in K $ such that $ \bual{x} = x',$ we can see that a path between two vertices $v,w$ in $x$ corresponds to a path between two vertices $\dual{\bual{v}}$ and $\dual{\bual{w}}$ in $\dual{K'}^{(1)} \cap \dual{x'}.$ Therefore if $\dual{K'}^{(1)} \cap \dual{x'}$ is not connected then so is $\kg \cap x,$ which contradicts the assumption that $K$ is cell-connected.  Every cell in $K'$ is included in a $R$-cell since this is clear for the cells in $\bual{K}$ and every $(R-1)$-cell $\dual{v} \in \dual{\partial K}$ is contained in $\bual{v} \in (K')^{[R]},$ which is then pure of dimension $R.$ Finally $K'$ is also non-singular since $K$ is graph-based.

We then turn to the point \ref{thmdualcobsptb}. Proving this statement amounts to verify that all the properties shown for the point \ref{thmdualcobspta} are still valid when restricted to cells not included in the sets $\bual{J_0} \sqcup \dual{J_0}.$ Axioms \ref{cccrank}, \ref{cccenough} and \ref{cccdiamond} are simply a consequence of the cc structure of $K'$ and the fact that whenever $x',y' \in K'_0$ are two cells such that $x' \subset y'$ and $z' \in K'$ satisfies $x' \subset z' \subset y'$ then $ z' \in K'_0.$ Axiom \ref{cccinter} is also verified for $K'_0$ since if $x',y' \in K'_0$ then $x' \cap y'$ contains no maximal cell of $J_0$ and therefore cannot be included in $\bual{J_0} \sqcup \dual{J_0}.$ We have that $K'_0$ is clearly graph-based since so is $K'$ and the set of edges of $K'_0$ is a subset of the set of edges of $K'.$ 
Also $K_0'$ is cell-connected since so is $K'.$ The purity of $K'_0$ is verified if we can show that for every cell $x'$ of $K'_0$ contained in a maximal cell of $K'$ of the form $\bual{v}$ where $v \in J_0^{[0]},$ we have that $x'$ is also contained in a maximal cell of the form $\bual{w},$ where $w \in \kz \setminus (J_0^{[0]} \sqcup J^{[0]}).$ This is ensured by the non-degeneracy of $(K,J_0 \sqcup J)$ together with the connectedness of the cells in $K_0'.$
The cc $K'_0$ is also non-singular since the set of $(R-1)$-cells of $K_0'$ is a subset of the set of $(R-1)$-cells of $K'$ and their set of cofaces in $K_0'$ is identical to their set of cofaces in $K'.$

Now comes a key point of this proof where the locality of $(K,J)$ enters crucially: showing that $K'_0$ is non-pinching. First we can see that there is no $K'_0$-pinch in $K_0'.$ Indeed $K'$ has no $K'$-pinch and $\dg{K'_0} \cap \bual{x'} = \dg{K'} \cap \bual{x'}$ for all $x' \in K'_0$ since if two maximal cells of $K'_0$ contain $x'$ and share a sub-maximal cell $y' \in K'$ then $y'$ is also in $K'_0.$ It is therefore sufficient to show that $\partial K'_0$ contains no $\partial K'_0$-pinch. For this we first prove the following relation:

\begin{equation} \label{claim6thmdual}
\partial K_0' = \dual{\partial K \setminus (J \sqcup J_0)} \sqcup \dual{K_J} \sqcup \dual{K_{J_0}}.
\end{equation}

In order to prove (\ref{claim6thmdual}) we first observe that $\dual{\partial K \setminus (J \sqcup J_0)} \leq \partial K_0'$ as a consequence of the following: every vertex $v \in (\partial K \setminus (J \sqcup J_0)^{[0]}$ satisfies that $y' := \cdual{v}{\partial K} \in (\partial K \setminus (J \sqcup J_0))^{[R-1]}$ and $\bual{v}$ is the only $R$-cell of $K'$ containing $y'.$ Since $(K, J \sqcup J_0)$ is non-singular, every $(R-1)$-cell of $K_0'$ not in $\dual{ \partial K \setminus (J \sqcup J_0)}$ is of the form $y' = \cdual{e}{K}$ for some $e \in \ko$ such that $\abs{e \cap (J \sqcup J_0)^{[0]}} \leq 1.$ Hence $\abs{\cface{y'}} = 1$ if and only if $e \in K_{(J \sqcup J_0)}$ and this proves (\ref{claim6thmdual}).

We can now turn back to showing that there is no $\partial K'_0$-pinch on $\partial K'_0.$ Since we assume that $K' \in \nsc,$ the relation (\ref{claim6thmdual}) implies that it is sufficient to prove that $\dual{K_{J_0}}$ is non-pinching. By Lemma \ref{lemmidsec}, we have that the map $ (M_{J_0}^K) \ni \E_{J_0}^x\mapsto x \in K_{J_0}$ is an isomorphism of posets. By Proposition \ref{propdual}, $(\dual{M_{J_0}^K}, \rk_{\dual{M_{J_0}^K}})$ is a closed cc and for all $x \in K_{J_0}$ we have
$$\rk_{\dual{M_{J_0}^K}}( \dual{\E_{J_0}^x} ) = R - 1 - \rk_{M_{J_0}^K}(\E_{J_0}^x) = R - \rk_K(x) = \rk_{\dual{K}}(\dual{x}) =: \rk_{\dual{K_{J_0}}}(\dual{x}).$$ 
Hence we have that $(\dual{K_{J_0}}, \rk_{\dual{K_{J_0}}})$ is a cc and $\dual{K_{J_0}} \cong \dual{M_{J_0}^K}.$ Therefore $\dual{K_{J_0}}$ is non-pinching since $M_{J_0}^K$ is cell-connected by locality of $(K - J).$
We can then conclude that the point \ref{thmdualcobpt1} is verified for all possible $J.$

The point \ref{thmdualcobpt2} is shown using (\ref{claim6thmdual}) (with $J_0 = \emptyset$) and observing that $J':= \dual{ \partial K \setminus J}$ is an element of $\mlc^{R-1}$ since it is the dual of a union of local non-pinching closed cell complexes. Finally $\dual{(K -J)}$ is exactly collared since, using the notation $L = \partial K \setminus J,$ the map 
$$ \{ \{y,z\} ~|~ L^{[R-1]} \ni y \subset z \in \kR\} \ni \{y,z\} \longmapsto y \in L^{[R-1]}$$
is a bijection and we have the isomorphism $\cong$ in the following derivation: 
\begin{align*}
M_{\dual{L}}^{\bual{K \setminus J} \sqcup \dual{L}} &= \{ \E_{\dual{L}}^{\bual{x}} ~|~ x \in L \}\\
 &= \{ \{ \dual{e} \in \bual{K}^{[1]}_{\dual{L}} ~|~ \dual{e} \subset \bual{x} \} ~|~ x \in L \} \\
 &= \{ \{ \{y,z\} ~|~ y \in L^{[R-1]}, ~ z \in \kR_L ~:~ x \subset y \subset z \} ~|~ x \in L \} \\
 &\cong \{ \cdual{x}{L} ~|~ x \in L \} = \dual{L}.
\end{align*}
\end{proof}

\newpage
\vspace{3cm}

\section{Composition of cobordisms and causal structure} \label{seccausalstruct}

\vspace{2cm}

The main goal of this chapter is to define a category whose morphisms are defined from the notion of cobordism introduced in Chapter \ref{secdualcob} and whose objects are associated with boundary components of cobordisms. 
Defining a category in particular requires to define the composition of two cobordisms whenever these cobordisms correspond to two morphisms satisfying that the target object of the first is the source object of the second. The definition of an object must then include enough information about the structure of any cobordism having this object as its image or as its source in order to guarantee that it is possible to "compose" any two such cobordisms so that the union of the cells of both cobordism can be given the structure of a cobordism. 

Defining how to compose cobordisms will therefore require a precise description of the structure of cobordism around their boundary components. The fundamental tools we introduce in order to describe such structure are two types of cc-homomorphisms called reductions and collapses, introduced in Section \ref{ssecredncol}. A reduction between two cc allows to consider the first cc as a subdivision of the second and a collapse corresponds to the notion dual to that of a reduction when defined on closed cc. Proposition \ref{propprecpartialorder}, the main result of this first section, shows that reductions and collapses induce partial orders on the set of pure cc.

Section \ref{ssectrans} introduces a number of technical assumptions defined on relative cc constituting the notion of "uniformity". These assumptions allow to introduce a cc called transition $J(K)$ associated to a uniform relative cc $(K,J).$ The main result of this second section implies that using the uniformity assumption one can relate both a boundary component $J$ of a cc $K$ in $\nsc$ and the associated midsection $M_J^K$ to the transition $J(K)$ using respectively a reduction $\sphi_J^K$ and a collapse $\cphi_J^K.$

The maps $\sphi_J^K$ and $\cphi_J^K$ satisfy the additional property of being "compatible" poset homomorphisms, a notion we define in Section \ref{sseccompnorthhom} together with the notion of orthogonal poset homomorphisms. Section \ref{sseccompnorthhom} also includes other notions and technical results introduced in preparation for the next section.


In Section \ref{ssecslice}, we prove a number of important results related to the notion of slice, a type of cobordism having in particular all its vertices on the boundary. These results include the "Correspondence Theorem" \ref{thmcorresp}, which establishes the existence of a bijection between slices and sequences of reductions and collapses called "slice sequences".

Section \ref{sseccompcob} starts by introducing the notion dual to that of slice sequences called connecting sequences. Using the Correspondence Theorem we are then able to define the union of to elements in $\nsc$ intersecting on one of their boundary components in such a way that it defines a connecting sequence and this can be directly translated into a composition operation on cobordisms.

Finally, in Section \ref{sseccatcausalcob} we introduce braket notation to define a category whose morphisms correspond to a specific type of cobordisms called causal cobordisms. 
We then use the previous results of this chapter to explain how this category can be interpreted as having connecting sequences and slice sequences as its set of objects and morphisms corresponding to cobordisms constructed as a succession of slices. We finish by showing that the duality map defines a duality on this category and that it admits two other idempotent functors, one of which is analogous to a "time reversal" transformation. Since the composition of these three idempotent functors also gives the identity, we named these functors "time reversal", "parity transformation" and "charge conjugation", the latter being the functor associated to the duality map.


 


\subsection{Reductions and collapses} \label{ssecredncol}

In this section, we introduce two types of surjective cc-homomorphisms called reductions and collapses.
A reduction can be understood as a map $\sphi : J \dans L$ between two cc obeying certain conditions that allow to interpret $J$ as a subdivision of the cc $L.$ The barycentric subdivision $\bdiv{L}$ defined in Chapter \ref{secgennot} will turn out to be a particular example of a subdivision. A collapse can be understood as a map $\cphi: J \dans L$ such that if $J, L \in \mlc$ then the dual map $\dual{\cphi}: \dual{J} \dans \dual{L},$ defined by $\dual{\cphi}(\dual{x}) = \dual{\cphi(x)}$ for $x \in J,$ is a reduction. 
Reductions and collapses are nevertheless both defined using five conditions allowing these maps to be defined on arbitrary cc.

These conditions may at first seem arbitrary, we therefore illustrated in Figure \ref{figcondred} some motivations for introducing four of these conditions for the case of a reduction using an example in dimension 1.
The second condition \ref{cond2red}, which is not represented on this picture, can be interpreted as follows. If all the co-faces of a given cell $x$ are mapped to the same cell $y$ then the image of $x$ has to be $y.$ This is a desirable property for a notion of subdivision that may be a consequence of the other four conditions in general, but we have not been able to prove this except for the case of cc in $\mlc$ as shown in Lemma \ref{lemcond2redcol}.
A number of basic results related to these maps are included in this section, among which Lemma \ref{lemdualredcol} shows that a reduction is indeed dual to a collapse for the case of closed cc. The more involved proof of Proposition \ref{propprecpartialorder} shows that reductions and collapses
have the important property of inducing partial orders on the set of pure cc. 

In the following sections the domain $J$ and image $L$ of a reduction and a collapse will generally consist of elements in $\mlc.$ A direct consequence of one of their five defining conditions is that $\Rk(J) = \Rk(L)$ and by Lemma \ref{lemcond2redcol} only four of these conditions are needed to characterize these maps for the case of closed cc. We close Section \ref{ssecredncol} by discussing the notions of relative reduction and relative collapse as maps between relative cc, as introduced in Section \ref{ssecrelcc}. The notion of relative reduction will be used to consider subdivisions of a given boundary component $J$ of a cc $K \in \nsc$ as a map acting trivially on the cells in $K \setminus J.$ 

As noted later in Remark \ref{remcond3redncol}, condition \ref{cond3red} (resp. \ref{cond3col}) in particular implies that if a map $\phi: J \dans K$ is a reduction (resp. a collapse) then it is in particular a cc-homomorphisms, i.e.  $\phi$ satisfies that for all $y \in \phi(J),$ there exists $x \in J^{[\rk_K(y)]}$ such that $\phi(x) = y.$ 

For the purpose of this section, let us recall the following notations from Section \ref{ssecbdiv} for the set of cells respectively above and below a cell $x$ in a poset $K:$
$$ A(x) := \{ y \in K ~|~ x \subset y \}, \quad B(x) := \{ y \in K ~|~ y \subset x \}. $$
It will also be convenient to denote by $S^{[r]}$ the set of $r$-cells in a subset $S$ of a cc of rank $R \geq r.$

Let us start by introducing the notion of reduction.

\begin{defi}[Reduction and subdivision]\label{defreduction}
Let $J,K,$ be cc. We say that a map $\sphi: J \dans K$ is a \textit{reduction} if $\sphi$ is a surjective poset homomorphism satisfying the following five conditions. 
\begin{enumerate} [label=(\textbf{r\arabic*}),topsep=2pt, parsep=2pt, itemsep=1pt]
\item \label{cond1red} $\abs{\sphi^{-1}(v)} = 1,$ for all $v \in \kz$ (i.e. $\sphi|^{\kz}$ is injective).
\item \label{cond2red} $\sphi(\bigwedge \cface{x}) = \bigwedge \sphi( \cface{x})$ for all $x \in J;$
\item \label{cond3red} $A( \sphi(x) )^{[r]} \subset \sphi(A(x)^{[r]}),$ for all $x \in J, ~r \geq 0;$
\item \label{cond4red} if $x \in J$ satisfies $\rk_J(x)= \rk_K(\sphi(x)) - 1$ then $\abs{ \cface{x} \cap \sphi^{-1}(\sphi(x))} = 2;$
\item \label{cond5red} if $x \in J$ satisfies $\rk_J(x)= \rk_K(\sphi(x))$ then 
$$\abs{ \cface{x} \cap \sphi^{-1}(y)} = 1, \quad \forall y \in \cface{\sphi(x)}.$$
\end{enumerate}
We shall use the notation $J \red_\sphi K$ if $\sphi: J \dans K$ is a reduction and call \textit{$J$ a subdivision of $K$} or equivalently say that \textit{$K$ is a reduction of $J,$} written $J \red K,$ if there exists a reduction $\sphi : J \dans K.$ A cell $y \in K$ will be called \textit{subdivided by} $\sphi$ if $\abs{\sphi^{-1}(y)} > 1.$ For example, a subdivision does not subdivide vertices.
\end{defi}

\begin{figure}[!h]
\centering
\includegraphics[scale=0.72]{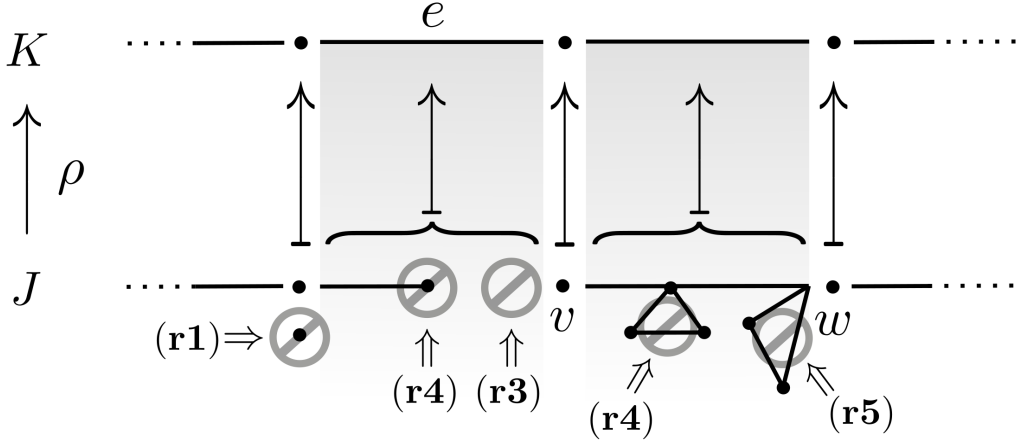}
\caption{\label{figcondred} In this picture, we illustrate what undesirable features the conditions \ref{cond1red}, \ref{cond3red}, \ref{cond4red} and \ref{cond5red} of a reduction are preventing, in the case of a reduction $\sphi: J \dans K$ between two 1-dimensional cell complexes. The vertical arrows indicate which cells in $J$ are mapped to each cell in $K,$ for example all the cells contained in the grey areas are mapped to an edge in $K.$ More explicitly, condition \ref{cond1red} prevents more than two vertices to be mapped to the same vertex. Condition \ref{cond3red} ensures that the vertex $v \in J$ is contained in an edge in the pre-image of the edge $e \in K$ since $\sphi(v) \in e.$ Condition \ref{cond4red} does not allow branchings or boundaries in the pre-image of a cell and condition \ref{cond5red} with $x = w \in J^{[0]}$ ensures in particular that the type of branching on $w$ shown on the figure does not occur or that a cell is not subdivided multiple times.}
\end{figure}

Since the definition of reductions involves the greatest lower bounds $\bigwedge \cface{x}$ and $\bigwedge \sphi( \cface{x})$ it is worth noting that such an element always exists by the following remark.

\begin{rema} \label{remaxred} 
As noted in Remark \ref{remcface} $\bigwedge \cface{x}$ always exists. As concerns $\bigwedge \sphi( \cface{x})$  we first note that since $\cface{ \sphi(x)} \subset \sphi( \cface{x} )$ by \ref{cond3red}, a lower bound $b$ of $\sphi(\cface{x})$ is also a lower bound of $\cface{\sphi(x)}.$ This implies that $b \subset \bigwedge \cface{\sphi(x)},$ which exists by Remark \ref{remcface}. If $\cface{x} \in \{\emptyset, \{y\}\}$ for some $y \in J$ then $\bigwedge \sphi(\cface{x})$ is equal to $\emptyset$ or $y,$ respectively, so we can assume $\abs{ \cface{x}} \geq 2.$ If $\sphi(x) = \bigwedge \cface{ \sphi(x)},$ which is a lower bound of $\sphi(\cface{x})$ since $\cface{\sphi(x)}\subset \sphi(\cface{x}),$ then every lower bound of $\sphi( \cface{x} )$ is contained in $\sphi(x),$ hence $\bigwedge \sphi(\cface{x}) = \sphi(x).$ If we assume that $\sphi(x) \subsetneq \bigwedge \cface{ \sphi(x)}$ then $\cface{ \sphi(x)} = \{ \sphi(y)\}$ for some $y \in K$ and therefore $\bigwedge \cface{ \sphi(x)} = \sphi(y).$ Using again that $\cface{ \sphi(x)} \subset \sphi(\cface{x}),$ we have that $y \in \cface{x}.$ If $y' \in \cface{x} \setminus \{y\}$ then we have
\begin{equation} \label{equreaxred}
\sphi(x) \subset \sphi(y') \notin \cface{\sphi(x)}.
\end{equation}
There are then two possibilities. If there exists $y' \in \cface{x}$ such that $\sphi(x) = \sphi(y')$ then $\bigwedge \sphi(\cface{x}) = \sphi(x)$ and, if not, then (\ref{equreaxred}) implies that $\sphi(y) \subset \sphi(y')$ for all $y' \in \cface{x}$ and therefore $\bigwedge \sphi( \cface{x} ) = \sphi(y),$ which concludes our argument, showing that $\bigwedge \sphi(\cface{x})$ always exists.
\end{rema}

\begin{lem}
The barycentric subdivision of a cc $K$ is a subdivision of $K,$ i.e. $\bdiv{K} \red K.$
\begin{proof}
Define $\sphi : \bdiv{K} \dans K$ by $\sphi( \{x_0 \subsetneq \dots \subsetneq x_r\} ) := x_r,$ clearly a surjective poset homomorphism and injective on $\sphi^{-1}(\kz) = \{ \{v\} ~|~ v \in \kz\}.$
Condition \ref{cond2red} follows from the definition of $\sphi,$ which, for the case $\cface{x_r} \neq \emptyset,$ implies that 
$$\sphi(\cface{\{ x_0 \subsetneq \dots \subsetneq x_r\}}) = \sphi(\{\{ x_0 \subsetneq \dots \subsetneq x_r \subsetneq y\} ~|~ y \in \cface{x_r}\}) = \cface{x_r},$$ 
hence we get
$$ \bigwedge \sphi( \cface{\{ x_0 \subsetneq \dots \subsetneq x_r\}} ) = \bigwedge \cface{x_r} = \cface{\sphi(\{x_0 \subsetneq \dots \subsetneq x_r\subset \bigwedge \cface{x_r}\})}.$$ The case $\cface{x_r} = \emptyset$ holds trivially.

Let $\{x_0 \subsetneq \dots \subsetneq x_r\} \in \bdiv{K}$ and $y \in K$ such that $\sphi(\{x_0 \subsetneq \dots \subsetneq x_r\}) \subset y$ then by \ref{cccenough} there exists a cell of $\bdiv{K}$ of the form $ \sigma = \{x_0 \subsetneq \dots \subsetneq x_r \subset \dots \subset y \}$ such that $\rk_{\bdiv{K}}(\sigma) = \rk_K(y)$  and $\sphi(\sigma) = y$ and therefore $\sphi$ also satisfies condition \ref{cond3red} of a reduction.
For condition \ref{cond4red}, if $r := \rk_{\bdiv{K}}(\{x_0 \subsetneq \dots \subsetneq x_r\}) = \rk_K(x_r) + 1$ then there exists a unique $0 \leq i \leq r$ such that $\rk_K(x_i) = \rk_K(x_{i+1})-2$ and therefore
\begin{align*}
& \abs{\cface{\{x_0 \subsetneq \dots \subsetneq x_r\}} \cap \sphi^{-1}(x_r) } \\
&= \abs{\{ \{x_0 \subsetneq \dots \subsetneq x_i \subsetneq y \subsetneq x_{i+1} \subsetneq \dots \subsetneq x_r \} ~|~ y \in \cface{x_i} \cap \face{x_{i+1}} \}}\\
&= \abs{\cface{x_i} \cap \face{x_{i+1}}} = 2.
\end{align*}

Finally, for condition \ref{cond5red}, if $ r:=  \rk_{\bdiv{K}}( \{x_0 \subsetneq \dots \subsetneq x_r \}) = \rk_K(x_r)$ then $\rk(x_i) = i $ for all $0 \leq i \leq r.$ Hence for all $x_{r+1} \in \cface{x_r},$ we have
\begin{align*}
\abs{\cface{\{x_0 \subsetneq \dots \subsetneq x_r\}} \cap \sphi^{-1}(x_{r+1}) }
= \abs{\{ \{x_0 \subsetneq \dots \subsetneq x_{r+1} \} \} } = 1.
\end{align*}

\end{proof}
\end{lem}

Let us now turn to the notion dual to that of a reduction for the case of closed cc. Since reductions are defined on general cc we decided to formulate the definition of collapse also using five conditions and to prove in Lemma \ref{lemdualredcol} that these are indeed dual notions when  defined on elements in $\mlc.$ 

Thus, in analogy with the interpretation of \ref{cond4red} in Figure \ref{figcondred}, the dual condition \ref{cond4col} of a collapse states that if a $r$-cell collapses into a cell of rank $r-1$ then it must be the result of an identification of two of its faces.

\begin{defi}[Collapse and expansion] \label{defcollapse}
Let $J,K,$ be cc. We say that $\cphi: J \dans K$ is a \textit{collapse} if $\cphi$ is a surjective poset homomorphism and 
\begin{enumerate} [label=(\textbf{c\arabic*}),topsep=2pt, parsep=2pt, itemsep=1pt]
\item \label{cond1col}$\abs{\cphi^{-1}(z)} = 1,$ for all maximal cell $z \in K;$
\item \label{cond2col}$\cphi(x) = \bigvee \cphi(\face{x})$ for all $x \in J$ such that $\rk_J(x) \geq 1;$
\item \label{cond3col}$B(\cphi(x))^{[r]} \subset \cphi(B(x)^{[r]})$ for all $x \in J, ~r \geq 0;$
\item \label{cond4col}if $x \in J$ satisfies $\rk_J(x) = \rk_K(\cphi(x)) + 1$ then $\abs{ \face{x} \cap \cphi^{-1}(\cphi(x))} = 2;$
\item \label{cond5col}if $x \in J$ satisfies $\rk_J(x)= \rk_K(\cphi(x))$ then 
$$\abs{ \face{x} \cap \cphi^{-1}(y)} = 1, \quad \forall y \in \face{\cphi(x)}.$$
\end{enumerate}
We will use the notation $J \col_\cphi K$ if $\cphi: J \dans K$ is a collapse and say that \textit{ $J$ is an expansion of $K$} or equivalently that \textit{$K$ is a collapse of $J,$} written $J \col K,$ if there exists a collapse $ \cphi : J \dans K.$
\end{defi}


Since the definition of collapse involves a least upper bound $\bigvee \cphi(\face{x}),$ we will also justify that condition \ref{cond2col} makes sense.

\begin{rema} \label{remaxcol} 
To see that if $\rk(x)_J \geq 1$ then $\bigvee \cphi(\face{x})$ exists in general, we first note that if $\cface{\cphi(x)} = \emptyset$ then $\cphi(x)$ is a vertex and therefore $ \cphi(y) = \cphi(x)$ for all $y \in \face{x}$ implying that $\bigvee \cphi(\face{x}) = \cphi(x).$ On the other hand if $\face{\cphi(x)} \neq \emptyset$ then $\abs{\face{\cphi(x)}} \geq 2$ and $\bigvee \face{\cphi(x)} = \cphi(x).$ Moreover, the fact that $ \face{ \cphi(x)} \subset \cphi(\face{x})$ by condition \ref{cond3col} implies that $\bigvee \cphi(\face{x}) \subset b$ for all upper bounds $b$ of $\cphi(\face{x}).$ Hence we have that $\cphi(x)$ is an upper bound for $\cphi(\face{x})$ which is contained in every upper bound of $\cphi(\face{x}),$ in other words, $\bigvee \cphi(\face{x}) = \cphi(x).$ 
\end{rema}

Let us point out that a cc $J$ can both be a reduction and a collapse of a given cc $K,$ as illustrated in Figure \ref{figexredncol}.


The next remark includes several observations to which we will refer several times in the next proofs. 

\begin{rema}\label{remcond3redncol} 
A consequence of condition \ref{cond3red} is that $\rk_K(\sphi(x)) \geq \rk_J(x),$ for a reduction \mbox{$\sphi: J \dans K$} and \ref{cond3col} implies that $\rk_K(\cphi(x)) \leq \rk_J(x)$ for a collapse $\cphi : J \dans K.$ This can be seen for example for the case of a reduction $\sphi$ by considering the instance $r = \rk_K(\sphi(x))$ in condition \ref{cond3red} according to which
$$ \{\sphi(x)\} = A( \sphi(x) )^{[r]} \subset \sphi(A(x)^{[r]}).$$
This implies $A(x)^{[r]} \neq \emptyset$ and therefore $\rk_J(x) \leq r = \rk_K(\sphi(x))$ for all $x \in J.$ The case of a collapse is similar. This argument also shows that conditions \ref{cond3red} and \ref{cond3col} in particular imply that reductions and collapses are cc-homomorphisms.
\end{rema}

\begin{figure}
\centering
\includegraphics[scale=0.4]{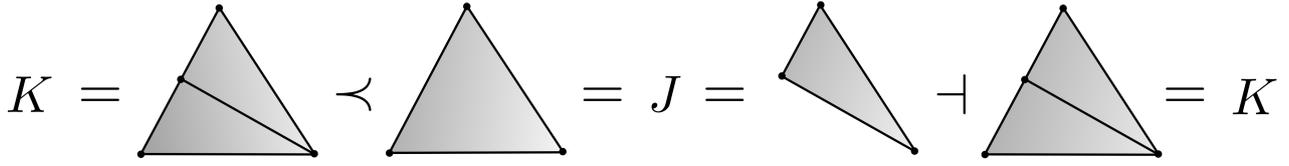}
\caption{\label{figexredncol} Example of cell complexes $J$ and $K$ such that $J \red K$ and $J \col K.$}
\end{figure}


\begin{rema} \label{remrksredncol} 
If $J \red K$ or $J \col K$ and $J$ is pure then $\Rk(J) = \Rk(K).$ Indeed, for the case $J \red K,$ $\sphi = \sphi_J^K$ is in particular a cc-homomorphism hence there exists a $\Rk(K)$-cell in $\sphi^{-1}(z),$ for all $z \in K^{[\Rk(K)]}$ and therefore $\Rk(J) \geq \Rk(K).$ Remark \ref{remcond3redncol} also implies $\rk_J(x) \leq \rk_K(\sphi(x))$ for all $x \in J$ and we get $\Rk(J) = \Rk(K).$ 

For the case $J \col K$ we need to assume that $J$ is pure. Figure \ref{fignonequidimcol} illustrates why this hypothesis is necessary. By Remark \ref{remcond3redncol} $\cphi = \cphi_J^K$ is a cc-homomorphism and we also directly get the inequality $\Rk(J) \geq \Rk(K).$ Since $\cphi^{-1}(z) = \{w\} \subset J^{[\rk_K(z)]}$ for any maximal cell $z \in K$ then $w$ is a maximal cell of $J$ as otherwise if $w \subsetneq z'$ then $ z = \cphi(w) \subset \cphi(z')$ and therefore $\cphi(z') = z$ by maximality of $z,$ which is a contradiction. Since $J$ is pure, $\rk_J(w) = \Rk(J),$ and the previous argument implies in particular that $\Rk(J) \leq \Rk(K)$ and therefore $\Rk(J) = \Rk(K).$
\end{rema}

\begin{figure}
\centering
\includegraphics[scale=0.42]{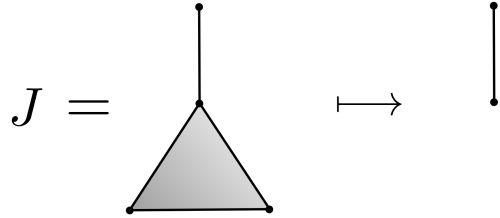}
\caption{\label{fignonequidimcol} Example of a collapse between cell complexes of different ranks. By Remark \ref{remrksredncol} this requires the domain of the collapse $J$ to be not pure. The triangle in $J$ and all the cells contained in it are mapped to the lower vertex.}
\end{figure}

We also mention the following technical result that will be used in the proof of Lemma \ref{lemcond2redcol}. It implies in particular that if there exists a reduction from a pure cc $J$ to $K \in \mlc$ then $\partial J = \emptyset.$

\begin{lem} \label{claimbdryredcol}
Let $J$ be a pure cc and let $K \in \mlc.$ If $\sphi: J \dans K$ is a reduction then
$\abs{ \cface{x} } \geq 2$ for all non-maximal cell $x$ in  $J.$
\begin{proof}
To prove this claim, we assume by contradiction that there exists $x \in J$ such that $\cface{x} =\{y\}.$ Since $\cface{\sphi(x)} \subset \sphi(\cface{x})$ by condition \ref{cond3red}, we have that $\abs{\cface{\sphi(x)}} \leq 1,$ hence $\cface{\sphi(x)} = \emptyset,$ since $K \in \mlc.$ The inclusion $\sphi(x) \subset \sphi(y)$ then implies that $\sphi(x) = \sphi(y).$ We also have that $y$ is a maximal cell, since if there exists $z \in J$ such that $x \subsetneq y \subsetneq z$ then the diamond property for $J$ implies that $\abs{\cface{x}} \geq 2$ which is a contradiction. Since both $\sphi(x) =  \sphi(y)$ and $y$ are maximal cells in respectively $K$ and $J$, we have that $\rk_J(y) = R = \rk_K(\sphi(x))$ by Remark \ref{remrksredncol} and therefore $$\rk_J(x) = \rk_J(y) - 1 = \rk_K(\sphi(x)) -1.$$ Condition \ref{cond4red} then implies that $$\abs{\cface{x} \cap \sphi^{-1}(\sphi(x))} = 2,$$ which is also a contradiction.
\end{proof}
\end{lem}


\begin{lem} \label{lemcond2redcol}
Let $J$ be a pure cc and let $K \in \mlc.$ If $\sphi: J \dans K$ is a reduction then condition \ref{cond3red} implies condition \ref{cond2red}. Similarly, if $\cphi: J \dans K$ is a collapse then condition \ref{cond3col} implies condition \ref{cond2col}.
\begin{proof}
Let us first prove the statement for reductions.
Since $K \in B$ then either $\abs{\cface{\sphi(x)}} \geq 2$ and $ \bigwedge \cface{\sphi(x)} = \sphi(x)$  or $\cface{\sphi(x)} = \emptyset.$ The same is also true for the cells in $J,$ by Lemma \ref{claimbdryredcol}.

Moreover, we have the two following properties. If $ \bigwedge S, \bigwedge \sphi(S)$ exists then since $\sphi$ is a poset homomorphism we have $\sphi(\bigwedge S) \subset \bigwedge \sphi(S)$ and if $S, T \in \pws{\kz}$ such that $S \subset T$ and $\bigwedge S, \bigwedge T$ exist and are non-empty then $\bigwedge T \subset \bigwedge S.$ Using these properties and the fact that condition \ref{cond3red} implies that $\cface{\sphi(x)} \subset \sphi(\cface{x}),$ for the case $\abs{ \cface{\sphi(x)} } \geq 2$ we have that $\abs{\cface{x}} \geq 2$ and
$$\sphi(x) = \sphi( \bigwedge \cface{x} ) \subset \bigwedge \sphi(\cface{x}) \subset \bigwedge \cface{\sphi(x)} = \sphi(x),$$
hence $\sphi(\bigwedge \cface{x}) = \sphi(x) = \bigwedge \sphi( \cface{x} ).$ If $\cface{ \sphi(x)} = \emptyset$ and $\abs{\cface{x}} \geq 2$ then $\sphi(y) = \sphi(x)$ for all $y \in \cface{x}$ by maximality of $\sphi(x)$ and therefore we can also conclude that
$$ \sphi( \bigwedge \cface{x}) = \sphi(x) = \bigwedge \sphi(\cface{x})$$
and this concludes the proof for the case of reductions.

If $\cphi : J \dans K$ is a collapse then the dual statement of Lemma \ref{claimbdryredcol} is clearly true, namely that every cell which is not a vertex has more than two faces. The rest of the proof then follows using the dual version of the arguments above.
\end{proof} 
\end{lem}

\begin{lem} \label{lemdualredcol}
Let $K, L \in \mlc.$ Then $L \red_\sphi K$ if and only if $\dual{L} \col_{\dual{\sphi}} \dual{K},$ where $ \dual{\sphi}: \dual{J} \dans \kb$ is defined by
\begin{equation*}
x \longmapsto \dual{\sphi}(x) := \dual{\sphi(x)}.
\end{equation*}
\begin{proof}
By Remark \ref{remrksredncol}, $\Rk(K) = \Rk( L ) =:R.$ Let $\sphi : L \dans K$ be a reduction. Then map $\dual{\sphi}$ is well defined by Lemma \ref{claim2propdual}, is clearly surjective and is a poset homomorphism since if $\dual{x} \subset \dual{y}$ then $y \subset x$ and
$$ \dual{\sphi}(\dual{x}) = \dual{\sphi(x)} \subset \dual{\sphi(y)} = \dual{\sphi}(\dual{y}).$$

First, if $\dual{y} \in \dual{K}^{[R]}$ then $y = v \in \kz$ and $\abs{\sphi^{-1}(v)}=1,$ therefore 
$$\abs{(\dual{\sphi})^{-1}(\dual{y})} = \abs{\dual{ \sphi^{-1}(v)}} = 1,$$ which shows that $\dual{\sphi}$ satisfies condition \ref{cond1col} of a collapse (we recall that according to our notation we have $\dual{S} =\{ \dual{s} ~|~ s \in S\}$ for a subset $S \subset K$).

Let $y \in L^{[r]}.$ We have the following relation
\begin{align*}
(\dual{\sphi})^{-1}( \dual{y} ) &= \{ \dual{x} \in \dual{K} ~|~ \dual{\sphi}(\dual{x}) = \dual{ \sphi(x)} = \dual{y} \}\\
&= \{ \dual{x} ~|~ x \in \sphi^{-1}(y) \} =: \dual{\sphi^{-1}(y)}.
\end{align*}
Therefore, if $\dual{y} \subset \dual{\sphi}(\dual{x}) = \dual{\sphi(x)}$ then $\sphi(x) \subset y,$ hence there exists an $r$-cell $w \in \sphi^{-1}(y)$ such that $x \subset w.$ This means that $\dual{w} \in \dual{\sphi}^{-1}(\dual{y})$ has rank $R - r = \rk_{\dual{L}}(\dual{y})$ and $ \dual{w} \subset \dual{x}$ and $\dual{\sphi}$ satisfies condition \ref{cond3col}. 

The map $\dual{\sphi}$ satisfies condition \ref{cond4col} since if $\rk_{\dual{L}}(\dual{y}) = \rk_{\dual{K}}(\dual{\sphi}(\dual{y})) + 1$ then $\rk_L(y) = \rk_K(\sphi(y)) - 1$ and we have the following relation:
$$ \abs{\face{\dual{y}} \cap \dual{\sphi}^{-1}(\dual{\sphi}(\dual{y})} = \abs{ \dual{\cface{y}} \cap \dual{ \sphi^{-1}(\sphi(y))}} = \abs{\dual{ \cface{y} \cap \sphi^{-1}(\sphi(y))}} = \abs{ \cface{y} \cap \sphi^{-1}(\sphi(y))} = 2,$$
by condition \ref{cond4red} of a reduction.

Finally, condition \ref{cond5col} is checked as follows. If $\rk_{\dual{J}}(\dual{x}) = \rk_{\dual{K}}(\dual{\sphi}(\dual{x}))$ then $\rk_J(x) = \rk_K(\sphi(x))$ and if $\dual{y} \in \face{\dual{\sphi}(\dual{x})}$ then $y \in \cface{\sphi(x)}$ and
$$ \abs{\face{\dual{y}} \cap \dual{\sphi}^{-1}(\dual{y})} = \abs{\dual{\cface{y} \cap \sphi^{-1}(y)}} = 1.$$

By Lemma \ref{lemcond2redcol} condition \ref{cond2col} is automatically satisfied and we have therefore proven that $\dual{\sphi}$ is a collapse. 

The reverse implication is proven using the same arguments while reversing the inclusion relations between cells and replacing vertices with maximal cells (i.e. $R$-cells) and vice versa.
\end{proof}
\end{lem}

The following proposition is an important property of reductions and collapses. It will be used in particular in Proposition \ref{propredbdry} where we show that subdivisions of a boundary component defines an operation on elements in $\nsc.$ A sizeable part of the proof of the next result is about proving that condition \ref{cond2red} is transitive and by Lemma \ref{lemcond2redcol} this is only necessary for the case of pure cc with non-empty boundary.

\begin{prop} \label{propprecpartialorder}
$\red$ and $\col$ are  partial orders on the set of pure cc, up to cc-isomorphism.
\begin{proof}
We will show this result for $\red.$ The proof applies also to $\col$ using the same arguments while reversing inclusion relations between cells and replacing vertices with maximal cells. Also, the arguments in the following proof involving condition \ref{cond2red} are simplier for the case of collapses, since condition \ref{cond2col} can be expressed without involving a least upper bound, as noted in Remark \ref{remaxcol}.

The relation $\red$ is clearly reflexive. The main part of this proof focuses on showing that $\red$ is transitive. For this, we assume that $J \red_{\sphi_1} K,$ $K \red_{\sphi_2} L$ and check the five conditions for $\sphi_2 \circ \sphi_1.$

Condition \ref{cond1red} for $\sphi_2 \circ \sphi_1$ is checked easily: if $y \in L^{[0]}$ we have $\abs{\sphi_2^{-1}(y)} = 1$ since $\sphi_2$ is a reduction, i.e. $ v = \sphi_2^{-1}(y) \in \kz$ and we have $\abs{\sphi_1^{-1}(v)} = 1$ since $\sphi_1$ is a reduction, hence $\abs{(\sphi_2 \circ \sphi_1)^{-1}(y)} = 1.$

Proving condition \ref{cond3red} for $\sphi_2 \circ \sphi_1$ is relatively direct as well. If $y \in L^{[r]}, x \in J$ are such that $\sphi_2 \circ \sphi_1 (x) \subset y$ then since $\sphi_2$ is a reduction there exists $w' \in \sphi_2^{-1}(y)^{[r]}$ such that $\sphi_1(x) \subset w'.$ Also, since $\sphi_1$ is a reduction there exists $w \in \sphi^{-1}(w')^{[r]} \subset (\sphi_2 \circ \sphi_1)^{-1}(y)^{[r]}$ such that $x \subset w$ and this proves that $$A( \sphi_2(\sphi_1(x)))^{[r]} \subset \sphi_2(\sphi_1(A(x)^{[r]})).$$

The transitivity of the second condition is more technical. We first note that condition \ref{cond2red} for $\sphi_1$ implies the following equality
\begin{equation} \label{eqtransred}
\sphi_2( \sphi_1 ( \bigwedge \cface{x})) = \sphi_2(\bigwedge \sphi_1( \cface{x})).
\end{equation}
In addition, since $\sphi_2$ is a poset homomorphism, we have
\begin{equation} \label{incltransred}
\sphi_2(\bigwedge \sphi_1(\cface{x}) \subset \bigwedge \sphi_2(\sphi_1(\cface{x})),
\end{equation} 
hence it remains to show the reversed inclusion. By condition \ref{cond3red} we have in particular that 
\begin{equation}\label{inclr3transred}
\cface{\sphi_1(x)} \subset \sphi_1(\cface{x})
\end{equation}
which impies that 
\begin{equation}\label{equtransred}
\bigwedge \sphi_2(\sphi_1(\cface{x})) \subset \bigwedge \sphi_2( \cface{ \sphi_1(x)}) = \sphi_2( \bigwedge \cface{ \sphi_1(x)}),
\end{equation}
where the equality is due to condition \ref{cond2red} for $\sphi_2.$

The inclusion (\ref{inclr3transred}) implies that $\abs{\cface{\sphi_1(x)}} \leq \abs{\cface{x}}.$ Therefore we can split the remainder of the proof of the transitivity of \ref{cond2red} in four cases:
\begin{itemize} [topsep=2pt, parsep=2pt, itemsep=3pt]
\item $\abs{\cface{x})} \geq \abs{\sphi_1(x))} \geq 2,$
\item $\abs{\cface{\sphi_1(x)}} = \abs{\cface{x}} = 1,$
\item $\abs{\cface{x}} > 0, ~\abs{\cface{\sphi_1(x)}} = 0$ (uses the assumption of $J$ pure!),
\item $\abs{\cface{x}} \geq 2, ~ \abs{\cface{\sphi_1(x)}} =1.$
\end{itemize}

In the first case, when $\abs{\cface{\sphi_1(x)}} \geq 2$ and $\abs{\cface{x}} \geq 2,$ the Lemma \ref{lemcface} directly gives
$$ \bigwedge \cface{ \sphi_1(x) } = \sphi_1(x) = \sphi_1( \bigwedge \cface{x})$$ and we get the desired equality using  (\ref{eqtransred}), (\ref{incltransred}) and (\ref{equtransred}). 

In the second case we get $\abs{ \cface{x}} = \cface{\sphi_1(x)} = 1,$ hence we can pick $y \in \cface{x}$ to be the unique co-face of $x.$ Then the inclusion (\ref{inclr3transred}) implies that $\cface{\sphi_1(x)} \subset \sphi_1(\cface{x}) =\{y\}.$ Therefore we obtain $\cface{\sphi_1(x)} = \{\sphi_1(y)\}$ and
$$ \bigwedge \cface{\sphi_1(x)} = \sphi_1(y) = \sphi_1(\bigwedge \cface{x}).$$
We can then also concludes the argument using (\ref{eqtransred}), (\ref{incltransred}) and (\ref{equtransred}).

Third, if $\cface{x} \neq \emptyset$ and $\cface{\sphi_1(x)}= \emptyset,$ an argument  similar as in the proof of Lemma \ref{claimbdryredcol} shows that $\abs{\cface{x}} \geq 2,$ by the purity of $J.$ We also have that $\sphi(y) = \sphi(x)$ for all $y \in \cface{x},$ by maximality of $\sphi(x).$ Therefore we get
$$ \sphi_2(\sphi_1(\bigwedge \cface{x})) = \sphi_2(\sphi_1(x)) = \bigwedge \{\sphi_2( \sphi_1(x))\} = \bigwedge \sphi_2(\sphi_1(\cface{x})),$$
which corresponds exactly to condition \ref{cond2red} for $\sphi_2 \circ \sphi_1.$

For the last case we have $ \abs{\cface{x}} \geq 2$ and $\abs{\cface{\sphi_1(x)}} =1$ and condition \ref{cond2red} for $\sphi_1$ implies that $$ \sphi_2( \bigwedge \sphi_1( \cface{x} )) = \sphi_2( \sphi_1( \bigwedge \cface{x})) = \sphi_2(\sphi_1(x)).$$  This implies that there exists $y \in \cface{x}$ such that $\sphi_1(y) = \sphi_1(x).$ Therefore $$\sphi_2(\sphi_1(\bigwedge \cface{x})) = \sphi_2(\sphi_1(x)) \subset \bigwedge \sphi_2(\sphi_1(\cface{x})),$$ and we can conclude using the inclusion (\ref{incltransred}).

We will show condition \ref{cond5red} for $\sphi_2 \circ \sphi_1$ and finish the proof of transitivity with condition \ref{cond4red} below. 
If we assume that $\rk_J(x) = \rk_L(\sphi_2(\sphi_1(x))),$ it implies that $\rk_K(\sphi_1(x)) = \rk_J(x)$ since the rank is increasing under subdivision as observed in Remark \ref{remrksredncol}. Hence conditions \ref{cond5red} for $\sphi_1$ and $\sphi_2$ imply respectively the two following conditions:

\begin{equation}\label{enumcond5_1}
\abs{ \cface{x} \cap \sphi_1^{-1}(y)} = 1 \quad \text{for all } y \in \cface{ \sphi_1(x)};
\end{equation}

\begin{equation}\label{enumcond5_2}
\abs{ \cface{\sphi_1(x)} \cap \sphi_2^{-1}(y)} = 1 \quad \text{for all } y \in \cface{ \sphi_2(\sphi_1(x))}.
\end{equation}

Using that 
\begin{equation}
f( A \cap f^{-1}(B)) = f(A) \cap B
\end{equation}
for a function $f : S \dans T$ between two sets $S,T$ and $A \subset S,~ B \subset T$ and the fact that condition \ref{cond3red} implies $\cface{\sphi_1(x)} \subset \sphi_1(\cface{x}),$ we get the following inequality for all $y \in \cface{ \sphi_2( \sphi_1(x))}:$
\begin{align*}
\abs{ \cface{x} \cap (\sphi_2 \circ \sphi_1)^{-1}(y)} &\geq \abs{ \sphi_1( \cface{x} \cap (\sphi_2 \circ \sphi_1)^{-1}(y))} \\
&= \abs{ \sphi_1(\cface{x}) \cap \sphi_2^{-1}(y)}\\
&\geq \abs{ \cface{x} \cap \sphi_2^{-1}(y)} = 1,
\end{align*}
where the last equality is due to condition (\ref{enumcond5_1}).

Moreover if we suppose that for some $y \in \cface{\sphi_2(\sphi_1(x))}$ there are two distinct elements 
$$w,w' \in \cface{x} \cap (\sphi_2 \circ \sphi_1)^{-1}(y)$$ 
then since $\cface{ \sphi_2(\sphi_1(x)} \subset \sphi_2(\cface{\sphi_1(x)})$ by condition \ref{cond3red}, there exists $y' \in \cface{\sphi_1(x)} \cap \sphi_2^{-1}(y)$ and this element is unique by condition (\ref{enumcond5_2}). This implies that
$$ \{w,w'\} \subset \abs{ \cface{x} \cap \sphi_1^{-1}(y')}$$ which contradicts condition (\ref{enumcond5_1}) and therefore we have proven \ref{cond5red} for $\sphi_2 \circ \sphi_1.$ 

Finally, we show condition \ref{cond4red} for $\sphi_2 \circ \sphi_1$ as follows. If $x \in J$ satisfies 
$$\rk_J(x) = \rk_L(\sphi_2(\sphi_1(x))) - 1$$
then we can have the two following possible cases.

The first possible case is
$$rk_J(x) = \rk_K(\sphi_1(x)), \quad \rk_K(\sphi_1(x)) = \rk_L(\sphi_2(\sphi_1(x))) - 1,$$
which respectively implies condition (\ref{enumcond5_1}) as well as condition \ref{cond4red} for $\sphi_2,$ i.e.
\begin{equation}\label{equsphi2proppreord}
\abs{ \cface{\sphi_1(x)} \cap \sphi_2^{-1}(\sphi_2(\sphi_1(x)))} = 2.
\end{equation}
Every element $w \in \cface{x} \cap (\sphi_2 \circ \sphi_1)^{-1}(\sphi_2(\sphi_1(x)))$ satisfies $\sphi(x) \subset \sphi(w)$ and $\sphi_2(\sphi_1(w)) = \sphi_2(\sphi_1(x)),$ hence
$$r + 1 = \rk_J(w) \leq \rk_K(\sphi_1(w)) \leq \rk_L(\sphi_2(\sphi_1(w))) = \rk_L(\sphi_2(\sphi_1(x))) = r + 1.$$
Therefore $\rk_K(\sphi_1(w)) = r + 1$ and $\sphi_1(w) \in \cface{\sphi_1(x)}.$  Condition (\ref{enumcond5_1}) implies that $\sphi_1^{-1}(\sphi_1(w)) = \{w\}$ for every such $w$ and using (\ref{equsphi2proppreord}) we get
$$ \abs{ \cface{x} \cap (\sphi_2 \circ \sphi_1)^{-1}(\sphi_2(\sphi_1(x)))} = \abs{ \cface{\sphi_1(x)} \cap \sphi_2^{-1}( \sphi_2(\sphi_1(x)))} = 2.$$

The second possible case is
$$rk_J(x) = \rk_K(\sphi_1(x)) - 1, \quad \rk_K(\sphi_1(x)) = \rk_L(\sphi_2(\sphi_1(x))).$$
This case implies condition \ref{cond4red} for $\sphi_1:$
\begin{equation}\label{equsphi1proppreord}
\abs{\cface{x} \cap \sphi_1^{-1}(\sphi_1(x))} = 2.
\end{equation}
If $w \in \cface{x} \cap (\sphi_2 \circ \sphi_1)^{-1}(\sphi_2(\sphi_1(x)))$ then $\sphi(x) \subset \sphi(w)$ and $\rk_K(\sphi(w)) = r + 1 = \rk_K(\sphi_1(x))$ by the same argument as above, therefore $\sphi_1(w) = \sphi_1(x).$
This implies that
$$\cface{x} \cap (\sphi_2 \circ \sphi_1)^{-1}(\sphi_2(\sphi_1(x))) = \cface{x} \cap \sphi_1^{-1}(\sphi_1(x))$$
and we get the desired result by (\ref{equsphi1proppreord}). This concludes the proof that $\sphi_2 \circ \sphi_1$ is a reduction and therefore $\red$ is transitive.

Finally, the relation $\red$ is antisymmetric by the following argument. Suppose $ J \red_\phi K$ and $K \red_\gamma J$  (which by Remark \ref{remrksredncol} implies $\Rk(J) = \Rk(K)$) and define $\phi_r := \phi|^{K^{[r]}}$ and $\gamma_r := \gamma|^{J^{[r]}}.$ We will prove by induction on $r$ that $\phi_r : J^{[r]} \dans \kr$ is a bijection for all $0 \leq r \leq R$ implying that $\phi$ is an isomorphism by Lemma \ref{lemisom}. For the case $r=0$ we know that $\phi_0$ and $\gamma_0$ are surjective and this implies that $\abs{J^{[0]}} = \abs{\kz}.$ Hence $\phi_0$ and $\gamma_0$ are also injective and therefore bijective. Suppose that $\phi_s : J^{[s]} \dans K^{[s]}$ and $\gamma_s : K^{[s]} \dans J^{[s]}$ are bijective for all $s = 0,1 \dots r-1.$ As noted in Remark \ref{remcond3redncol} since $\phi$ is a reduction we have $\phi^{-1}( \kr ) \subset J^{(r)}$ and $J^{(r-1)} = \phi^{-1}(K^{(r-1)}).$ Hence $\phi^{-1}(\kr) \subset J^{[r]}$ and $\abs{\kr} \leq \abs{J^{[r]}}.$ By the same argument for $\gamma$ we also get the reverse inequality and therefore $\phi_r : J^{[r]} \dans \kr$ and $\gamma_r : \kr \dans J^{[r]}$ are indeed bijections. 
\end{proof}
\end{prop}

The proof of antisymmetry in Proposition \ref{propprecpartialorder} only requires condition \ref{cond3red} of reductions (and the same is true for collapse). As a consequence, the third condition together with the surjective poset homomorphism property are sufficient to define a partial order on the set of cc. 

Before giving the definitions of relative reduction and collapse, we recall that by Definition \ref{defirelccc} a relative cc-homomorphism or a homomorphism of $J$ relative to $K$ is a homomorphism $\phi :(K,J) \dans (H,L)$ of relative cc (i.e. a cc-homomorphism $\phi : K \dans H$ such that $\phi(J) = L$) such that $\phi|_{K \setminus J} : K \setminus J \dans H \setminus L$ is a poset isomorphism.

\begin{defi}[Relative reduction and collapse] \label{defirelredncol}
Let $(K,J), (H,L)$ be relative cc. A \textit{ relative reduction (respectively a relative collapse)} or a \textit{reduction (respectively collapse) of $J$ relative to $K$} is a cc-homomorphism $\phi : (K,J) \dans (H,L)$  of $J$ relative to $K$ such that $\phi : K \dans H$ is a reduction (respectively a collapse). We will also write $(K,J) \red (H,L)$ and say that $(K,J)$ is \textit{a subdivision of $L$ relative to $H$} if there exists a reduction $\sphi : (K,J) \dans (H,L)$ of $J$ relative to $K.$ Similarly we will write $(K,J) \col (H,L)$ and say that $(K,J)$ is \textit{an expansion of $L$ relative to $H$} if there exists a collapse $\cphi: (K,J) \dans (H,L)$ of $J$ relative to $K.$
\end{defi}

Clearly Definition \ref{defirelredncol} implies in particular that a relative cc-homomorphism $\phi : (K,J) \dans (H,L)$ is a relative reduction (respectively a relative collapse) if and only if $\phi|_J : J \dans L$ is a reduction (respectively a collapse). By abuse of notation, we will sometimes also use the notation $\phi$ to designate $\phi|_J$ when $\phi : (K,J) \dans (H,L)$ is a relative homomorphism.

It will be convenient to use the following notation on a cc $K$ in order to specify how it is related to a relative cc-homomorphism.

\begin{defi}[$K_\phi$ and $K^\phi$] \label{defikphi}
Let $\phi: J \dans L$ be a surjective cc-homomorphism and let $H$ be a cc such that $ L \leq H.$ If there exists a cc $K \geq J$ and a relative cc-homomorphism $\phi^H : (K,J) \dans (H,L)$ such that $\phi^H|_J = \phi$ then $K$ is uniquely determined by $H$ and $\phi : J \dans L$ up to cc-isomorphism, we will therefore use the notation $K = H^\phi.$

Similarly, let $\phi: J \dans L$ be a cc-homomorphism and let $K$ be a cc such that $ J \leq K.$ If there exists a cc $H \geq L$ and a relative cc-homomorphism $\phi_K : (K,J) \dans (H,L)$ such that $\phi_K|_J = \phi$ then $H$ is uniquely determined by $K$ and $\phi : J \dans L$ up to cc-isomorphism, we will therefore use the notation $H = K_\phi.$
\end{defi}

The uniqueness in Definition \ref{defikphi} follows from the observation that, say in the first case, whenever we take a relative cc $(K',J)$ such that there exists a relative cc-homomorsphism $ \phi_{K'} : (K',J) \dans (H,L)$ then by Definition \ref{defirelccc} the map $\psi : K \dans K'$ defined by $\psi|_J = \id_J$ and $\psi|_{K \setminus J} = \phi_{K'}^{-1} |_{H \setminus L} \circ \phi_K|_{K \setminus J},$ is a poset isomomorphism such that $\psi^{-1}$ is also a poset isomomorphism, hence $\psi$ is a cc-isomorphism by Lemma \ref{lemisom}. Moreover the isomorphism class of $K$ does not depend on the choice of the realtive cc-homomorphisms $\phi_K$ and $\phi_{K'}.$

Let us make the following remark about the previous definitions considering the example of a reduction  $\sphi : J \dans L.$ In general it is wrong that given a relative cc $(H ,L)$ there exists a cc $K$ such that $(K, J) \red_{\sphi} (H,L)$ is a relative reduction or that given a relative cc $(K,J)$ there exists a cc $H$ such that $(K, J) \red_{\sphi} (H,L)$ is a relative reduction.

One reason that this is wrong for the latter case is that it is possible for the reduction $\sphi$ to reduce two $r$-cells $x_1,x_2$ in $J$ each sharing a face $y_1,y_2 \in J^{[r-1]}$ with a $r$-cell in $y \in K_J$ into one cell $x =\sphi(x_1) = \sphi(x_2),$ hence $\{x_1,x_2\} \subset x \cap y,$ violating the intersection Axiom \ref{cccinter} for $H.$

Similarly, a reason for the first case is that a relative subdivision of $(H,L)$ can subdivide a cell $x$ of $L$ contained in a cell of $y \in H_L$ only if every cell $w$ of $L$ strictly containing $x$ is also subdivided in such a way that any intersection of an element in $\sphi^{-1}(y)$ with an element in $\sphi^{-1}(w)$ contains at most one cell, which is not the case in general. 

We close this section with a remark about the interpretation of the notion of relative cc- automorphisms in our framework. 

\begin{rema}\label{remautomrelsubd} 
Consider two cc $J,L$ such that for example $J \red L.$ If $J, L$ satisfy the stronger condition that there exist $K,H$ such that $(K,J) \red (H,L)$ then this introduces some freedom in the choice of $K$ and $H$ in the following sense. If we take any relative cc-automorphism $\alpha : (K,J) \dans (K_\alpha,J)$ (i.e. a relative cc-homomorphism $\alpha : (K,J) \dans (K_\alpha,J) $ such that $\alpha|_J: J \dans J$ is an cc-automorphism) and any relative cc-automorphism $\beta : (H,L) \dans (H_\beta, L)$ then  $(K_\alpha, J) \red (H_\beta, L).$ Similarly if $(K,J) \col (H,L)$ then $(K_\alpha, J) \col (H_\beta, L).$ But in general $K \not \cong K_\alpha.$ This freedom in the choice of $K$ and $H$ when obtaining a relative reduction could be interpreted in the context of cobordisms as a discrete notion of gauge symmetry.
\end{rema}

\subsection{Transition and uniform relative cc} \label{ssectrans}

An important motivation for introducing reductions and collapses is to establish a relation between a boundary component $J$ of a cc $K \in \nsc$ and the associated midsection $M_J^K.$ In order to do that, we need to add a number of additional assumptions on the relative cc $(K,J)$ that will go under the name of "uniformity". These include several rather technical assumptions related to a poset called "transition", written $J(K),$ constructed from the cells in the collar $K_J.$ The uniformity of $(K,J)$ in particular implies that $J(K) \in \mlc$ and in Proposition \ref{propredbdry} we show that in the case $K \in \nsc^{R+1}$ and $J \leq \partial K,$ $J \in \mlc^R$ the assumption of uniformity also implies that there exist a reduction $\sphi_J^K$ and a collapse $\cphi_J^K$ such that $$J \red_{\sphi_J^K} J(K) \loc_{\cphi_J^K} M_J^K.$$



In view of the definition of relative reduction and collapse given at the end of the last section, it is natural to ask if there exists a "coarsest" relative reduction of a relative cell complex $(K,J)$ i.e. some maximal cell complex $J(K)$ with respect to the partial order $\red$ such that $(K,J) \red (K',J(K))$ for some $K'.$ In order to keep aside any question related to the existence of $J(K)$ when referring to it, it will be more convenient to use an explicit definition of $J(K)$ simply as a poset that one can associate to a relative cell complex. 

\begin{defi}[Transition] \label{defreducedbdry}
Let $(K,J)$ be a relative cc. For $x \in K_J,$ we define the \textit{reduced cell of $x$} to be
$$ x_J := \{ e \cap J^{[0]} ~|~ e \in \E_J^x \}.$$
We then define the \textit{transition of $(K,J)$} to be the poset
$$J(K) := \{ x_J ~|~ x \in K_J \}.$$
\end{defi}

\begin{figure}[!h]
\centering
\includegraphics[scale=0.42]{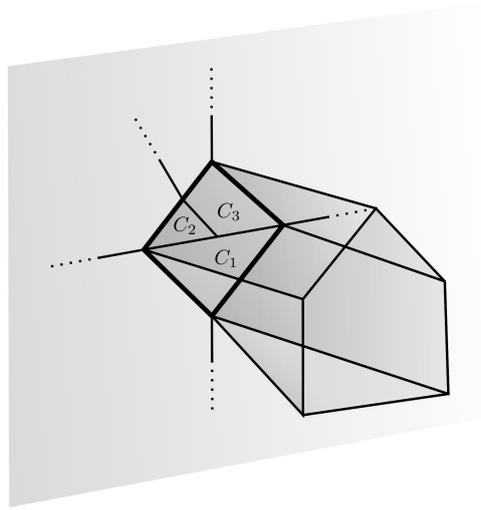}
\caption{\label{figtrans} This picture represents an example of a reduced cell $x_J$ made of four vertices and four edges (thicker lines) where $x$ is a 3-cell with 9 faces including $C_1, C_2$ and $C_3.$ The cell $x$ can be seen as the only 3-cell represented from the collar $K_J$ of a 3-cc $K$ where $J$ is a 2-cc made of all the cells whose vertices are drawn on the grey plane. For example $J \cap x$ contains the three 2-cells $C_1, C_2$ and $C_3$ and their neighbouring 2-cells in $J$ contain the "unending" edges represented using the dots $\cdot \cdot \cdot  \cdot \cdot.$}
\end{figure}

An example of a reduced cell is given in Figure \ref{figtrans}. We would like to also associate a rank function to $J(K),$ but it is not clear in general how to associate a rank to a reduced cell as it can be constructed from different cells in $K_J.$ Therefore we introduce the following notion of pure relative cc.

Let us first introduce some notations that will make the following formulations more convenient. If $A \subset J^{[0]},$ we can consider the set 
$$ K_J(A) := \{y \in K_J ~|~ A \subset y\}.$$
We also define $K^-_J(A)$ to be the set of minimal elements in $K_J(A).$

\begin{defi}[Pure relative cc and relative rank]\label{defpurerelcc}
A relative cc  $(K,J)$ is said to be \textit{pure} if $J \cap x$ is pure for all $x \in K_J$ and, dually, for all $x \in J,$ all cells in $K^-_J(x)$ have equal rank $r(x).$ It is clear that in this case $r(x) \geq 1$ for all $x \in J$ and we define $\rk_J^K(x) := r(x) - 1$ to be the \textit{relative rank of $x \in J$ in $K.$} 
\end{defi}
This definition in particular implies that $\rk_J^K(x) \geq \rk_K(x) = \rk_J(x)$ for all $x \in J.$

The two conditions of relative purity from Definition \ref{defpurerelcc} are always satisfied for example in the case of simplicial complexes, but there are instances of cc that do not satisfy either of the two conditions. An example of a relative cc $(K,J)$ that does not satisfy the first condition is given in Figure \ref{fignopure}, in which there exists $x \in K_J$ such that $J \cap x$ is not pure.   
Figure \ref{fignocopure} provides an example of relative cc $(K,J)$ that does not satisfy the second "dual condition" of pure relative cc, i.e. such that there exists $x \in K_J$ (the central vertex) such that $K_J^-[x]$ contains two elements with different ranks (the thicker edge and the wavy triangle). One can in fact relate the example of Figure \ref{fignocopure} to the example of the previous Figure \ref{fignopure} by using the duality map and defining a subdivision using the 2-cells $C_1,C_2,C_3.$ 

\begin{figure}
\centering
\includegraphics[scale=0.42]{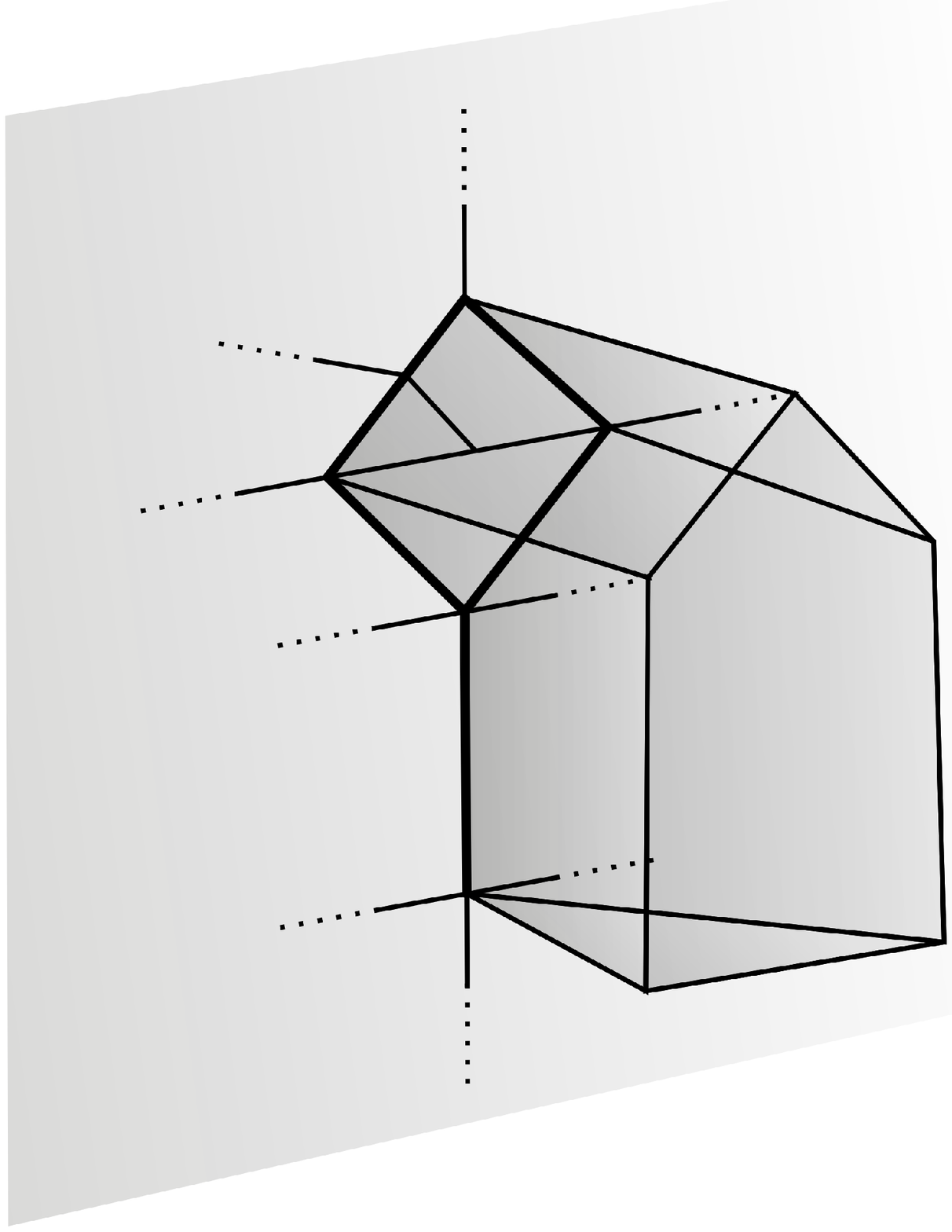}
\caption{\label{fignopure}  An example of a 3-cell $x \in K_J$ of a relative cc that does not satisfy the first condition for pure relative cc since $J \cap x,$ whose edges are represented with thicker lines, is not pure.}


\includegraphics[scale=0.47]{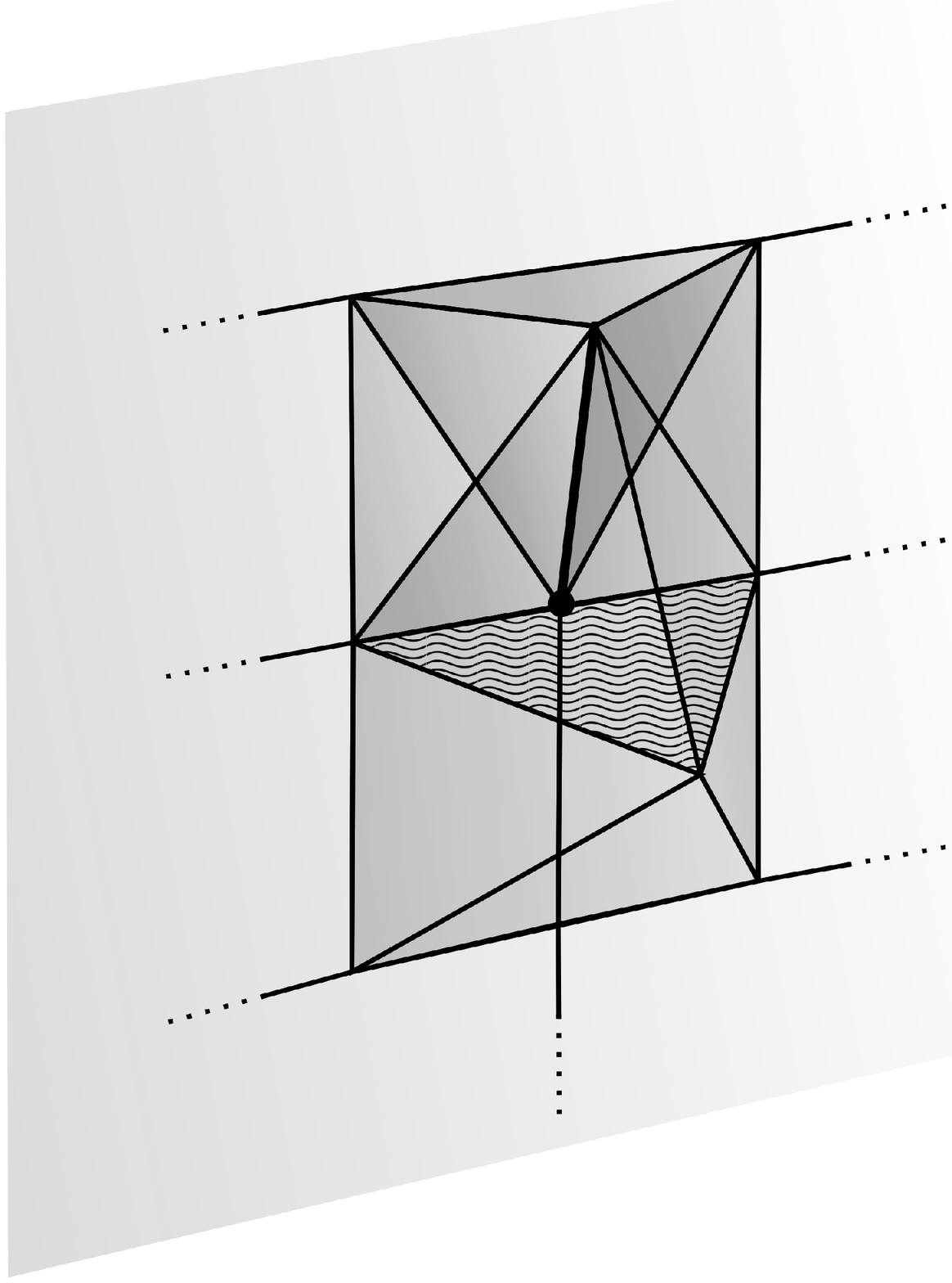}
\caption{\label{fignocopure}  An example of a cell $x \in J$ that does not satisfy the second condition of purity for the relative cc $(K,J),$ i.e. $x$ is the central vertex (black dot) and is both included in an edge (thicker line) and a 2-cell (wavy triangle) in $K_J.$ There are five 3-cells in this picture: four tetrahedra and one 3-cell below the wavy triangle having seven vertices.}
\end{figure}

Since we would like to manipulate reduced cells as the cells of the transition, we need to ensure that these contain "enough vertices". Indeed if $x \in K_J$ then $x_J \subset x \cap J^{[0]},$ but in general there can be vertices in $x \cap J^{[0]}$ not contained in an edge in $\E_J^x.$ 

We  therefore add an additional condition to our framework,  namely that the cells of a relative cc $(K,J)$ satisfy what we will refer to as the \textit{uniformity condition}:
\begin{equation} \label{conduniform}
J \cap x \subset J \cap y \quad \text{if and only if} \quad x_J \subset y_J, \quad \forall x,y \in K_J. \tag{\textbf{U}}
\end{equation}

We shall now see what the assumption of purity and condition (\ref{conduniform}) imply for the cells of $J(K).$ For this, we also introduce the following set associated to a cell $x \in J(K):$
$$K_J[x] = \{ y \in K_J ~|~ y_J = x\}$$
and denote by $K_J^-[x]$ the set of minimal elements in $K_J[x].$
The next lemma provides us with more information about the cells in $K_J^-[x]$ for $x \in J(K).$ In particular it implies that if $y \in K_J$ contains $x \in J(K)$ then there exists an element $x' \in K_J^-[x]$ such that $\rk_J^K(x) = \rk_K(x') - 1$ and $x \subset x' \subset y.$ This property is reminiscent of condition \ref{cond3col} from the definition of a collapse and it will indeed allow us to show that the map $\cphi_J^K$ we will define in Proposition \ref{proptransit} is a collapse.

\begin{lem}\label{lemunifrelcc}
Let $(K,J)$ be a pure relative cc such that (\ref{conduniform}) is satisfied and let $x \in J(K).$ Then for all $x' \in K_J^-[x]$ and all maximal cells $w \in J \cap x'$ we have
$$  \rk_J^K(x) := \rk_K(x') - 1 = \rk_K(w) .$$ 
Moreover,  for all $x' \in K_J^-[x]$ we have that $K_J(x) = K_J(x' \cap J^{[0]})$ and
$$K_J^-(x) = K_J^-[x] = K_J^-(x' \cap J^{[0]}),$$
which implies that for any given $y' \in K_J(x)$ there exists $x' \in K_J^-[x]$ such that $x' \subset y'.$
\begin{proof}
First, it is clear that $K_J[x] \subset K_J(x).$ An element $x' \in K_J^-[x]$ satisfies that $x \subset x',$ which implies the inclusion 
$$K_J(x) \supset K_J(x' \cap J^{[0]}).$$
The reverse inclusion is also valid by the following argument. Note that an element $y' \in K_J(x)$ satisfies
$$ x'_J = x \subset y'_J$$ 
and (\ref{conduniform}) applied to $(K,J)$ gives
$$ x' \cap J^{[0]} \subset y' \cap J^{[0]},$$
which in particular implies that $y \in K_J(x' \cap J^{[0]}).$ Therefore we indeed have $K_J(x) = K_J(x' \cap J^{[0]}).$

Moreover, a minimal element in $K_J[x]$ is also a minimal element in $K_J(x).$ This is a consequence of the fact that if $x' \in K_J[x]$ and  $y' \in K_J(x)$ satisfy $y' \subset x'$ then $$x \subset y'_J \subset x'_J = x$$
and therefore $y'_J = x$ i.e. $y \in K_J[x].$
As a consequence of the results above, the purity of $(K,J)$ implies that for all $x' \in K_J^-[x]$ we have
$$K_J^-[x] \subset K_J^-(x) = K_J^-(x' \cap J^{[0]}).$$

Therefore if $x' \in K_J^-[x]$ then $x' \in K_J^-(x' \cap J^{[0]})$ and since $J \cap x'$ is pure, Axiom \ref{cccenough} implies that a maximal cell $w$ of $J \cap x'$ satisfies
$$\rk_K(w) = \rk_K(x') - 1.$$ 

Suppose by contradiction that there exists $w' \in K_J^-(x) \setminus K_J^-[x]$ and let $x' \in K_J^-[x].$ Then $w' \cap x'$ cannot be an element of $K_J$ as it would in particular contradict the minimality of $x',$ hence $w' \cap x' \in J.$ But this implies that $J \cap w' = J \cap x'$ which by condition (\ref{conduniform}) implies $x = x'_J = w'_J.$ This contradicts our assumptions and concludes the proof.
\end{proof}
\end{lem}

As a consequence of Lemma \ref{lemunifrelcc}, if $(K,J)$ is a uniform relative cc then the map
$$ J(K) \ni x_J \longmapsto \Rk(J \cap x) $$
is well defined and we have $\Rk(J \cap x) = \rk_K(x) - 1$ whenever $x \in K_J^-[x_J].$ We can therefore associate the rank function 
$$ \rk_{J(K)}( x_J) := \Rk(J \cap x)$$
to $J(K).$

\begin{defi}[Uniform relative cc]
A relative cc $(K,J)$ is called \textit{uniform} if the following four conditions are fulfilled:
\begin{itemize}[topsep=2pt, parsep=2pt, itemsep=3pt]
\item $(K,J)$ non-degenerate,
\item $(K,J)$ is pure,
\item the uniformity condition (\ref{conduniform}) holds for $(K,J),$
\item $(J(K), \rk_{J(K)}) \in \mlc.$
\end{itemize}
\end{defi}

It is nevertheless not clear to us whether the uniformity condition (\ref{conduniform}) and the assumption that $J(K) \in \mlc$ are independent from the assumptions of purity and non-degeneracy. From the result of Proposition \ref{proppolytope} on polytopes, it is reasonable to believe that every non-degenerate pure relative cc $(K,J)$ where $K \in \nsc^R$ for $R \leq 4$ is uniform.

Let us introduce the corresponding notion for cobordisms.

\begin{defi}[Uniform cobordism]
A cobordism $(K - J) \in \cob,$ will be called \textit{uniform} if $(K, \partial K)$ uniform. 
\end{defi}

Note that in general, the dual of a uniform cobordism is not uniform. The next lemma provides an extension of the definition of \textit{reduced cell $x_J$} for cells  $x \in J.$ 

\begin{lem} \label{lemredcell}
If $(K,J)$ is uniform and $x \in J$ then there exists a unique cell $x_J \in J(K)$ such that
$$K_J^-[x_J] = K_J^-(x) \quad \text{and} \quad \rk_{J(K)}(x_J) = \rk_J^K(x).$$
$x_J$ will be called the reduced cell associated to $x.$
\begin{proof}
In order to prove the first statement, it is sufficient to show that if $y,y'$ are distinct elements in $K_J^-(x)$ then $y_J = y'_J.$ 
Since $y$ and $y'$ are minimal elements in $K_J(x)$ then $y \cap y'$ cannot be contained in $K_J,$ so $y \cap y' \in J.$ Hence $J \cap y = J \cap y'$ and condition (\ref{conduniform}) implies that $y_J = y'_J.$

Also, if $y \in K_J^-(x)$ then by definition $\rk_J^K(x) = \rk_K(y) - 1$ and by Lemma \ref{lemunifrelcc} we have $\rk_K(y) = \rk_K(w) + 1$ for all maximal cells in $J \cap y.$ The second statement then follows since $y \in K_J^-[x_J]$ by the first statement and therefore
$$\rk_{J(K)}(x_J) = \rk_{J(K)}(y_J) = \rk_K(w)  = \rk_J^K(x).$$
\end{proof} 
\end{lem}

\begin{prop}\label{proptransit}
Let $K \in \nsc^{R+1}$ and $J \leq \partial K,$ $J \in \mlc^R$ such that $(K,J)$ is uniform and local.
Then we have the two following results.
\begin{enumerate}[label=\arabic*) , topsep=2pt, parsep=2pt, itemsep=1pt]
\item \label{proptransitred} The poset homomorphism $\sphi_J^K: J \dans J(K)$ defined by
$$\sphi_J^K(x) = x_J$$
is a reduction.

\item \label{proptransitcol}The poset homomorphism $ \cphi_J^K: M_J^K \dans J(K)$ defined by $$\cphi_J^K(\E^x_J) = x_J$$ is a collapse.
\end{enumerate}

\begin{proof}
In order to prove the point \ref{proptransitred} we first prove that $\sphi_J^K: J \dans J(K)$ is a reduction. The map $\sphi_J^K$ is clearly a surjective poset homomorphism. Condition \ref{cond1red} follows from the fact that if a vertex $v \in J^{[0]}$ satisfies $v_J \in J(K)^{[0]}$ then we have $v_J = \{v\}.$

Condition \ref{cond3red} is obtained in the following way. Suppose $x \in J$ satisfies
$$ x_J \subset y \in J(K)^{[r]}.$$
For $y' \in K_J^-[y]$ we then have that $\Rk(J \cap y') = r$ and $x \in J \cap y'.$ By purity of $(K,J)$ there exists  an $r$-cell $x'$ in $J \cap y'$ such that $x \subset x'.$ If we name $w = x_J'$ then Lemma \ref{lemredcell} implies that
$$K_J^-[w] = K_J^-(x').$$
Since $x'$ is a maximal cell of $J \cap y'$ we have that $y' \in K_J^-(x')$ and by the previous relation this implies
$$ y = y'_J = w = x'_J,$$
which proves condition \ref{cond3red} for $\sphi_J^K.$
Condition \ref{cond2red} then follows from Lemma \ref{lemcond2redcol}.

We then turn to condition \ref{cond4red}. Let $x \in J$ such that $\rk_J(x) = \rk_{J(K)}(x_J) - 1$ and let $x' \in K_J^-[x_J].$ Then $\rk_K(x') = \rk_{J(K)}(x_J) + 1,$ hence $\rk_K(x') = \rk_K(x) - 2.$ Since $x \subset x',$ Axiom \ref{cccdiamond} for $K$ provides us with two cells $y,y'$ such that 
$$ \cface{x} \cap \face{x'} = \{y,y'\}.$$
By minimality of $x',$ we have that $y$ and $y'$ are in $J$ and 
$$x' \in K_J^-(y) \cap K_J^-(y') = K_J^-[y_J] \cap K_J^-[y'_J],$$
where the equality is due to Lemma \ref{lemredcell}. Also, if $w \in \cface{x}$ satisfies $w_J = x_J$ then $w \subset x'$ because $x' \in K_J^-[x_J].$ Since $\rk_J(w) = \rk_J(x) + 1 = \rk_K(x') - 1,$ we have $w \in \face{x'}$ and therefore $w \in \{y,y'\}.$ This completes the proof that
$$ \abs{ \cface{x} \cap (\sphi_J^K)^{-1}(x_J)} = 2.$$

Then condition \ref{cond5red} can be proven by taking $x \in J$ such that $\rk_J(x) = \rk_{J(K)}(x_J).$ This implies that $x$ is a maximal cell of $J \cap x'$ for all $x' \in K_J^-[x_J].$ If we take $y \in \cface{x_J},$  $y' \in K_J^-[y]$ and $x' \in K_J^-[x_J]$ such that $x' \subset y'$ then we obtain $$\rk_K(y') = \rk_{J(K)}(y) + 1 = \rk_K(x) + 2.$$ 
Therefore there exists $x'' \in \cface{x}$ such that
$$ \cface{x} \cap \face{y'} = \{x',x''\}.$$
If $x'' \in K_J$ then $x' \cap J = x'' \cap J = x$ which implies $y' \cap J = x$ but this contradicts the fact that $\Rk(J \cap y') = \rk_J(x) + 1.$
Therefore $x'' \in J,$ hence  $x''_J = y'_J = y$ and we have $$\cface{x} \cap (\sphi_J^K)^{-1}(y) = \{x''\}.$$ 

Next we prove that $\cphi_J^K : M_J^K \dans J(K)$ is a collapse. Condition \ref{cond1col} is clear since $\abs{\cface{z}} = 1$ for all $z \in J^{[R-1]}$ and therefore $K_J^-[z] = K_J^-(z)$ contains a unique cell $z'$ which is a maximal cell of $K,$ in other words $$(\cphi_J^K)^{-1}(z) = \{ \E_J^{z'}\}.$$

Condition \ref{cond3col} is a consequence of Lemma \ref{lemunifrelcc}, since if $x \in J(K)$ and $\E_J^{y'} \in M_J^K$ are such that $x \subset y_J'$ then there exists $x' \in K_J^-[x]$ such that $x' \subset y'.$ Hence we have  $\rk_{M_J^K}(\E_J^{x'}) = \rk_{J(K)}(x)$ and Lemma \ref{lemmidsec} implies that $\E_J^{x'} \subset \E_J^{y'}.$ Condition \ref{cond2col} is again a consequence of Lemma \ref{lemcond2redcol}.

In order to prove condition \ref{cond4col}, we take $\E_J^{x'} \in M_J^K$ such that
$$\rk_{M_J^K}(\E_J^{x'}) = \rk_{J(K)}(x'_J) + 1 = \rk_K(x') - 1.$$
If $w \in J \cap x'$ is maximal cell then Axiom \ref{cccdiamond} for $K$ gives
$$ \cface{w} \cap \face{x'} = \{y, y'\},$$
for $y,y' \in K.$ The maximality of $w$ then implies that $\{y,y'\} \subset K_J.$ 
We therefore have that 
$$w \subset y \cap y' \in J,$$ 
which by Axiom \ref{cccrank} implies that $w = y \cap y'.$  Therefore $y_J = y'_J = x_J$ and we indeed have
$$\abs{ \face{\E_J^{x'}} \cap (\cphi^{-1})(x)} = 2.$$

Finally, we can prove condition \ref{cond5col} in the following way. Suppose that $\E_J^{x'} \in M_J^K$ satisfies $ \rk_{M_J^K}(\E_J^{x'}) = \rk_{J(K)}(x'_J).$ If $w$ is a maximal cell in $J \cap x'$ we have that $\rk_K(w) = \rk_K(x') - 1,$ i.e. $w \in \face{x'}.$ By Lemma \ref{lemredcell} we also have that
$$x' \in K_J^-(w) = K_J^-[w_J].$$
If we fix an element $y \in \face{x}$ where $x := x'_J = w_J$ then by Lemma \ref{lemunifrelcc} there exists $y' \in K_J^-[y]$ such that $y' \subset x'$ and
$$\rk_K(y') = \rk_{J(K)}(y) + 1 = \rk_{J(K)}(x)  = \rk_K(x') - 1.$$
This implies that $y' \in \face{x'},$ i.e. $\E_J^{y'} \in \face{\E_J^{x'}}$  by Lemma \ref{lemmidsec}. Moreover if $\E_J^{w'} \in \face{\E_J^{x'}} \setminus \{\E_J^{y'}\}$ satisfies $w'_J = y$ then by minimality of $y'$ we have $y' \cap w' \in J.$ Hence it must be that $J \cap (y' \cap w')$ has a unique maximal cell $u$ and
$$  \{y' , w'\} \subset \cface{u} \cap \face{x'} .$$
This is a contradiction with Axiom \ref{cccdiamond} for $K$ since $J \cap x'$ is a pure cc such that
$$ \Rk(J \cap x') = \rk_K(x') - 1 = \rk_K(y') = \rk_K(u) + 1$$ and therefore there exists a maximal cell $w$ of $J \cap x'$ such that $w \in \cface{u} \cap \cface{x'}\setminus \{y',w'\}.$
Thus we can indeed conclude that
$$ \face{\E_J^{x'}} \cap (\cphi_J^K)^{-1}(y) = \{\E_J^{y'}\}.$$

\end{proof}
\end{prop}

\subsection{Compatible, reflective and orthogonal homomorphisms} \label{sseccompnorthhom}

For this section and what follows, we introduce the notation $x \vee y$ to designate $\bigvee \{x,y\},$ the \textit{least upper bound (or supremum) of ${x,y}$}, i.e. the smallest cell containing both $x$ and $y$ in a cc $K$ in which an upper bound of $x$ and $y$ exists. If no cell contains both $x$ and $y$ we simply write $x \vee y = \all.$

In addition to being a reduction and a collapse, we will see in this section that the reduction and collapse $\sphi_J^K$ and $\cphi_J^K$ introduced in Proposition \ref{proptransit} of Section \ref{ssectrans} have a property that will be described as "$\sphi_J^K$ is compatible with $\cphi_J^K$" and constitutes the first definition of this section. Compatibility is therefore in particular a condition on two poset homomorphisms having the same image and this section also introduces the notions of reflectivity and orthogonality, which involve a condition on two homomorphisms having the same domain. These conditions will later on be key in defining specific sequences of reductions and collapses between elements in $\mlc^R$ that will be used to construct cells of dimension $R+1.$ The notion of "augmented poset" introduced in the later part of this section will also be used for these constructions, together with a number of technical results proven here.

\begin{defi}[Compatible homomorphisms, $\cpa$]  \label{defcoorthom}
Let $j : J \dans I$  and $l : L \dans I$ be cc-homomorphisms. Then we say that $j$ is \textit{compatible with} $l$ if the three following conditions are satisfied:
\begin{enumerate}[label=\arabic*) , topsep=2pt, parsep=2pt, itemsep=1pt]
\item \label{cpa} For all $w \in I$ either $\abs{j^{-1}(w)} = 1$ or $\abs{l^{-1}(w)}=1.$
\item \label{cpau} If $x,x' \in J$ are such that $x \vee x' \neq \all,$ then $$j(x \vee x') = j(x) \vee j(x').$$
\item \label{cpai} If $y,y' \in L$ are such that $y \cap y' \neq \emptyset$ then $$l(y \cap y') = l(y) \cap l(y').$$
\end{enumerate}
We write $j \cpa l$ to denote that $j$ is compatible with $l.$ Note that if $j$ and $l$ are poset isomorphisms then $j$ is compatible with $l.$
\end{defi}

The next corollary provides typical examples of compatible homomorphisms.

\begin{cor} \label{cortransitmidsec}
Let $K \in \nsc^R$ and $J \leq \partial K,$ $J \in \mlc^{R-1}$ such that $(K,J)$ is uniform and local then
$$J \red_{\sphi_J^K} J(K) \loc_{\cphi_J^K} M_J^K \quad \text{and} \quad \sphi_J^K \cpa \cphi_J^K.$$
\begin{proof}
The first statement is a direct consequence of Proposition \ref{proptransit}.
As for the first condition of compatibility, suppose by contradiction that there exists $x \in J(K)$ such that there are two distinct elements $x_1,x_2 \in J$ satisfying $(x_1)_J = (x_2)_J = x,$ as well as two distinct elements $\E_J^{y_1},\E_J^{y_2} \in M_J^K$ such that $(y_1)_J = (y_2)_J = x.$ By the conditions \ref{cond3red}, \ref{cond3col}, \ref{cond4red} and \ref{cond4col} we can assume that 
$$ r:= \rk_{J(K)}(x) = \rk_J(x_i)  = \rk_{M_J^K}(\E_J^{y_i})  = \rk_K(y_i)-1,$$
for $i=1,2$ where the two last equalities come from the definition of the midsection $M_J^K.$ This implies that $y_1,y_2 \in K_J^-[x]$ and therefore $y_1 \cap y_2 \in J^{[r]}.$ But we have $\{x_1, x_2\} \subset J \cap (y_1 \cap y_2)$ so this contradicts the intersection axiom \ref{cccinter} of $K.$

In order to prove the condition \ref{cpau} of compatibility, take $x,x' \in J$ such that $x \vee x' \neq \all.$ We need to show that 
$$\sphi_J^K(x \vee x') = (x \vee x')_J = x_J \vee x'_J = \sphi_J^K(x) \vee \sphi_J^K(x'),$$ where the map $J \ni x \mapsto x_J \in J(K)$ is defined in Lemma \ref{lemredcell}, which also states that $K_J^-[x_J] = K_J^-(x)$ and $\rk(y)-1=\rk(x)$ for all $y \in K_J^-(x).$ 

The inclusion $x_J \vee x'_J \subset (x \vee x')_J$ is clear (and in particular $x_J \vee x'_J \neq \all$) since if $y \in K_J^-(x \vee x')$ then in particular $x\subset y$ and $x' \subset y$ hence the definition of reduced cells easily leads $x_J \subset y_J$ and $x'_J \subset y_J$ and we have $$x_J \vee x'_J \subset y_J = (x \vee x')_J.$$ 

Now consider $w \in K_J^-[x_J \vee x'_J],$ hence we have that $x_J \subset w_J$ and $x'_J \subset w_J.$ The uniformity condition \ref{conduniform} implies that  $J \cap y' \subset J \cap w$ for all $y' \in K_J^-(x) \cup K_J^-(x')$ which in particular implies that $x \subset w$ and $x' \subset w.$  Hence we have $x \vee x' \subset w$ and this implies that $y_J = (x \vee x')_J \subset w_J$ which leads the desired equality.

The condition \ref{cpai} of compatibility is a direct consequence of Lemma \ref{lemmidsec} since if $\E_J^y,\E_J^{y'} \in M_J^K$ are such that $\E_J^y \cap \E_J^{y'} \neq \emptyset,$ then
\begin{align*}
\cphi_J^K(\E_J^y \cap \E_J^{y'}) &= \cphi_J^K(\E_J^{y \cap y'}) = (y \cap y')_J\\
&= \{ e \cap J^{[0]} ~|~ e \in \E_J^{y \cap y'} \} = \{e \cap J^{[0]} ~|~ e \in \E_J^y \cap \E_J^{y'}\}\\
&= \{e \cap J^{[0]} ~|~ e \in \E_J^y \} \cap \{e \cap J^{[0]} ~|~  e \in \E_J^{y'}\} \\
&= x_J \cap y_J = \cphi_J^K(\E_J^y) \cap \cphi_J^K(\E_J^{y'}).
\end{align*} 

\end{proof}
\end{cor}

Next we  define the notions of reflective and orthogonal cc-homomorphisms. We will see later that midsections of slices provide examples of orthogonal collapses, as a consequence of the Correspondence Theorem \ref{thmcorresp}. Reflective reductions will typically be used in the context of unions of cell complexes in $\nsc$ as introduced in Definition \ref{defunioncc}.

\begin{defi}[Reflective ($\refl$) and Orthogonal ($\perp$) homomorphisms] \label{defrefnorth}
Let $j : I \dans J$ and $l : I \dans L$ be two cc-homomorphisms.
\begin{enumerate}[label=\arabic*) , topsep=2pt, parsep=2pt, itemsep=1pt]
\item \label{condoref} We say that $j$ and $l$ are \textit{reflective}, noted $j \refl l,$ if whenever two elements $w,w' \in I$ satisfy that both $j(w) \vee j(w') \neq \all$ and $l(w) \vee l(w') \neq \all$ then we have that $w \vee w' \neq \all,$ as well as the two equalities  
$$ j(w \vee w') = j(w) \vee j(w') \quad \text{ and } \quad l(w \vee w') = l(w) \vee l(w').$$
\item \label{condorth} Dually, we say that $j$ and $l$ are \textit{orthogonal}, noted $j \perp l,$ if whenever two elements $w,w' \in I$ satisfy that both $j(w) \cap j(w') \neq \emptyset$ and $l(w) \cap l(w') \neq \emptyset$ then we have that $w \cap w' \neq \emptyset,$  as well as the two equalities $$ j(w \cap w') = j(w) \cap j(w') \quad \text{ and } \quad l(w \cap w') = l(w) \cap l(w').$$
\end{enumerate}
Note that if $j$ and $l$ are poset isomorphisms then $j$ and $l$ are reflective and orthogonal.
\end{defi}

In the next lemma we introduce the notion of augmented poset with a corresponding rank function. These will be key tools in the proof of the Correspondence Theorem \ref{thmcorresp} in the next section.

\begin{lem} \label{lemsqcup}
Let $M,J,L$ be mutually disjoint cc and let $\phi_J : M \dans J$ and $\phi_L:M \dans L$ be two orthogonal cc-homomorphisms. Then the map from $M$ to $\pws{J^{[0]} \sqcup L^{[0]}}$ defined by 
\begin{equation}\label{equlemcorresp}
x \mapsto \phi_J(x) \sqcup \phi_L(x) 
\end{equation}
is an injective poset homorphism such that the inverse is a poset homomorphism. The poset
$$ (\phi_J \sqcup \phi_L)(M) := \{ \phi_J(m) \sqcup \phi_L(m) ~|~ m \in M\}$$
will be called the \textit{augmented poset associated to $\phi_J$ and $\phi_L$} and will be assigned the rank function
$$\rk_{(\phi_J \sqcup \phi_L)(M)}(x \sqcup y) = \rk_M(m(x,y)) + 1,$$
where $m(x,y)$ is the unique element in $M$ such that 
$$\phi_J(m(x,y)) \sqcup \phi_L(m(x,y)) = x \sqcup y.$$ 
\begin{proof}
It is sufficient to show that for all $x,x' \in M,$ $$\phi_J(x) \sqcup \phi_L(x) \subset \phi_J(x') \sqcup \phi_L(x') \quad \Rightarrow \quad x \subset x'.$$
Suppose by contradiction that there exists $x,x' \in M$ such that $\phi_J(x) \sqcup \phi_L(x) \subset \phi_J(x') \sqcup \phi_L(x')$ and there exists $v \in x \setminus x'.$ Hence we get $\phi_I(v)\subset \phi_I(x)\subset \sphi_I(x'), ~I=J,L$ and this implies that $\phi_J(v) \cap \phi_J(x') \neq \emptyset$ and $\phi_L(v) \cap \phi_L(x') \neq \emptyset.$ The orthogonality condition \ref{condorth} implies that $v \in x'$ which is a contradiction.
\end{proof}
\end{lem}

The next lemma indicates how compatibility, reflectivity and orthogonality of cc-homomorphisms behave under duality, and is a relatively direct consequence of the definitions.

\begin{lem} \label{lemdualcompnorth}
Let $I,J,L \in \mlc^R.$ Then the two following statements hold.
\begin{enumerate}[label=\Alph*) , topsep=2pt, parsep=2pt, itemsep=1pt]
\item \label{lemdualcomp} If $j: J \dans I$ and $\l : L \dans I$ are cc-homomorphisms then
$$ j \cpa l \quad \text{if and only if} \quad \dual{j} \apc \dual{l}.$$
\item \label{lemdualorth} If $j: I \dans J$ and $l : I \dans L$ are cc-homomorphisms then
$$ j \perp l \quad \text{if and only if} \quad \dual{j} \refl \dual{l}.$$
\end{enumerate}

\begin{proof}
For the claim \ref{lemdualcomp}, the first condition of the compatibility of $\dual{l}$ with $\dual{j}$ is obtained from the fact that if $A$ is subset of cells in a closed cc then $\abs{\dual{A}} = \abs{A}$ by Lemma \ref{claim2propdual}. As for the other conditions for compatibility and the claim \ref{lemdualorth}, they follow from the fact that if $x,y$ are cells in a closed cc the $\dual{x}\cap \dual{y} = \dual{x \vee y}$ and equivalently $\dual{x}\vee \dual{y} = \dual{x \cap y}.$ More precisely if $ x \vee y = \all$ then no maximal contains both $x$ and $y$ hence $\dual{x} \cap \dual{y} = \emptyset,$ and if otherwise $\dual{x} \cap \dual{y}$ is the set of maximal cells containing both $x$ and $y,$ i.e. containing $x \vee y.$
\end{proof}
\end{lem}

\subsection{Slices and slice sequences} \label{ssecslice}

In this section we prove an number of important results related to the notion of slice, a uniform cobordism having in particular the property of having all vertices on its boundary. Proposition \ref{propslice} shows that a slice has exactly two boundary components and Proposition \ref{propredbdry} proves that one can define the subdivision of a boundary component of an element in $\nsc$ using a relative reduction. Using these two propositions we prove the main result of this chapter: the Correspondence Theorem \ref{thmcorresp}. This theorem establishes a bijection between slices (up to cc-isomorphism) and sequences of the form 
$$J \red J' \loc M \col L' \der L$$
where $J,J'M,L',L \in \mlc,$ introduced in Definition \ref{defsliceseq} as "slice sequences".

The notion of slice we introduce next is to be interpreted as a fundamental building bloc of what will become our notion of causal cobordism. 

\begin{defi}[Slice]\label{defslice}
A cobordism $(S - J)$ is a \textit{slice} if it is uniform and  
\begin{equation}\label{condslice}
S^{[0]} = (\partial S)^{[0]}, \quad J \notin \{ \partial S, \emptyset\}.
\end{equation}
We will denote the set of slices by $\cob_s.$
\end{defi}

An example of a sub-cc of a slice of rank 3 composed of three maximal cells was given in Figure \ref{slicepic} from Section \ref{susecmidsection} on midsections. We give a similar illustration of a slice in Figure \ref{figslice}. 

In particular every maximal cell of a slice $S$ in the previous definition has at least one vertex in $J$ and one vertex in $\partial S \setminus J$ by the non-degeneracy of $(S,\partial S).$ 

The following result establishes a number of important properties of slices.
Note that this result does not use the assumption that $(S - J)$ is uniform, except for the fact that $(S, \partial S)$ is non-degenerate.

\begin{prop} \label{propslice}
Let $(S - J)$ be a slice and let $L := \partial S \setminus J.$ Then $M_L^S = M_J^S$ and $(S - L)$ is a slice. Moreover $J$ and $L$ are connected, i.e. $\partial S$ has only two connected components.
\begin{proof}
Proving that $M_L^S = M_J^S$ amounts to prove that $S_L = S_J.$ The latter is clear from the property (\ref{condslice}), which implies that $\emptyset \neq x \cap J^{[0]} \neq x $ if and ony if $ \emptyset \neq x \cap L^{[0]} \neq x,$ for all $x \in S.$

We next show that $L$ is connected.  Suppose by contradiction that $L = L_1 \sqcup L_2$ where $L_i \in \mlc$ for $i= 1,2$ and let $R =\Rk(L) = \Rk(J).$ The case $R = 0$ is clear since the connectedness of $S$ implies that there is a path from $L_1$ to $L_2$ going through $J,$ which means that there is a vertex in $J$ included in two edges and this is a contradiction with the definition of a boundary cell, i.e. that a sub-maximal cell of the boundary - here a 0-cell - is included in only one maximal cell. The strategy of the proof for the case $R \geq 3$ is to show that our hypothesis implies that there exists an edge of $S$ having one vertex in $L_1$ and the other vertex in $L_2,$ which contradicts the non-degeneracy of $(S,L).$

By connectedness of $S,$ there exists a path $p : v_1 \pathto v_2$ in $S$ from $v_1 \in L_1^{[0]}$ to $v_2 \in L_2^{[0]}$ and we may suppose that $p$ is simple and $\abs{p^{[0]} \cap L_i^{[0]}} = 1$ for $i= 1,2,$ that is to say, $p$ intersects $L$ only on two vertices. Let $\{v_1,w\}$ be the first edge of $p.$ If $w \in L_2$ we are done, so we suppose that $w \in J^{[0]}.$ By our assumptions, $p$ enters a connected component $J_0$ of $J$ from the vertex $v$ and leaves $J_0$ through an edge of the form $\{w',v_2\},$ where $w \in J_0^{[0]}.$ By locality of $(S,J)$ the midsection $M := M_{J_0}^S$ is connected, hence there exists a path $p_M :\{w,v_1\} \pathto \{w',v_2\}$ in $M.$ Each edge of the path $p_M$ corresponds to a 2-cell in $S_J$ which by locality intersects $J$ in a simple path $p_0,$ or a single vertex $v,$ in which case we set $p_0 = \emptyset_v$ (cf. definitions in Section \ref{ssecgraphncc}). So using each such 2-cell $C$ in $p_M$ we can successively deform $p$ using a path move of the form $m_C^{p_1}$ where $p_1 \subset p$ is the simple path $p_1 = \{v_1,w\} \ast p_0.$ At each iteration the resulting path might cross $L$ in more than two vertices. If it does so by also crossing through $L_2$ this would imply that there is an edge of $S$ crossing from $L_1$ to $L_2$ which would produce the desired contradiction. We can therefore assume that deforming $p$ can only increase the number of vertices it crosses in $L_1,$ in which case we can simply re-define $v_1$ to be the last vertex the path crosses in $L_1.$ In this way we can successively deform and shorten the path $p$ until the last step, corresponding to the last edge in $p_M$ which again produces an edge of $S$ crossing from $L_1$ to $L_2,$ which concludes our argument.

In order to show that $(S-L) \in \cob$ we note that since $(S,L)$ is non-degenerate and $M_L^S = M_J^S$ it only remains to show that $x \cap L$ is connected for all $x \in S_L.$ Assume on the contrary that $L \cap x = X_1 \sqcup X_2$ for some $x \in S_L,$ where $X_i \leq L$ for $i= 1,2$ and $X_1$ and $X_2$ are not connected. Then since $S$ is cell-connected there exists a path $p : v_1 \pathto v_2$ in $S \cap x$ from $v_1 \in X_1$ to $v_2 \in X_2.$ By locality of $(S,J)$ we have that $\E_L^x = \E_J^x$ is a connected cell of $M_L^S = M_J^S.$ Therefore the previous argument again implies that there must exist an edge of the form $\{v_1, v_2\}$ where $v_i \in X_i, ~ i = 1,2,$ which is a contradiction.
As a consequence $(S-L) \in \cob$ is clearly also a slice. 

By exchanging the roles of $J$ and $L$ we have that $J$ is also connected, which concludes our proof.
\end{proof}
\end{prop}



As a consequence of Proposition \ref{propslice}, a cobordism $(S - L)$ is a slice if and only if $S \in \nsc,$ $S^{[0]} = (\partial S)^{[0]},$ $(S, \partial S)$ is uniform and local and $\partial S$ has exactly two connected components, one of which is $L.$ We will therefore \textit{also call such cc $S$ a slice} and denote \textit{the set of slices} by $\nsc_s.$ We will also denote the \textit{midsection associated to a slice} $S$ by 
$$M^S := M_J^S = M_L^S.$$  

The remainder of this section focuses on a how to characterize a slice using the notion of a slice sequence. The correspondence between these two objects will also allow us to provide  sufficient conditions for the composition of combordisms, as done in the next section, and will be at the basis for the construction of the category of causal cobordisms given in Section \ref{sseccatcausalcob}.

\begin{defi}[Slice sequence] \label{defsliceseq}
A sequence of local mutually disjoint cell complexes $M,J,J',L,L'$ in $\mlc$ and maps
$$J \red_{\sphi_J} J' \loc_{\cphi_J} M \col_{\cphi_L} L' \der_{\sphi_L} L$$
will be called a \textit{slice sequence} if $\sphi_J, \cphi_J, \sphi_L, \cphi_L$ satisfy 
$$ \sphi_J \cpa \cphi_J, \quad \sphi_L \cpa \cphi_L, \quad  \cphi_J \perp \cphi_L.$$ 
\end{defi}

The next results are included as a preparation for the Correspondence Theorem \ref{thmcorresp}.

\begin{lem}\label{lemredbdry}
Let $K \in \nsc$ and let $J'$ be a connected component of $\partial K.$ If $\sphi: J \dans J'$ is a reduction where $J$ is disjoint from $K$ then the map from $K_{J'}$ to $\pws{\left(\kz \setminus(J')^{[0]}\right) \sqcup J^{[0]}}$ defined by
\begin{equation} \label{maplemredbdry}
 y' \longmapsto \left( y' \setminus (J')^{[0]} \right) \sqcup \left( \bigcup_{x \in J, ~ \sphi(x) \subset y'} x \right)
\end{equation}
is an injective poset homomorphism such that the inverse is also a poset homomorphism.
\begin{proof}
Let us denote by $\phi$ the map (\ref{maplemredbdry}), which clearly is a poset homomorphism since so is $\sphi.$

The injectivity of $\phi$ is a direct consequence of condition \ref{cond1red} and the observation that any vertex $v \in \sphi^{-1}( (J')^{[0]} \cap y')$ satisfies that $\phi(v) \subset y'$ hence $v$ is included in $\left( \bigcup_{x \in J, ~ \sphi(x) \subset y'} x \right).$

From the last observation one also get that the inverse of $\phi$ is a poset homomorphism since if $x',y' \in K_{J'}$ satisfy that $\phi(x') \subset \phi(y')$ then it implies that $x' \setminus (J')^{[0]} \subset y' \setminus (J')^{[0]}$ and
$$ \sphi^{-1}( (J')^{[0]} \cap x') \subset \sphi^{-1}( (J')^{[0]} \cap y').$$
Therefore \ref{cond1red} implies that $ (J')^{[0]} \cap x' \subset (J')^{[0]} \cap y'$ and it follows that $x' \subset y'.$
\end{proof}
\end{lem}

\begin{prop} \label{propredbdry}
Let $K \in \nsc$ and let $J'$ be a connected component of $\partial K$ such that $(K , J')$ is uniform. If $\sphi : J \dans J'$ is a reduction such that $\sphi \circ \sphi_{J'}^K : J \dans J'(K)$ is compatible with $\cphi_{J'}^K : M_{J'}^K \dans J'(K)$ and $J$ is disjoint from $K$ then the cc $K^\sphi$ (introduced in Definition \ref{defikphi}) is an element in $\nsc$ and $J$ is a connected component of $\partial K^\sphi,$ in other words 
$$( K^\sphi, J) \red_\sphi (K, J')$$
defines a relative reduction. Moreover, we have that $\sphi_J^{K^\sphi} = \sphi \circ \sphi_{J'}^K$ and $(K^\sphi, J)$ is uniform.
\begin{proof}
Since Proposition \ref{propprecpartialorder} implies that $\sphi \circ \sphi_{J'}^K$ is also a reduction we can without loss of generality assume that $\sphi_{J'}^K = \id_{J'(K)},$ i.e. $J' = J'(K)$ and $\sphi = \sphi_{J'}^K \circ \sphi.$

Define $(J')^c := K \cap ( \kz \setminus (J')^{[0]}),$ i.e. the sub-cc of $K$ containing the cells in $K$ not intersecting $J'.$ Let $\phi$ be the map defined by (\ref{maplemredbdry}) which by Lemma \ref{lemredbdry} defines a poset isomorphism 
$$ \phi : K_{J'} \dans \phi(K_{J'}).$$
We define $(K^\sphi, \rk_{K^\sphi})$ by
\begin{align*}
K^\sphi &:= (J')^c \sqcup J \sqcup \phi(K_{J'}) \\
\rk_{K^\sphi} &:= \rk_K|_{(J')^c} + \rk_J + \rk_K|_{K_{J'}} \circ \phi^{-1}.
\end{align*}
It then follows that $K_J^\sphi = \phi(K_{J'}).$ Also, since $J' = \cphi_{J'}^K(M_{J'}^K)$ as $\sphi_{J'}^K = \id_{J'}$ we have that
\begin{equation}\label{condapropredbdry}
\text{for all } y' \in K_{J'}, (y')_{J'} = y' \cap ({J'})^{[0]}.
\end{equation}

Let's now check the axioms of cc for $K^\sphi.$ To make our notations more explicit, for a cell $x$ in a cc $L$ we will denote by $\kcface{x}{L}$ the set of co-faces of $x$ in $L$ and similarly for faces. Since these axioms are satisfied for $K$ they are automatically satisfied when all cells involved are in $J$ or $(J')^c \sqcup \phi(K_J)$ by the poset isomorphism property of $\phi.$ This allows us to only check the axioms in the case of inclusions of cells between the sets $J$ and $(J')^c \sqcup \phi(K_J).$

To check Axiom \ref{cccrank} we only need to consider the case where $x \in J,$ $y \in K^\sphi_{J}$ and $x \subsetneq y.$ In this case we have that $\sphi(x) \subsetneq \phi^{-1}(y)$ and therefore
$$ \rk_{K^\sphi}(x) \leq \rk_{K}(\sphi(x)) < \rk_K( \phi^{-1}(y) ) = \rk_{K^\sphi}(y).$$

To check Axiom \ref{cccinter} for the non-direct cases amounts to show that $x \cap y \in K^\sphi$ for $x \in K_J^\sphi \sqcup J$ and $y \in K_J^\sphi.$ If $x \in K_J^\sphi$ and $x',y' \in K_{J'}$ are such that $\phi(x') = x $ and $\phi(y') = y$ then we have two possible cases: either $\E_{J'}^{x'} \cap \E_{J'}^{y'} \neq \emptyset$ or $\E_{J'}^{x'} \cap \E_{J'}^{y'} = \emptyset.$ In the first case one has that $x' \cap y' \in K_{J'}$ and $$\phi(x' \cap y') = \phi(x') \cap \phi(y') = x \cap y \in K^\sphi_J.$$
For the case $\E_{J'}^{x'} \cap \E_{J'}^{y'} = \emptyset,$ then by property (\ref{condapropredbdry}), we have that $$x' \cap y' = (x')_{J'} \cap (y')_{J'} \in J'(K) = J'.$$
By Proposition \ref{proptransit} we know that $\cphi_{J'}^K$ is a collapse, hence by \ref{cond3col} we can pick $x'_0 \in B(x')^{[r+1]}$ and $y'_0 \in B(y')^{[r+1]}$ where $r = rk(x' \cap y') = \rk_{M_{J'}^K}(\E_{J'}^{x'_0}) = \rk_{M_{J'}^K}(\E_{J'}^{y'_0}).$ By Lemma \ref{lemunifrelcc} this implies that $x'_0 ,y'_0 \in K^-_{J'}[x' \cap y'],$ and since $\E_{J'}^{x'} \cap \E_{J'}^{y'} = \emptyset,$ we also have that $\E_{J'}^{x_0'} \cap \E_{J'}^{y_0'} = \emptyset.$ The latter implies that $\abs{(\cphi_{J'}^K)^{-1}(x'\cap y')}>1$ which by condition \ref{cpa} of compatibility implies that there exists $x \in J$ such that $\sphi^{-1}(x' \cap y')=\{w_0\}.$ Therefore we obtain $$\phi(x') \cap \phi(y') = \bigcup_{w\in \sphi^{-1}(B(x' \cap y'))} w = w_0 \in J.$$
For the case where $x \in J$ and if $y' = \phi^{-1}(y) \in K_{J'}$ then one has to show that $$x \cap \left(  \bigcup_{w \in J ~:~ \sphi(w) \subset y'} w \right) \in J.$$ Let $w_0 := (y')_{J'},$ then $\{w \in J ~|~ \sphi(w) \subset y'\} = \sphi^{-1}(w_0),$ hence we need to show that there exists a unique maximal cell in $B(x) \cap \sphi^{-1}(w_0).$ Suppose by contradiction that $w_1,w_2$ are two such maximal cells. Since we have $w_i \subset x$ for $i=1,2,$ we have that $w_1 \vee w_2 \neq \all.$ Condition \ref{cpau} of compatibility then implies that $$\sphi_{J'}^K(w_1 \vee w_2) = \sphi(w_1) \vee \sphi(w_2) \subset w_0,$$ where the inclusion is due to the fact that $\sphi(w_i) \subset w_0$ for $i=1,2.$ But this implies that $w_1 \vee w_2 \in \sphi^{-1}(B(w_0))$ and contradicts the maximality of $w_1$ and $w_2.$ This finishes the proof of Axiom \ref{cccinter} for $K^\sphi.$

Axiom \ref{cccenough} can be inferred from the case $J \ni x \subsetneq y \in K^\sphi_J$ which is proven as follows. By the poset isomorphism property of $\phi$ we have that $\sphi(x) \subset \phi^{-1}(y) \in K_{J'}.$ Therefore by Axiom \ref{cccenough} for $K$ we can suppose without loss of generality that $\rk_K(y) = \rk_K(\sphi(x)) +1.$ If $\rk_J(x) = \rk_K(\sphi(x))$ then we have that $y \in \kcface{x}{K^\sphi}$ as desired. If $\rk_J(x) < \rk_K(\sphi(x))$ then, as noted in Remark \ref{remcond3redncol}, there exists $w \in \sphi^{-1}(\sphi(x))$ such that $\rk_J(w) = \rk_K(\sphi(x)).$ Therefore using \ref{cccenough} for $J$ we can find a co-face of $x$ contained in $w$ which is then also contained in $y.$

For Axiom \ref{cccdiamond} it is also sufficient to focus on the case where $x \in J, y \in K_J^\sphi,$ $x \subsetneq y$ and 
$$\rk_{K^\sphi}(x) = \rk_{K^\sphi}(y) - 2.$$
Since $\rk_{K^\sphi}(y) = \rk_K(\phi^{-1}(y)),$ this leave two possibilities: either we have $\rk_K(\sphi(x))= \rk_{K^\sphi}(x)$ or $\rk_K(\sphi(x)) = \rk_{K^\sphi}(x) +1.$ In the first case where $\rk_K(\sphi(x)) = \rk_{K^\sphi}(x),$ Axiom \ref{cccdiamond} for $K$ implies that 
$$\abs{ \kcface{\sphi(x)}{K} \cap \kface{\phi^{-1}(y)}{K}} = 2.$$
To conclude this first case, it remains to show that there is a bijection between $\kcface{\sphi(x)}{K} \cap \kface{\phi^{-1}(y)}{K}$ and $\kcface{x}{K^\sphi} \cap \kface{y}{K^\sphi}.$
If $w' \in \kcface{\sphi(x)}{K} \cap \kface{\phi^{-1}(y)}{K}$ then either $w' \in K_{J'}$ or $w' \in J'.$ Since $w' \in K_{J'}$ if and only if
$ \phi(w') \in \kcface{x}{K^\sphi} \cap \kface{y}{K^\sphi}$ and $\phi$ is bijection, we are good for this case.
If $w '\in J'$ then $w' \in \kcface{\sphi(x)}{J'}$ and \ref{cond5red} implies that the map 
$$ \kcface{x}{J} \cap \sphi^{-1}(\kcface{\sphi(x)}{K} ) \ni w \longmapsto \sphi(w) \in \kcface{\sphi(x)}{J'} $$
is a bijection, and this concludes the case  $\rk_K(\sphi(x)) = \rk_{K^\sphi}(x).$ 

In the second case where $\rk_K(\sphi(x)) = \rk_{K^\sphi}(x) + 1$ condition \ref{cond4red} implies that
$$\abs{ \kcface{x}{J} \cap \sphi^{-1}(\sphi(x))} = 2.$$
Moreover, if $w \in \kcface{x}{K^\sphi} \cap \kface{y}{K^\sphi}$ then $\sphi(w) \subset \phi^{-1}(y)$ and 
$$\rk_J(x) + 1 = \rk_J(w) \leq \rk_K(\sphi(w)) < \rk_K(y) = \rk_J(x) + 2.$$
Therefore $\rk_K(\sphi(w)) = \rk_K(\sphi(x))$ and $\sphi(w) = \sphi(x).$ Thus $w \in \kcface{x}{J} \cap \sphi^{-1}(\sphi(x))$ and this also proves that $$\abs{\kcface{x}{K^\sphi} \cap \kface{y}{K^\sphi}}=2,$$
which concludes the proof of \ref{cccdiamond} and shows that $K^\sphi$ is a cc

The cc $K^\sphi$ is graph-based if $K^\sphi_J$ is graph-based, which is the case since $e \in (K^\sphi_J)^{[1]}$ if and only if $e' := \phi^{-1}(e) \in \ko_{J'}$ and therefore 
$$\abs{e} = \abs{e' \setminus (J')^{[0]} \sqcup \sphi^{-1}(J^{[0]} \cap e)} = 2$$
by condition \ref{cond1red}. 

The previous argument also implies that  $K^\sphi_J$ is cell-connected. Since $K$ is connected and cell-connected, this shows that $K^\sphi$ is connected and cell-connected.

One has that $K^\sphi$ is clearly pure and $R := \Rk(K^\sphi) = \Rk(K).$ And $K^\sphi$ is non-branching if $\abs{\cface{y'}} = 2$ for all $y' \in (K^\sphi_J)^{[R-1]}$ which is a direct consequence of the poset isomorphism property of $\phi.$

A $K^\sphi$-pinch cannot occur in $(J')^c$ nor $K^\sphi_J$ since the poset isomorphism property of $\phi$ would produce a $K$-pinch. If $x \in J$ then
$$ \dg{K^\sphi} \cap \cdual{x}{K^\sphi} = \phi(\cdual{\sphi(x)}{K}) \sqcup \{ \{\phi(z),\phi(z')\} \subset \kR ~|~ \sphi(x) \subset z \cap z' \in K^{[R-1]}\} \cong \dg{K} \cap \cdual{\sphi(x)}{K}.$$
Therefore there also cannot be a $K^\sphi$-pinch in $J$ since this would contradict that $K$ has no $K$-pinch on $J'.$
This in addition shows that $K^\sphi \in \nsc.$

By Remark \ref{remrksredncol} we have $\Rk(J) = \Rk(J') = R - 1$ and every cell in $J$ is included in a cell $y \in J^{[R-1]}.$ By Remark \ref{remcond3redncol} such a $y \in J^{[R-1]}$ satisfies that $\sphi(y) \in (J')^{[R-1]}.$ Therefore we have
$$\abs{\kcface{y}{K^\sphi}} = \abs{\kcface{\sphi(y)}{K}} = 1,$$
hence $J$ is indeed a connected component of $\partial K^\sphi.$

We have that $\sphi_J^{K^\sphi} = \sphi$ by the following arguments. 
If $v'$ is the vertex in $e \cap J'$ for $e \in  \E'_{J'}$ then by \ref{cond1red} we have that $\sphi^{-1}(\{v'\})$ contains exactly one vertex $v$ from $J.$ We will therefore consider the vertices $v'$ and $v$ as identical to simplify the discussion. With this convention, a cell $y \in J'(K)$ can be seen as having its vertices in $J.$ Moreover, if $y' \in K_{J'}^-[y]$ then $\phi(y') \in (K^\sphi)_J^-(x)$ for all $x \in J$ such that $\sphi(x) \subset y.$ Since we also assume that $J' = J'(K)$ and $\sphi_{J'}^K = \id_{J'}$ we have that $y = y_{J'}$ for all $y \in J'.$ Therefore, if we take $x \in J$ and $y' \in K_{J'}^-[\sphi(x)_{J'}]$ then
$$ x_J = \phi(y')_J = y'_{J'} = \sphi(x)_{J'} = \sphi(x)$$
and this proves that $\sphi = \sphi_J^{K^\sphi}.$
This also directly implies that $J(K^\sphi) = J'(K)$ and that $(K^\sphi, J)$ is uniform.
\end{proof}
\end{prop}

\begin{thm}[Correspondence] \label{thmcorresp}
There is a one-to-one correspondence between slice sequences and slices, up to cc-isomorphism. More precisely, a sequence
$$J \red_{\sphi_J} J' \loc_{\cphi_J} M \col_{\cphi_L} L' \der_{\sphi_L} L$$
is a slice sequence if and only if the poset $(S, \rk_S)$ defined by
$$S :=  (\cphi_J \sqcup \cphi_L)(M) \sqcup J \sqcup L,$$
$$\rk_S :=  \rk_{(\cphi_J \sqcup \cphi_L)(M)} + \rk_J + \rk_L, $$ 
is an element in $\nsc_s$ such that
$$\partial S = J \sqcup L, \quad M = M^S, \quad J(S) = J', \quad L(S) = L',$$
$$ \phi_I = \phi_I^S ~~\text{ for all }~~ \phi \in \{\sphi,\cphi\}, ~~ I \in \{J,L\}.$$

\begin{figure}
\centering
\includegraphics[scale=0.47]{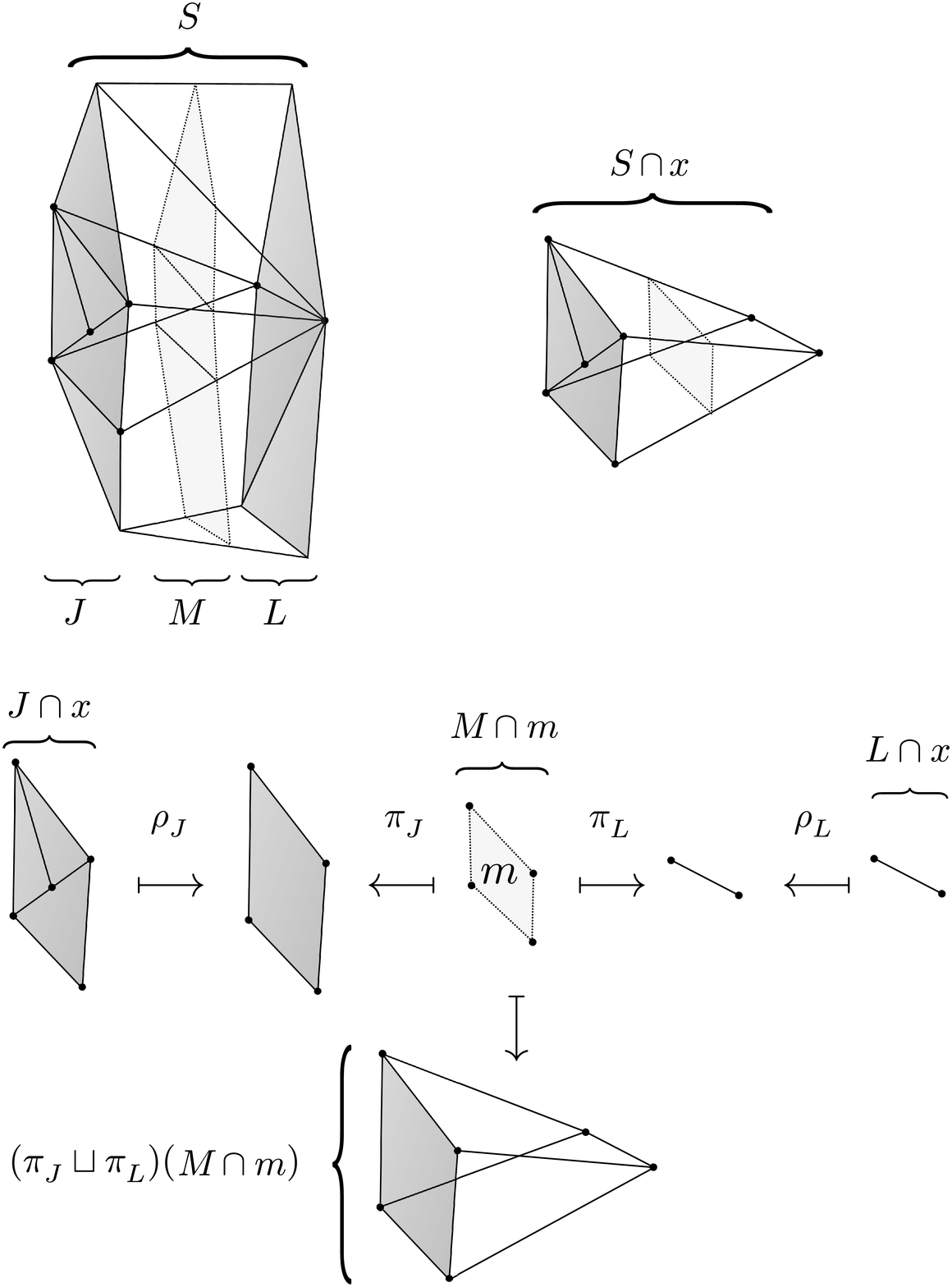}
\caption{\label{figslice}  On the top of this illustration we represented the portion of a slice $S$ of rank 3 with boundary components $J$ and $L$ and midsection $M = M^S.$ For the sake of clarity, we did not draw the 2-cells in the collar $S_J = S_L.$ We singled out one maximal cell $x$ of $S$ and represented the cells in $S \cap x.$ The lower part of the figure explains how, starting with a slice sequence as in Theorem \ref{thmcorresp}, the cells in $M \cap m$ for some $m \in M$ are used to construct the cells in $S \cap x$ using the reductions $\sphi_J, \sphi_L$ and the collapses $\cphi_J, \cphi_L.$ We also specified how the cells in $(\cphi_J \sqcup \cphi_L)(M \cap m)$ look like in this case.}
\end{figure}

\begin{proof}
We first show that a slice sequence is uniquely defined from a slice. This can be obtained directly as a consequence of Corollary \ref{cortransitmidsec} if we also show that if $S$ is a slice with boundary components $J$ and $L$ then 
$$ \cphi_J^S \perp  \cphi_L^S.$$
 
In order to prove that $\cphi_J^S$ and $\cphi_L^S$ are orthogonal, i.e. satisfy the point \ref{condorth} of Definition \ref{defrefnorth}, let us consider $m_1,m_2 \in M^S$ such that $\cphi_J^S(m_1) \cap \cphi_J^S(m_2) \neq \emptyset$ and $\cphi_L^S(m_1) \cap \cphi_L^S(m_2) \neq \emptyset.$  By Proposition \ref{propslice} we have that $m_i = \E_J^{y_i} = \E_L^{y_i}$ for some $y_i \in S_J = S_L,$ for $i=1,2,$ and the definition of $\cphi_J^S$ in Proposition \ref{proptransit} implies that $\cphi_J^S(m_i) = (y_i)_J$ for $i=1,2.$ By assumption we therefore have that $(y_1)_J \cap (y_2)_J \neq \emptyset$ and $(y_1)_L \cap (y_2)_L \neq \emptyset,$ hence $y_1 \cap y_2 \neq \emptyset.$ Lemma \ref{lemmidsec} implies $\E_J^{y_1 \cap y_2} = \E_J^{y_1} \cap \E_J^{y_2}$ hence we get the following inclusion for $I=J,L$ :
\begin{align*}
(y_1 \cap y_2)_I &= \{ e \cap J^{[0]} ~|~ e \in \E_J^{y_1 \cap y_2}\} = \{ e \cap J^{[0]} ~|~ e \in \E_J^{y_1} \cap \E_J^{y_2}\}\\
&\subset \{ e \cap J^{[0]} ~|~ e \in \E_J^{y_1}\} \cap \{ e \cap J^{[0]} ~|~ e \in  \E_J^{y_2}\} = (y_1)_I \cap (y_2)_I.
\end{align*}
To prove the reverse inclusion, and conclude that $\cphi_J^S$ and $\cphi_L^S$ are orthogonal, suppose by contradiction that there exists a vertex $v$ in $(y_1)_I \cap (y_2)_I \setminus	(y_1 \cap y_2)_I$ for $I = J$ or $I=L.$ If $\rk(y_1 \cap y_2) = 1$ we directly obtain a contradiction as in this case $v \in y_1 \cap  y_2 = (y_1 \cap y_2)_I.$ If $\rk(y_1 \cap y_2) \geq 2,$ we have in particular that $v$ is by assumption contained in an edge $e \subset \E_I^{y_1}$ and since $v \notin (y_1 \cap y_2)_I,$ there is no edge in $\E_I^{y_1 \cap y_2}$ containing $v.$ As a consequence, if $x$ is a minimal cell in $A(v) \cap S_I^{y_1 \cap y_2}$ then $\rk(x) \geq 2.$ But $x$ is then also a minimal cell in $S(v)$ and this contradicts the assumption that $(S,\partial S)$ is pure, since both $e$ and $x$ are minimal cells and $1 = \rk(e) \neq \rk(x).$

We then show how to obtain a slice from a slice sequence. By Proposition \ref{propredbdry}, it is sufficient to consider the case $J = J'$ and $L = L'$ since if in this case we obtain a slice $S$ (satisfying the desired conditions) the latter result using $\sphi = \sphi_I$ for  $I = J,L$ indeed leads $\sphi_I^S = \sphi_I$ and $I(S) = I'.$

Our goal is then to prove that $(S, \rk_S) \in \nsc_s$ where we recall that $(S, \rk_S)$ is defined by
$$S :=  (\cphi_J \sqcup \cphi_L)(M) \sqcup J \sqcup L,$$
$$\rk_S :=  \rk_{(\cphi_J \sqcup \cphi_L)(M)} + \rk_J + \rk_L, $$
and $(\cphi_J \sqcup \cphi_L)(M)$ is the augmented poset associated to $\cphi_J$ and $\cphi_L.$
As usual we start by proving that $S$ is a cc. By Lemma \ref{lemsqcup} the Axioms \ref{cccrank} and \ref{cccenough} are directly derived from the same axioms for the cc $J,L$ and $M.$ 

As for Axiom \ref{cccinter}, it is sufficient to show that if $m,m' \in M,$ then
$$x := (\cphi_J \sqcup \cphi_L)(m) \cap (\cphi_J \sqcup \cphi_L)(m') \in S \cup \{ \emptyset \}.$$
If either $(\cphi_J(m) \cap \cphi_J(m') = \emptyset$ or $(\cphi_L(m) \cap \cphi_L(m') = \emptyset$ then Axiom \ref{cccinter} for respectively $L$ and $J$ implies that $x \in L$ or $x \in J.$ If both $(\cphi_J(m) \cap \cphi_J(m')$ and $(\cphi_L(m) \cap \cphi_L(m')$ are non-empty, then as a consequence of the orthogonality of $\cphi_J$ and $\cphi_L$ we have that $m \cap m' \neq \emptyset$ and 
$$x = (\cphi_J \sqcup \cphi_L)( m \cap m') \in S.$$

And in order to show Axiom \ref{cccdiamond} for $S$ the only cases that are not directly obtained from Axiom \ref{cccdiamond} for $J,L$ or $M$ are inclusions of the form $w \subsetneq (\cphi_J \sqcup \cphi_L)(m),$ for some $m\in M$ and $w \in J \sqcup L$ such that 
$$\rk_S(w) = \rk_S((\cphi_J \sqcup \cphi_L)(m)) - 2 = \rk_M(m) - 1.$$ 
To deal with these cases we set $(\cphi_J \sqcup \cphi_L)(m):= x \sqcup y$ and assume without loss of generality that $w \in J,$ hence $w \subset x = \cphi_J(m).$ As noted in Remark \ref{remcond3redncol}, condition \ref{cond3col} implies that $\rk_J(\cphi_J(m)) \leq \rk_M(m).$ Therefore we have
$$ \rk_S(w) = \rk_J(w) \leq \rk_J(x)  \leq \rk_M(m)$$
and this leaves us with two possibilities: either $\rk_J(x) = \rk_M(m) - 1$ or $\rk_J(x) = \rk_M(m).$ In the first case we have $w = x,$ so that \ref{cond4col} implies
$$ \abs{\face{m} \cap (\cphi_J)^{-1}(x)} = 2.$$
By Lemma \ref{lemsqcup} an element $m' \in M$ belongs to $\face{m} \cap (\cphi_J)^{-1}(x)$ if and only if it satisfies that $\rk_S(\cphi_J(m') \sqcup \cphi_L(m')) = \rk_S(x) + 1$  and $\cphi_J(m') = x.$ Moreover any such $m'$ satisfies $\cphi_L(m') \subset y$ hence $m' \in \face{m} \cap (\cphi_J)^{-1}(x)$ if and only if
$$ \cphi_J(m') \sqcup \cphi_L(m') \in \cface{x} \cap \face{ x \sqcup y}.$$
Therefore we indeed have
$$ \abs{\face{m} \cap (\cphi_J)^{-1}(x)}  = \abs{ \cface{x} \cap \face{ x \sqcup y}}$$
which concludes the case when $x = w.$

The second case is when $w \subsetneq x = \cphi_J(m)$ and therefore $\rk_J(\cphi_J(m)) = \rk_M(m).$ In this case \ref{cond5col} implies
$$ \abs{\face{m} \cap (\cphi_J)^{-1}(x')} = 1 \quad \forall x' \in \face{x}.$$
Since $w \in \face{x}$ we have $\face{m} \cap (\cphi_J)^{-1}(w) = \{w'\}$ and therefore
$$ \cface{w} \cap \face{ x \sqcup y} = \{ x, \cphi_J(w') \sqcup \cphi_L(w')\}.$$
This concludes the proof that $S$ is a cc.

$S$ is graph-based since every edge in $(\cphi_J \sqcup \cphi_L)(M)$ is of the form $\cphi_J(v) \sqcup \cphi_L(v)$ for some vertex $v \in M^{[0]}$ and $\rk_J(\cphi_J(v)) = \rk_L(\cphi_L(v)) = 0$ by Remark \ref{remcond3redncol}.

$S$ is also clearly connected and cell-connected, as every cell $\cphi_J(m) \sqcup \cphi_L(m) \in (\cphi_J \sqcup \cphi_L)(M)$ satisfies that $J \cap \cphi_J(m)$ and $L \cap \cphi_L(m)$ are connected and can be linked by an edge of the form $\cphi_J(v) \sqcup \cphi_L(v)$ where $v \in m.$

By Remark \ref{remrksredncol} we have $\Rk(J) = \Rk(M) = \Rk(L) = R.$ It follows that $S$ is pure of rank $R+1$ since  the surjectivity of $\cphi_J$ and $\cphi_L$ implies that every cell in $J \sqcup L$ is contained in a cell of the form $\cphi_J(m) \sqcup \cphi_L(m)$ for some $m \in M^R$ and hence
$$ \rk_S(\cphi_J(m) \sqcup \cphi_L(m)) = \rk_M(m) + 1 = R + 1.$$

$S$ is non-branching as a consequence of Lemma \ref{lemsqcup} for the cells in $(\cphi_J \sqcup \cphi_L)(M)$ and since $J$ and $L$ are also non-branching. 

By Lemma \ref{lemsqcup} and the closeness of $M$ every $R$-cell in $S \setminus (J \sqcup L)$ is contained in two $(R+1)$-cells. Therefore every $R$ cell contained in one $(R+1)$-cell must be contained in either $J$ or $L$ and by the surjectivity of $\cphi_J$ and $\cphi_L$ it is the case for all maximal cells of $J$ and $L,$ hence $\partial S = J \sqcup L.$

$S$ is non-pinching since an $S$-pinch would imply an $M$-pinch by Lemma \ref{lemsqcup} and a $\partial S$-pinch would imply a $J$-pinch or an $L$-pinch by the previous derivation, all of which are excluded since $M,J,L$ are non-pinching.

As a consequence we indeed have that $S \in \nsc.$
It remains to show that $S$ is a slice. First, we have that $\kz = J^{[0]} \sqcup L^{[0]}$ since every cell in $(\cphi_J \sqcup \cphi_L)(M)$ has  rank higher than or equal to one. Clearly  $S_J = S_L = (\cphi_J \sqcup \cphi_L)(M)$ and by Lemma \ref{lemsqcup} this implies $M = M^S.$ $(S,\partial S)$ is also clearly non-degenerate and it is pure by the following argument. If, say, $y \in S_J$ then $y = \cphi_J(m) \cap \cphi_L(m)$ for some $m \in M$ and $J \cap y  = J \cap \cphi_J(m)$ is a pure cc. Hence if $y$ has minimal rank among elements in $S_J(x)$ then $\cphi_J(m) = x.$ By  condition \ref{cond3col} we can assume that $\rk_M(m) = \rk_J(x)$ therefore
$$ \rk_S(y)  = \rk_M(m) + 1 = \rk_J(x) + 1 = \rk_J^S(x) + 1,$$
and this shows that all elements in $S^-(x)$ have same rank.

Concerning uniformity, for all $\cphi_J(m) \sqcup \cphi_L(m) \in (\cphi_J \sqcup \cphi_L)(M) = S_{\partial S}$ we have that
$$ (\cphi_J(m) \sqcup \cphi_L(m))_J = \{ v \in \cphi_J(m) ~|~ \exists w \in m, ~ \cphi(w) = v \} = \cphi_J(m) = \left( \cphi_J(m) \sqcup \cphi_L(m)  \right) \cap J^{[0]}$$
and similarly for $L.$ Since we also assumed that $J(S) = J \in \mlc$ and $L(S) = L \in \mlc$ it follows that $(S,\partial S)$ is indeed uniform.

Finally, we have that $\cphi_I = \cphi_I^S$ for $I \in \{J,L\}$ by the following argument. It is sufficient to show this for the case $I=J$ in which for any $m \in M$ we have 
$$\E_J^x = \{ \cphi_J(v) \sqcup \cphi_L(v) ~|~ v \in m \},$$
where $x = \cphi_J(m) \sqcup \cphi_L(m).$ By identifying each vertex $v \in m$ with the edge $\cphi_J(v) \sqcup \cphi_L(v) \in (M^S)^{[0]}$ we obtain $\E_J^x = m$ and therefore
$$ \cphi_J^S(\E_J^x) = x_J = x \cap J^{[0]} =  \cphi_J(m),$$
where in the second equality we used the assumption $J = J(S).$
\end{proof}
\end{thm}


\subsection{Connecting sequences and union of cc} \label{sseccompcob}

In this section we start by introducing the notion dual to that of a slice sequence called a "connecting sequence" consisting of reductions and collapses from  and to elements $M,J,I,L,M' \in \mlc$ related in the following way:
$$M \col J \der I \red L \loc M'.$$
Such a sequence is designed to define the union of elements $K,H \in \nsc$ such that $I = K \cap H  \in \mlc$ is a connected component of the boundary of $K$ and $H.$ In this case  $M,M' \in \mlc$ are chosen as $M = M_I^K,$ $M' = M_I^H$ and the cc $J,L \in \mlc$ are the associated transitions i.e. $J = I(K)$ and $L = I(K).$


Connecting sequences are thus used in Definition \ref{defconnbdry} to give a precise meaning to the notion of cobordisms "connecting through" a given boundary component. It will then be natural to give the definition of the union of two cc in $\nsc$ sharing a boundary component and prove in Corollary \ref{corunioncc} that it defines an element in $\nsc$ as a consequence of the Correspondence Theorem shown in the previous section.

\begin{defi}[Connecting sequence] \label{defconnseq}
A sequence of mutually disjoint local cell complexes $M, J,$ $I,L,M'$ in $\mlc$ and maps
$$M \col_{\cphi_J} J \der_{\sphi_J} I \red_{\sphi_L} L \loc_{\cphi_L} M'$$
will be called a \textit{connecting sequence} if $\cphi_J, \sphi_J, \cphi_L, \sphi_L$ satisfy
$$ \sphi_J \cpa \cphi_J, \quad \sphi_L \cpa \cphi_L, \quad  \sphi_J \refl  \sphi_L.$$
\end{defi}

\begin{defi}[Connecting boundary components] \label{defconnbdry}
Let $K, H \in \nsc^R$ and $J,L \in \mlc^{R-1}$ be connected components of $\partial K$ and $ \partial H,$ respectively, such that $(K,J)$ and $(H,L)$ are uniform. Then we say that \textit{$K$ and $H$ connect through} $I \in \mlc^{R-1}$ if there exist cell complexes $K^{\sphi_J},H^{\sphi_L} \in \nsc^R$ and relative reductions $(I,K^{\sphi_J}) \red_{\sphi_J} (J,K)$ and $(I, H^{\sphi_L}) \red_{\sphi_L} (L,H)$ such that 
\begin{equation} \label{connectsequdef}
M_J^K \col_{\cphi^K_J} J(K) \der_{\sphi_J \circ \sphi_J^K} I \red_{\sphi_L \circ \sphi_L^H} L(H) \loc_{\cphi_L} M_L^H
\end{equation}
is a connecting sequence, where $\cphi^K_J, \sphi_J^K, \cphi_L^H$ and $\sphi_L^H$ are defined in Proposition \ref{proptransit}. In this case we say that $\sphi_J$ and $\phi_L$ \textit{are connecting $K$ and $H$ through $I.$}
\end{defi}

The next definition gives a candidate object for what should be the union of two cell complexes in $\nsc$ with "similar" boundary components. It will be shown in the following lemma that it indeed uniquely determines an element in $\nsc.$ 
For the purpose of defining the composition of cobordisms as morphisms in a category, it is sufficient to consider the case $J \cong I \cong L$ in the following definition. 

\begin{defi}[Union of cell complexes with connecting boundary components]\label{defunioncc}
Let $K, H \in \nsc^R$ and $J,L \in \mlc^{R-1}$ be boundary components of respectively $K$ and $H$ such that $(K,J)$ and $(H,L)$ are uniform. Let 
$$(K^{\sphi_J},I) \red_{\sphi_J} (K, J), \quad (H^{\sphi_L},I) \red_{\sphi_L} (H,L)$$
be two relative reductions connecting $K$ and $H$ through $I$ and such that  $K^{\sphi_J} \cap H^{\sphi_L} = I.$ We define the \textit{union of $K$ and $H$ by $\sphi_J$ and $\sphi_L$ } to be the poset
$$K \ccup{\sphi_J ~ \sphi_L} H := K^{\sphi_J} \cup H^{\sphi_L}$$
with rank function 
$$\rk_{\sphi_J, \sphi_L} = \rk_{K \ccup{\sphi_J ~ \sphi_L} H} := \rk_{K^{\sphi_J}} + \rk_{H^{\sphi_L}} - \rk_I.$$
In the particular case were $J = I = L,$ i.e. $\sphi_J = \sphi_L = \id_I,$ we will simply write
$$K \ccup{\sphi_J ~ \sphi_L} H = K \ccup{I} H.$$
\end{defi}

\begin{cor}\label{corunioncc}
$(K \ccup{ \sphi_J ~ \sphi_L} H, \rk_{\sphi_J, \sphi_L} )$ as in Definition \ref{defunioncc} is an element of $\nsc^R$ uniquely determined by $K,H$ and the reductions $\sphi_J : I \dans J,$ $\sphi_L : I \dans L$ up to cc-isomorphism. In particular, if $(K - J) , (H - L) \in \cob$ are uniform cobordisms such that $\partial H = J \sqcup L$ and $K \cap H = J$ then
$$(K \ccup{J} H - L) \in \cob.$$
\begin{proof}
As pointed out in Definition \ref{defikphi}, $K' := K^{\sphi_J}$ and $H' := H^{\sphi_L}$ are uniquely determined by $K,H, \sphi_J$ and $\sphi_L$, hence it remains to show that $K' \cup H' \in \nsc^R.$ Notice first that $\rk_{\sphi_J, \sphi_L}|_I = \rk_I.$

The central argument of this proof is the following. By our assumptions, and considering that we are using the same notation as in Definition \ref{defconnbdry}, we obtain the connecting sequence (\ref{connectsequdef}). As a consequence, and by Lemmas \ref{lemdualcompnorth} and \ref{lemdualredcol}, the following sequence
$$\dual{M_J^K} \red_{\sphi'_J} \dual{J(K)} \loc_{ \cphi'_J} \dual{I} \col_{\cphi'_L} \dual{L(H)} \der_{ \sphi'_L} \dual{M_L^H}$$
where 
$$\sphi'_J := \dual{\cphi_J^K}, \quad \cphi'_J := \dual{\sphi_J \circ \sphi_J^K}, \quad \cphi'_L := \dual{\sphi_L \circ \sphi_L^H}, \quad \sphi'_L := \dual{\cphi_L^H},$$
is a slice sequence. Therefore Theorem \ref{thmcorresp} provides us with a slice $S$ of rank $R$ such that 
$$ S = \dual{M_J^K} \sqcup (\cphi_J' \sqcup \cphi'_L)(\dual{I}) \sqcup \dual{M_L^H}, \quad  \partial S = \dual{M_J^K} \sqcup \dual{M_L^H}, \quad M^S = \dual{I}.$$
Lemmas \ref{lemmidsec} and \ref{lemsqcup} imply that there is an isomorphism of posets
$$ S \cong \dual{K_I^{\sphi_J}} \sqcup \dual{I} \sqcup \dual{H_I^{\sphi_L}}.$$
Hence by Lemma \ref{claim2propdual} the duality map $ x \longmapsto \dual{x}$ defines a poset anti-isomorphism between $S$ and $K_I^{\sphi_J} \sqcup I \sqcup H_I^{\sphi_L},$ i.e. a bijection between the two posets which reverses inclusion relations. Moreover, for any element $x \in (\cphi_J' \sqcup \cphi'_L)(\dual{I})$ we have
$$ \rk_S(x) = \rk_{\dual{I}}(x) + 1 = (R - 1) - \rk_I(\dual{x}) = R - \rk_{\sphi_J, \sphi_L}(\dual{x}),$$
therefore $\rk_{\sphi_J, \sphi_L}(\dual{x}) = R - \rk_S(x)$ for all $x \in S.$

By the same arguments as in the proof of Proposition \ref{propdual} this implies that the axioms of cc as well as all properties defining elements in $\nsc$ (i.e. graph-based, pure, non-branching, non-pinching, connected, cell-connected) are satisfied for the cells in $K_I^{\sphi_J} \sqcup I \sqcup H_I^{\sphi_L}$ as the same properties are true for the cells in $S.$
\end{proof}
\end{cor}

\begin{rema}
If $\alpha : (K,J) \dans (K_\alpha,J)$ and $ \beta : (H,L) \dans (H_\beta,L)$ are relative cc-automorphisms then $K_\alpha$ and $H_\beta$ also connect through $I$ but in general $ K \ccup{ \sphi_J ~ \sphi_L} H  \not \cong K_\alpha \ccup{\alpha \circ \sphi_J ~ \beta \circ \sphi_L} H_\beta.$
\end{rema}

\subsection{The category of causal cobordisms}  \label{sseccatcausalcob}



We saw in the previous section that the notion of connecting sequence incorporates the conditions determining when the composition (i.e. the union) of two cobordisms is defined. It is therefore natural to use connecting sequences as the set of objects in order to define a category having morphisms defined as a certain class of cobordisms. Let $J,L \in \mlc^R$ for some $R \geq 0$ and define $\Hom_{\mlc^R}(J,L)$ to be the set of cc-homomorphisms $\phi:J \dans L.$ The duality map then induces a bijection from $\Hom_{\mlc^R}(J, L)$ to $\Hom_{\mlc^R}(\dual{J}, \dual{L})$ defined by $\dual{\phi}(x) := \dual{\phi(x)}, ~ x \in J.$ Therefore, as a consequence of Lemma \ref{lemdualredcol}, the duality map defines a bijection between connecting sequences and slice sequences. In order to introduce a category on which the duality map acts as a functor it is therefore also natural to include slice sequences as objects of the category.

The category of causal cobordisms we obtain as a result will be defined using the so-called braket notation commonly used in quantum mechanics. 
In our context kets and bras are a convenient way to represent connecting and slice sequences and use them to define larger sequences. 


In order to define the braket notation it will be convenient to introduce a lighter notation for connecting and slice sequences as follows.

Let us use the notation $M \col \der J$  for $M,J \in \mlc^d$ to denote an object we call a \textit{semi-sequence} of \textit{dimension $d$} defined as a shorthand for the notation
$$M \col_\cphi J' \der_\sphi J,$$
where $J' \in \mlc^{d}$ is the \textit{transition of $M \col \der J$} and  $\cphi \apc \sphi$ are some given data of the semi-sequence.

We define the composition of two semi-sequences of a given dimension using the orthogonality condition as follows. If $M_0 \col \der J$ and $M_1 \col \der J$ are semi-sequences then their \textit{composition} is defined to be the expression
$$M_0 \col \der J  \red \loc M_1,$$
if this expression defines a connecting sequence, i.e.
$$M_0 \col \der J = M_0 \col_{\cphi_0} J_0 \der_{\sphi_0} J, \quad M_1 \col \der J = M_1 \col_{\cphi_1} J_1 \der_{\sphi_1} J \quad \text{and} 
\quad \sphi_0 \refl \sphi_1.$$
Similarly, the composition of two semi-sequences $M \col \der J$ and $M \col \der L$  is defined as
$$J \red \loc M \col \der L,$$
if the latter expression defines a slice sequence.
Unless stated otherwise, we set the convention that if $(K, J)$ is a uniform relative cc then the semi-sequence $M_J^K \col \der J$ stands for
$$ M_J^K \col_{\cphi_J^K} J(K) \der_{\sphi_J^K} J.$$

This notation will now be used to introduce the braket notation and give the definitions of states and sequences of states.

\begin{defi}[Braket notation, sets of states $\sts$ and sequences of states $\seqs$] \label{defbraket}
We introduce the following braket notation, where each \textit{ket} $\ket{J_1, J_2, J_3}$ and \textit{bra} $\bra{J_1, J_2, J_3}$ is \text{labelled} by three elements $J_i \in \mlc, ~i= 1,2,3,$ satisfying the following conditions:
\begin{itemize}
\item $\ket{M, L , M'}$ denotes the \textit{ingoing state} corresponding to the connecting sequence $$M \col \der L \red \loc M',$$
\item $\bra{J, M, L}$ denotes the \textit{outgoing state} corresponding to the slice sequence 
$$J \red \loc M \col \der L,$$
\item $ \braket{J', M, J}{ N , L , N'}$ indicates that $M \col \der J$ is equal to $N \col \der L,$
\item $\ktbr{N', J ,N}{L , M , L'}$ indicates that $J \red \loc N$ is equal to $L \red \loc M.$
\end{itemize}
We define the \textit{set of states} $\sts$ to be the union of the \textit{set of ingoing states} $\ists$ and the \textit{set of outgoing states} $\osts.$ 
We define a \textit{sequence of states} to be an expression $s_0 \dots s_n$ where $s_i \in \sts, ~i = 0, \dots, n$ where $n \in \N$ is the \textit{length} of the sequence and such that $s_0 \dots s_n$ constitutes an expression composed of alternating ingoing and outgoing states in accordance with the above conditions. We denote by $\seqs$ the \textit{set of sequences of states} where for each state $s$ we denote by $\id_s := s$ the sequence of states of length 1 containing $s$ called the \textit{identity sequence on $s.$} If $\sigma := s_0 \dots s_n \in \seqs$ we define the \textit{image of $\sigma$} by $\im(\sigma) := s_0$  and the \textit{domain of $\sigma$} by $\Dom(\sigma) := s_n.$ 

Two sequences of the form $\sigma = s_0 \dots s_n, \gamma = s_n \dots s_m,$ where $n,m \in \N, ~n \leq m$ can be \textit{composed} to form the sequence
$$\sigma \circ \gamma := s_0 \dots s_m.$$
By Remark \ref{remrksredncol} all the labels of states belonging to a given sequence of states are elements in $\mlc$ of the same rank.
We define the \textit{set of $d$-dimensional states} $\sts_d$ to be the set of states labelled by elements in $\mlc^d.$ Similarly, the \textit{set of $d$-dimensional sequences} $\seqs_d$ is defined as the set of sequences of states with labels in $\mlc^d.$
\end{defi}


This provides us with all the notions required to define the category of causal cobordisms. 

\begin{defi}[Category $\Cob_d$ of causal cobordism of rank $d$]
Let $d \in \N^*.$ The \textit{category of $d$-dimensional causal cobordisms} $\Cob_d$ is defined by having the set of states $\sts_{d-1}$ as its set of objects and set of morphisms from $a$ to $b$ defined by
$$\homc(a,b) = \{ \sigma \in \seqs ~|~ a = \Dom(\sigma), ~ b = \im(\sigma)\}.$$ 
The composition of morphisms in $\Cob_d$ is defined via the composition of sequences and the identity on $s \in \sts_{d-1}$ is given by $\id_s \in \homc(s,s).$
\end{defi}

Since it is clear that the latter defines a category, it remains to justify how it is related to cobordisms.
Unless specified otherwise, we will conventionally assume that the elements in $\mlc$ used to label any two given states in a sequence of states are mutually disjoint cc unless the two states are consecutive states in the sequence of states, in which case they only share one semi-sequence.

By the Correspondence Theorem \ref{thmcorresp}, we have a correspondence:
$$\bra{ J, M, L} \quad \text{is associated to} \quad (S - L ) \in \cob_s, \quad\text{where}\quad \partial S = J \sqcup L,~M^S \cong M,$$ 
and where the slice $S$ corresponds to the slice sequence
$$J \red \loc M \col \der L.$$
In this case we consider the two following bras as identical:
$$ \bra{ J, M, L}  = \bra{S - L}.$$
Then for example if $\bra{S - L}, \bra{S' - L'}$ are outgoing states we call $\ket{ M, L, M'}$ an \textit{interpolating ingoing state} if 
$$ \braket{S - L}{ M, L, M'} \bra{S' - L'} \in \seqs.$$
The latter requirement is equivalent to stating that $\partial S' = L \sqcup L'$ and  $M \col \der L \red \loc M'$ is the connecting sequence
$$M^S \col_{\cphi_L^S} L(S) \der_{\sphi_L^S} L \red_{\sphi_L^{S'}} L(S') \loc_{\cphi_L^{S'}} M^{S'}.$$ 
Corollary \ref{corunioncc} implies that one can define the union $ (S \ccup{L} S' - L') \in \cob.$ We can therefore define the \textit{composition of $\bra{S - L}$ with $\bra{S' -L'}$} via the ingoing state $\ket{M, L, M'}$ by
$$  \bra{ S \ccup{L} S' - L'} := \braket{S - L}{ M, L, M'} \bra{S' - L'},$$
which in particular implies that $M = M^S$ and $M' = M^{S'}.$
More generally, we recursively define outgoing states of the form $\bra{K - J}$ where $K$ is a union of slices $S_0, \dots, S_n$ using the relation:
$$\bra{ S_0 \ccup{L_i}_{i = 1, \dots, n} S_i - L_n} := \braket{ S_0 \ccup{L_i}_{i = 1, \dots, n-1} S_i - L_{n-1}}{M_{n-1}, L_{n-1} , M_{n}} \bra{ S_n - L_n },$$
where $\partial S_i = L_{i-1} \sqcup L_i$ and $M_i = M^{S_i}$ for $i= 1, \dots, n.$

At this point it is natural to introduce the following notion of causal cobordism.

\begin{defi}[Causal cobordisms $\cob_c$] \label{defcauscob}
A cobordism $(K - J) \in \cob$ is said to be \textit{causal} if there exist slices $S_0, \dots, S_n$ and interpolating ingoing states $\ket{M^i, L_i, M^{i+1}}$ for $i = 1, \dots, n$ such that
$$ (K - J) = ( S_0 \ccup{L_i}_{i = 1, \dots, n}  S_i - L_n).$$
The set of causal cobordisms is denoted by $\cob_c.$ We will also call a bra of the form $\bra{K - J}$ where $(K - J) \in \cob_c$ a causal cobordism.
For each $J \in \mlc$ we conventionally define an element called the \textit{empty causal cobordism at $J$} denoted $(J - J)$ to be a causal cobordism with boundary $J$ and rank $\Rk(J) + 1.$ 
\end{defi}

Empty cobordisms are simply introduced to make some expressions of sequences of states more general. We will in particular use the bra $\bra{J - J}$ in an expression of the form:
$$ \ket{M, J, M'} \braket{J - J}{N, L, N'}  \quad \text{implying that} \quad \ket{M, J, M'} = \ket{N, L, N'}.$$

It will also be convenient to introduce a notation indicating that a slice is either the first or the last slice of a causal cobordism as follows. Let us consider a causal cobordism $(K - J) \in \cob_c$ given by
$$(K - J) = ( S_0 \ccup{L_i}_{i = 1, \dots, n} S_i - L_n),$$
for some slices $S_0, \dots, S_n$ and ingoing states $\ket{M^i, L_i, M^{i+1}}, ~i = 1, \dots, n$ such that $J = L_n.$ Writing $L := L_0$ we have that $\partial K = J \sqcup L.$
We say that a slice $\bra{S - J} = \bra{J',  M^S , J}$  is the \textit{ ingoing slice of $\bra{K - J}$} and use the notation 
$$\bra{J', M^S, J } \ins \bra{K - J},$$ if $(S - J) = (S_n - L_n).$
Similarly, we say that a slice $\bra{S - L'}= \bra{L , M^S , L'}$ is the \textit{outgoing slice of $\bra{K - J}$} and use the notation
$$\bra{L , M^S, L'} \outs \bra{K - J},$$
if $S = S_0.$

Using this notation, we can more explicitly express the sets of morphisms between states as follows. The expressions associated to morphisms in $\Cob_d$ are of four different kinds depending on whether the image and domain are ingoing or outgoing states:
\begin{align*}
&\homc( \ket{M, J, M'} , \ket{N', L, N}) := \left\lbrace ~ \ket{N', L, N} \braket{K - J}{M, J, M'} ~:~ \bra{K - J} \in \cob_c  \right\rbrace,\\[1em]
&\homc( \bra{J', M, J} , \ket{N', L, N}) := \left\lbrace ~ \ktbr{N', L, N}{K - J} ~:~ \bra{K - J} \in \cob_c,~ \bra{J',M,J} \ins \bra{K - J} \right\rbrace,\\[1em]
&\homc( \ket{M, J, M'} , \bra{ L', N, L}) := \left\lbrace ~ \braket{K - J}{M, J, M'} ~:~ \bra{K - J} \in \cob_c, ~ \bra{L', N, L} \outs \bra{K - J} \right\rbrace, \\[1em]
&\homc( \bra{J', M, J} , \bra{ L', N, L}) := \left\lbrace ~ \bra{K - J} \in \cob_c ~:~  \begin{array}{c} \bra{J',M,J} \ins \bra{K - J},\\ 
\bra{L', N, L} \outs \bra{K - J} \end{array} \right\rbrace.
\end{align*}


It is seen that only the morphisms in the last set $\homc( \bra{J', M, J} , \bra{ L', N, L})$ correspond exactly to the notion of cobordisms we defined in Definition \ref{defcauscob}. The other sets of morphisms include additional ingoing states corresponding to connecting sequences associated with one or both of their boundary components. The morphisms of the first set $\homc( \ket{M, J, M'} , \ket{N', L, N})$ actually define the \textit{ sub-category of $d$-dimensional ingoing states} $\Cob_d^{in}$ with $\ists_d$ as the set of objects. Similarly, one can single out the \textit{sub-category of $d$-dimensional outgoing states} $\Cob_d^{out}$ having $\osts_d$ as objects and morphisms defined by the last set $\homc( \bra{J', M, J} , \bra{ L', N, L}).$

We can now turn to the definition of the duality map for elements in $\Cob^d.$ For this, it is sufficient to specify how the duality map acts on sequences of states in $\seqs_d$ as done in the following definition.

\begin{defi}[Duality map on $\sts$ and $\seqs$]
The dual of a sequence of states is defined via the following rules:
\begin{itemize}
\item if $s = \ket{M, L, M'}$ then $\dual{s} = \bra{\dual{M'}, \dual{L}, \dual{M}}$ is the outgoing state characterized by the slice sequence
$$\dual{M'} \red \loc \dual{L} \col \der \dual{M},$$
that we call the \textit{dual sequence} of $M \col \der L \red \loc M';$
\item if $s = \bra{J, M , L}$ then $\dual{s} = \ket{\dual{L}, \dual{M}, \dual{J}}$ is the ingoing state characterized by the connecting sequence
$$ \dual{L} \col \der \dual{M} \red \loc \dual{J},$$
\item if $ \sigma = \sigma_1 \sigma_2$ is a sequence of states then
$$ \dual{\sigma} = \dual{\sigma_2} ~~ \dual{\sigma_1}.$$
\end{itemize}
\end{defi}
The definition of the duality map on $\sts$ and $\seqs$ for example implies:
\begin{align}\label{exdualmapbrkt}
\dual{ \braket{ J, M , L}{M , L , M'} } &= \dual{ \ket{M, L, M'}} ~~ \dual{\bra{J, M ,L}} = \braket{\dual{M'}, \dual{L}, \dual{M}}{\dual{L}, \dual{M}, \dual{J}}.
\end{align}

As a direct consequence of the above we obtain the following result.

\begin{cor}[Duality map on $\Cob^d$]
Let $a,b \in \sts^d$ and let $\sigma$ be a morphism in $\Cob^d.$ We have the following equivalence:
$$\sigma \in \homc(a,b)  \quad \text{if and only if} \quad \dual{\sigma} \in \homc(\dual{b}, \dual{a}).$$
\end{cor}

In other words, $\Cob_d$ is self-dual as a category under the contravariant functor $\charge = \dual{( \cdot )}$ induced by the duality map.  It is also clear that $\charge$ can be restricted to a functor between the sub-categories $\Cob_d^{in}$ and $\Cob_d^{out}.$

Let us furthermore note the following contravariant functor $\tim = \rev{(\cdot)}$ in $\Cob_d.$ For any state $s \in \sts,$ let the \textit{reverse state} $\rev{s}$  be the state defined by the same labels in reversed order, for example
$$ \rev{\ket{M',L,M}} = \ket{M,L,M'}.$$
Then $s \in \ists$ if and only if $\rev{s} \in \ists$ and similarly for $\osts.$ We also define the \textit{reverse sequence of states} $\rev{\sigma}$ of a sequence $\sigma = \sigma_1 ~ \sigma_2$ by 
$$\rev{\sigma} = \rev{\sigma_2}~ \rev{\sigma_1},$$
which evidently implies that
$$ \sigma \in \homc(a,b) \quad \text{if and only if} \quad \rev{\sigma} \in \homc(b,a).$$

Finally, it is worth mentioning a third (covariant) functor $\parity = (\cdot)^*$ defined as follows. The functor $\parity$ acts on states as a map $s \longmapsto s^*,$ where $s^*$ is the state corresponding to the sequence dual to the sequence characterizing $s$ labelled by the dual of the labels of $s$ in the same order. For example,
$$ \ket{M', L, M}^* = \bra{ \dual{M'}, \dual{L}, \dual{M}}.$$
The functor $\parity$ acts on morphisms in $\Cob_d$ simply by acting on each states of the corresponding sequence of states, i.e. if $\sigma = \sigma_1 ~ \sigma_2$ then 
$$\sigma^* = \sigma_1^* ~ \sigma_2^*.$$
This in turn implies the equivalence:
$$ \sigma \in \homc(a,b) \quad \text{if and only if} \quad \sigma^* \in \homc(a^*, b^*).$$

These functors satisfy $\tim^2 = \charge^2 = \parity^2 = \id_{\Cob_d}$ and are in particular related by the relation $\charge =  \parity ~ \tim ,$ as can be seen for the specific example used in (\ref{exdualmapbrkt}):
\begin{align*}
\left( \rev{\left( \braket{ J, M, L}{M, L, M'} \right)} \right)^* &= \left( \ktbr{M',L,M}{M,L,J} \right)^*\\
&= \braket{ \dual{M'}, \dual{L}, \dual{M}}{\dual{L}, \dual{M}, \dual{J}} = \dual{ \braket{ J, M, L}{M, L, M'} }.
\end{align*}

It is then suggestive to view $\tim$ as the analogue of a "time reversal" transformation, $\parity$ as the analogue of a "parity transformation"  and the duality map $\charge$ as the analogue of a "charge transformation".

\newpage

\addcontentsline{toc}{section}{References}

\bibliographystyle{plain}
\bibliography{biblio}


\end{document}